\documentclass[10pt]{article}

\usepackage[margin = 1in]{geometry}  
\usepackage[english]{babel}  
\usepackage[final]{showkeys}  

\usepackage{graphicx}  
\graphicspath{{./images/}}
\usepackage{tikz}  
\usepackage[small, labelsep = endash, labelfont = bf, up, up]{caption}  
\usepackage{subcaption}
\usepackage{tabularray}  
\usepackage{titlesec}  
\usepackage{amsmath, amssymb, amsthm, mathtools}  
\usepackage{amsfonts}  
\usepackage{bm}  
\usepackage[mathscr]{euscript}  
\usepackage{xcolor}  
\definecolor{ao(english)}{rgb}{0.0, 0.5, 0.0}
\usepackage[bookmarks = false, colorlinks = true, allcolors = ao(english)]{hyperref}  
\usepackage{soul}  
\usepackage{comment}  

\newtheorem{theorem}{Theorem}[section]  
\newtheorem{corollary}[theorem]{Corollary}  
\newtheorem{proposition}[theorem]{Proposition}  
\newtheorem{lemma}[theorem]{Lemma}  
\newtheorem{definition}[theorem]{Definition}  
\newtheorem{remark}[theorem]{Remark}  
\newtheorem{example}[theorem]{Example}  

\newcommand{\N}{\mathbb{N}}  
\newcommand{\Z}{\mathbb{Z}}  
\newcommand{\R}{\mathbb{R}}  
\newcommand{\C}{\mathbb{C}}  
\newcommand{\smallerrelaux}[2]{\vcenter{\hbox{{\scalebox{.75}{$#1#2$}}}}}  
\newcommand{\smallerrel}[1]{\mathrel{\mathpalette\smallerrelaux{#1}}}
\newcommand{\smallin}{\smallerrel{\in}}

\titleformat{\section}[block]{\bf\filcenter}{\thesection.}{5pt}{}  
\titleformat{\subsection}[block]{\bf}{\thesubsection.}{5pt}{}  
\titleformat{\subsubsection}[block]{\bf}{\thesubsubsection.}{5pt}{}  
\makeatletter  
\g@addto@macro\bfseries{\boldmath}
\makeatother

\addto\captionsenglish{}  

\title{INSTABILITY OF QUADRATIC BAND DEGENERACIES \\ AND THE EMERGENCE OF DIRAC POINTS}
\author{J. CHABAN\thanks{Department of Applied Physics and Applied Mathematics, Columbia University} \hspace{5pt} AND \hspace{5pt} M. I. WEINSTEIN\thanks{Department of Applied Physics and Applied Mathematics and Department of Mathematics, Columbia University}}
\date{\today}

\makeatletter  
\def\@maketitle{
    \begin{center}
    {\large \bf \@title} \\
    \vspace{\baselineskip}
    \let\footnote\thanks
    {\@author} \\
    \vspace{1.5\baselineskip}
    {\@date}
    \end{center}}
\makeatother

\begin{document}

\maketitle

\begin{abstract}
\noindent Consider the Schr\"{o}dinger operator ${H = -\Delta + V}$, where the potential $V$ is real, $\Z^2$-periodic, and additionally invariant under the symmetry group of the square. We show that, under typical small linear deformations of $V$, the quadratic band degeneracy points occurring over the high-symmetry quasimomentum ${\bm M}$ (see \cite{keller2018spectral, keller2020erratum}) each split into two separated degeneracies over perturbed quasimomenta ${\bm D}^+$ and ${\bm D}^-$, and that these degeneracies are Dirac points. The local character of the degenerate dispersion surfaces about the emergent Dirac points are tilted, elliptical cones. Correspondingly, the dynamics of wavepackets spectrally localized near either ${\bm D}^+$ or ${\bm D}^-$ are governed by a system of Dirac equations with an advection term. Symmetry-breaking perturbations and induced band topology are also discussed.
\end{abstract}

\section{Introduction}
\label{sec:intro}

Wave propagation in energy-conserving, periodic media is determined by the {\it band structure} of the relevant self-adjoint Hamiltonian operator. Degenerate points within the band structure are energy-quasimomentum pairs at which two consecutive dispersion surfaces touch; novel wave dynamics behavior may arise from such degeneracies. For example, the corresponding Floquet-Bloch eigenstates may be multivalued with respect to variations in quasimomentum in a neighborhood of the degeneracy, contributing to a {\it Berry phase} in the dynamics of semiclassical wavepackets; see \cite{niu1996, carles-sparber2012, WLW2017}. Further, band structure degeneracies may seed topological behavior; opening a band gap via time-reversal symmetry-breaking perturbations can lead to nonzero {\it Chern numbers} associated with the emergent isolated bands. For a system arising by the insertion of an unbounded line defect, or {\it edge}, the spectral gap may be populated by energy eigenvalues whose corresponding eigenstates are {\it edge states} which propagate along the edge, but are localized transverse to it; see, e.g., \cite{HR08, FLW16, FLW16a, FLW17-mem, LWZ19, drouotweinstein2020, drouot2021}. These are quantified, in a sense, by a {\it spectral flow} equal to the difference of bulk Chern numbers via the bulk-edge correspondence; see, e.g., \cite{H93, HR08, drouot2021} and references cited therein.

Well-known examples of band structure degeneracies include {\it Dirac points}. These are conical degeneracies in the band structures of Hamiltonians corresponding to periodic media with honeycomb symmetry. Examples include the Schr\"{o}dinger equation \cite{fefferman2012honeycomb} (or its tight-binding approximation \cite{W1947, FLW17}) governing the material graphene, a hexagonal arrangement of carbon atoms in the plane, and  classical waves (governed, e.g., by Maxwell's equations \cite{LWZ19, CW21}). In such media, the band structure exhibits conical touchings of dispersion surfaces over high-symmetry quasimomenta ${\bm K}$ and ${\bm K}^\prime$, situated at vertices of the hexagonal Brillouin zone; see also \cite{grushin2009, berkolaiko2014symmetry}. The envelope of wavepackets spectrally localized about these Dirac points evolves according to an effective (homogenized) two-dimensional Dirac equation \cite{FW14}. Leveraging the presence of band structure degeneracies is a central idea in the exploration and application of naturally occurring and engineered materials; see, e.g., \cite{top-photonics19}.

In \cite{keller2018spectral, keller2020erratum}, the presence of quadratic touchings of dispersion surfaces, or {\it quadratic band degeneracies}, over the high-symmetry quasimomentum ${{\bm M} = [\pi, \, \pi]^\mathsf{T}}$, located at a vertex of the square Brillouin zone (see Figure \ref{fig:kmow-2}), was studied for continuum Schr\"{o}dinger operators with real, $\Z^2$-periodic potentials which are additionally invariant under the symmetries of the square (i.e., dihedral group of order 8); see Definition \ref{def:sql-pot}.

{\it In this paper, we explore periodic Schr\"odinger operators arising from a linear deformation of the square period cell. We prove that, under small, non-dilational linear deformations, a quadratic band degeneracy point over ${\bm M}$ splits into two separate degeneracies over perturbed quasimomenta ${\bm D}^+$ and ${\bm D}^-$ and that these emergent degeneracies are Dirac points. The local character of the degenerate dispersion surfaces over ${\bm D}^+$ and ${\bm D}^-$ are tilted, elliptical cones, corresponding to effective two-dimensional Dirac Hamiltonians with an advection term. Our main results are stated in Theorems \ref{thm:M-srf} and \ref{thm:M-dgn}, and in the discussion of Section \ref{sec:break-PC}.}

\begin{figure}[t]
\centering
\begin{subfigure}{0.33\textwidth}
    \subcaption{$\quad$}
    \vspace{0.1cm}
    \includegraphics[height = 4.3cm]{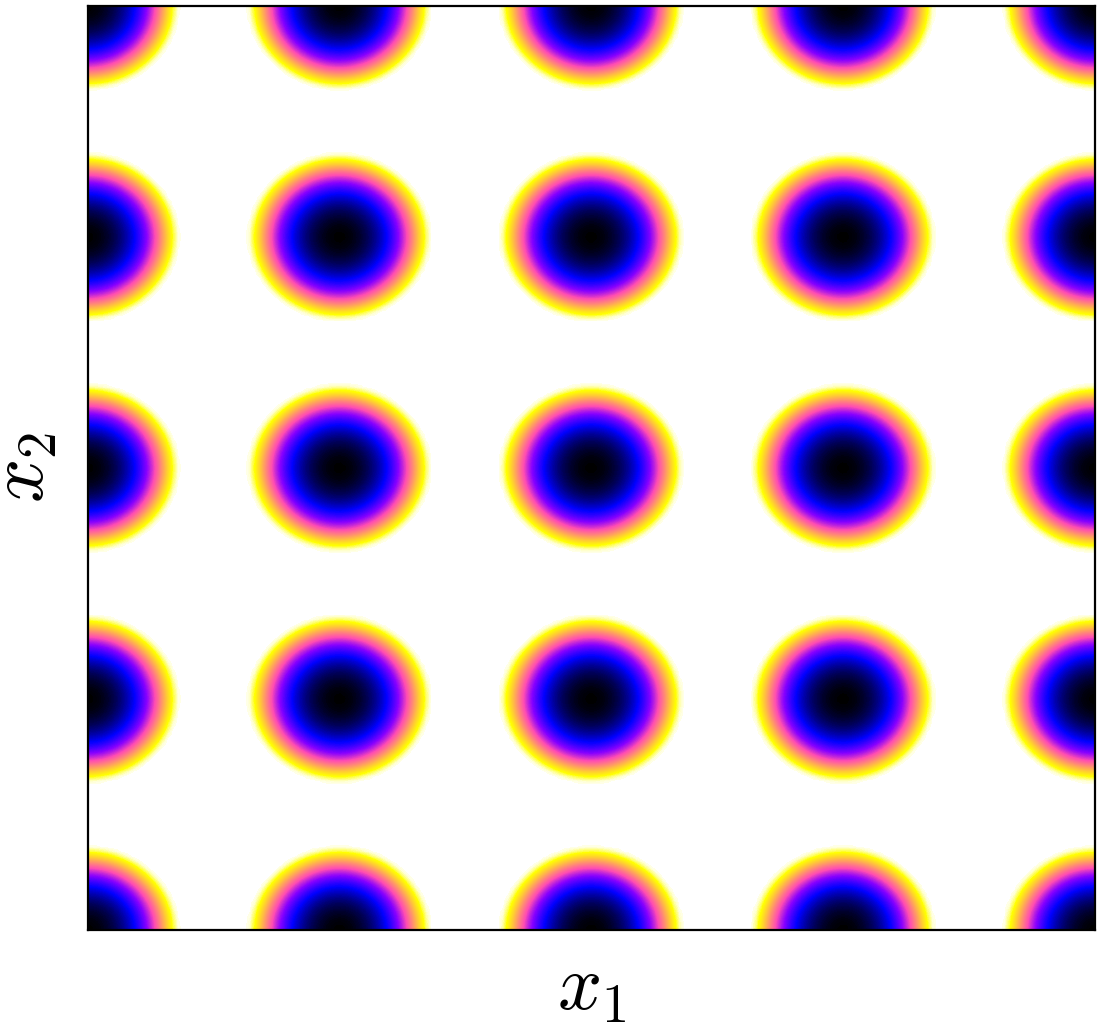}
    \label{fig:intro-1a}
\end{subfigure}
\hspace{0.2cm}
\begin{subfigure}{0.42\textwidth}
    \subcaption{$\qquad$}
    \vspace{0.1cm}
    \includegraphics[height = 4.3cm]{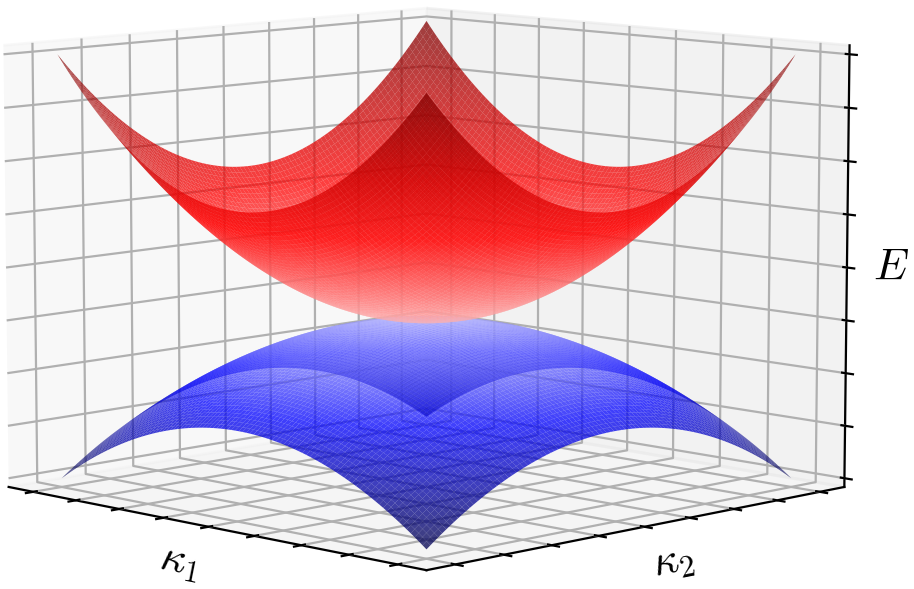}
    \label{fig:intro-1b}
\end{subfigure} \\
\vspace{0.2cm}
\begin{subfigure}{0.33\textwidth}
    \subcaption{$\quad$}
    \vspace{0.1cm}
    \includegraphics[height = 4.3cm]{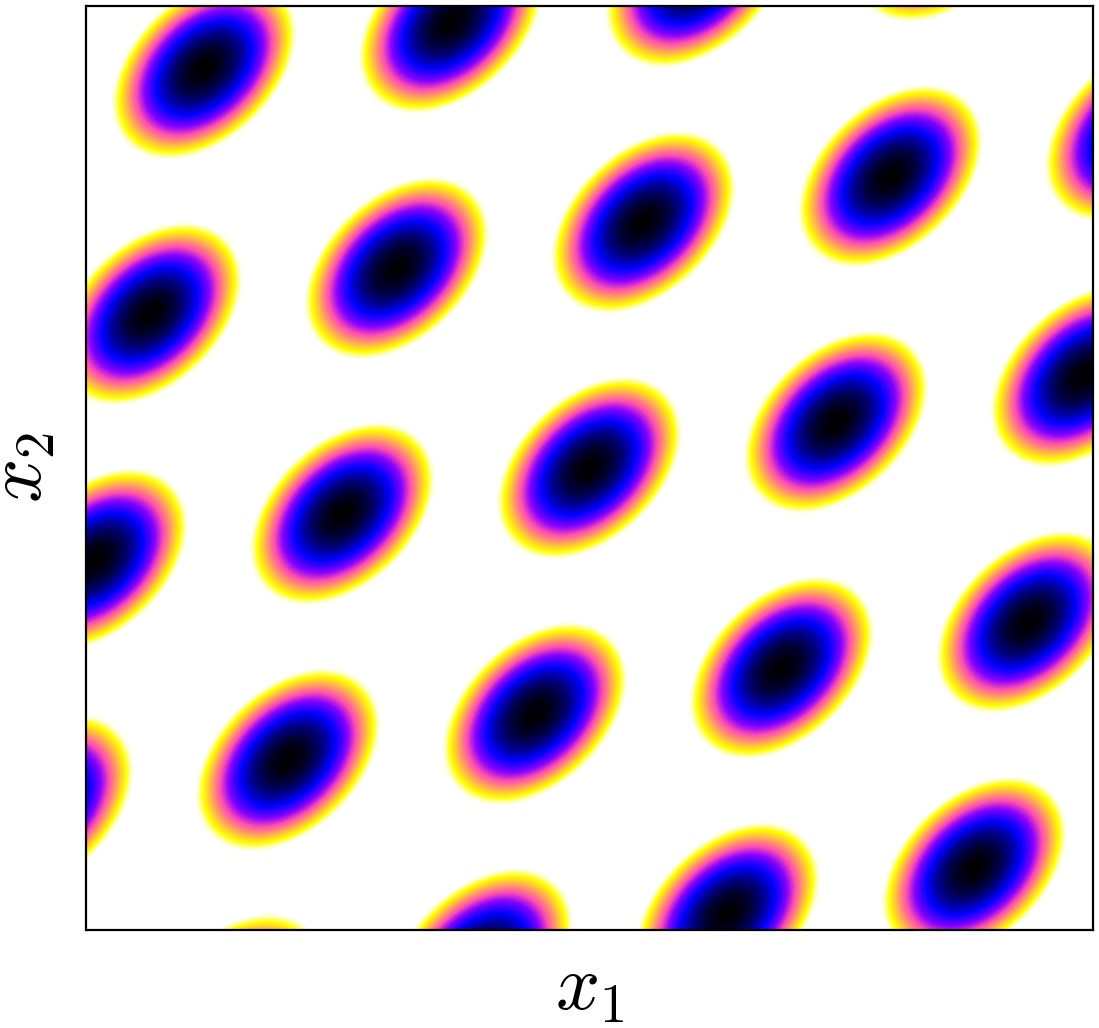}
    \label{fig:intro-1c}
\end{subfigure}
\hspace{0.2cm}
\begin{subfigure}{0.42\textwidth}
    \subcaption{$\qquad$}
    \vspace{0.1cm}
    \includegraphics[height = 4.3cm]{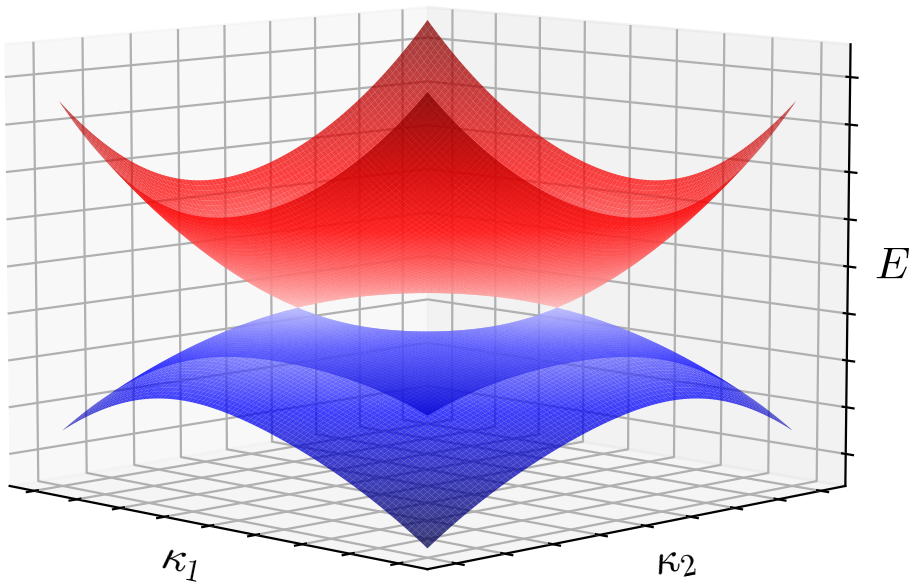}
    \label{fig:intro-1d}
\end{subfigure}
\vspace{0.3cm}
\caption{{\bf (a)} Square lattice potential constructed as the sum of translates of a compactly-supported potential. {\bf (b)} Band structure of the Schr\"{o}dinger operator ${H_V = -\Delta + V}$, where $V$ is a square lattice potential, near a quadratic band degeneracy point ${(E_S, \, {\bm M})}$. The dispersion surfaces are approximately described by the eigenvalue mappings of ${H^{\bm M}_{\rm eff}(\bm\kappa)}$ in \eqref{eq:kmow-Heff}, where ${{\bm \kappa} = {\bm k} - {\bm M}}$; see Theorem \ref{thm:kmow} and Remark \ref{rmk:kmow-eff}. {\bf (c)} Periodic potential obtained by a linear deformation $T$ of the square lattice potential in panel (\subref{fig:intro-1a}). {\bf (d)} Band structure of the ``pushforward'' deformed Schr\"odinger operator ${T_* H_V = -\nabla \cdot (T^\mathsf{T} T)^{-1} \nabla + V}$ near the former quadratic band degeneracy at ${(E_S, \, {\bm M})}$. The dispersion surfaces are approximately described by ${H^{\bm M}_{\rm eff}({\bm \kappa}; \tau_0, {\bm \tau})}$ in \eqref{eq:modelHeff}, where again ${{\bm \kappa} = {\bm k} - {\bm M}}$ and $\tau_0$, ${\bm \tau}$ are parameters, defined in \eqref{eq:def-tau}, deriving from $T$; see Theorem \ref{thm:M-srf} and Remark \ref{rmk:M-srf-eff}. When ${|{\bm \tau}| \neq 0}$, the dispersion surfaces of ${H^{\bm M}_{\rm eff}({\bm \kappa}; \tau_0, {\bm \tau})}$ contain a pair of tilted, elliptical Dirac points; see Theorem \ref{thm:M-dgn}, as well as Figure \ref{fig:intro-2}.}
\label{fig:intro-1}
\end{figure}

Differing local character of the band structure near its degenerate points gives rise to contrasting dynamics of wavepackets spectrally localized about these degeneracies. The evolution of such wavepackets is governed, on large but finite timescales, by an {\it effective (homogenized) ${2 \times 2}$ matrix Hamiltonian} whose two dispersion relations represent a magnification of the band structure near the degeneracy. For example, the envelope of a wavepacket spectrally supported near a quadratic band degeneracy point evolves according to a Schr\"{o}dinger-type effective Hamiltonian. In contrast, the envelope of a wavepacket spectrally localized about a conical degeneracy will, in general, be governed by a Dirac Hamiltonian with an advection term.

\subsection{Summary of results}
\label{sec:summary}

Let $V$ be a real-valued, $\Z^2$-periodic potential such that ${V(R^\mathsf{T} {\bm x}) = V({\bm x})}$, where $R$ is the $\pi/2$ rotation matrix; see \eqref{eq:def-rot}. Note that, since ${R^2 = -I}$, this implies that ${V(-{\bm x}) = V({\bm x})}$; see Remark \ref{rmk:min-sym}. Additionally, assume ${V(\sigma_1 {\bm x}) = V({\bm x})}$, where $\sigma_1$ is the first Pauli matrix, representing a reflection about the line ${x_1 = x_2}$; see \eqref{eq:def-pauli}. We call such  potentials {\it square lattice potentials}; see Definition \ref{def:sql-pot} and the example in Figure \ref{fig:intro-1}, panel (\subref{fig:intro-1a}). In \cite{keller2018spectral}, it was shown that the band structure of the Schr\"{o}dinger operator
\begin{equation}
H_V = -\Delta + V
\end{equation}
contains quadratic band degeneracy points: energy-quasimomentum pairs ${(E_S, \, {\bm M})}$ at which two consecutive dispersion surfaces touch quadratically. More specifically, the band structure of $H_V$ in the vicinity of ${(E_S, \, {\bm M})}$ is approximated by the two dispersion relations of an effective Hamiltonian $H^{\bm M}_{\rm eff}$ with Fourier symbol
\begin{align}
\label{eq:kmow-Heff}
H^{\bm M}_{\rm eff}({\bm \kappa}) & = (1 - \alpha_0) (\kappa_1^2 + \kappa_2^2) \, I - 2 \alpha_1 \kappa_1 \kappa_2 \, \sigma_1 - \alpha_2 (\kappa_1^2 - \kappa_2^2) \, \sigma_2 \\
& = (1 - \alpha_0) |{\bm \kappa}|^2 \, I - \alpha_1 ({\bm \kappa} \cdot \sigma_1 {\bm \kappa}) \, \sigma_1 - \alpha_2 ({\bm \kappa} \cdot \sigma_3 {\bm \kappa}) \, \sigma_2. \nonumber
\end{align}
Here, ${{\bm \kappa} = {\bm k} - {\bm M}}$ denotes the quasimomentum displacement from ${\bm M}$, and $\sigma_1$, $\sigma_2$, and $\sigma_3$ are the standard Pauli matrices; see Section \ref{sec:notation}. The real parameters $\alpha_0$, $\alpha_1$, and $\alpha_2$ are defined in terms of the two Floquet-Bloch eigenstates spanning the two-dimensional eigenspace associated with ${(E_S, \, {\bm M})}$. The two dispersion surfaces of $H^{\bm M}_{\rm eff}$ (i.e., graphs of the two eigenvalue mappings ${{\bm \kappa} \mapsto \varepsilon_\pm({\bm \kappa})}$ of $H^{\bm M}_{\rm eff}(\bm\kappa)$) are plotted in Figure \ref{fig:intro-1}, panel (\subref{fig:intro-1b}). Together, they provide a magnification (or blow-up) of a small neighborhood of the band structure of $H_V$ near ${(E_S, \, {\bm M})}$.

\begin{figure}[t]
\centering
\begin{subfigure}{0.42\textwidth}
    \subcaption{$\qquad$}
    \vspace{0.1cm}
    \includegraphics[height = 4.3cm]{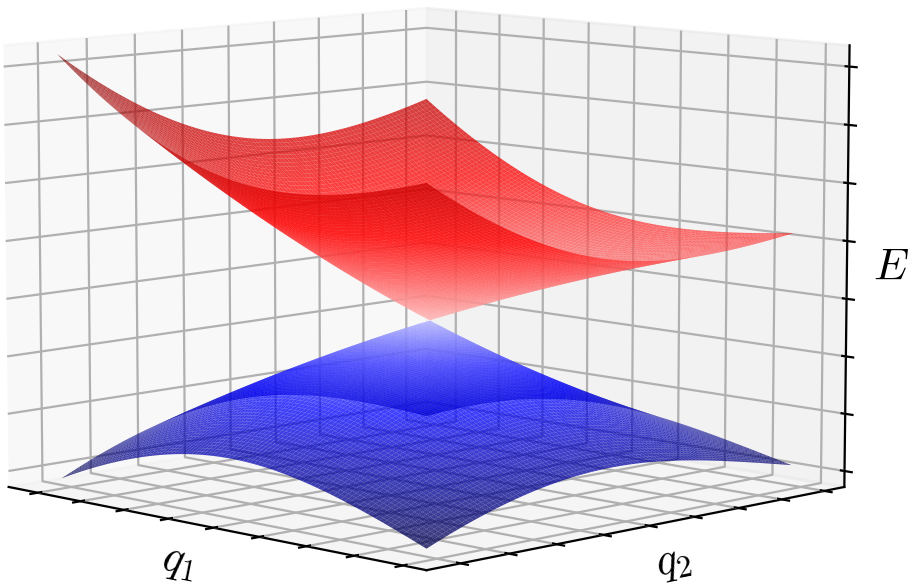}
    \label{fig:intro-2a}
\end{subfigure}
\hspace{0.4cm}
\begin{subfigure}{0.42\textwidth}
    \subcaption{$\qquad$}
    \vspace{0.1cm}
    \includegraphics[height = 4.3cm]{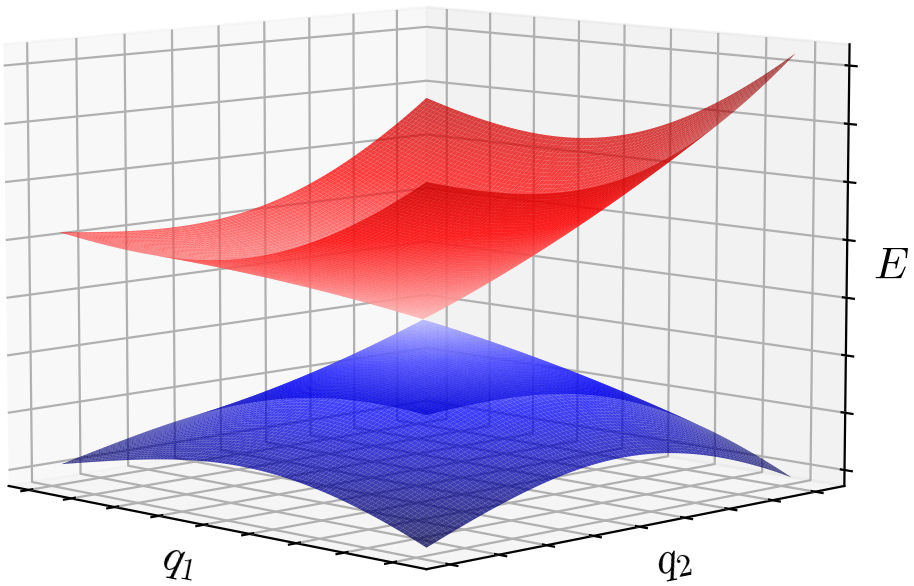}
    \label{fig:intro-2b}
\end{subfigure}
\vspace{0.3cm}
\caption{Local band structure of the ``pushforward'' deformed Schr\"{o}dinger operator ${T_* H_V}$ near the two emergent Dirac points {\bf (a)} ${(E_D(\tau_0, {\bm \tau}), \, {\bm D}^-(\tau_0, {\bm \tau}))}$ and {\bf (b)} ${(E_D(\tau_0, {\bm \tau}), \, {\bm D}^+(\tau_0, {\bm \tau}))}$ pictured in Figure \ref{fig:intro-1}, panel (\subref{fig:intro-1d}). Their existence is the assertion of Theorem \ref{thm:M-dgn}. The dispersion surfaces are approximately described by the eigenvalue mappings of {\bf (a)} ${H^{{\bm D}^-}_{\rm eff}({\bm q})}$ in \eqref{eq:Dirac-hams}, where ${{\bm q} = {\bm k} - {\bm D}^-(\tau_0, {\bm \tau})}$, and {\bf (b)} $H^{{\bm D}^+}_{\rm eff}({\bm q})$ in \eqref{eq:Dirac-hams}, where instead ${{\bm q} = {\bm k} - {\bm D}^+(\tau_0, {\bm \tau})}$; see Theorem \ref{thm:dir-srf} and Remark \ref{rmk:dir-srf-eff}, which provide a general characterization of Dirac points.}
\label{fig:intro-2}
\end{figure}

Our point of departure is a Schr\"odinger operator $H_V$, where $V$ is a square lattice potential, which has a quadratic band degeneracy point at ${(E_S, \, {\bm M})}$. We consider the $T\Z^2$-periodic deformed Schr\"{o}dinger operator
\begin{equation}
H_{V \circ \, T^{-1}} = -\Delta + V \circ \, T^{-1}
\end{equation}
corresponding to a periodic medium with  fundamental cell distorted by an invertible linear transformation $T$; see Figure \ref{fig:intro-1}, panel (\subref{fig:intro-1c}). The analysis proceeds by first transforming, by the change of variables ${{\bm x} \to T^{-1}{\bm x}}$, the  Schr\"odinger operator $H_{V \circ \, T^{-1}}$ into the  {\it ``pushforward''} Schr\"{o}dinger operator ${T_* H_V \equiv \nabla \cdot (T^\mathsf{T} T)^{-1} \nabla + V}$, which is  again $\Z^2$-periodic. In particular, the fundamental cell of ${T_* H_V}$ is independent of $T$.

To determine the effect of a small deformation on the local band structure near ${(E_S, \, {\bm M})}$, we implement a Schur complement/Lyapunov-Schmidt reduction strategy, which identifies the local band structure near $(E_S, \, {\bm M})$ with an approximate effective Hamiltonian. In terms of ${\bm \kappa}$, the quasimomentum displacement from ${\bm M}$, the  effective Hamiltonian has the same general structure of \eqref{eq:kmow-Heff}:
\begin{equation}
\label{eq:modelHeff}
H^{\bm M}_{\rm eff}({\bm \kappa}; \tau_0, {\bm \tau}) = \bigl( \beta_0 \tau_0 + (1 - \alpha_0) |{\bm \kappa}|^2 \bigr) \, I + (\beta_1 \tau_1 - \alpha_1 {\bm \kappa} \cdot \sigma_1 {\bm \kappa}) \, \sigma_1 + (\beta_2 \tau_3 - \alpha_2 {\bm \kappa} \cdot \sigma_3 {\bm \kappa}) \, \sigma_2.
\end{equation}
The parameters $\tau_0$, ${{\bm \tau} = (\tau_1, \, \tau_3)}$ are coefficients of ${I - (T^\mathsf{T} T)^{-1}}$ with respect to the basis of real Pauli matrices. The real parameters $\beta_0$, $\beta_1$, and $\beta_2$ are again defined in terms of the Floquet-Bloch eigenstates spanning the eigenspace of $H_V$ at the quadratic band degeneracy point ${(E_S, \, {\bm M})}$. There is no $\sigma_3$ contribution to \eqref{eq:modelHeff} since the linear deformation preserves symmetry under spatial inversion and complex conjugation. Theorem \ref{thm:M-srf} implies that, under suitable nondegeneracy assumptions on $H_V$, the band structure of ${T_* H_V}$ near $(E_S, \, {\bm M})$ consists of two dispersion surfaces which are approximately given by the graphs of the two eigenvalue maps ${{\bm \kappa} \mapsto \varepsilon_\pm({\bm \kappa}; \tau_0, {\bm \tau})}$ of $H^{\bm M}_{\rm eff}({\bm \kappa}; \tau_0, {\bm \tau})$. Hence, except in the case of a pure dilation (corresponding to ${|{\bm \tau}| = 0}$), the band structure perturbs to one which, for ${\bm \kappa}$ sufficiently small, has a local energy gap about $E_S$; see Figure \ref{fig:intro-1}, panel (\subref{fig:intro-1d}).

Further, the expression for $H^{\bm M}_{\rm eff}({\bm \kappa}; \tau_0, {\bm \tau})$ in \eqref{eq:modelHeff} suggests (and Theorem \ref{thm:M-dgn} establishes) that, for $\tau_0$, $|{\bm \tau}|$ small and ${|{\bm \tau}| \neq 0}$, the quadratic band degeneracy point at ${(E_S, \, {\bm M})}$ splits into two nearby degeneracies ${(E_D(\tau_0, {\bm \tau}), \, {\bm D}^\pm(\tau_0, {\bm \tau}))}$  with quasimomenta ${\bm D}^\pm(\tau_0, {\bm \tau})$ approximately determined by the vanishing of off-diagonal entries of $H^{\bm M}_{\rm eff}({\bm \kappa}; \tau_0, {\bm \tau})$. In general, ${\bm D}^-(\tau_0, {\bm \tau})$ is located at the inversion of ${\bm D}^+(\tau_0, {\bm \tau})$ with respect to ${\bm M}$: ${{\bm D}^-(\tau_0, {\bm \tau}) = 2 {\bm M} - {\bm D}^+(\tau_0, {\bm \tau})}$.

Notably, the emergent degeneracies ${(E_D(\tau_0, {\bm \tau}), \, {\bm D}^\pm(\tau_0, {\bm \tau}))}$ are Dirac points; see Section \ref{sec:dir-pt}. By Theorem \ref{thm:dir-srf}, the generic local character of the band structure about each twofold degeneracy is a  tilted elliptical cone with upper and lower branches; see Figure \ref{fig:intro-2}. In terms of the quasimomentum displacement ${{\bm q} = {\bm k} - {\bm D}^\pm}$ from a Dirac point, the effective Hamiltonians governing the dynamics of wavepackets spectrally supported near ${(E_D, \, {\bm D}^\pm)}$ are Dirac Hamiltonians with Fourier symbols
\begin{equation}
\label{eq:Dirac-hams}
H^{{\bm D}^+}_{\rm eff}({\bm q}) = ({\bm \gamma}^+_0 \cdot {\bm q}) \, I + ({\bm \gamma}^+_1 \cdot {\bm q}) \, \sigma_1 + ({\bm \gamma}^+_2 \cdot {\bm q}) \, \sigma_2 \quad \text{and} \quad H^{{\bm D}^-}_{\rm eff}({\bm q}) = - H^{{\bm D}^+}_{\rm eff}({\bm q}).
\end{equation}
For special classes of deformations (e.g., deforming the fundamental cell from a square to a rhombus, or compressing it along a coordinate axis), the motion of the Dirac point quasimomenta during the deformation ${(\tau_0, {\bm \tau}) \mapsto {\bm D}^\pm(\tau_0, {\bm \tau})}$ is constrained to directions of preserved symmetry; see Section \ref{sec:special-def}.
 
In Section \ref{sec:break-PC}, we further consider how the band structure near each type of degeneracy (i.e. quadratic band degeneracies and Dirac points) changes under a small (i.e., size $\delta$) breaking of symmetries: (i)  spatial inversion (parity) symmetry, or (ii) complex conjugation (time-reversal) symmetry. This leads to more general effective Hamiltonians depending on $\delta$, which now include contributions from $\sigma_3$. 

\subsection{Relation to previous work and further discussion}
\label{sec:previous}

Haldane and Raghu \cite{HR08} theoretically demonstrated, in the context of the two-dimensional Maxwell equations, that unidirectional edge states, which are robust against localized (even large) perturbations, may be realized at an interface between two periodic media when time-reversal symmetry ($\mathcal{C}$) is broken. In this work, the two bulk media are honeycomb structures which are known to have conical degeneracies (Dirac points) in their band structures when $\mathcal{C}$ is unbroken. Here, two energy bands touch over the two independent high-symmetry quasimomenta ${\bm K}$ and ${\bm K}^\prime$. Breaking $\mathcal{C}$ leads to both media opening a common spectral gap about the ``Dirac energy''. While the union of spectrum of the bulk structures is contained in the spectrum of the interfaced structure, remarkably, the spectrum of the interfaced structure is filled with ``edge (or interface) spectra'' corresponding to solutions of the spectral problem which are localized about the interface and propagate (i.e., are plane wave-like) parallel to the edge. These edge states are robust against localized (even large) perturbations of the interface. 
 
Underlying this phenomenon is a topological explanation of the type which explains a central phenomenon in condensed matter physics: the integer quantum Hall effect \cite{TKNN1982}, observed for the motion of electrons in a two-dimensional quantum material in the presence of a perpendicular magnetic field. Due to the $\mathcal{C}$-breaking, gap-opening perturbation, the (now isolated) bands acquire integer-valued topological indices (i.e., Chern numbers). Since, for the system considered, the Chern numbers of the two structures differ, edge state eigenvalues populate the bulk spectral gap. For energies in this range, an algebraic count (spectral flow) is given by the difference of the total Chern numbers associated with all bands below a fixed energy in the gap (bulk-edge correspondence principle); see, e.g., \cite{drouot2019, drouot2019a, drouotweinstein2020, drouot2021}.
 
These remarkable developments have inspired a huge amount of experimental activity aimed at realizing analogous phenomena in many wave physics settings. The first such experimental work, in an electromagnetic setting, was that of Wang et al. in \cite{wang2008-prl}. In this work, the point of departure was a square lattice medium possessing quadratic band degeneracies over the high-symmetry point ${\bm M}$. Time-reversal symmetry ($\mathcal{C}$) was then broken via the gyromagnetic effect, for which the medium's dielectric tensor is Hermitian but not real-valued, and robust unidirectional edge states were observed experimentally and in energy ranges anticipated by theory.

Motivated by the application of quadratic band degeneracies in \cite{wang2008-prl}, Chong et al. \cite{chong2008effective} used group representation theory considerations to formally obtain the effective Hamiltonians which arise under various symmetry-breaking perturbations. Together, the work on quadratic band degeneracies in \cite{keller2018spectral, keller2020erratum} and the present work on deformations provide a rigorous foundation. 

In \cite{drouot2021ubiq}, Drouot studied the general question of eigenvalue degeneracies of ${N \times N}$ Hermitian matrices which depend on parameters, arising  via homogenization or ``low-energy" approximations of  tight-binding models, as Fourier symbols of effective Hamiltonians. He proved that, generically, the local band structure about a degeneracy is, locally, a tilted cone. Since families of Hermitian matrices with degenerate eigenvalues arise via a spectral localization (Schur complement/Lyapunov-Schmidt reduction) of continuum periodic Schr\"{o}dinger operators about degeneracies (Sections \ref{sec:set-up} and \ref{sec:break-PC} implement this approach), we expect the degeneracies of generic continuum Schr\"{o}dinger operators to be conical; see  \cite[Conjecture 1]{drouot2021ubiq}. Our results on the splitting-off of Dirac points from a symmetry-induced quadratic band degeneracy point under linear deformation of the potential are a manifestation of this phenomenon. 

Further, we remark that \cite{BZ2023} studied an effective Hamiltonian arising for twisted bilayer graphene \cite{WKML2023} in the presence of an in-plane magnetic field. For special directions of the magnetic field, Dirac points were shown to merge into and depart from a quadratic band degeneracy point as the ``twist angle'' varies. Finally, it is noted in \cite{kuchment16} that non-isotropic (here referred to as tilted) Dirac cones arise in structures which lack honeycomb symmmetry, such as a class of periodic structures known as {\it graphynes}.

\begin{figure}[t]
\centering
\begin{subfigure}{0.36\textwidth}
    \subcaption{$\quad$}
    \vspace{0.1cm}
    \includegraphics[height = 4.5cm]{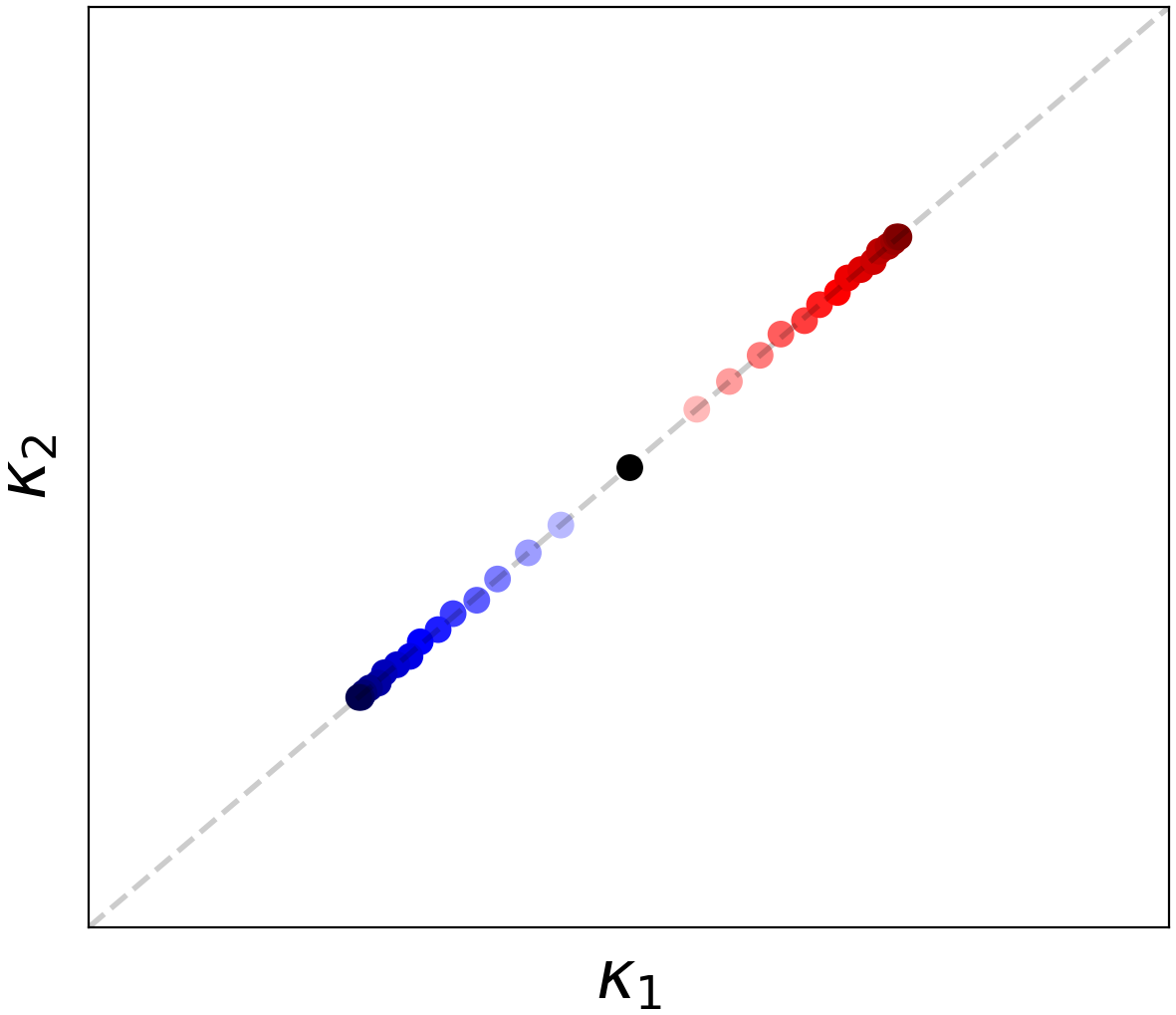}
    \label{fig:intro-3a}
\end{subfigure}
\hspace{0.2cm}
\begin{subfigure}{0.36\textwidth}
    \subcaption{$\quad$}
    \vspace{0.1cm}
    \includegraphics[height = 4.5cm]{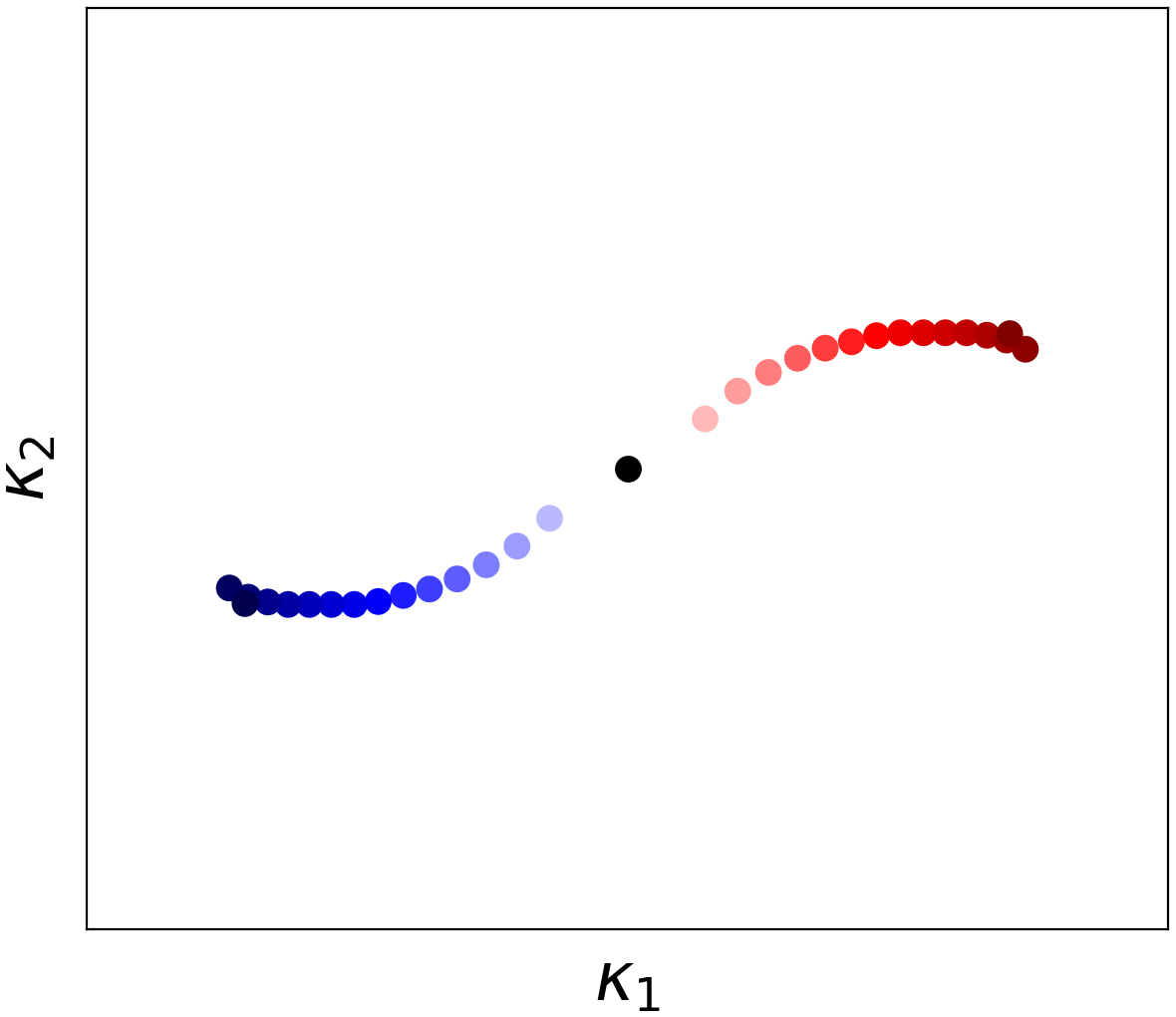}
    \label{fig:intro-3b}
\end{subfigure}
\vspace{0.3cm}
\caption{Trajectories of emergent Dirac point quasimomenta ${{\bm D}^+(\tau_0, {\bm \tau})}$ (red) and ${{\bm D}^-(\tau_0, {\bm \tau})}$ (blue) (see Figure \ref{fig:intro-2}) during a deformation  as $|{\bm \tau}|$ increases (with $\tau_0$ fixed). ${{\bm D}^-(\tau_0, {\bm \tau})}$ is the inversion 
 of ${\bm D}^+(\tau_0, {\bm \tau})$ with respect to ${\bm M}$ (black dot); ${\bm D}^-(\tau_0, {\bm \tau})=2{\bm M}-{\bm D}^+(\tau_0, {\bm \tau})$. Trajectories (${{\bm \kappa} = {\bm k} - {\bm M}}$ the displacement from ${\bm M}$) plotted over ${{\bm M} + \mathcal{B}}$, with ${\mathcal{B} = [-\pi, \, \pi] \times [-\pi, \, \pi]}$ the fixed Brillouin zone of ${T_* H_V}$. The deformation corresponding to {\bf (a)} preserves diagonal reflection symmetry, such as that in panel (\subref{fig:intro-1c}) of Figure \ref{fig:intro-1}, resulting in a symmetric splitting of Dirac points constrained to the diagonal ${\kappa_1 = \kappa_2}$; see Theorem \ref{thm:M-dgn-ref-1}. The deformation in {\bf (b)} breaks all symmetries, resulting in a curved splitting trajectories.}
\label{fig:intro-3}
\end{figure}

\subsection{Generalizations}
\label{sec:general}

We believe the results of this paper can be generalized to deformations of spectral problems of the type
\begin{equation}
( -\nabla_{\bm x} \cdot A({\bm x}) \nabla_{\bm x} + V({\bm x})) \, \psi({\bm x}) = E \, \rho(\bm x) \, \psi({\bm x}),
\end{equation}
for scalar functions $A$, $V$, and $\rho$  having the symmetries of our admissible class of $\Z^2$-periodic potentials. 
We may also allow $A$ to be matrix-valued, where the symmetry ${A(R^\mathsf{T} {\bm x}) = A({\bm x})}$ is instead replaced by ${A(R^\mathsf{T} {\bm x}) = R^\mathsf{T} A({\bm x}) R}$ \cite[Theorem 1]{LWZ19}. Further generalizations to vectorial problems, such as Maxwell's equations in slab geometry (with planar translation invariance), are also possible; see, e.g., \cite{guglielmon2021landau, barsukova-etal-2024}.

\subsection{Large deformations}
\label{sec:large}

It is of interest to understand the possible motions of Dirac points under large deformations of a periodic structure. We have numerically investigated the trajectories followed by Dirac point quasimomenta emanating from a quadratic band degeneracy when the operator is subject to a particular class of deformations; see, e.g., Figure \ref{fig:intro-3}. These simulations indicate that Dirac points persist for fairly large deformations. However, it is generally possible, for example, that Dirac degeneracies collide, annihilate and open a local spectral gap. Some numerical explorations of such phenomena for tight-binding models are presented in \cite{montambaux2009merging, hasegawa2012merging}.

\subsection{Edge states}
\label{sec:edge-states} 

Effective Hamiltonians, which characterize the local band structure near a degeneracy, depend on parameters computable from a basis for the degenerate eigenspace. These determine Chern numbers of the isolated bands, which arise due to symmetry breaking; see the discussion in Section \ref{sec:previous} and \cite{chong2008effective}. In the forthcoming work \cite{CMW-inprogress}, we analytically construct and numerically investigate edge states in media which interpolate, across a domain wall, deformations of media with quadratic band degeneracies, study effective equations governing the localization of edge states, and discuss these results in the context of the bulk-edge correspondence principle; see Section \ref{sec:band-top}.

\subsection{Outline of article}
\label{sec:outline}

The remainder of the article is organized as follows:

\begin{itemize}
\item In Section \ref{sec:fb-thy}, we briefly review Floquet-Bloch theory.

\item In Section \ref{sec:kmow}, we review results of \cite{keller2018spectral, keller2020erratum} on quadratic band degeneracies of Schr\"{o}dinger operators ${H_V = -\Delta + V}$, where $V$ is a square lattice potential. In particular, we state results concerning the existence and local character of quadratic band degeneracies ${(E_S, \, {\bm M})}$ (Section \ref{sec:quad-dgn}).

\item In Section \ref{sec:dir-pt}, we review criteria for the existence of Dirac points in the band structures for a class of second-order periodic, elliptic operators.

\item In Section \ref{sec:deform}, we introduce our model of linearly deformed square lattice media. It is given by the deformed Schr\"{o}dinger operator ${H_{V \circ \, T^{-1}} = - \Delta + V \circ \, T^{-1}}$, where $T$ is a fixed invertible linear transformation. We note the equivalence between the band structures of $H_{V \circ \, T^{-1}}$ and a $\Z^2$-periodic ``pushforward'' operator ${T_* H_V = -\nabla \cdot (T^\mathsf{T} T)^{-1} \nabla + V}$. Next, we discuss the symmetries preserved and broken by $T_* H_V$ (equivalently $H_{V \circ \, T^{-1}}$) relative to the undeformed operator $H_V$ (Section \ref{sec:deform-syms}). Finally, we present examples of deformations (Section \ref{sec:examples}).

\item Section \ref{sec:results} contains the statements of our main results. Theorem \ref{thm:M-srf} presents, for small deformations, the form of dispersion surfaces of $T_* H_V$ in a neighborhood of ${(E_S, \, {\bm M})}$ and the corresponding effective Hamiltonian. The form of the effective Hamiltonian suggests that the quadratic band degeneracy point $(E_S, \, {\bm M})$ of the undeformed Hamiltonian typically splits into two nearby degeneracies. Theorem \ref{thm:M-dgn} shows that $(E_S, \, {\bm M})$ indeed perturbs to two nearby Dirac points ${(E_D, \, {\bm D}^\pm)}$, about each of which the band structure is locally a tilted, elliptical cone; see Theorem \ref{thm:dir-srf}. For special classes of deformations (e.g., deforming the unit cell from a square to a rhombus, or compressing it along a coordinate axis), the displacement of the degenerate quasimomenta ${{\bm M} \to {\bm D}^\pm}$ is more constrained (Section \ref{sec:special-def}).

\item Section \ref{sec:set-up} initiates the set-up for the proofs of our main results. Via a Lyapunov-Schmidt reduction/Schur complement scheme, we find that the band structure of ${T_* H_V}$ near ${(E_S, \, {\bm M})}$ is determined by a self-adjoint ${2 \times 2}$ matrix-valued, analytic function ${\mathcal{M}(\varepsilon; {\bm \kappa}, \tau_0, {\bm \tau})}$. Here, ${(\varepsilon, \, {\bm \kappa})}$ is the energy-quasi- momentum displacement from ${(E_S, \, {\bm M})}$, and $\tau_0$, ${{\bm \tau} = (\tau_1, \, \tau_3)}$ are parameters deriving from $T$.

\item In Sections \ref{sec:pf-M-srf} and \ref{sec:pf-M-dgn}, we prove Theorems \ref{thm:M-srf} and \ref{thm:M-dgn}, respectively,  via a detailed study of 
${\mathcal{M}(\varepsilon; {\bm \kappa}, \tau_0, {\bm \tau})}$ by complex analysis methods and the implicit function theorem.

\item Section \ref{sec:break-PC} summarizes the effect of perturbations which break spatial inversion (parity) or complex conjugation (time-reversal) symmetries; such perturbations typically lift band degeneracies, resulting in a local energy gap. Model perturbations which break parity, but preserve time-reversal symmetry (Section \ref{sec:break-P}), and which preserve parity, but break time-reversal symmetry (Section \ref{sec:break-C}), are separately discussed. We further discuss the implications for band topology, particularly calculations of the Chern numbers associated with nondegenerate bands, which play a central role in the study of edge states in related structures (Section \ref{sec:band-top}).

\item Finally, appendices contain the technical details of several propositions.
\end{itemize}

\subsection{Notation, conventions, and some linear algebra}
\label{sec:notation}

\begin{itemize}
\item Pauli matrices:
\begin{equation}
\label{eq:def-pauli}
\sigma_0 \equiv I =
\begin{bmatrix}
1 & 0 \\
0 & 1
\end{bmatrix} \! , 
\quad \sigma_1 \equiv 
\begin{bmatrix}
0 & 1 \\
1 & 0
\end{bmatrix} \! ,
\quad \sigma_2 \equiv
\begin{bmatrix}
0 & -i \\
i & 0
\end{bmatrix} \! ,
\quad \sigma_3 \equiv
\begin{bmatrix}
1 & 0 \\
0 & -1
\end{bmatrix} \! .
\end{equation}
We sometimes find it convenient to arrange the latter three Pauli matrices into a ``vector'':
\begin{equation}
\label{eq:def-pauli-vec}
{\bm \sigma} \equiv [\sigma_1, \, \sigma_2, \, \sigma_3]^\mathsf{T}.
\end{equation}
For ${{\bm h} \smallin \R^3}$, we will often use the shorthand notation: ${{\bm h} \cdot {\bm \sigma} = h_1 \sigma_1 + h_2 \sigma_2 + h_3 \sigma_3}$.
\item (Clockwise) $\pi/2$-rotation matrix:
\begin{equation}
\label{eq:def-rot}
R \equiv i \sigma_2 =
\begin{bmatrix}
0 & 1 \\
-1 & 0
\end{bmatrix} \! .
\end{equation}
\end{itemize}

\begin{lemma}
\label{lem:2d-ker}
For any ${2 \times 2}$ matrix $A$, ${{\rm dim}({\rm ker}(A)) = 2}$ if and only if $A$ is the zero matrix.
\end{lemma}

\begin{lemma}
\label{lem:Ahermsym}
Given $a$, ${b \smallin \R}$ and ${c \smallin\ C}$, consider the $2\times2$ Hermitian matrix $A = \begin{bmatrix}
a & c^* \\
c & b
\end{bmatrix}$ with ${c = c_1 + i c_2}$. Then 
\begin{equation}
\label{eq:exp-Aherm}
A = \frac{a + b}{2} I + c_1 \sigma_1 + c_2 \sigma_2 + \frac{a - b}{2} \sigma_3.
\end{equation}
Further, if $A$ is real and symmetric, then ${c_2 = 0}$:
\begin{equation}
\label{eq:exp-Asym}
A = \frac{a + b}{2} I + c_1 \sigma_1 + \frac{a - b}{2} \sigma_3.
\end{equation}
\end{lemma}

\noindent In fact, the set of matrices ${\{\sigma_j\}_{0 \, \leq \, j \, \leq \, 3}}$ is a basis for the vector space ${({\rm Herm}(2), \, \R)}$, where ${\rm Herm}(2)$ is the set of ${2 \times 2}$ Hermitian matrices. This basis is orthonormal with respect to the inner product
\begin{equation}
\langle A, \, B \rangle = \frac{1}{2} {\rm tr}(A^* B).
\end{equation}

\begin{lemma}
\label{lem:Adegen}
Suppose $A = \begin{bmatrix}
a & c^* \\
c & b
\end{bmatrix}$ has an degenerate eigenvalue of multiplicity two. Then ${A = aI}$, ${a \smallin \R}$.
\end{lemma}

\begin{lemma}
\label{lem:chain}
{\rm (Chain rule.)}
Let $f({\bm x})$ and ${\bm g}({\bm x})$ be differentiable scalar- and vector-valued functions of ${{\bm x} \smallin \R^n}$, respectively. For $A$ such that ${A^\mathsf{T} A = I}$, define
\begin{equation}
\mathscr{M}_A[f]({\bm x}) \equiv f(A{\bm x}) \quad \text{and} \quad \mathscr{M}_A[{\bm g}]({\bm x}) \equiv {\bm g}(A{\bm x}).
\end{equation}
Then
\begin{equation}
\nabla \mathscr{M}_A[f]({\bm x}) = A^\mathsf{T} \mathscr{M}_A[\nabla f]({\bm x}), \quad \text{equivalently} \quad \mathscr{M}_A[\nabla f]({\bm x}) = A \nabla \mathscr{M}_A[f]({\bm x}),
\end{equation}
and
\begin{equation}
\nabla \cdot \mathscr{M}_A[{\bm g}]({\bm x}) = \mathscr{M}_A[\nabla \cdot A {\bm g}]({\bm x}), \quad \text{equivalently} \quad \mathscr{M}_A[\nabla \cdot {\bm g}]({\bm x}) = \nabla \cdot A^\mathsf{T} \mathscr{M}_A[{\bm g}]({\bm x}).
\end{equation}
\end{lemma}

\bigskip

\section*{Acknowledgements}

The authors wish to thank Jeremy Marzuola for many stimulating discussions. We also thank  an anonymous referee for helpful comments. MIW and JC were supported in part by NSF grant DMS-1908657, DMS-1937254, and Simons Foundation Math + X Investigator Award \# 376319 (MIW). Part of this research was completed during the 2023-24 academic year, when M. I. Weinstein was a Visiting Member in the School of Mathematics - Institute of Advanced Study, Princeton, supported by the Charles Simonyi Endowment, and a Visiting Fellow in the Department of Mathematics at Princeton University.

\bigskip

\section{Floquet-Bloch theory}
\label{sec:fb-thy}

\setcounter{equation}{0}
\setcounter{figure}{0}

We briefly review the spectral theory of periodic, elliptic operators; see \cite{kuchment16, RS4-1978} and references cited therein. We focus, in particular, on the self-adjoint Schr\"{o}dinger operator ${H_V = -\Delta + V}$, where $V$ is real-valued and periodic with respect to a two-dimensional (Bravais) lattice $\Lambda$. We later introduce the class of periodic {\it square lattice potentials}; Schr\"{o}dinger operators with such potentials are the main object of study in this article.

Let ${\bm v}_1$, ${{\bm v}_2 \smallin \R^2}$ be linearly independent and define the {\it (Bravais) lattice} ${\Lambda \equiv \Z {\bm v}_1 \oplus \Z {\bm v}_2 \subset \R^2}$. The dual lattice is ${\Lambda^* \equiv \Z {\bm k}_1 \oplus \Z {\bm k}_2 \subset (\R^2)^*}$, where ${\bm k}_1$, ${{\bm k}_2 \smallin (\R^2)^*}$ satisfy ${{\bm k}_j \cdot {\bm v}_k = 2 \pi \delta_{j, k}}$ for $j$, ${k \smallin \{ 1, \, 2 \}}$. Introduce a choice of fundamental cell for ${(\R^2)^* / \Lambda^*}$, the Brillouin zone $\mathcal{B}$, consisting of all elements of $(\R^2)^*$ which are closer to ${{\bm 0} \smallin \Lambda^*}$ than to any other element of $\Lambda^*$. Then, $(\R^2)^*$ is tiled by the set of all $\Lambda^*$-translates of $\mathcal{B}$.

For each ${{\bm k} \smallin (\R^2)^*}$, let $L^2_{\bm k}$ denote the subspace of ${\bm k}$-pseudoperiodic functions: 
\begin{equation}
L^2_{\bm k} = L^2_{\bm k}(\R^2/\Lambda) \equiv \{ f \smallin L^2_{\rm loc}(\R^2) : f({\bm x} + {\bm v}) = e^{i{\bm k}\cdot{\bm v}} f({\bm x}) \ {\rm a.e.} \ {\bm x} \smallin \R^2, \ {\bm v} \smallin \Lambda\}.
\end{equation}
Note that ${L^2_{\bm k} = L^2_{{\bm k} \, + \, \tilde{\bm k}}}$ for any ${\tilde{\bm k} \smallin \Lambda^*}$. The space $L^2(\R^2)$ has a decomposition into fibers: ${L^2(\R^2) = \int_\mathcal{B}^\oplus L^2_{\bm k} \, {\rm d}{\bm k}}$.

Since $V$ is $\Lambda$-periodic, the operator $H_V$ commutes with translation by elements of $\Lambda$. Thus, $H_V$ acts in each $L^2_{\bm k}$ space, mapping a dense subspace of $L^2_{\bm k}$ to itself; we denote this operator by $H_{\bm k}$. Hence, the spectral properties of $H_V$ acting in $L^2(\R^2)$ can be reduced to a study of the family of spectral properties for $H_{\bm k}$, where ${{\bm k} \smallin \mathcal{B}}$. 

Thus, consider the {\it Floquet-Bloch eigenvalue problem}
\begin{gather}
\label{eq:floquet-bloch-evp_general}
H_{\bm k} \Phi = E \Phi, \quad \Phi \smallin L^2_{\bm k}, \quad {\bm k} \smallin \mathcal{B}.
\end{gather}
Equivalently, we may set ${\Phi({\bm x}) = e^{i {\bm k} \cdot {\bm x}} \phi({\bm x})}$, where ${\phi \smallin L^2(\R^2 / \Lambda) = L^2_{\bm 0}}$, and consider the eigenvalue problem
\begin{gather}
\label{eq:floquet-bloch-evp_general_alt}
H({\bm k}) \phi = E \phi, \quad \phi \smallin L^2(\R^2 / \Lambda), \quad {\bm k} \smallin \mathcal{B}.
\end{gather}
Note that, for each ${\bm k}$, ${H({\bm k}) \equiv e^{-i{\bm k}\cdot{\bm x}} \, H_{\bm k} \, e^{i{\bm k}\cdot{\bm x}} = -(\nabla + i{\bm k})^2 + V}$ acts on $L^2(\R^2/\Lambda)$. 

For each ${\bm k} \smallin \mathcal{B}$, $H_{\bm k}$ acting in $L^2_{\bm k}$ (equivalently, $H(\bm k)$ acting in $L^2(\R^2/\Z^2)$) is self-adjoint and has compact resolvent. Therefore, the eigenvalue problems \eqref{eq:floquet-bloch-evp_general} and \eqref{eq:floquet-bloch-evp_general_alt} have a discrete sequence of eigenvalues
\begin{equation}
E_1({\bm k}) \leq E_2({\bm k}) \leq \cdots \leq E_b({\bm k}) \leq \cdots
\end{equation}
of finite multiplicity and tending to infinity. The corresponding eigenstates for \eqref{eq:floquet-bloch-evp_general} ({\it Floquet-Bloch states}) are  $\Phi_b(\bm x,\bm k)$ and those for \eqref{eq:floquet-bloch-evp_general_alt} are $\phi_b(\bm x,\bm k)=e^{-i\bm k\cdot \bm x}\Phi_b(\bm x,\bm k)$. The eigenvalue maps ${{\bm k} \to E_b({\bm k})}$, ${b \geq 1}$, are Lipschitz continuous and called {\it dispersion relations}; their graphs over $\mathcal{B}$ are called {\it dispersion surfaces}. 

As ${\bm k}$ varies over $\mathcal{B}$, each $E_b({\bm k})$ sweeps out a real subinterval of $\R$. The spectrum of $H_V$ acting on $L^2(\R^2)$ is the union of all such subintervals:
\begin{equation}
{\rm spec}(H) = \bigcup_{b \, \geq \, 1} E_b(\mathcal{B}).
\end{equation}
The collection of eigenvalue maps/dispersion relations and corresponding eigenstates ${\{(E_b({\bm k}), \, \Phi_b({\bm x}, {\bm k}) \}_{b \, \geq \, 1}}$ is called the {\it band structure} of $H_V$.

\bigskip

\section{Square lattice media and quadratic band degeneracies}
\label{sec:kmow}

\setcounter{equation}{0}
\setcounter{figure}{0}

In this section, we review results in \cite{keller2018spectral} on the presence of quadratic band degeneracies in the band structures of Schr\"{o}dinger operators for a class of periodic {\it square lattice potentials}.

\subsection{Square lattice potentials}
\label{sec:sql-pot}

The  {\it square lattice} in $\R^2$ is given by
\begin{equation}
\label{eq:def_sql}
\Z^2 = \Z {\bm v}_1 \oplus \Z {\bm v}_2, \quad \text{where} \quad {\bm v}_1 \equiv [1, 0]^\mathsf{T}, \quad {\bm v}_2 \equiv [0, 1]^\mathsf{T}.
\end{equation}
Our choice of fundamental cell is the unit square ${\Omega \equiv [-1/2, \, 1/2] \times [-1/2, \, 1/2]}$. The dual lattice is
\begin{equation}
\label{eq:def-sql-dual}
(\Z^2)^* = \Z {\bm k}_1 \oplus \Z {\bm k}_2, \quad \text{where} \quad {\bm k}_1 \equiv 2 \pi \, [1, 0]^\mathsf{T}, \quad {\bm k}_2 \equiv 2 \pi \, [0, 1]^\mathsf{T},
\end{equation}
and the Brillouin zone is ${\mathcal{B} \equiv [-\pi, \, \pi] \times [-\pi, \, \pi]}$.

\begin{figure}[t]
\centering
\begin{subfigure}{0.38\textwidth}
    \subcaption{$\ $}
    \label{fig:kmow-1a}
    \vspace{0.1cm}
    \begin{tikzpicture}[x = 1.75cm, y = 1.75cm]
    \draw[white] (0, 0) rectangle (3.5, 3.5);
    \draw[black, thick, ->, opacity = 0.5] (0.25, 1.75) -- (3.25, 1.75) node[anchor = north] {$x_1$};
    \draw[black, thick, ->, opacity = 0.5] (1.75, 0.25) -- (1.75, 3.25) node[anchor = east] {$x_2 \,$};
    \fill[black] (0.75, 0.75) circle (0.1cm);
    \fill[black] (0.75, 1.75) circle (0.1cm) node[anchor = north east, opacity = 0.5] {$-1$};
    \fill[black] (0.75, 2.75) circle (0.1cm);
    \fill[black] (1.75, 0.75) circle (0.1cm) node[anchor = north east, opacity = 0.5] {$-1$};
    \fill[black] (1.75, 1.75) circle (0.1cm) node[anchor = north east, opacity = 0.5] {$0$};
    \fill[black] (1.75, 2.75) circle (0.1cm) node[anchor = north east, opacity = 0.5] {$1$};
    \fill[black] (2.75, 0.75) circle (0.1cm);
    \fill[black] (2.75, 1.75) circle (0.1cm) node[anchor = north east, opacity = 0.5] {$1$};
    \fill[black] (2.75, 2.75) circle (0.1cm); 
    \draw[red!20!blue, thick, opacity = 0.5] (1.25, 1.25) rectangle (2.25, 2.25);
    \fill[red!20!blue, opacity = 0.3] (1.25, 1.25) rectangle (2.25, 2.25);
    \draw[red, very thick, ->] (1.75, 1.75) -- (2.75, 1.75) node[anchor = south east] {${\bm v}_1$};
    \draw[red, very thick, ->] (1.75, 1.75) -- (1.75, 2.75) node[anchor = north west] {${\bm v}_2 \,$};
    \end{tikzpicture}
\end{subfigure}
\hspace{0.4cm}
\begin{subfigure}{0.38\textwidth}
    \subcaption{$\ $}
    \label{fig:kmow-1b}
    \vspace{0.1cm}
    \begin{tikzpicture}[x = 1.75cm, y = 1.75cm]
    \draw[white] (0, 0) rectangle (3.5, 3.5);
    \draw[black, thick, ->, opacity = 0.5] (0.25, 1.75) -- (3.25, 1.75) node[anchor = north] {$k_1$};
    \draw[black, thick, ->, opacity = 0.5] (1.75, 0.25) -- (1.75, 3.25) node[anchor = east] {$k_2 \,$};
    \fill[black] (0.75, 0.75) circle (0.1cm);
    \fill[black] (0.75, 1.75) circle (0.1cm) node[anchor = north east, opacity = 0.5] {$-2\pi$};
    \fill[black] (0.75, 2.75) circle (0.1cm);
    \fill[black] (1.75, 0.75) circle (0.1cm) node[anchor = north east, opacity = 0.5] {$-2\pi$};
    \fill[black] (1.75, 1.75) circle (0.1cm) node[anchor = north east, opacity = 0.5] {$0$};
    \fill[black] (1.75, 2.75) circle (0.1cm) node[anchor = north east, opacity = 0.5] {$2\pi$};
    \fill[black] (2.75, 0.75) circle (0.1cm);
    \fill[black] (2.75, 1.75) circle (0.1cm) node[anchor = north east, opacity = 0.5] {$2\pi$};
    \fill[black] (2.75, 2.75) circle (0.1cm);
    \draw[red!60!yellow, thick, opacity = 0.5] (1.25, 1.25) rectangle (2.25, 2.25);
    \fill[red!60!yellow, opacity = 0.3] (1.25, 1.25) rectangle (2.25, 2.25);
    \draw[red, very thick, ->] (1.75, 1.75) -- (2.75, 1.75) node[anchor = south east] {${\bm k}_1$};
    \draw[red, very thick, ->] (1.75, 1.75) -- (1.75, 2.75) node[anchor = north west] {${\bm k}_2 \,$};
    \end{tikzpicture}
\end{subfigure}
\vspace{0.3cm}
\caption{{\bf (a)} The square lattice $\Z^2$ with lattice vectors ${{\bm v}_1 = [1, 0]^\mathsf{T}}$ and ${{\bm v}_2 = [0, 1]^\mathsf{T}}$ (red). The fundamental cell ${\Omega = [-1/2, \, 1/2] \times [-1/2, \, 1/2]}$ (shaded purple) is indicated. {\bf (b)} The dual lattice $(\Z^2)^*$ with dual lattice vectors ${{\bm k}_1 = 2\pi \, [1, 0]^\mathsf{T}}$ and ${{\bm k}_2 = 2\pi \, [0, 1]^\mathsf{T}}$ (red). The fundamental cell, or first Brillouin zone ${\mathcal{B} = [-\pi, \, \pi] \times [-\pi, \, \pi]}$ (shaded orange) is indicated; see also Figure \ref{fig:kmow-2}.}
\label{fig:kmow-1}
\end{figure}

Our class of square lattice potentials is defined in terms of several symmetry operators:
\begin{definition}
\label{def:syms} 
{\rm (Symmetry operators.)}
For any function $f$, we define:
\begin{align}
\mathcal{P}[f]({\bm x}) & \equiv f(-{\bm x}) & \text{(spatial inversion/parity with respect to ${{\bm x} = {\bm 0}}$)}, \\
\mathcal{C}[f]({\bm x}) & \equiv f({\bm x})^* & \text{(complex conjugation/time-reversal)}, \\
\mathcal{R}[f]({\bm x}) & \equiv f(R^\mathsf{T}{\bm x}) & \text{($\pi/2$ rotation about ${{\bm x} = {\bm 0}}$)}, \\
\Sigma_1[f]({\bm x}) & \equiv f(\sigma_1 {\bm x}) & \text{(reflection about ${x_1 = x_2}$)}. 
\end{align}
The matrices $R$ and $\sigma_1$ are defined in \eqref{eq:def-rot} and \eqref{eq:def-pauli}.
\end{definition}

\begin{definition}
\label{def:sql-pot}
{\rm (Square lattice potential.)}
We say that a smooth potential $V$ is a {\it square lattice potential} if it is $\Z^2$-periodic and invariant under the symmetries introduced in Definition \ref{def:syms}.
\end{definition}

\begin{remark}
\label{rmk:min-sym}
Since ${\mathcal{R}^2 = \mathcal{P}}$, the full group of symmetries of square lattice potentials is actually generated by a proper subset of the symmetry operators introduced above.
\end{remark}

\begin{remark}
Note that $\mathcal{R}$ and $\Sigma_1$ can be composed to generate a new symmetry operator:
\begin{equation}
\Sigma_3[f]({\bm x}) \equiv {\mathcal{R} \Sigma_1}[f]({\bm x}) = f(\sigma_3 {\bm x}).
\end{equation}
The matrix $\sigma_3$ is defined in \eqref{eq:def-pauli}. Hence, square lattice potentials are also invariant under $\Sigma_3$.
\end{remark}

\begin{remark}
We believe it is possible, without much difficulty, to considerably relax the assumption that square lattice potentials are smooth.
\end{remark}

\noindent A simple, nontrivial example of a square lattice potential (in the sense of Definition \ref{def:sql-pot}) is a sum of $\Z^2$-translates of a fixed, rapidly-decaying, radially symmetric function; see Figure \ref{fig:intro-1}. Other examples are discussed in \cite{keller2018spectral}.

\begin{remark}
\label{rmk:sql-pot-ref}
Note that $\Sigma_1$ invariance is not completely necessary; analogs of both Theorem \ref{thm:kmow} (ahead) and our main results (see Section \ref{sec:results}) hold without it. However, under this additional symmetry, the form of the effective Hamiltonian governing a neighborhood of the degeneracy is simpler. Further, as we shall see in Theorems \ref{thm:M-dgn-ref-1} and \ref{thm:M-dgn-ref-2}, deformations which are constrained by additional symmetries lead to more constrained "unfolding" of the quadratic degeneracy.
\end{remark}

\noindent In Appendix \ref{apx:fourier}, a characterization of the Fourier series of such potentials is discussed.

\subsection{Analysis of Schr\"{o}dinger operators with square lattice potentials}
\label{sec:sql-op}

We consider Schr\"{o}dinger operators
\begin{equation}
H_V \equiv -\Delta + V,
\end{equation}
where $V$ is a square lattice potential in the sense of Definition \ref{def:sql-pot}.

\begin{proposition}
{\rm (Symmetries of $H_V$.)}
Let $V$ denote a square lattice potential; see Definition \ref{def:sql-pot}. The operator ${H_V = -\Delta+V}$ acting on $L^2(\R^2)$, with dense domain $C^2_0(\R^2)$, commutes with $\Z^2$-translations, as well as with $\mathcal{P}$, $\mathcal{C}$, $\mathcal{R}$, $\Sigma_1$, and $\Sigma_3$ symmetries.
\end{proposition}

The notion of a {\it high-symmetry quasimomentum} plays an important role in this study. For us, a high symmetry quasimomentum, ${{\bm k}_\star \smallin \mathcal{B}}$ modulo $\Lambda^*$, is one for which ${\mathcal{R} L^2_{{\bm k}_\star} = L^2_{{\bm k}_\star}}$. One checks easily that ${\bm k}_\star$ is a high symmetry quasimomentum exactly when ${\bm k}_\star$ is one of the following two quasimomenta:
\begin{equation}
{\bm \Gamma} = [0, \, 0]^\mathsf{T} \quad \text{or} \quad {\bm M} = [\pi, \, \pi]^\mathsf{T}.
\end{equation}

\begin{figure}[t]
\centering
\begin{tikzpicture}[x = 1.25cm, y = 1.25cm]
    \draw[white] (0, 0) rectangle (4, 4);
    \draw[black, thick, ->, opacity = 0.5] (0.25, 2) -- (3.85, 2) node[anchor = north] {$k_1$};
    \draw[black, thick, ->, opacity = 0.5] (2, 0.25) -- (2, 3.85) node[anchor = east] {$k_2 \,$};
    \draw[black, opacity = 0.5] (0.5, 1.9) -- (0.5, 2) node[anchor = north east] {$-\pi$} -- (0.5, 2.1);
    \draw[black, opacity = 0.5] (2, 2) node[anchor = north east] {$0$};
    \draw[black, opacity = 0.5] (3.5, 1.9) -- (3.5, 2) node[anchor = north east] {$\pi$} -- (3.5, 2.1);
    \draw[black, opacity = 0.5] (1.9, 0.5) -- (2, 0.5) node[anchor = north east] {$-\pi$} -- (2.1, 0.5);
    \draw[black, opacity = 0.5] (1.9, 3.5) -- (2, 3.5) node[anchor = north east] {$\pi$} -- (2.1, 3.5);    
    \draw[red!60!yellow, thick, opacity = 0.5] (0.5, 0.5) rectangle (3.5, 3.5);
    \fill[red!60!yellow, opacity = 0.3] (0.5, 0.5) rectangle (3.5, 3.5);
    \fill[black] (2, 2) circle (0.1cm) node[anchor = south west] {${\bm \Gamma}$};
    \fill[black] (3.5, 3.5) circle (0.1cm) node[anchor = south west] {${\bm M}$};
    \fill[black] (0.5, 3.5) circle (0.1cm) node[anchor = south east] {$R^3 {\bm M} = -R {\bm M}$};
    \fill[black] (0.5, 0.5) circle (0.1cm) node[anchor = north east] {$R^2 {\bm M} = -{\bm M}$};
    \fill[black] (3.5, 0.5) circle (0.1cm) node[anchor = north west] {$R {\bm M}$};
\end{tikzpicture}
\vspace{0.3cm}
\caption{The Brillouin zone ${\mathcal{B} = [-\pi, \, \pi] \times [-\pi, \, \pi]}$; see also Figure \ref{fig:kmow-1}, panel (\subref{fig:kmow-1b}). The high-symmetry points ${{\bm \Gamma} = [0, 0]^\mathsf{T}}$ and ${{\bm M} = [\pi, \pi]^\mathsf{T}}$ are identified, along with the rotations of ${\bm M}$, equivalent to ${\bm M}$ modulo $(\Z^2)^*$.}
\label{fig:kmow-2}
\end{figure}

We shall focus on the band structure in a neighborhood of the high-symmetry point ${\bm M}$. Since ${\mathcal{R} L^2_{\bm M} = L^2_{\bm M}}$ and $\mathcal{R}$ is a normal operator, $L^2_{\bm M}$ decomposes into a direct sum of eigenspaces of $\mathcal{R}$. We make this decomposition explicit: Since ${\mathcal{R}^4 = 1}$, its eigenvalues $\varsigma$ satisfy ${\varsigma^4 = 1}$ and are given by $\{ 1, \, i, \, -1, \, -i\}$. Thus, we have the orthogonal decomposition
\begin{align}
\label{eq:L2M-dsum}
L^2_{\bm M} & = L^2_{{\bm M}, \, 1} \oplus L^2_{{\bm M}, \, i} \oplus L^2_{{\bm M}, \, -1} \oplus L^2_{{\bm M}, \, -i}, \quad \text{where} \\
\label{eq:def-L2M-rot}
L^2_{{\bm M}, \, \varsigma} & \equiv \{ f \smallin L^2_{\bm M} : \mathcal{R}[f]({\bm x}) = \varsigma f({\bm x}) \ {\rm a.e.} \ {\bm x} \smallin \R^2 \} .
\end{align}
Note that $\mathcal{PC}$ maps $L^2_{{\bm M}, i}$ to $L^2_{{\bm M}, -i}$.

\subsection{Quadratic band degeneracies}
\label{sec:quad-dgn}

A {\it quadratic band degeneracy point}, or simply a {\it quadratic band degeneracy}, is an energy-quasimomentum pair ${(E_S, \, {\bm k}_S)}$ at which exactly two consecutive dispersion surfaces touch quadratically. In \cite{keller2018spectral}, it is proven that quadratic band degeneracies occur in the band structures of Schr\"{o}dinger operators ${H_V = -\Delta + V}$, with $V$ a square lattice potential, at the high-symmetry quasimomentum ${{\bm k}_S = {\bm M}}$; see Theorem \ref{thm:kmow} below.

The locally quadratic character follows from the existence of an energy-quasimomentum pair ${(E_S, \, {\bm M})}$ for which the following structure of the corresponding Floquet-Bloch eigenspace holds:

\begin{enumerate}
\renewcommand{\theenumi}{Q\arabic{enumi}}
\item \label{itm:quad-dgn-1} $E_S$ is a multiplicity two $L^2_{\bm M}$ eigenvalue of $H_V$,
\item \label{itm:quad-dgn-2} $E_S$ is a simple $L^2_{{\bm M},+i}$ eigenvalue of $H_V$ with normalized eigenstate $\Phi_1$,
\item \label{itm:quad-dgn-3} $E_S$ is a simple $L^2_{{\bm M},-i}$ eigenvalue of $H_V$ with normalized eigenstate ${\Phi_2 = \mathcal{PC}[\Phi_1]}$,
\item \label{itm:quad-dgn-4} $E_S$ is neither an $L^2_{{\bm M},+1}$ eigenvalue nor an $L^2_{{\bm M},-1}$ eigenvalue of $H_V$.
\end{enumerate}
Introduce orthogonal projections ${\Pi^\parallel = \Phi_1 \langle \Phi_1, \, \cdot \rangle + \Phi_2 \langle \Phi_2, \, \cdot \rangle }$ and ${\Pi^\perp = 1 - \Pi^\parallel}$, and define the resolvent ${\mathscr{R}(E_S) = \Pi^\perp (H_V - E_S)^{-1} \Pi^\perp}$, which maps ${{\ker}_{L^2_{\bm M}}(H_V - E_S)^\perp \to H^2_{\bm M}}$. Further, define the parameters
\begin{equation}
\label{eq:def-a-par}
\begin{aligned}
\alpha_0 & \equiv \langle \partial_{x_1} \Phi_1 , \, \mathscr{R}(E_S) \partial_{x_1} \Phi_1 \rangle \smallin \R, \\
\alpha_1 & \equiv \langle \partial_{x_1} \Phi_1, \, \mathscr{R}(E_S) \partial_{x_2} \Phi_2 \rangle \smallin \R, \\ 
\alpha_2 & \equiv i \langle \partial_{x_1} \Phi_1 , \, \mathscr{R}(E_S) \partial_{x_1} \Phi_2 \rangle \smallin \R.
\end{aligned}
\end{equation}

\begin{enumerate}
\renewcommand{\theenumi}{Q\arabic{enumi}}
\setcounter{enumi}{4}
\item \label{itm:quad-dgn-5} (Nondegeneracy condition.) ${\alpha_1 \neq 0}$ and ${\alpha_2 \neq 0}$.
\end{enumerate}

\begin{remark}
\label{rmk:a-par-small}
In \cite[Appendix C]{keller2018spectral}, the nondegeneracy condition \ref{itm:quad-dgn-5} is explicitly verified for small amplitude potentials, under a generically satisfied nondegeneracy condition on its Fourier coefficients.
\end{remark}

\noindent The following result was proved in  \cite{keller2018spectral}.
\begin{theorem}
\label{thm:kmow}
{\rm (Quadratic band degeneracies in square lattice media.)}
Let $V$ denote a square lattice potential (Definition \ref{def:sql-pot}) and  ${H_V = -\Delta + V}$ the associated Schr\"{o}dinger operator. Then, the following hold:
\begin{enumerate}
\item \label{itm:kmow-exist} The band structure of $H_V$ generically\footnote{The precise meaning of the term ``generic'' is stated in Theorem 6.1 of \cite{keller2018spectral} and summarized in Remark \ref{rmk:kmow-generic} of Appendix \ref{apx:fourier}.} admits energy-quasimomentum pairs ${(E_S, \, {\bm M})}$ satisfying properties \ref{itm:quad-dgn-1} -- \ref{itm:quad-dgn-5}.
\item \label{itm:kmow-srf} Suppose ${(E_S, \, {\bm M})}$ satisfies properties \ref{itm:quad-dgn-1} -- \ref{itm:quad-dgn-5}. Then, there exists ${\kappa^\star > 0}$ such that the dispersion surfaces of $H_V$ containing ${(E_S, \, {\bm M})}$ are described, for ${|{\bm \kappa}| < \kappa^\star}$, by
\begin{align}
\label{eq:kmow}
E_\pm({\bm M} + {\bm \kappa}) - E_S = (1 - \alpha_0) |{\bm \kappa}|^2 + Q_4({\bm \kappa}) \pm \sqrt{(\alpha_1 {\bm \kappa} \cdot \sigma_1 {\bm \kappa})^2 + (\alpha_2 {\bm \kappa} \cdot \sigma_3 {\bm \kappa})^2 + Q_6({\bm \kappa})}.
\end{align}
The functions $\mathcal{Q}_n({\bm \kappa})$, $n= 4,6$, are analytic in a neighborhood of ${{\bm \kappa} = {\bm 0}}$ and satisfy: 
\begin{equation}
\label{eq:kmow-Q-bd}
Q_n({\bm \kappa}) = O(|{\bm \kappa}|^n) \ \ \text{as} \ \ |{\bm \kappa}| \to 0.
\end{equation}
Moreover, they possess the following symmetries:
\begin{equation}
\label{eq:kmow-Q-sym}
\begin{aligned}
Q_n({\bm \kappa}) = Q_n(R^\mathsf{T} {\bm \kappa}), \quad Q_n({\bm \kappa}) = Q_n(\sigma_1 {\bm \kappa}), \quad \text{and} \quad Q_n({\bm \kappa}) = Q_n(\sigma_3 {\bm \kappa}).
\end{aligned}
\end{equation}
Here, $R$ denotes the $\pi/2$-rotation matrix and $\sigma_1$, $\sigma_3$ are Pauli matrices; see Section \ref{sec:notation}.
\end{enumerate}
\end{theorem}

\begin{remark}
\label{rmk:kmow-eff}
{\rm (Effective Hamiltonian about ${(E_S, \, {\bm M})}$.)}
The dispersion surfaces of $H_V$ which touch at ${(E_S, \, {\bm M})}$ are locally approximated by the dispersion relations of the effective Hamiltonian $H^{\bm M}_{\rm eff}$, whose ${2 \times 2}$ matrix Fourier symbol is:
\begin{equation}
\label{eq:kmow-Heff_2}
H^{\bm M}_{\rm eff}({\bm \kappa}) \equiv (1 - \alpha_0) |{\bm \kappa}|^2 I - (\alpha_1 {\bm \kappa} \cdot \sigma_1 {\bm \kappa}) \sigma_1 - (\alpha_2 {\bm \kappa} \cdot \sigma_3 {\bm \kappa}) \sigma_2.
\end{equation}
The two dispersion relations of $H^{\bm M}_{\rm eff}$ are given by:
\begin{equation}
\label{eq:kmow-Heff-srf}
\varepsilon_\pm({\bm \kappa}) = (1 - \alpha_0) |{\bm \kappa}|^2 \pm \sqrt{(\alpha_1 {\bm \kappa} \cdot \sigma_1 {\bm \kappa})^2 + (\alpha_2 {\bm \kappa} \cdot \sigma_3 {\bm \kappa})^2}.
\end{equation}
These coincide with the expressions \eqref{eq:kmow} of part \ref{itm:kmow-srf} of Theorem \ref{thm:kmow}, omitting higher order terms.
\end{remark}

\noindent Theorem \ref{thm:kmow} was proved in \cite{keller2018spectral}. The form of the effective Hamiltonian was argued by group representation considerations in \cite{chong2008effective}.

\bigskip

\section{Dirac points for second order periodic elliptic operators}
\label{sec:dir-pt}

\setcounter{equation}{0}
\setcounter{figure}{0}

Our main results (see Section \ref{sec:results}) establish that quadratic band degeneracy points perturb, under typical small deformations, to a pair of nearby Dirac (conical) points. In this section, we discuss sufficient conditions  for the existence of Dirac points in the band structure of an self-adjoint elliptic operator of the general form:
\begin{equation}
\label{eq:def-L}
L \equiv -\nabla \cdot A \nabla + V, 
\end{equation}
where
(a) $V$ is a sufficiently smooth, real-valued function which is periodic with respect to a two-dimensional lattice $\Lambda$ and inversion symmetric, and 
(b) $A$ is a sufficiently smooth, ${2 \times 2}$ real symmetric matrix-valued function which is $\Lambda$-periodic, inversion symmetric, and uniformly positive-definite. 
Our discussion of $L$ is closely related to that of \cite{fefferman2012honeycomb}; see also \cite{LWZ19}. In contrast to these references, we obtain a (tilted) Dirac cone without any assumptions on rotational invariance of $L$. This is consistent with the genericity of Dirac points, discussed in \cite{drouot2021ubiq}; see also our discussion in the introduction.

The operator $L$ commutes with translation by elements of $\Lambda$. Thus, the Floquet-Bloch theory of Section \ref{sec:fb-thy} applies to $L$, and hence $L$ has a band structure. Further, $L$ commutes with $\mathcal{PC}$ (Definition \ref{def:syms}) and, moreoever, for any ${{\bm k} \smallin \mathcal{B}}$, $\mathcal{PC}$ maps $L^2_{\bm k}$ to itself. Since $\mathcal{PC}$ satisfies ${(\mathcal{PC})^2 = 1}$ its eigenvalues are $\pm 1$. This induces the orthogonal decomposition:
\begin{align}
\label{eq:L2k-PC}
L^2_{\bm k} & = Y_{{\bm k},+1} \oplus Y_{{\bm k},-1}, \quad \text{where} \\
Y_{{\bm k},\varsigma} & \equiv \{ f \smallin L^2_{\bm k} : \mathcal{PC}[f] = \varsigma f \}, \ \ \varsigma = \pm 1.
\end{align}

A {\it Dirac point} is a locally conical touching of consecutive dispersion surfaces. In analogy with the case of quadratic band degeneracy points, Dirac points are a consequence of the following structure of the Floquet-Bloch eigenspace at an energy-quasimomentum pair ${(E_D, \, {\bm D})}$, which we shall assume:
\begin{enumerate}
\renewcommand{\theenumi}{D\arabic{enumi}}
\item \label{itm:dir-pt-1} $E_D$ is a multiplicity two $L^2_{\bm D}$ eigenvalue of $L$, 
\item \label{itm:dir-pt-2} $E_D$ is a simple $Y_{{\bm D}, +1}$ eigenvalue of $L$ with normalized eigenstate $\tilde{\Phi}^{\bm D}_1$,
\item \label{itm:dir-pt-3} $E_D$ is a simple $Y_{{\bm D}, -1}$ eigenvalue of $L$ with normalized eigenstate $\tilde{\Phi}^{\bm D}_2$.
\end{enumerate}
The eigenspace ${{\rm ker}(L - E_D)}$ has an orthonormal basis ${\{ \tilde{\Phi}^{\bm D}_1, \, \tilde{\Phi}^{\bm D}_2 \}}$. We use a more convenient orthonormal basis ${\{ \Phi^{\bm D}_1, \, \Phi^{\bm D}_2 \}}$, chosen so that ${\Phi^{\bm D}_2 = \mathcal{PC}[\Phi^{\bm D}_1]}$. Such a basis can be constructed by setting
\begin{equation}
\Phi^{\bm D}_1 \equiv \frac{1}{\sqrt{2}} (\tilde{\Phi}^{\bm D}_1 + \tilde{\Phi}^{\bm D}_2) \quad \text{and} \quad \Phi^{\bm D}_2 \equiv \frac{1}{\sqrt{2}} (\tilde{\Phi}^{\bm D}_1 - \tilde{\Phi}^{\bm D}_2).
\end{equation}
Further, we define
\begin{equation}
\label{eq:def-gam-par}
\begin{aligned}
{\bm \gamma}_0 & \equiv \langle \Phi^{\bm D}_1, \, -i \nabla \cdot (A \Phi^{\bm D}_1) - i A \nabla \Phi^{\bm D}_1 \rangle, \\
{\bm \gamma}_1 & \equiv {\rm Re} \, \langle \Phi^{\bm D}_1, \, -i \nabla \cdot (A \Phi^{\bm D}_2) - i A \nabla \Phi^{\bm D}_2 \rangle , \\
{\bm \gamma}_2 & \equiv -{\rm Im} \, \langle \Phi^{\bm D}_1, \, -i \nabla \cdot (A \Phi^{\bm D}_2) - i A \nabla \Phi^{\bm D}_2 \rangle .
\end{aligned}
\end{equation}

\begin{enumerate}
\renewcommand{\theenumi}{D\arabic{enumi}}
\setcounter{enumi}{3}
\item \label{itm:dir-pt-4} (Nondegeneracy condition.) ${\bm \gamma}_1$, ${\bm \gamma}_2 \neq {\bm 0}$.
\end{enumerate}

\begin{remark}
Both Dirac points and quadratic band degeneracy points satisfy \ref{itm:dir-pt-1} -- \ref{itm:dir-pt-3}. However, for quadratic points, the parameters ${\bm \gamma}_1$ and ${\bm \gamma}_2$ vanish by symmetry considerations; see \cite[Proposition 4.10]{keller2018spectral}.
\end{remark}

\begin{theorem}
\label{thm:dir-srf}
{\rm (Properties \ref{itm:dir-pt-1} -- \ref{itm:dir-pt-4} imply Dirac points.)}
Consider the periodic, elliptic operator $L$, defined in \eqref{eq:def-L}. Assume that the energy-quasimomentum pair ${(E_D, \, {\bm D})}$ satisfies hypotheses \ref{itm:dir-pt-1} -- \ref{itm:dir-pt-4}. Then, there exists ${q^\star > 0}$ such that the dispersion surfaces of $L$ containing ${(E_D, \, {\bm D})}$ are described, for ${|{\bm q}| < q^\star}$, by:
\begin{equation}
\label{eq:dir-srf}
E_\pm({\bm D} + {\bm q}) - E_D = {\bm \gamma}_0 \cdot {\bm q} + Q_2({\bm q}) \pm \sqrt{({\bm \gamma}_1 \cdot {\bm q})^2 + ({\bm \gamma}_2 \cdot {\bm q})^2 + Q_3({\bm q})}.
\end{equation}
The functions $Q_n({\bm q})$, ${n = 2}$, $3$, are analytic in a neighborhood of ${{\bm q} = {\bm 0}}$ and satisfy:
\begin{equation}
\label{eq:dir-srf-Q-bd}
Q_n({\bm q}) = O(|{\bm q}|^n) \ \ \text{as} \ \ |{\bm q}| \to 0.
\end{equation}
\end{theorem}

\noindent A sketch of the proof (see also Remark \ref{rmk:dir-srf-eff} below) is given in Appendix \ref{apx:pf-dir-srf}.
 
\begin{remark}
{\rm (Relation to \cite{fefferman2012honeycomb} and the proof of Theorem \ref{thm:dir-srf}.)}
Theorem \ref{thm:dir-srf} extends results in \cite[Theorem 4.1]{fefferman2012honeycomb}, which give sufficient conditions for the existence of  Dirac (conical) points in the band structure of   Schr\"odinger operators with {\it honeycomb lattice potentials}, {\it i.e.} ${L = -\Delta + V}$, where $V$ is periodic with respect to the equilaterial triangular lattice, $\mathcal{PC}$-invariant, and $2\pi/3$-rotationally invariant. In this case, the Dirac points appear at the high-symmetry quasimomenta at the vertices (${\bm K}$ and ${\bm K}^\prime$ points) of the hexagonal Brillouin zone and the local behavior, where the two bands touch, is given by
\begin{equation}
E_\pm({\bm D} + {\bm q}) = \pm |v_D| |{\bm q}| \, (1 + O(|{\bm q}|)) \ \ \text{as} \ \ |{\bm q}| \to 0,
\label{eq:Dcone-honey}
\end{equation}
with ${|v_D| > 0}$. The expression \eqref{eq:Dcone-honey} arises from \eqref{eq:dir-srf} from the observations, in this case, that ${{\bm \gamma}_0 = {\bm 0}}$, and ${\bm \gamma}_1$, ${\bm \gamma}_2$ are orthogonal and of equal length. Further, in \cite[Theorem 9.1]{fefferman2012honeycomb}, the persistence of Dirac points against small $\mathcal{PC}$-invariant perturbations (which may break $2\pi/3$-rotational invariance) is proven.
\end{remark}

\noindent The proof in \cite[Theorem 9.1]{fefferman2012honeycomb} applies with minor modifications in the more general setting of Theorem \ref{thm:dir-srf} and we therefore omit the full details.

\begin{remark}
\label{rmk:dir-srf-eff}
{\rm (Effective Hamiltonian about ${(E_D, \, {\bm D})}$.)}
The dispersion surfaces of $L$ containing ${(E_D, \, {\bm D})}$ are locally approximated by those of an effective Hamiltonian $L^{\bm D}_{\rm eff}$ with Fourier symbol
\begin{equation}
\label{eq:def-Leff}
L^{\bm D}_{\rm eff}({\bm q}) \equiv ({\bm \gamma}_0 \cdot {\bm q}) I + ({\bm \gamma}_1 \cdot {\bm q}) \sigma_1 + ({\bm \gamma}_2 \cdot {\bm q}) \sigma_2.
\end{equation}
The two dispersion relations of $L^{\bm D}_{\rm eff}$ are given by
\begin{equation}
\label{eq:Leff-srf}
\varepsilon_\pm({\bm q}) = {\bm \gamma}_0 \cdot {\bm q} \pm \sqrt{({\bm \gamma}_1 \cdot {\bm q})^2 + ({\bm \gamma}_2 \cdot {\bm q})^2}. 
\end{equation}
\end{remark}

\noindent Generally, tuning the parameters ${\bm \gamma}_1$, ${\bm \gamma}_2$ deforms the circular cross sections of a right, circular cone into ellipses, while tuning the parameter ${\bm \gamma}_0$ tilts the axis.

\bigskip

\section{Deformed square lattice media}
\label{sec:deform}

\setcounter{equation}{0}
\setcounter{figure}{0}

\subsection{Deformed square lattice potentials and Schr\"{o}dinger operators}
\label{sec:deform-op}

Fix a square lattice potential $V$ in the sense of Definition \ref{def:sql-pot}, and consider the associated Schr\"{o}dinger operator ${H_V = - \Delta + V}$. Suppose $T$ is a real, invertible ${2 \times 2}$ matrix. Then, the {\it deformed square lattice potential} ${V \circ \, T^{-1}}$ models a linear deformation of the medium described by $V$. The {\it deformed Schr\"{o}dinger operator}
\begin{equation}
\label{eq:def-HT}
H_{V \circ \, T^{-1}} \equiv -\Delta + V \circ \, T^{-1}
\end{equation}
is periodic with respect to the lattice ${T \Z^2 = \Z T {\bm v}_1 \oplus \Z T {\bm v}_2}$ with fundamental cell $T \Omega$. The dual lattice is ${(T \Z^2)^* = \Z (T^{-1})^\mathsf{T} {\bm k}_1 \oplus \Z (T^{-1})^\mathsf{T} {\bm k}_2}$ with fundamental cell (Brillouin zone) $(T^{-1})^\mathsf{T} \mathcal{B}$, since
\begin{equation}
\label{eq:def-Tk}
(T^{-1})^\mathsf{T}{\bm k}_j \cdot T {\bm v}_k = {\bm k}_j \cdot T^{-1} T {\bm v}_k = {\bm k}_j \cdot {\bm v}_k = 2 \pi \delta_{jk}, \quad j, \, k \smallin \{ 1, \, 2 \}.
\end{equation}
Figure \ref{fig:deform-2} displays the fundamental cell and Brillouin zone for example deformations; see Section \ref{sec:examples}. 

The band structure of $H_{V \circ \, T^{-1}}$ is determined by the family of Floquet-Bloch eigenvalue problems:
\begin{equation}
\label{eq:fb-evp-HT}
H_{V \circ \, T^{-1}} \Phi = E \Phi, \quad \Phi \smallin L^2_{\bm k}, \quad {\bm k} \smallin (T^{-1})^\mathsf{T} \mathcal{B}.
\end{equation}
By a change of coordinates, we see that the band structure of $H_{V \circ \, T^{-1}}$ is obtainable from the band structure of a $\Z^2$-periodic operator $T_*H_V$, given in the following proposition, proved in Appendix \ref{sec:ch-var}:

\begin{proposition}
\label{prop:fb-evp-HT-TH}
Let ${{\bm k} \smallin \mathcal{B}}$. Then ${(\Phi \circ T^{-1}, \, E)}$ is an $L^2_{(T^{-1})^\mathsf{T}{\bm k}}(\R^2/T\Z^2)$ eigenpair of $H_{V \circ \, T^{-1}}$ if and only if ${(\Phi, \, E)}$ is an $L^2_{\bm k}(\R^2/\Z^2)$ eigenpair of 
\begin{equation}
\label{eq:def-TH}
T_* H_V \equiv -\nabla \cdot (T^\mathsf{T} T)^{-1} \nabla + V.
\end{equation}
\end{proposition}

\noindent {\bf N.B.} Since $(T^\mathsf{T} T)^{-1}$ is symmetric and positive definite, $T_* H_V$ is elliptic.  For the remainder of this article, we shall study the $\Z^2$-periodic ellipitic operator $T_* H_V$. Moreover, for notational simplicity, we shall frequently write $H$ for $H_V$ and $T_* H$ for $T_* H_V$. \\

We consider the scenario of small deformations (i.e, where ${I - (T^\mathsf{T} T)^{-1}}$ has small matrix norm), and therefore find it useful to express $T_* H$ as a perturbation of $H$:
\begin{equation}
\label{eq:def-TH_2}
T_* H = H + \nabla \cdot (I - (T^\mathsf{T} T)^{-1}) \nabla.
\end{equation}
Since the matrix norm of $I-(T^\mathsf{T} T)^{-1}$ is small, ${\nabla \cdot (I - (T^\mathsf{T} T)^{-1}) \nabla}$ is a bounded linear operator from ${H^2(\R^2) \to L^2(\R^2)}$ of small norm.

By Lemma \ref{lem:Ahermsym}, the real symmetric matrix ${I - (T^\mathsf{T} T)^{-1}}$  has an expansion in real Pauli matrices:
\begin{equation}
\label{eq:def-tau}
I - (T^\mathsf{T} T)^{-1} = \tau_0 I + \tau_1 \sigma_1 + \tau_3 \sigma_3.
\end{equation}
The coefficients $\tau_0$, $\tau_1$, and $\tau_3$ are expressible in terms of the action of $T$ on the square lattice vectors ${\bm v}_1$, ${\bm v}_2$:
\begin{equation}
\label{eq:tau-exp}
\tau_0 = 1 - \frac{1}{2} \frac{|T{\bm v}_1|^2 + |T{\bm v}_2|^2}{{\rm det}(T)^2}, \quad \tau_1 = \frac{T{\bm v}_1 \cdot T{\bm v}_2}{{\rm det}(T)^2}, \quad \tau_3 = \frac{1}{2} \frac{|T{\bm v}_1|^2 - |T{\bm v}_2|^2}{{\rm det}(T)^2}.
\end{equation}
Therefore, 
\begin{equation}
\label{eq:def-H-tau}
T_* H = H^{\tau_0, \tau_1, \tau_3} \equiv H + \tau_0 \Delta + \nabla \cdot (\tau_1 \sigma_1 + \tau_3 \sigma_3) \nabla.
\end{equation}
We use the compressed notation
\begin{equation}
\label{eq:def-H-tau_2}
H^{\tau_0, {\bm \tau}} \equiv H^{\tau_0, \tau_1, \tau_3}, \quad \text{where} \quad {\bm \tau} \equiv [\tau_1, \tau_3]^\mathsf{T}.
\end{equation}
As we shall see, variations with respect to $\tau_0$ alone (isotropic dilation)  preserve the quadratic degeneracy, while those with respect to $\bm \tau$ cause the the quadratic degeneracy to split. Hence, in our parameterization \eqref{eq:def-H-tau_2}, we  have separated out $\tau_0$ from $\bm\tau$. It is convenient to express ${\bm \tau}$ in polar coordinates as
\begin{align}
\label{eq:def-tau-polar}
& {\bm \tau} = |{\bm \tau}|  [\cos(\varphi), \, \sin(\varphi)]^\mathsf{T}, \quad
\text{where} \\
& |{\bm \tau}| \equiv \sqrt{\tau_1^2 + \tau_3^2} \geq 0, \quad \varphi \equiv {\rm Arg}(\tau_1 + i \tau_3) \smallin (-\pi, \, \pi].
\end{align}
Thus, 
\begin{equation}
\label{eq:def-H-phi}
H^{\tau_0, {\bm \tau}} = H + \tau_0 \Delta + |{\bm \tau}| \nabla \cdot \sigma_\varphi \nabla, \quad \text{where} \quad \sigma_\varphi = \cos(\varphi) \sigma_1 + \sin(\varphi) \sigma_3.
\end{equation}

\begin{remark}
\label{rmk:tau-pd}
Since $(T^\mathsf{T} T)^{-1}$ is positive definite, we must have that ${\tau_0 < 1}$ and ${|{\bm \tau}| < 1 - \tau_0}$. Indeed, the eigenvalues of $(T^\mathsf{T} T)^{-1}$, given by
\begin{equation}
\label{eq:tau-eig}
\lambda_\pm(\tau_0, {\bm \tau}) = (1 - \tau_0) \pm |{\bm \tau}|
\end{equation}
must both be positive.
\end{remark}

\begin{figure}[!t]
\centering
\SetTblrInner{hspan = even}
\begin{tblr}{hline{1} = {2-Z}{solid}, hline{2-Z} = {solid}, vline{1} = {2-Z}{solid}, vline{2-Z} = {solid}}
& $\mathcal{P}$ & $\mathcal{C}$ & $\mathcal{R}$ & $\Sigma_1$ & $\Sigma_3$ \\
{\bf a.} $\tau_1 = 0$, $\tau_3 = 0$ & $\checkmark$ & $\checkmark$ & $\checkmark$ & $\checkmark$ & $\checkmark$ \\
{\bf b.} $\tau_1 \neq 0$, $\tau_3 = 0$ & $\checkmark$ & $\checkmark$ & - & $\checkmark$ & - \\
{\bf c.} $\tau_1 = 0$, $\tau_3 \neq 0$ & $\checkmark$ & $\checkmark$ & - & - & $\checkmark$ \\
{\bf d.} $\tau_1 \neq 0$, $\tau_3 \neq 0$ & $\checkmark$ & $\checkmark$ & - & - & - \\
\end{tblr}
\vspace{0.3cm}
\caption{Linear deformations $T$ according to the ${(\tau_0, \, {\bm \tau})}$ parameterization and the symmetries each preserves. (Recall that $\Sigma_3 = \mathcal{R} \Sigma_1$ so that it is not possible for exactly two of these three symmetries to hold.) {\bf (a)} Orthogonal transformation composed with a dilation or contraction. {\bf (b)} Tilt-like deformation of the unit cell. {\bf (c)} Uniaxial dilation or contraction-like deformation. {\bf (d)} General linear deformation.}
\label{fig:deform-1}
\end{figure}

\subsection{Symmetries of deformed Schr\"{o}dinger operators}
\label{sec:deform-syms}

The operator $H^{\tau_0, {\bm \tau}}$ is $\Z^2$-translation invariant and $\mathcal{P}$-, $\mathcal{C}$-, and $\mathcal{PC}$-invariant for all ${(\tau_0, \, {\bm \tau})}$. However, deformations for which ${|{\bm \tau}| \neq 0}$ are not $\mathcal{R}$-invariant. In the special case ${|{\bm \tau}| = 0}$, we have ${T_*H = -(1-\tau_0)\Delta + V}$ (see \eqref{eq:def-tau}), which leads to a spectral problem for a new square lattice potential. A detailed summary of the symmetries of $H^{\tau_0, {\bm \tau}}$ is presented in the table of Figure \ref{fig:deform-1}.
 
\begin{figure}[p]
\centering
\begin{subfigure}{0.33\textwidth}
    \subcaption{$\ $}
    \vspace{0.1cm}
    \begin{tikzpicture}[x = 1.5cm, y = 1.5cm]
        \draw[white] (0, 0) rectangle (3.5, 3.5);
        \draw[black, thick, ->, opacity = 0.5] (0.25, 1.75) -- (3.25, 1.75) node[anchor = north] {$x_1$};
        \draw[black, thick, ->, opacity = 0.5] (1.75, 0.25) -- (1.75, 3.25) node[anchor = east] {$x_2 \,$};
        \draw[black, thick, dashed] (1, 1) rectangle (2.5, 2.5);
        \draw[red!20!blue, thick, opacity = 0.5] (0.75, 0.75) -- (2.75, 0.75) -- (2.75, 2.75) -- (0.75, 2.75) -- cycle;
        \fill[red!20!blue, opacity = 0.3] (0.75, 0.75) -- (2.75, 0.75) -- (2.75, 2.75) -- (0.75, 2.75) -- cycle;
    \end{tikzpicture}
\end{subfigure}
\hspace{0.2cm}
\begin{subfigure}{0.33\textwidth}
    \subcaption{$\ $}
    \vspace{0.1cm}
    \begin{tikzpicture}[x = 1.5cm, y = 1.5cm]
        \draw[white] (0, 0) rectangle (3.5, 3.5);
        \draw[black, thick, ->, opacity = 0.5] (0.25, 1.75) -- (3.25, 1.75) node[anchor = north] {$k_1$};
        \draw[black, thick, ->, opacity = 0.5] (1.75, 0.25) -- (1.75, 3.25) node[anchor = east] {$k_2 \,$};
        \draw[black, thick, dashed] (1, 1) rectangle (2.5, 2.5);
        \draw[red!60!yellow, thick, opacity = 0.5] (1.19, 1.19) -- (2.31, 1.19) -- (2.31, 2.31) -- (1.19, 2.31) -- cycle;
        \fill[red!60!yellow, opacity = 0.3] (1.19, 1.19) -- (2.31, 1.19) -- (2.31, 2.31) -- (1.19, 2.31) -- cycle;
    \end{tikzpicture}
\end{subfigure} \\
\vspace{0.2cm}
\begin{subfigure}{0.33\textwidth}
    \subcaption{$\ $}
    \vspace{0.1cm}
    \begin{tikzpicture}[x = 1.5cm, y = 1.5cm]
        \draw[white] (0, 0) rectangle (3.5, 3.5);
        \draw[black, thick, ->, opacity = 0.5] (0.25, 1.75) -- (3.25, 1.75) node[anchor = north] {$x_1$};
        \draw[black, thick, ->, opacity = 0.5] (1.75, 0.25) -- (1.75, 3.25) node[anchor = east] {$x_2 \,$};
        \draw[black, thick, dashed] (1, 1) rectangle (2.5, 2.5);
        \draw[red!20!blue, thick, opacity = 0.5] (0.87, 0.87) -- (2.34, 1.116) -- (2.63, 2.63) -- (1.16, 2.34) -- cycle;
        \fill[red!20!blue, opacity = 0.3] (0.87, 0.87) -- (2.34, 1.116) -- (2.63, 2.63) -- (1.16, 2.34) -- cycle;
    \end{tikzpicture}
\end{subfigure}
\hspace{0.2cm}
\begin{subfigure}{0.33\textwidth}
    \subcaption{$\ $}
    \vspace{0.1cm}
    \begin{tikzpicture}[x = 1.5cm, y = 1.5cm]
        \draw[white] (0, 0) rectangle (3.5, 3.5);
        \draw[black, thick, ->, opacity = 0.5] (0.25, 1.75) -- (3.25, 1.75) node[anchor = north] {$k_1$};
        \draw[black, thick, ->, opacity = 0.5] (1.75, 0.25) -- (1.75, 3.25) node[anchor = east] {$k_2 \,$};
        \draw[black, thick, dashed] (1, 1) rectangle (2.5, 2.5);
        \draw[red!60!yellow, thick, opacity = 0.5] (1.11, 1.11) -- (2.71, 0.79) -- (2.39, 2.39) -- (0.79, 2.71) -- cycle;
        \fill[red!60!yellow, opacity = 0.3] (1.11, 1.11) -- (2.71, 0.79) -- (2.39, 2.39) -- (0.79, 2.71) -- cycle;
    \end{tikzpicture}
\end{subfigure} \\
\vspace{0.2cm}
\begin{subfigure}{0.33\textwidth}
    \subcaption{$\ $}
    \vspace{0.1cm}
    \begin{tikzpicture}[x = 1.5cm, y = 1.5cm]
        \draw[white] (0, 0) rectangle (3.5, 3.5);
        \draw[black, thick, ->, opacity = 0.5] (0.25, 1.75) -- (3.25, 1.75) node[anchor = north] {$x_1$};
        \draw[black, thick, ->, opacity = 0.5] (1.75, 0.25) -- (1.75, 3.25) node[anchor = east] {$x_2 \,$};
        \draw[black, thick, dashed] (1, 1) rectangle (2.5, 2.5);
        \draw[red!20!blue, thick, opacity = 0.5] (1, 0.75) -- (2.5, 0.75) -- (2.5, 2.75) -- (1, 2.75) -- cycle;
        \fill[red!20!blue, opacity = 0.3] (1, 0.75) -- (2.5, 0.75) -- (2.5, 2.75) -- (1, 2.75) -- cycle;
    \end{tikzpicture}
\end{subfigure}
\hspace{0.2cm}
\begin{subfigure}{0.33\textwidth}
    \subcaption{$\ $}
    \vspace{0.1cm}
    \begin{tikzpicture}[x = 1.5cm, y = 1.5cm]
        \draw[white] (0, 0) rectangle (3.5, 3.5);
        \draw[black, thick, ->, opacity = 0.5] (0.25, 1.75) -- (3.25, 1.75) node[anchor = north] {$k_1$};
        \draw[black, thick, ->, opacity = 0.5] (1.75, 0.25) -- (1.75, 3.25) node[anchor = east] {$k_2 \,$};
        \draw[black, thick, dashed] (1, 1) rectangle (2.5, 2.5);
        \draw[red!60!yellow, thick, opacity = 0.5] (1, 1.19) -- (2.5, 1.19) -- (2.5, 2.31) -- (1, 2.31) -- cycle;
        \fill[red!60!yellow, opacity = 0.3] (1, 1.19) -- (2.5, 1.19) -- (2.5, 2.31) -- (1, 2.31) -- cycle;
    \end{tikzpicture}
\end{subfigure}
\vspace{0.3cm}
\caption{Deformed fundamental cells $T \Omega$ (left column, shaded purple) and deformed Brillouin zones $(T^{-1})^\mathsf{T} \mathcal{B}$ (right column, shaded orange) corresponding to the example deformations of Section \ref{sec:examples}. The dashed square indicates the boundary of the undeformed fundamental cell $\Omega$ (left column) and the undeformed Brillouin zone $\mathcal{B}$ (right column). {\bf (a, b)} An isotropic dilation with ${\zeta = 1.33}$; see Example \ref{ex:volum}. {\bf (c, d)} A tilt deformation with ${\xi = 0.20}$ (radians); see Example \ref{ex:tilt}. {\bf (e, f)} A unixaial dilation (along ${x_1 = 0}$) with ${\eta = 1.33}$; see Example \ref{ex:unax}.}
\label{fig:deform-2}
\end{figure}

Our analytical results concern, for ${(\tau_0, \, {\bm \tau})}$ in a neighborhood of ${(0, \, {\bm 0})}$, the Floquet-Bloch eigenvalue problems for the deformed Schr\"{o}dinger operator $H^{\tau_0, {\bm \tau}}$:
\begin{equation}
H^{\tau_0, {\bm \tau}} \Phi = E \Phi, \quad \Phi \smallin L^2_{\bm k}, \quad {\bm k} \smallin \mathcal{B}.
\end{equation}
We conclude this section with an assertion regarding a symmetry of the band structure of $H_V$, for square lattice potentials, with the respect to the high symmetry quasimomentum $\bm M$.

\begin{proposition}
\label{prop:gl-P-sym}
Let $V$ denote a square lattice potential. Then, if ${(E, \, \Phi)}$ is a $L^2_{{\bm M} + {\bm \kappa}}$ eigenpair of $H^{\tau_0, {\bm \tau}}$, both $(E, \, \mathcal{P}[\Phi])$ and $(E, \, \mathcal{C}[\Phi])$ are $L^2_{{\bm M} - {\bm \kappa}}$ eigenpairs of $H^{\tau_0, {\bm \tau}}$.
\end{proposition}

\begin{proof}
Note that if $H^{\tau_0, {\bm \tau}}\Phi=E\Phi$, then both $\mathcal{P}[\Phi]$ and 
$\mathcal{C}[\Phi]$ satisfy the same equation.
We investigate the pseudoperiodic boundary condition satisfied by these functions.
 Fix ${\bm \kappa}$. If  $\Phi\in L^2_{{\bm M} + {\bm \kappa}}$, then  note that, for any ${\bm v} \smallin \Z^2$,
\begin{equation}
\mathcal{P}[\Phi]({\bm x} + {\bm v}) = \Phi(-{\bm x} - {\bm v}) = e^{i({\bm M} + {\bm \kappa}) \cdot (-{\bm v})} \Phi(-{\bm x}) = e^{i({\bm M} - {\bm \kappa})\cdot{\bm v}}\mathcal{P}[\Phi]({\bm x})
\end{equation}
since ${{\bm M} = [\pi, \, \pi]^\mathsf{T}}$ implies ${e^{-i {\bm M} \cdot {\bm v}} = e^{i {\bm M} \cdot {\bm v}}}$. Hence, ${(\mathcal{P}[\Phi], \, E)}$ is an $L^2_{{\bm M} - {\bm \kappa}}$ eigenpair of $H^{\tau_0, {\bm \tau}}$. An analogous calculation shows that ${\mathcal{C}[\Phi] \smallin L^2_{{\bm M} - {\bm \kappa}}}$.
\end{proof}

\subsection{Examples}
\label{sec:examples}

In this section, we present three examples of deformations; see also Figure \ref{fig:deform-2}.

\begin{example}[Isotropic dilation or contraction]
\label{ex:volum}
For ${\zeta > 0}$, consider:
\begin{equation}
\label{eq:def-volum}
T(\zeta) =
\begin{bmatrix}
\zeta & 0 \\
0 & \zeta
\end{bmatrix} \! .
\end{equation}
This one-parameter class of deformations corresponds to deformed Schr\"{o}dinger operators with
\begin{equation}
\label{eq:volum-par}
\tau_0(\zeta) = -\frac{1}{\zeta^2}, \quad \tau_1(\zeta) = 0, \quad \tau_3(\zeta) = 0.
\end{equation}
\end{example}

\begin{example}[Tilt]
\label{ex:tilt}
For ${\xi \smallin (-\pi/4, \, \pi/4)}$, consider:
\begin{equation}
\label{eq:def-tilt}
T(\xi) =
\begin{bmatrix}
\cos(\xi) & \sin(\xi) \\
\sin(\xi) & \cos(\xi)
\end{bmatrix} \! .
\end{equation}
In this case,
\begin{equation}
\label{eq:tilt-par}
\tau_0(\xi) = -\tan^2(2\xi), \quad \tau_1(\xi) = \sec(2\xi) \tan(2\xi), \quad \tau_3(\xi) = 0.
\end{equation}
\end{example}

\begin{example}[Uniaxial dilation or contraction]
\label{ex:unax}
For ${\eta \smallin (0, \, 1)}$, consider:
\begin{equation}
\label{eq:def-unax}
T(\eta) =
\begin{bmatrix}
1 & 0 \\
0 & \eta
\end{bmatrix} \! .
\end{equation}
Here we have
\begin{equation}
\label{eq:unax-par}
\tau_0(\eta) = \frac{1}{2} - \frac{1}{2 \eta^2}, \quad \tau_1(\eta) = 0, \quad \tau_3(\eta) = \frac{1}{2 \eta^2} - \frac{1}{2}.
\end{equation}
\end{example}

\section{Main results}
\label{sec:results}

\setcounter{equation}{0}
\setcounter{figure}{0}

In this section, we state our main analytical results concerning the band structure of the deformed Schr\"{o}dinger operator $H_{V \circ \, T^{-1}}$. By \eqref{eq:def-tau-polar} -- \eqref{eq:def-H-phi}, this is reduced to the study of the $\Z^2$-periodic operator
\begin{equation}
\label{eq:def-H-phi_2}
H^{\tau_0, {\bm \tau}} = H + \tau_0 \Delta + |{\bm \tau}| \nabla \cdot \sigma_\varphi \nabla, \quad \text{where} \quad \sigma_\varphi = \cos(\varphi) \sigma_1 + \sin(\varphi) \sigma_3.
\end{equation}
Thus, we study the family of Floquet-Bloch eigenvalue problems:
\begin{equation}
\label{eq:fb-evp_1}
H^{\tau_0, {\bm \tau}} \, \Phi = E \, \Phi, \quad \Phi \smallin L^2_{\bm k}, \quad \text{where} \quad {\bm k} \smallin \mathcal{B}=[-\pi, \, \pi] \times [-\pi, \, \pi].
\end{equation}

Recall that ${{\bm M} = [\pi, \, \pi]^\mathsf{T}}$ is a high-symmetry quasimomentum of ${H^{0, {\bm 0}} = H}$, and suppose ${(E_S, \, {\bm M})}$ is a quadratic band degeneracy point with eigenspace ${{\rm ker}_{L^2_{\bm M}}(H - E_S)}$ spanned by normalized eigenstates $\{ \Phi_1, \, \Phi_2 \}$, such that hypotheses \ref{itm:quad-dgn-1} -- \ref{itm:quad-dgn-5} are satisfied. To state our main results, we require an additional nondegeneracy hypothesis on the quadratic band degeneracy beyond conditions \ref{itm:quad-dgn-1} -- \ref{itm:quad-dgn-5}, which ensures that the effect of the deformation enters at leading order. Define the parameters
\begin{equation}
\label{eq:def-b-par}
\begin{aligned}
\beta_0 & \equiv - \langle \partial_{x_1} \Phi_1, \, \partial_{x_1} \Phi_1 \rangle = - \lVert \partial_{x_1} \Phi_1 \rVert^2 \smallin \R, \\
\beta_1 & \equiv - \langle \partial_{x_1} \Phi_1, \, \partial_{x_2} \Phi_2 \rangle \smallin \R, \\
\beta_2 & \equiv - i \langle \partial_{x_1} \Phi_1, \, \partial_{x_1} \Phi_2 \rangle \smallin \R.
\end{aligned}
\end{equation}

\noindent In addition to \ref{itm:quad-dgn-1} -- \ref{itm:quad-dgn-5} of Section \ref{sec:kmow}, we assume:

\begin{enumerate}
\renewcommand{\theenumi}{Q\arabic{enumi}}
\setcounter{enumi}{5}
\item \label{itm:quad-dgn-6} (Nondegeneracy condition.) Either ${\beta_1 \neq 0}$ or ${\beta_2 \neq 0}$.
\end{enumerate}

\begin{remark}
\label{rmk:b-par-small}
As noted in Remark \ref{rmk:a-par-small}, the nondegeneracy condition \ref{itm:quad-dgn-5} has been verified for generic small amplitude potentials; see \cite[Appendix C]{keller2018spectral}. In Appendix \ref{apx:b-par-small}, we present analogous computations which verify that hypothesis \ref{itm:quad-dgn-6} holds for small amplitude potentials; see Proposition \ref{prop:b-par-small}. In addition, we provide an example of a potential for which all parameters are nonzero; see Example \ref{ex:a-b-par-small}.
\end{remark}

\begin{remark}
\label{rmk:quad-dgn-6-fail}
If the hypothesis \ref{itm:quad-dgn-6} fails (i.e. ${\beta_1 = 0}$ and ${\beta_2 = 0}$), then we expect the splitting to occur at higher order.
\end{remark}

\subsection{Dispersion surfaces near ${(E_S, \, {\bm M})}$ under small deformation}
\label{sec:M-srf}

Our first result  describes how two dispersion surfaces of ${H^{0, {\bm 0}} = H}$, assumed to contain a quadratic band degeneracy point ${(E_S, \, {\bm M})}$, locally perturb under small deformation; see Figure \ref{fig:intro-1}.

\begin{theorem}
\label{thm:M-srf}
{\rm (Dispersion surfaces near $(E_S, \, {\bm M})$ under small deformation.)}
Suppose ${(E_S, \, {\bm M})}$ satisfies conditions \ref{itm:quad-dgn-1} -- \ref{itm:quad-dgn-5}, and is therefore a quadratic band degeneracy point of ${H^{0, {\bm 0}} = H}$; see Theorem \ref{thm:kmow}. Assume the additional nondegeneracy condition \ref{itm:quad-dgn-6}. Then, there exist $\kappa^\star$, $\tau_0^\star$, ${\tau^\star > 0}$ such that the dispersion surfaces of $H^{\tau_0, {\bm \tau}}$ near ${(E_S, {\bm M})}$ are described, for ${|{\bm \kappa}| < \kappa^\star}$, ${|\tau_0| < \tau_0^\star}$, and ${|{\bm \tau}| < \tau^\star}$, by:
\begin{align}
\label{eq:M-srf}
E_\pm({\bm M} + {\bm \kappa}; \tau_0, {\bm \tau}) - E_S & = \beta_0 \tau_0 + (1 - \alpha_0) |{\bm \kappa}|^2 + Q_4({\bm \kappa}; \tau_0, {\bm \tau}) \\
& \quad \pm \sqrt{(\beta_1 \tau_1 - \alpha_1 {\bm \kappa} \cdot \sigma_1 {\bm \kappa})^2 + (\beta_2 \tau_3 - \alpha_2 {\bm \kappa} \cdot \sigma_3 {\bm \kappa})^2 + \mathcal{Q}_6({\bm \kappa}; \tau_0, {\bm \tau})}. \nonumber
\end{align}
The functions $\mathcal{Q}_n({\bm \kappa}; \tau_0, {\bm \tau})$, ${n = 4}$, $6$, are analytic in a neighborhood of ${({\bm \kappa}; \tau_0, {\bm \tau}) = ({\bm 0}; 0, {\bm 0})}$ and satisfy:
\begin{align}
\label{eq:M-srf-Q-sym}
\mathcal{Q}_n({\bm \kappa}; \tau_0, {\bm \tau}) & = \mathcal{Q}_n(-{\bm \kappa}; \tau_0, {\bm \tau}), \\
\label{eq:M-srf-Q-bd}
\mathcal{Q}_n({\bm \kappa}; \tau_0, {\bm \tau}) & = O(|{\bm \kappa}|^n + |\tau_0|^{n/2} + |{\bm \tau}|^{n/2}) \ \ \text{as} \ \ |{\bm \kappa}|, \, |\tau_0|, \, |{\bm \tau}| \to 0.
\end{align}
\end{theorem}

\noindent Note that the expression \eqref{eq:M-srf} reduces to \eqref{eq:kmow} for the case where the linear deformation $T$ is an orthogonal transformation, corresponding to ${(\tau_0, \, {\bm \tau}) = (0, \, {\bm 0})}$.

\begin{remark}
\label{rmk:M-srf-eff}
{\rm (Effective Hamiltonian about ${(E_S, \, {\bm M})}$ under small deformation.)}
The dispersion surfaces of $H^{\tau_0, {\bm \tau}}$ containing ${(E_S, \, {\bm M})}$ are locally approximated by those of an effective Hamiltonian $H^{\bm M}_{\rm eff}(\tau_0, {\bm \tau})$ with Fourier symbol
\begin{equation}
\label{eq:M-eff}
H^{\bm M}_{\rm eff}({\bm \kappa}; \tau_0, {\bm \tau}) \equiv \bigl( \beta_0 \tau_0 + (1 - \alpha_0)|{\bm \kappa}|^2 \bigr) \, I + (\beta_1 \tau_1 - \alpha_1 {\bm \kappa} \cdot \sigma_1 {\bm \kappa}) \, \sigma_1 + (\beta_2 \tau_3 - \alpha_2 {\bm \kappa} \cdot \sigma_3 {\bm \kappa}) \, \sigma_2.
\end{equation}
The two dispersion relations of $H^{\bm M}_{\rm eff}(\tau_0, {\bm \tau})$ are given by
\begin{equation}
\label{eq:M-srf-eff}
\varepsilon_\pm({\bm \kappa}; \tau_0, {\bm \tau}) = \beta_0 \tau_0 + (1 - \alpha_0) |{\bm \kappa}|^2 \pm \sqrt{(\beta_1 \tau_1 - \alpha_1 {\bm \kappa} \cdot \sigma_1 {\bm \kappa})^2 + (\beta_2 \tau_3 - \alpha_2 {\bm \kappa} \cdot \sigma_3 {\bm \kappa})^2}. 
\end{equation}
These coincide with the expressions \eqref{eq:M-srf} of Theorem \ref{thm:M-srf}, omitting higher order terms.
\end{remark}

\subsection{Quadratic degeneracies split into Dirac points}
\label{sec:M-dgn}

We first observe that the approximate effective Hamiltonian \eqref{eq:M-eff} anticipates that, for ${|{\bm \tau}| \neq 0}$, a pair of twofold degeneracies perturb from the quadratic band degeneracy point ${(E_S, \, {\bm M})}$ of ${H^{0, {\bm 0}} = H}$:

A degenerate energy-quasimomentum pair ${(\varepsilon, \, {\bm \kappa}) = (\varepsilon_\star, \, {\bm \kappa}_\star)}$ occurs at ${{\bm \kappa} = {\bm \kappa}_\star}$ for which ${\varepsilon_+({\bm \kappa}_\star) = \varepsilon_-({\bm \kappa}_\star)}$. Thus, from \eqref{eq:M-srf-eff}, ${{\bm \kappa} = {\bm \kappa}_\star}$ satisfies the system of two equations:
\begin{equation}
\label{eq:M-eff-dgn}
\begin{aligned}
\beta_1 \tau_1 - \alpha_1 {\bm \kappa} \cdot \sigma_1 {\bm \kappa} & = 0, \\
\beta_2 \tau_3 - \alpha_2 {\bm \kappa} \cdot \sigma_3 {\bm \kappa} & = 0.
\end{aligned}
\end{equation}
Given any solution of \eqref{eq:M-eff-dgn}, the corresponding degenerate energy is given by
\begin{equation}
\varepsilon_\star = \beta_0 \tau_0 + (1 - \alpha_0) |{\bm \kappa}_\star|^2.
\end{equation}
The system \eqref{eq:M-eff-dgn} has two families of solutions. In polar coordinates ${\tau_1 = |{\bm \tau}| \cos(\varphi)}$ and ${\tau_3 = |{\bm \tau}| \sin(\varphi)}$:
\begin{enumerate}
\item If ${|{\bm \tau}| = 0}$, then
\begin{equation}
\label{eq:M-eff-dgn-sol-1}
(\varepsilon_\star, \, {\bm \kappa}_\star) = (\beta_0 \tau_0, \, {\bm 0}),
\end{equation}
corresponding to a persistent quadratic band degeneracy at $\bm M$ with perturbed, multiplicity two $L^2_{\bm M}$ eigenvalue ${E_M(\tau_0) = E_S + \beta_0 \tau_0}$.
\item If ${|{\bm \tau}| \neq 0}$, then 
\begin{equation}
\label{eq:M-eff-dgn-sol-2}
(\varepsilon_\star, \, {\bm \kappa}_\star) = (\beta_0 \tau_0 + (1 - \alpha_0) |{\bm \kappa}_D^{(0)}(\varphi)|^2 |{\bm \tau}|, \, \pm \sqrt{|{\bm \tau}|} {\bm \kappa}_D^{(0)}(\varphi)),
\end{equation}
where ${\bm \kappa}_D^{(0)}(\varphi)$ is defined as:
\begin{align}
\label{eq:def-kap-D-0-1}
{\bm \kappa}_D^{(0)}(\varphi) & \equiv |{\bm \kappa}_D^{(0)}(\varphi)| \, [\cos(\theta^{(0)}_D(\varphi)), \, \sin(\theta^{(0)}_D(\varphi))]^\mathsf{T}, \quad \text{where} \\
\label{eq:def-kap-D-0-2}
|{\bm \kappa}_D^{(0)}(\varphi)| & \equiv \sqrt[4]{\Bigl( \frac{\beta_2}{\alpha_2} \sin(\varphi) \Bigr)^2 + \Bigl( \frac{\beta_1}{\alpha_1} \cos(\varphi) \Bigr)^2}, \\ 
\label{eq:def-kap-D-0-3}
\theta^{(0)}_D(\varphi) & \equiv \frac{1}{2} {\rm Arg} \Bigl( \frac{\beta_2}{\alpha_2} \sin(\varphi) + i \frac{\beta_1}{\alpha_1} \cos(\varphi) \Bigr).
\end{align}
Note that the expressions \eqref{eq:def-kap-D-0-1} -- \eqref{eq:def-kap-D-0-3} are defined by hypothesis \ref{itm:quad-dgn-5} and \ref{itm:quad-dgn-6}.
\end{enumerate}

\noindent That these distinct twofold degeneracies of $H^{\bm M}_{\rm eff}(\tau_0, {\bm \tau})$ correspond to true degeneracies of $H^{\tau_0, {\bm \tau}}$ (and hence of $H_{V \circ \, T^{-1}}$), requires a much more refined analysis, and is the content of the following theorem.

\begin{theorem}
\label{thm:M-dgn}
{\rm (Quadratic degeneracies split into Dirac points.)}
Let ${(E_S, \, {\bm M})}$ denote a quadratic band degeneracy point of ${H^{0, {\bm 0}} = H}$ which satisfies \ref{itm:quad-dgn-1} -- \ref{itm:quad-dgn-5} (see Theorem \ref{thm:kmow}) and the further nondegeneracy condition \ref{itm:quad-dgn-6}. Then, the following hold:
\begin{enumerate}
\item \label{itm:M-dgn-1} (Pure dilation.) Let ${|{\bm \tau}| = 0}$. Then, there exists ${\tau_0^{\star \star} > 0}$ and an analytic function $E_M(\tau_0)$, defined for ${|\tau_0| < \tau_0^{\star \star}}$, such that ${(E_M(\tau_0), \, {\bm M})}$ is a twofold degenerate energy-quasimomentum pair of $H^{\tau_0, {\bm 0}}$. Here, 
\begin{align}
\label{eq:def-EM}
E_M(\tau_0) & \equiv E_S + \beta_0 \tau_0 + \varepsilon_M^{(1)}(\tau_0), \quad \text{where}\quad \varepsilon_M^{(1)}(\tau_0)  = O(|\tau_0|^2) \ \ \text{as} \ \ |\tau_0| \to 0.
\end{align}
In this case, ${(E_M(\tau_0), \, {\bm M})}$ satisfies conditions \ref{itm:quad-dgn-1} -- \ref{itm:quad-dgn-5} and is therefore a quadratic band degeneracy of $H^{\tau_0, {\bm 0}}$ in the sense of Theorem \ref{thm:kmow}.

\item \label{itm:M-dgn-2} Let ${|{\bm \tau}| \neq 0}$. Then, there exist $\tau_0^{\star \star}$, ${\tau^{\star \star} > 0}$ and analytic functions $E_D(\tau_0, s; \varphi)$, ${\bm D}^\pm(\tau_0, s; \varphi)$, defined for ${|\tau_0| < \tau_0^{\star \star}}$, ${|s| < \tau^{\star \star}}$, such that ${(E_D(\tau_0, |{\bm \tau}|; \varphi), \, {\bm D}^\pm(\tau_0, |{\bm \tau}|; \varphi))}$ are each twofold degenerate energy-quasimomentum pairs of $H^{\tau_0, {\bm \tau}}$. Further, we have the expansions
\begin{align}
\label{eq:def-ED}
E_D(\tau_0, |{\bm \tau}|; \varphi) & \equiv E_S + \beta_0 \tau_0 + (1 - \alpha_0) |{\bm \kappa}_D^{(0)}(\varphi)| |{\bm \tau}| + \varepsilon_D^{(1)}(\tau_0, |{\bm \tau}|; \varphi), \quad \text{where} \\
\label{eq:epsD-1-bd}
\varepsilon_D^{(1)}(\tau_0, |{\bm \tau}|; \varphi) & = O(|\tau_0|^2 + |{\bm \tau}|^2) \ \ \text{as} \ \ |\tau_0|, \, |{\bm \tau}| \to 0,
\end{align}
and
\begin{align}
\label{eq:def-Dpm}
{\bm D}^\pm(\tau_0, |{\bm \tau}|; \varphi) & \equiv {\bm M} \pm \sqrt{|{\bm \tau}|} \bigl( {\bm \kappa}_D^{(0)}(\varphi) + {\bm \kappa}_D^{(1)}(\tau_0, |{\bm \tau}|; \varphi) \bigr), \quad \text{where} \\
\label{eq:kapD-1-bd}
|{\bm \kappa}_D^{(1)}(\tau_0, |{\bm \tau}|; \varphi)| & = O(|\tau_0| + |{\bm \tau}|) \ \ \text{as} \ \ |\tau_0|, \, |{\bm \tau}| \to 0.
\end{align}
Here, ${\bm \kappa}_D^{(0)}(\varphi)$ is defined in \eqref{eq:def-kap-D-0-1}. Each of ${(E_D(\tau_0, |{\bm \tau}|; \varphi), \, {\bm D}^\pm(\tau_0, |{\bm \tau}|; \varphi))}$ satisfy conditions \ref{itm:dir-pt-1} -- \ref{itm:dir-pt-4} and are therefore (tilted) Dirac points of $H^{\tau_0, {\bm \tau}}$ in the sense of Theorem \ref{thm:dir-srf}.
\end{enumerate}
\end{theorem}

\noindent Note that ${\bm D}^-(\tau_0; |{\bm \tau}|)$ is the inversion with respect to ${\bm M}$ of ${\bm D}^+(\tau_0; |{\bm \tau}|)$: ${{\bm D}^-(\tau_0, {\bm \tau}) = 2 {\bm M} - {\bm D}^+(\tau_0, {\bm \tau})}$.

Conclusion \ref{itm:M-dgn-2} of Theorem \ref{thm:M-dgn} shows the bifurcation of two distinct Dirac points from a quadratic band degeneracy point in the scenario of small, non-dilational deformations. The locally tilted conical behavior of these is discussed in Theorem \ref{thm:dir-srf}. We note that this description applies to any Dirac points of $T_* H$, not just those emerging from small deformations; see Section \ref{sec:large} for a discussion of the analysis of large deformations.

Proposition \ref{prop:dir-pair} considers the scenario of Dirac points related by symmetry, and presents a resulting relationship between corresponding effective Hamiltonians.

\subsection{Special deformations}
\label{sec:special-def}

In this section, we remark on two particular classes of deformations for which the movement of Dirac points is constrained to be along a line. These are deformations where only one of the parameters $\tau_1$ or $\tau_3$ is nonzero. Throughout this section, we assume the hypotheses of Theorem \ref{thm:M-dgn}.

\begin{theorem}
\label{thm:M-dgn-ref-1}
{\rm (${\tau_1 \neq 0}$, ${\tau_3 = 0}$; Figure \ref{fig:M-dgn-ref}, panels (\subref{fig:M-dgn-ref-1a}, \subref{fig:M-dgn-ref-1b}).)}
Suppose ${\tau_1 \neq 0}$, ${\tau_3 = 0}$. Then, there exist $\tau_0^{\star \star}$, ${\tau_1^{\star\star} > 0}$ and analytic functions $E_D(\tau_0, \tau_1)$, ${\bm D}^\pm(\tau_0, \tau_1)$, defined for ${|\tau_0| < \tau_0^{\star \star}}$, ${|\tau_1| < \tau_1^{\star \star}}$, such that ${(E_D(\tau_0, \tau_1), \, {\bm D}^\pm(\tau_0, \tau_1))}$ are each twofold degenerate energy-quasimomentum pairs of $H^{\tau_0, \, (\tau_1, \, 0)}$. We have the expansions
\begin{align}
E_D(\tau_0, \tau_1) & = E_S + \beta_0 \tau_0 + (1 - \alpha_0) \tilde{\kappa}_D^{(0)}{}^2 |\tau_1| + \varepsilon_D^{(1)}(\tau_0, \tau_1), \quad \text{where} \\
\varepsilon_D^{(1)}(\tau_0, \tau_1) & = O(|\tau_0|^2 + |\tau_1|^2) \ \ \text{as} \ \ |\tau_0|, \, |\tau_1| \to 0,
\end{align}
and
\begin{align}
{\bm D}^\pm(\tau_0, \tau_1) & = {\bm M} \pm \sqrt{|\tau_1|} (\tilde{\kappa}_D^{(0)} + \tilde{\kappa}_D^{(1)}(\tau_0, \tau_1)) \, {\bm e}_1, \quad \text{where} \\
\tilde{\kappa}_D^{(1)}(\tau_0, \tau_1) & = O(|\tau_0| + |\tau_1|) \ \ \text{as} \ \ |\tau_0|, \, |\tau_1| \to 0.
\end{align}
Here,
\begin{equation}
\tilde{\kappa}_D^{(0)} = \sqrt{\Bigl|\frac{\beta_1}{\alpha_1}\Bigr|}
\end{equation}
and ${\bm e}_1$ is the normalized eigenvector satisfying:
\begin{equation}
\sigma_1 {\bm e}_1 = {\rm sgn}(\beta_1 \alpha_1 \tau_1) {\bm e}_1.
\end{equation}
In particular, the degenerate quasimomenta ${\bm D}^\pm(\tau_0, \tau_1)$ are constrained to one of the two diagonals bisecting ${\bm M}$: either $\kappa_1 = \kappa_2$ or $\kappa_1 = -\kappa_2$.
\end{theorem}

\begin{figure}[t]
\centering
\begin{subfigure}{0.35\textwidth}
    \subcaption{$\quad$}
    \vspace{0.1cm}
    \includegraphics[height = 4.3cm]{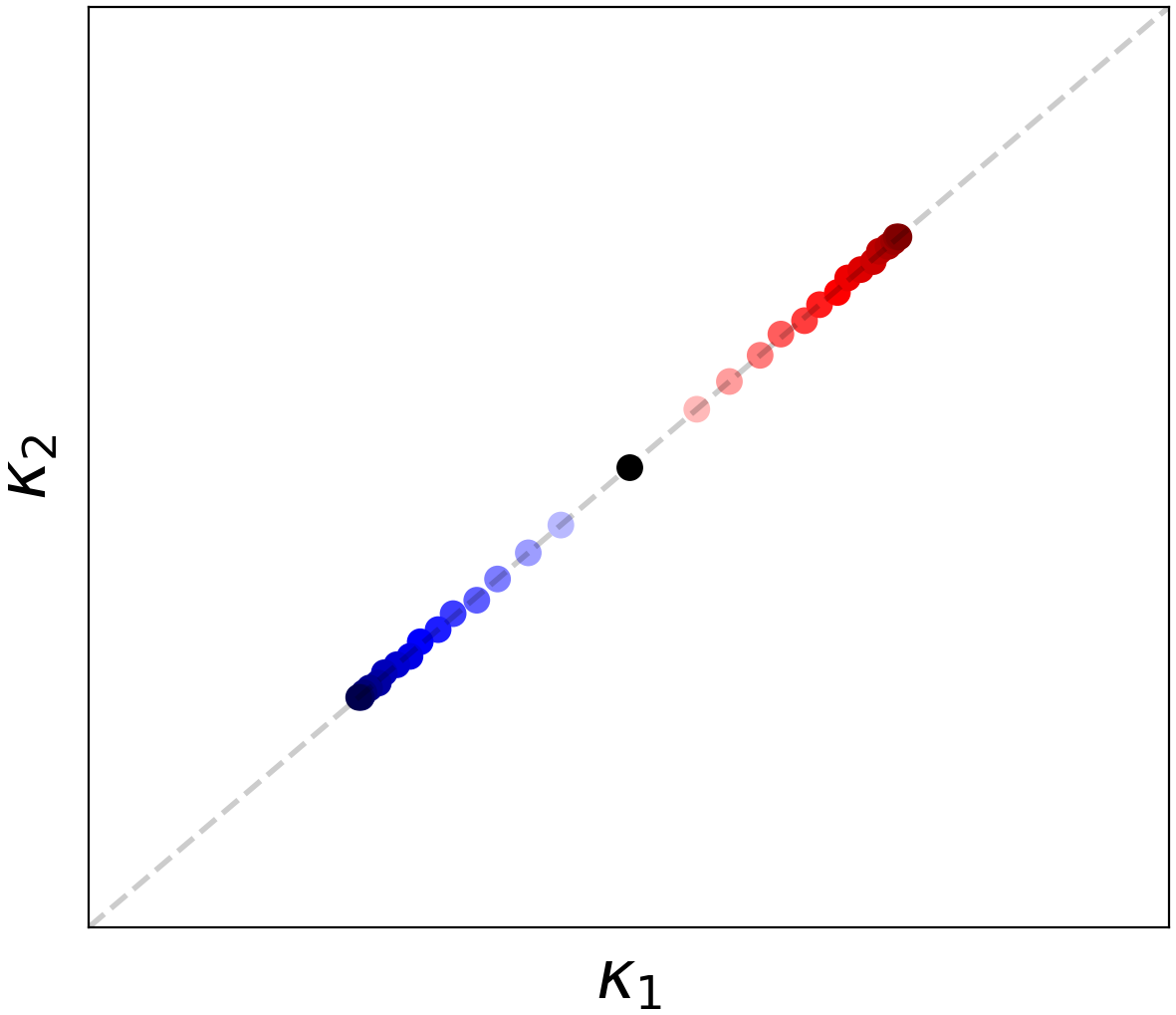}
    \label{fig:M-dgn-ref-1a}
\end{subfigure}
\hspace{0.2cm}
\begin{subfigure}{0.35\textwidth}
    \subcaption{$\quad$}
    \vspace{0.1cm}
    \includegraphics[height = 4.3cm]{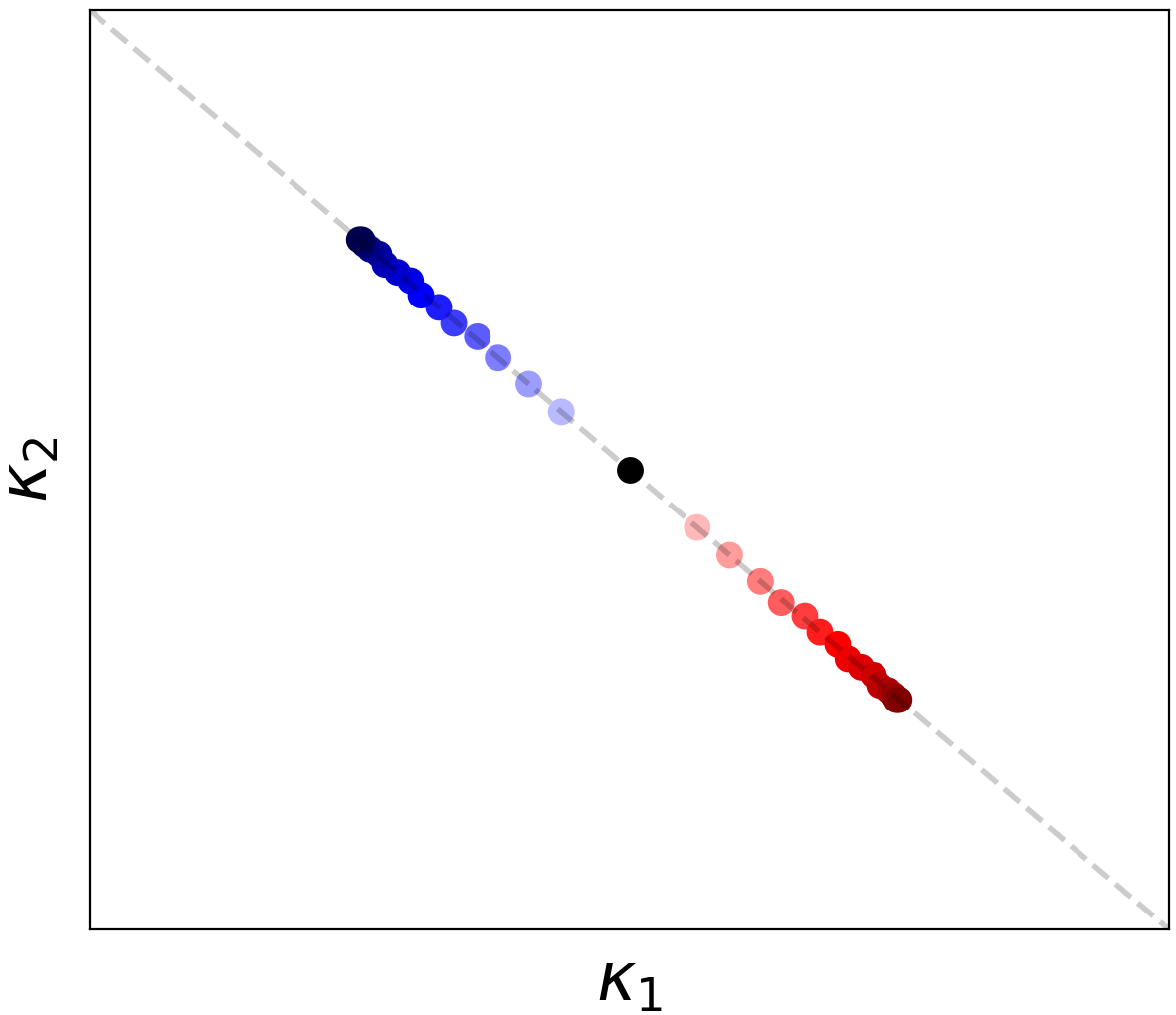}
    \label{fig:M-dgn-ref-1b}
\end{subfigure}  \\
\vspace{0.2cm}
\begin{subfigure}{0.35\textwidth}
    \subcaption{$\quad$}
    \vspace{0.1cm}
    \includegraphics[height = 4.3cm]{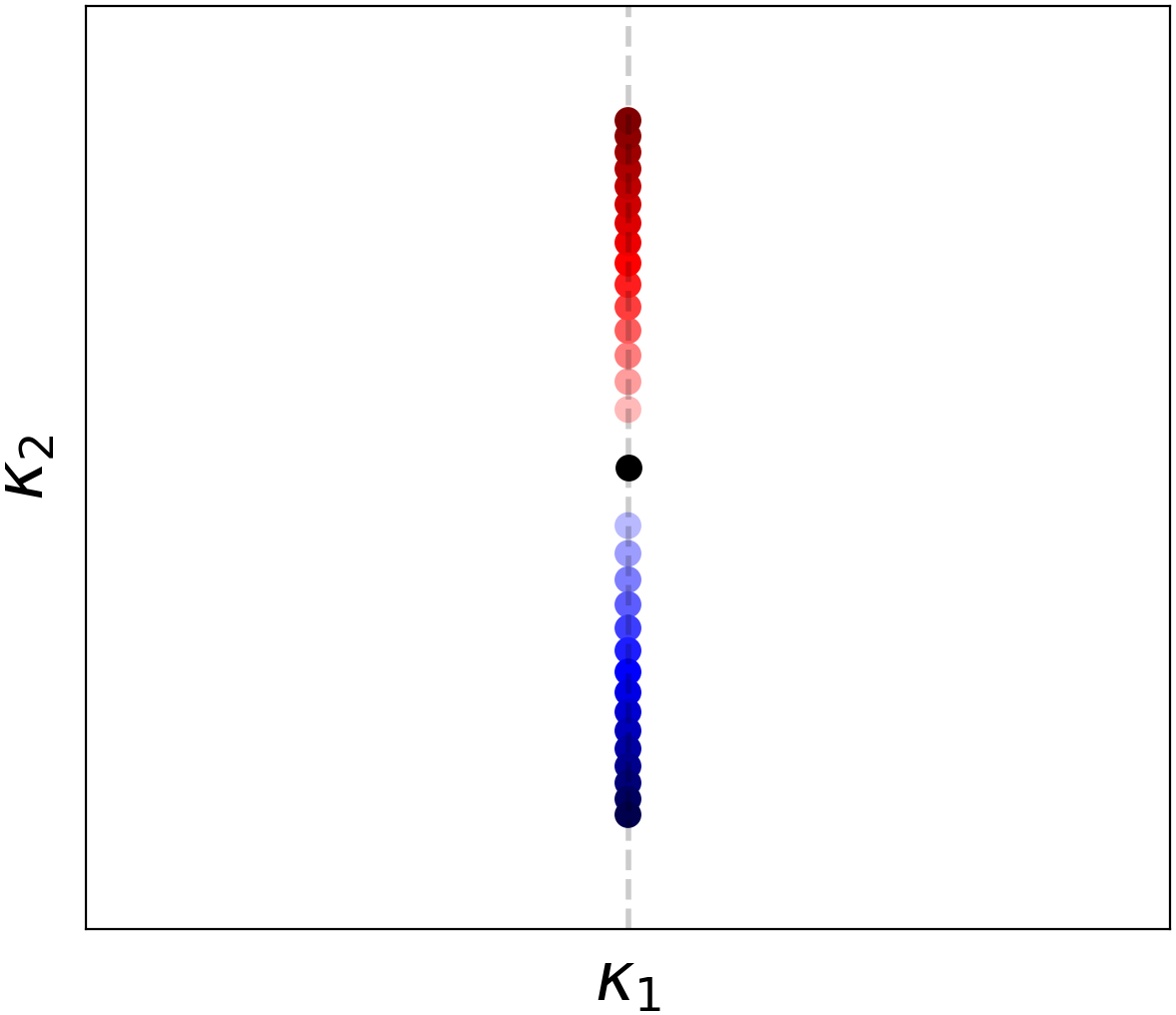}
    \label{fig:M-dgn-ref-2a}
\end{subfigure}
\hspace{0.2cm}
\begin{subfigure}{0.35\textwidth}
    \subcaption{$\quad$}
    \vspace{0.1cm}
    \includegraphics[height = 4.3cm]{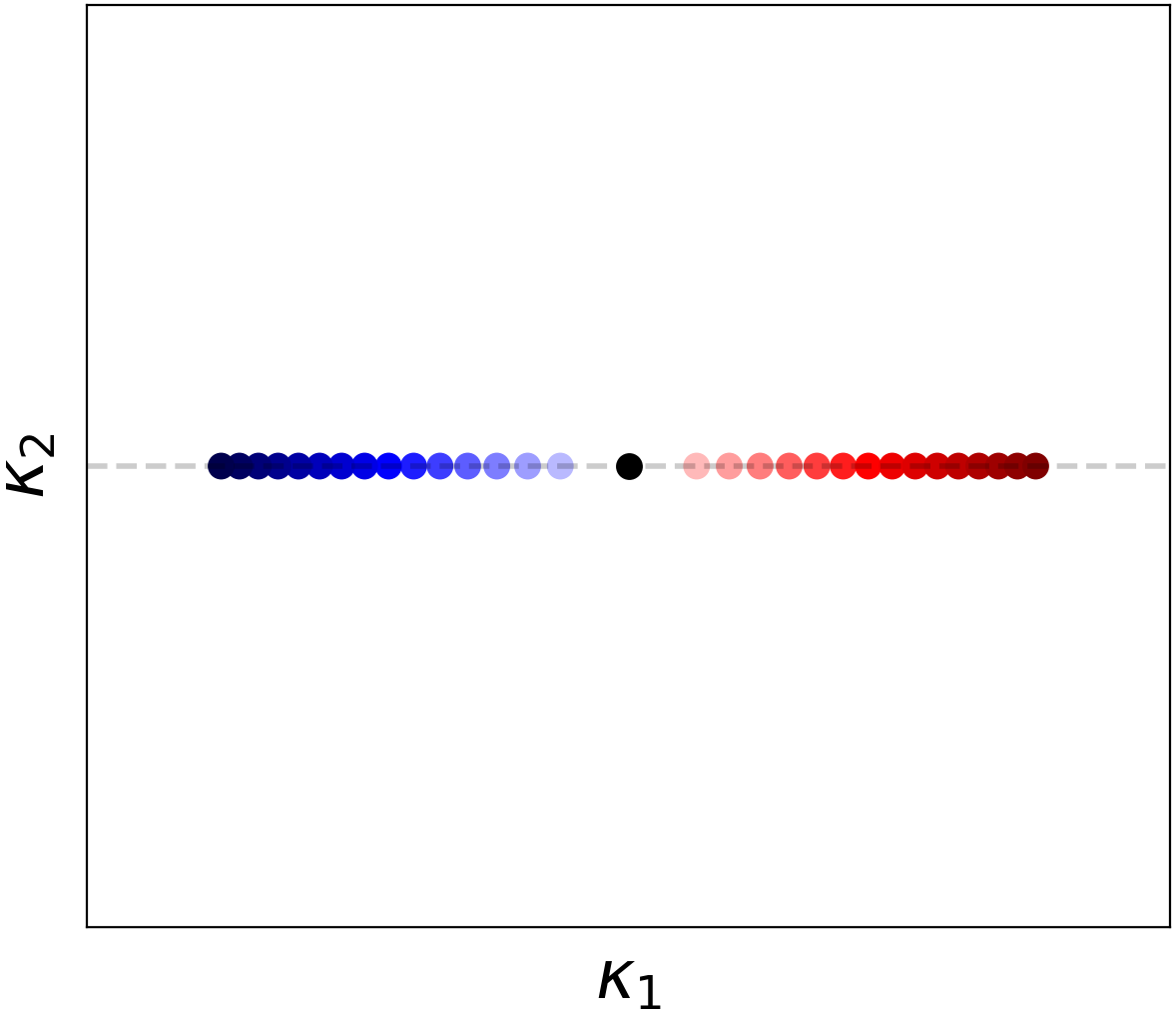}
    \label{fig:M-dgn-ref-2b}
\end{subfigure}
\vspace{0.3cm}
\caption{Trajectories of quasimomenta ${\bm D}^+(\tau_0, {\bm \tau})$ (red) and ${\bm D}^-(\tau_0, {\bm \tau})$ (blue) associated with Dirac points emerging from special (reflection symmetry-preserving) deformations. The top row illustrates Theorem \ref{thm:M-dgn-ref-1} with ${\tau_3 = 0}$ and {\bf (a)} ${{\rm sgn}(\alpha_1 \beta_1 \tau_1) = +1}$, {\bf (b)} ${{\rm sgn}(\alpha_1 \beta_1 \tau_1) = -1}$. The bottom row illustrates Theorem \ref{thm:M-dgn-ref-2} with ${\tau_1 = 0}$ and {\bf (c)} ${{\rm sgn}(\alpha_2 \beta_2 \tau_3) = +1}$, {\bf (d)} ${{\rm sgn}(\alpha_2 \beta_2 \tau_3) = -1}$.}
\label{fig:M-dgn-ref}
\end{figure}

\begin{theorem}
\label{thm:M-dgn-ref-2}
{\rm (${\tau_1 = 0}$, ${\tau_3 \neq 0}$; Figure \ref{fig:M-dgn-ref}, panels (\subref{fig:M-dgn-ref-2a}, \subref{fig:M-dgn-ref-2b}).)}
Suppose ${\tau_1 = 0}$, ${\tau_3 \neq 0}$. Then, there exist $\tau_0^{\star \star}$, ${\tau_3^{\star\star} > 0}$ and analytic functions $E_D(\tau_0, \tau_3)$, ${\bm D}^\pm(\tau_0, \tau_3)$, defined for ${|\tau_0| < \tau_0^{\star \star}}$, ${|\tau_3| < \tau_3^{\star \star}}$, such that $(E_D(\tau_0, \tau_3), \, {\bm D}^\pm(\tau_0, \tau_3))$ are each twofold degenerate energy-quasimomentum pairs of $H^{\tau_0, \, (0, \, \tau_3)}$. We have the expansions
\begin{align}
E_D(\tau_0, \tau_3) & = E_S + \beta_0 \tau_0 + (1 - \alpha_0) \tilde{\kappa}_D^{(0)}{}^2 |\tau_3| + \varepsilon_D^{(1)}(\tau_0, \tau_3), \quad \text{where} \\
\varepsilon_D^{(1)}(\tau_0, \tau_3) & = O(|\tau_0|^2 + |\tau_3|^2) \ \ \text{as} \ \ |\tau_0|, \, |\tau_3| \to 0,
\end{align}
and
\begin{align}
{\bm D}^\pm(\tau_0, \tau_3) & = {\bm M} \pm \sqrt{|\tau_3|} (\tilde{\kappa}_D^{(0)} + \tilde{\kappa}_D^{(1)}(\tau_0, \tau_3)) \, {\bm e}_3, \quad \text{where} \\
\tilde{\kappa}_D^{(1)}(\tau_0, \tau_3) & = O(|\tau_0| + |\tau_3|) \ \ \text{as} \ \ |\tau_0|, \, |\tau_3| \to 0.
\end{align}
Here,
\begin{equation}
\tilde{\kappa}_D^{(0)} \equiv \sqrt{\Bigl|\frac{\beta_2}{\alpha_2}\Bigr|}
\end{equation}
and ${\bm e}_3$ is the normalized eigenvector satisfying:
\begin{equation}
\sigma_3 {\bm e}_3 = {\rm sgn}(\beta_2 \alpha_2 \tau_3) {\bm e}_3.
\end{equation}
In particular, the degenerate quasimomenta ${\bm D}^\pm(\tau_0, \tau_3)$ are constrained to one of the two coordinate axes bisecting ${\bm M}$: either ${\kappa_1 = 0}$ or ${\kappa_2 = 0}$.
\end{theorem}

\begin{remark}
In the case of each special deformation, the additional hypothesis \ref{itm:quad-dgn-6} can be slightly relaxed. For example, the value of $\beta_1$ is unimportant in the case ${\tau_1 = 0}$, ${\tau_3 \neq 0}$.
\end{remark}

\subsection{Rotational symmetry breaking}
\label{sec:break-rot}

Results analogous to Theorems \ref{thm:M-dgn} 
can be obtained when rotational invariance is broken. To illustrate this, consider the effect of perturbation by an even potential which breaks $\pi/2$-rotational symmetry. Consider the Schr\"{o}dinger operator ${H_{V \, + \, \delta W} = -\Delta + V +\delta W}$, where $V$ is a square lattice potential (Definition \ref{def:sql-pot}) and $W$ is $\Z^2$-periodic, and such that
\begin{equation}
\mathcal{P}[W]({\bm x}) = W(-{\bm x}) = W({\bm x}) \quad \text{and} \quad \mathcal{R}[W]({\bm x}) = W(R^\mathsf{T} {\bm x}) \neq W({\bm x}).
\end{equation}
By \cite[Theorem 4.3]{keller2018spectral}, the quadratic degeneracy does not persist. In this case, the effective Hamiltonian \eqref{eq:kmow-Heff_2} \eqref{eq:Leff-quad} governing the band structure near ${(E_S, \, {\bm M})}$ is:
\begin{equation}
\label{eq:Heff-norot}
H^{\bm M}_{\rm eff}({\bm \kappa}; \delta) = \bigl( \delta \theta'_0 + (1 - \alpha_0) |{\bm \kappa}|^2 \bigr) \, I + (\delta \theta'_1 - \alpha_1 {\bm \kappa} \cdot \sigma_1 {\bm \kappa}) \, \sigma_1 + (\delta \theta'_2 - \alpha_2 {\bm \kappa} \cdot \sigma_3 {\bm \kappa}) \, \sigma_2
\end{equation}
where $\theta'_0$, $\theta'_1$ and $\theta'_2$ are real parameters given by the expressions:
\begin{equation}
\theta'_0 \equiv \langle \Phi_1, \, W \Phi_1 \rangle, \quad 
\theta'_1 \equiv {\rm Re} \, \langle \Phi_1, \, W \Phi_2 \rangle, \quad \theta'_2 \equiv -{\rm Im} \, \langle \Phi_1, \, W \Phi_2 \rangle.
\end{equation}
The effective Hamiltonian in \eqref{eq:Heff-norot} does not have a $\sigma_3$ term since $\mathcal{PC}$ symmetry is not broken. For ${\bm \kappa} = {\bm 0}$, this reduces to the conclusion \cite[Theorem 4.3]{keller2018spectral}. In analogy with the discussion at the start of Section \ref{sec:M-dgn}, the Hamiltonian \eqref{eq:Heff-norot} anticipates a splitting into distinct Dirac points.

\bigskip

\section{Set-up for the proofs of main results}
\label{sec:set-up}

\setcounter{equation}{0}
\setcounter{figure}{0}

Throughout this section:
\begin{enumerate}
\item $V$ denotes a square lattice potential in the sense of Definition \ref{def:sql-pot}.
\item ${(E_S, \, {\bm M})}$ denotes a quadratic band degeneracy point in the sense of Section \ref{sec:kmow}, satisfying conditions \ref{itm:quad-dgn-1} -- \ref{itm:quad-dgn-6}.
\item We use the operator $T_* H_V$ arising from $H_{V \circ \, T^{-1}}$ by change of variables:
\begin{equation}
\label{eq:H-phi_3}
T_* H_V = H^{\tau_0, {\bm \tau}} = -\Delta + V + \tau_0 \Delta + |{\bm \tau}| \nabla \cdot \sigma_\varphi \nabla, \quad \text{where} \quad \sigma_\varphi = \cos(\varphi) \sigma_1 + \sin(\varphi) \sigma_3,
\end{equation}
and consider $\tau_0$, $|{\bm \tau}|$ in a neighborhood of zero, which is independent of ${\varphi \smallin (-\pi, \, \pi]}$.
\end{enumerate}

\subsection{Reduction to a local analysis about ${(E_S, {\bm M})}$}
\label{sec:reduction}

We first show that for ${(\tau_0, \, {\bm \tau})}$ sufficiently small, the band structure of $H^{\tau_0, {\bm \tau}}$ near ${(E_S, \, {\bm M})}$ is determined by a ${2 \times 2}$ self-adjoint matrix-valued analytic function of the energy-quasimomentum displacement ${(\varepsilon, \, {\bm \kappa})}$ from ${(E_S, {\bm M})}$. This reduction serves as the starting point for the proofs of Theorems \ref{thm:M-srf} and \ref{thm:M-dgn}.

\begin{proposition}
\label{prop:red-M}
There exist $\varepsilon^\flat$, $\kappa^\flat$, $\tau_0^\flat$, ${\tau^\flat > 0}$, and a ${2 \times 2}$ matrix-valued function $\mathcal{M} = \left(\mathcal{M}_{ij}(\varepsilon; {\bm \kappa}, \tau_0, {\bm \tau})\right)$, where
 $1\le i,j\le 2$, defined for ${|\varepsilon| < \varepsilon^\flat}$, ${|{\bm \kappa}| < \kappa^\flat}$, ${|\tau_0| < \tau_0^\flat}$, ${|{\bm \tau}| < \tau^\flat}$, such that:
\begin{enumerate}
\item \label{itm:det-calM-0} ${E_S + \varepsilon}$ is an $L^2_{{\bm M} + {\bm \kappa}}$ eigenvalue of $H^{\tau_0, {\bm \tau}}$, with ${|\varepsilon| < \tilde{\varepsilon}}$, ${|{\bm \kappa}| < \tilde{\kappa}}$, if and only if
\begin{equation}
\label{eq:det-calM-0}
\det(\mathcal{M}(\varepsilon; {\bm \kappa}, \tau_0, {\bm \tau})) = 0.
\end{equation}
\item \label{itm:calM-0} ${E_S + \varepsilon}$ is an $L^2_{{\bm M} + {\bm \kappa}}$ eigenvalue of $H^{\tau_0, {\bm \tau}}$, with ${|\varepsilon| < \tilde{\varepsilon}}$, ${|{\bm \kappa}| < \tilde{\kappa}}$, of multiplicity two, if and only if 
\begin{equation}
\label{eq:calM-0}
\mathcal{M}_{i, j}(\varepsilon; {\bm \kappa}, \tau_0, {\bm \tau}) = 0,\quad i,j=1,2.
\end{equation}
\end{enumerate}
\end{proposition}
The essential properties of $\mathcal{M}(\varepsilon; {\bm \kappa}, \tau_0, {\bm \tau})$ are presented in Proposition \ref{prop:calM-pr} below. \\

\noindent {\it Proof of Proposition \ref{prop:red-M}.} Recall that 
\begin{equation}
\label{eq:H-phi-k_1}
H^{\tau_0, {\bm \tau}}({\bm k}) = - (\nabla + i{\bm k})^2 + V + \tau_0 (\nabla + i{\bm k})^2 + |{\bm \tau}| (\nabla + i{\bm k}) \cdot \sigma_\varphi (\nabla + i{\bm k}),
\end{equation}
which acts in $L^2(\mathbb{R}^2/\Z^2)$. We expand
$H^{\tau_0, {\bm \tau}}({\bm k})$ for ${\bm k}$ in a neighborhood of $\bm M$. Using the abbreviated notation
\begin{equation}
\label{eq:def_H-M}
\nabla_{\bm M} \equiv \nabla + i{\bm M}, \quad H({\bm M}) \equiv -\nabla_{\bm M} \cdot \nabla_{\bm M} + V,
\end{equation}
we obtain
\begin{equation}
\label{eq:H-phi-k_ex}
H^{\tau_0, {\bm \tau}}({\bm M} + {\bm \kappa}) = H({\bm M}) - 2i {\bm \kappa} \cdot \nabla_{\bm M} + |{\bm \kappa}|^2  + \tau_0 (\nabla_{\bm M} + i {\bm \kappa})^2 + |{\bm \tau}| (\nabla_{\bm M} + i {\bm \kappa}) \cdot \sigma_\varphi (\nabla_{\bm M} + i {\bm \kappa}).
\end{equation}

Next, via a Schur complement/Lyapunov-Schmidt reduction procedure, we reformulate the Floquet-Bloch eigenvalue problem:
\begin{equation}
\label{eq:fb_evp_2a}
 H^{\tau_0, {\bm \tau}}({\bm M} + {\bm \kappa}) \phi = E \phi, \quad \phi \in L^2(\mathbb{R}^2/\mathbb{Z}^2),
\end{equation}
for ${\bm \kappa} \equiv {\bm k} - {\bm M}$ small, and energies $E$ near $E_S$.  We seek $(\phi, \, E)$, $\phi \ne 0$, expanded as: 
\begin{align}
\label{eq:evp-asz_1}
\phi & = a_1 \phi_1 + a_2 \phi_2 + \phi^{(1)}, \quad a_1, \, a_2 \smallin \C, \quad \phi^{(1)} \smallin {\rm ker}(H_{\bm M} - E_S)^\perp, \\
E & = E_S + \varepsilon, \quad \varepsilon \smallin \R.
\end{align}
Here, ${{\rm ker}(H({\bm M}) - E_S) = {\rm span}\{ \phi_1, \, \phi_2 \}}$, and $a_1$, $a_2$, $\phi^{(1)}$, and $\varepsilon$ are to be determined.
 Substituting into the eigenvalue problem \eqref{eq:fb_evp_2a}, we obtain an equivalent problem: 
\begin{equation}
\label{eq:space}
\textrm{Determine\ \ $(a_1, a_2) \neq (0, 0)$, $\phi^{(1)}\in L^2(\R^2/\Z^2)$, and $\varepsilon\in\R$, }
\end{equation}
as functions of parameters ${\bm \kappa}$,$\tau_0$, and ${\bm \tau}$, 
such that
 \begin{align}
\label{eq:H-M-inhom}
(H({\bm M}) - E_S) \phi^{(1)} & = F(a_1, a_2, \phi^{(1)}, \varepsilon; {\bm \kappa}, \tau_0, {\bm \tau}), \\
F(a_1, a_2, \phi^{(1)}, \varepsilon; {\bm \kappa}, \tau_0, {\bm \tau}) & \equiv (\varepsilon + \mathcal{L}({\bm \kappa}; \tau_0, {\bm \tau})) (a_1 \phi_1 + a_2 \phi_2 + \phi^{(1)}).
\nonumber\end{align}
Here, 
\begin{align}
\label{eq:def_calL}
\mathcal{L}({\bm \kappa}; \tau_0, {\bm \tau}) \equiv 2i {\bm \kappa} \cdot \nabla_{\bm M} - |{\bm \kappa}|^2 - \tau_0(\nabla_{\bm M} + i {\bm \kappa})^2 - |{\bm \tau}| (\nabla_{\bm M} + i {\bm \kappa}) \cdot \sigma_\varphi (\nabla_{\bm M} + i {\bm \kappa}).
\end{align}

Now introduce the orthogonal projections:
\begin{align}
\label{eq:def_proj-M}
\Pi_{\bm M}^\parallel & \equiv \phi_1 \langle \phi_1, \, \cdot \rangle  + \phi_2 \langle \phi_2, \, \cdot \rangle : L^2(\R^2/\Z^2) \to {\rm ker}(H({\bm M}) - E_S), \\
\Pi_{\bm M}^\perp & \equiv 1 - \Pi_{\bm M}^\parallel: L^2(\R^2/\Z^2) \to {\rm ker}(H({\bm M})- E_S)^\perp.
\end{align}
 Then, \eqref{eq:space}, \eqref{eq:H-M-inhom} for $(a_1, a_2, \phi^{(1)}, \varepsilon)\in\C\times\C\times L^2(\R^2/\Z^2)\times\R$ is equivalent to:
\begin{align}
\label{eq:phi-1}
& (H({\bm M}) - E_S) \phi^{(1)} = \Pi_{\bm M}^\perp F(a_1, a_2, \phi^{(1)}, \varepsilon; {\bm \kappa}, \tau_0, {\bm \tau}), \\
\label{eq:red-M}
& \Pi_{\bm M}^\parallel F(a_1, a_2, \phi^{(1)}, \varepsilon; {\bm \kappa}, \tau_0, {\bm \tau}) = 0.
\end{align}

Our next steps are to solve \eqref{eq:phi-1} for $\phi^{(1)}$ as a functional of $a_1$, $a_2$, and $\varepsilon$, and to then substitute $\phi_1[a_1, a_2, \varepsilon]$ into \eqref{eq:red-M} to obtain a closed system of equations, depending on ${\bm \kappa}$, $\tau_0$, and ${\bm \tau}$, for the unknowns $a_1$, $a_2$, and $\varepsilon$. \\

\noindent {\bf N.B.} We occasionally suppress the dependence of functions on some or all of the parameters ${\bm \kappa}$, $\tau_0$, and ${\bm \tau}$.

\subsubsection{Constructing $(a_1,a_2,\varepsilon)\mapsto\phi^{(1)}[a_1,a_2,\varepsilon]$}
\label{sec:phi-1}

From \eqref{eq:phi-1}, we have 
\begin{align}
\label{eq:phi-1_1}
(H({\bm M}) - E_S) \phi^{(1)} - (\varepsilon + \Pi_{\bm M}^\perp \mathcal{L}({\bm \kappa})) \phi^{(1)} = \Pi_{\bm M}^\perp \mathcal{L}({\bm \kappa}) (a_1 \phi_1 + a_2 \phi_2).
\end{align}
Introduce the resolvent
\begin{equation}
\label{eq:res-M}
\mathscr{R}_{\bm M}(E_S) \equiv \Pi^\perp_{\bm M} (H_{\bm M} - E_S)^{-1} \Pi^\perp_{\bm M} : L^2(\R^2/\Z^2) \to  H^2(\R^2/\Z^2),
\end{equation}
which is a bounded linear operator. Applying $\mathscr{R}_{\bm M}(E_S)$ to  \eqref{eq:phi-1_1} yields
 \begin{equation}
\label{eq:phi-1_2}
\bigl(I - \mathscr{R}_{\bm M}(E_S) (\varepsilon + \mathcal{L}({\bm \kappa}))\bigr) \phi^{(1)} = \mathscr{R}_{\bm M}(E_S) \mathcal{L}({\bm \kappa}) (a_1 \phi_1 + a_2 \phi_2).
\end{equation}

Next, observe that there exists a constant ${C > 0}$ such that, for $|\varepsilon|$, $|{\bm \kappa}|$, $|\tau_0|$, and $|{\bm \tau}|$ sufficiently small, the  norm of $\mathscr{R}_{\bm M}(E_S)( \varepsilon + \mathcal{L}({\bm \kappa}))$, as an operator on $L^2_{\bm M}$, satisfies the bound
\begin{equation}
\label{eq:phi-1-op_bd}
\lVert \mathscr{R}_{\bm M}(E_S)( \varepsilon + \mathcal{L}({\bm \kappa})) \rVert \leq C (|\varepsilon| + |{\bm \kappa}| + |\tau_0| + |{\bm \tau}|).
\end{equation}
It follows that there exist $\varepsilon^\flat$, $\kappa^\flat$, $\tau_0^\flat$, ${\tau^\flat > 0}$ such that, for ${|\varepsilon| < \varepsilon^\flat}$, ${|{\bm \kappa}| < \kappa^\flat}$, ${|\tau_0| < \tau_0^\flat}$, ${|{\bm \tau}| < \tau^\flat}$, \eqref{eq:phi-1_2} is solvable for ${\phi^{(1)} \smallin H^2(\R^2/\Z^2)}$:
\begin{align}
\label{eq:phi-1_3}
\phi^{(1)} = \bigl(I - \mathscr{R}_{\bm M}(E_S) (\varepsilon + \mathcal{L}({\bm \kappa}))\bigr)^{-1} \mathscr{R}_{\bm M}(E_S) \mathcal{L}({\bm \kappa}) (a_1 \phi_1 + a_2 \phi_2).
\end{align}
For fixed parameters ${\bm \kappa}$, $\tau_0$, and ${\bm \tau}$, \eqref{eq:phi-1_3} defines a mapping
\begin{equation}
\label{eq:phi-1-map}
(a_1, a_2, \varepsilon) \mapsto \phi^{(1)}[a_1, a_2, \varepsilon] \smallin {\rm ker}(H_{\bm M} - E_S))^\perp.
\end{equation}

\subsubsection{The reduced equation}
\label{sec:red-M}

For fixed ${\bm \kappa}$, $\tau_0$, and ${\bm \tau}$, substituting the result \eqref{eq:phi-1-map} into \eqref{eq:red-M} yields
\begin{equation}
\label{eq:reduced}
\Pi_{\bm M}^\parallel F(a_1 \phi_1 + a_2 \phi_2, \phi^{(1)}[a_1, a_2, \varepsilon], \varepsilon) = 0.
\end{equation}
or equivalently
\[
\left\langle \Phi_j, F(a_1 \phi_1 + a_2 \phi_2, \phi^{(1)}[a_1, a_2, \varepsilon], \varepsilon)\right\rangle=0,\ j=1,2. 
\]

This is equivalent to a system of two linear, homogeneous equations for $a_1$ and $a_2$:
\begin{equation}
\label{eq:calM-a-0}
\mathcal{M}(\varepsilon; {\bm \kappa}, \tau_0, {\bm \tau}) \!
\begin{bmatrix}
a_1 \\
a_2
\end{bmatrix}
= 0.
\end{equation}
The matrix $\mathcal M$ is displayed in Section \ref{apx:calM-pr}.
This system is to be solved for $(a_1, a_2)\ne(0,0)$, and $\varepsilon$, with ${|\varepsilon| < \varepsilon^\flat}$, ${|{\bm \kappa}| < \kappa^\flat}$, ${|\tau_0| < \tau_0^\flat}$, ${|{\bm \tau}| < \tau^\flat}$.

 To prove  conclusion \ref{itm:det-calM-0} of Proposition \ref{prop:red-M}, note that \eqref{eq:calM-a-0} is has a non-trivial solution if and only if
\begin{equation}
\label{eq:det-calM-0_2}
\det\mathcal{M}(\varepsilon; {\bm \kappa}, \tau_0, {\bm \tau}) = 0.
\end{equation}
It follows from the ansatz \eqref{eq:evp-asz_1} that $E_S + \varepsilon$ is an $L^2_{{\bm M} + {\bm \kappa}}$ eigenvalue of $H^{\tau_0, {\bm \tau}}$ if and only if \eqref{eq:det-calM-0_2} holds.

Conclusion \ref{itm:calM-0} of Proposition \ref{prop:red-M} follows by noting that  
${E_S + \varepsilon}$ is a twofold degenerate $L^2_{{\bm M} + {\bm \kappa}}$ eigenvalue of $H^{\tau_0, {\bm \tau}}$ if and only if  $\varepsilon$ and ${\bm \kappa}$ are such that
${{\rm dim}({\rm ker}(\mathcal{M}(\varepsilon; {\bm \kappa}, \tau_0, {\bm \tau}))) = 2}$.
By Lemma \ref{lem:2d-ker}  
\begin{equation}
\label{eq:calM-0_2}
\mathcal{M}(\varepsilon; {\bm \kappa}, \tau_0, {\bm \tau}) = \left( \mathcal{M}_{ij}(\varepsilon; {\bm \kappa}, \tau_0, {\bm \tau})\right)_{1\le i,j\le 2}  = 0.
\end{equation}
 The proof of Proposition \ref{prop:red-M} is now complete.

\subsection[TEXT]{Properties of $\mathcal{M}(\varepsilon; {\bm \kappa}, \tau_0, {\bm \tau})$}
\label{sec:calM-pr}

We conclude this section by recording detailed properties of the ${2 \times 2}$ matrix $\mathcal{M}(\varepsilon; {\bm \kappa}; \tau_0, {\bm \tau})$, defined in \eqref{eq:calM-a-0}, and displayed in Appendix \ref{apx:calM-pr}. The proof  is given in Appendix \ref{apx:calM-pr}.

\begin{proposition}
\label{prop:calM-pr}
For ${|\varepsilon| < \varepsilon^\flat}$, ${|{\bm \kappa}| < \kappa^\flat}$, ${|\tau_0| < \tau_0^\flat}$ and  ${|{\bm \tau}| < \tau^\flat}$, the matrix $\mathcal{M}(\varepsilon; {\bm \kappa}, \tau_0, {\bm \tau})$ satisfies the following properties:
\begin{enumerate}
\item \label{itm:calM-sa} $\mathcal{M}(\varepsilon; {\bm \kappa}, \tau_0, {\bm \tau})$ is self-adjoint.
\item \label{itm:calM-an} The entries $\mathcal{M}_{j, k}(\varepsilon; {\bm \kappa}, \tau_0, {\bm \tau})$ define analytic functions; in particular, they have expansions in power series in $(\varepsilon; {\bm \kappa}, \tau_0, {\bm \tau})$ which converge uniformly on their domain.
\item \label{itm:calM-ev} By $\mathcal{P}$ (or $\mathcal{C}$) symmetry,
\begin{equation}
\label{eq:calM-ev}
\mathcal{M}(\varepsilon; {\bm \kappa}, \tau_0, {\bm \tau}) = \mathcal{M}(\varepsilon; -{\bm \kappa}, \tau_0, {\bm \tau}).
\end{equation}
\item \label{itm:calM-PC} By $\mathcal{PC}$ symmetry,
\begin{equation}
\label{eq:calM-PC}
\mathcal{M}(\varepsilon; {\bm \kappa}, \tau_0, {\bm \tau})_{1, 1} = \mathcal{M}(\varepsilon; {\bm \kappa}, \tau_0, {\bm \tau})_{2, 2}.
\end{equation}
\item \label{itm:calM-ex-1} $\mathcal{M}(\varepsilon; {\bm \kappa}, \tau_0, {\bm \tau})$ has the expansion
\begin{equation}
\label{eq:calM-ex-1}
\mathcal{M}(\varepsilon; {\bm \kappa}, \tau_0, {\bm \tau}) = \varepsilon I - H^{\bm M}_{\rm eff}({\bm \kappa}; \tau_0, {\bm \tau}) + \tilde{\mathcal{M}}(\varepsilon; {\bm \kappa}, \tau_0, {\bm \tau}).
\end{equation}
Here, $H^{\bm M}_{\rm eff}({\bm \kappa}; \tau_0, {\bm \tau})$ is the Fourier symbol of an effective Hamiltonian:
\begin{align}
\label{eq:Heff_2}
H^{\bm M}_{\rm eff}({\bm \kappa}; \tau_0, {\bm \tau}) = \bigl( \beta_0 \tau_0 + (1 - \alpha_0) |{\bm \kappa}|^2 \bigr) \, I + (\beta_1 \tau_1 - \alpha_1 {\bm \kappa} \cdot \sigma_1 {\bm \kappa}) \, \sigma_1 + (\beta_2 \tau_3 - \alpha_2 {\bm \kappa} \cdot \sigma_3 {\bm \kappa}) \, \sigma_2;
\end{align}
see \eqref{eq:M-eff} of Remark \ref{rmk:M-srf-eff}. The entries of the ${2 \times 2}$ matrix $\tilde{\mathcal{M}}(\varepsilon; {\bm \kappa}, \tau_0, {\bm \tau})$ satisfy, for $j, k = 1, 2$,
\begin{equation}
\label{eq:til-calM-bd}
\tilde{\mathcal{M}}_{j, k}(\varepsilon; {\bm \kappa}, \tau_0, {\bm \tau}) = O(|\varepsilon| |{\bm \kappa}|^2 + |{\bm \kappa}|^4 + |{\bm \kappa}|^2 |\tau_0| + |{\bm \kappa}|^2 |{\bm \tau}| + |\tau_0|^2 + |{\bm \tau}|^2)
\end{equation}
as $|\varepsilon|$, $|{\bm \kappa}|$, $|\tau_0|$, ${|{\bm \tau}| \to 0}$.
\end{enumerate}
\end{proposition}

\bigskip

\section{Dispersion surfaces near $(E_S, M)$ under small deformation; Proof of Theorem \ref{thm:M-srf}}
\label{sec:pf-M-srf}

\setcounter{equation}{0}
\setcounter{figure}{0}

By conclusion \ref{itm:det-calM-0} of Proposition \ref{prop:red-M}, for ${|\varepsilon| < \varepsilon^\flat}$, ${|{\bm \kappa}| < \kappa^\flat}$, ${|\tau_0| < \tau_0^\flat}$, and ${|{\bm \tau}| < \tau^\flat}$, ${E_S + \varepsilon}$ is an $L^2_{{\bm M} + {\bm \kappa}}$ eigenvalue of $H^{\tau_0, {\bm \tau}}$ if and only if 
\begin{equation}
\label{eq:det-calM-0_3}
\det(\mathcal{M}(\varepsilon; {\bm \kappa}, \tau_0, {\bm \tau})) = 0.
\end{equation}
Here, $\mathcal{M}(\varepsilon; {\bm \kappa}, \tau_0, {\bm \tau})$ is the ${2 \times 2}$ matrix introduced in \eqref{eq:calM-a-0}. To prove Theorem \ref{thm:M-srf}, we determine the locus of solutions ${\varepsilon = \varepsilon({\bm \kappa}; \tau_0, {\bm \tau})}$ to \eqref{eq:det-calM-0_3}.

 Noting \eqref{eq:Heff_2}, we center about zero energy by translating the energy $\varepsilon$ setting:
\begin{equation}
\label{eq:def-nu}
\nu \equiv \varepsilon - \beta_0 \tau_0 - (1 - \alpha_0) |{\bm \kappa}|^2, 
\end{equation}
and we define
 $\nu^\flat \equiv \varepsilon^\flat + \beta_0 \tau_0^\flat + (1 - \alpha_0) (\kappa^\flat{})^2$.
For 
\begin{equation}\label{eq:param-dom}
{|\nu| < \nu^\flat},\ {|{\bm \kappa}| < \kappa^\flat},\  {|\tau_0| < \tau_0^\flat},\ \textrm{and}\ {|{\bm \tau}| < \tau^\flat},
\end{equation}
 we study
\begin{equation}
\label{eq:def-calD}
\mathcal{D}(\nu; {\bm \kappa}, \tau_0, {\bm \tau}) \equiv \det \Big( \mathcal{M}(\nu + \beta_0 \tau_0 + (1 - \alpha_0) {\bm \kappa}\cdot{\bm \kappa} ; {\bm \kappa}, \tau_0, {\bm \tau}) \Big).
\end{equation}
The mapping ${\nu \mapsto \mathcal{D}(\nu; {\bm \kappa}, \tau_0, {\bm \tau})}$ extends to an analytic function in a  complex neighborhood defined by the inequalities \eqref{eq:param-dom}.

We next study the zeros ${\nu = \nu({\bm \kappa}; \tau_0, {\bm \tau})}$ of $\mathcal{D}(\nu; {\bm \kappa}, \tau_0, {\bm \tau})$. We first prove the existence of exactly two zeros of $\mathcal{D}(\nu; {\bm \kappa}, \tau_0, {\bm \tau})$ in a complex neighborhood of ${\nu = 0}$ via an application of Rouch\'{e}'s theorem, and  then obtain expressions for these zeros using the argument principle.

\begin{proposition}
\label{prop:calD-pr}
$\mathcal{D}(\nu; {\bm \kappa}, \tau_0, {\bm \tau})$ is defined and analytic on the domain \eqref{eq:param-dom}, and has the following properties:
\begin{enumerate}
\item \label{itm:calD-ev} By $\mathcal{P}$ (or $\mathcal{C}$) symmetry,
\begin{equation}
\label{eq:calD-ev}
\mathcal{D}(\nu; {\bm \kappa}, \tau_0, {\bm \tau}) = \mathcal{D}(\nu; -{\bm \kappa}, \tau_0, {\bm \tau}).
\end{equation}
\item \label{itm:calD-ex} $\mathcal{D}(\nu; {\bm \kappa}, \tau_0, {\bm \tau})$ has the expansion
\begin{equation}
\label{eq:calD-ex}
\mathcal{D}(\nu; {\bm \kappa}, \tau_0, {\bm \tau}) = \mathcal{D}^{(0)}(\nu; {\bm \kappa}, \tau_0, {\bm \tau}) + \mathcal{D}^{(1)}(\nu; {\bm \kappa}, \tau_0, {\bm \tau}).
\end{equation}
Here, $\mathcal{D}^{(0)}(\nu; {\bm \kappa}, \tau_0, {\bm \tau})$ is defined by
\begin{align}
\label{eq:def-calD-0}
\mathcal{D}^{(0)}(\nu; {\bm \kappa}, \tau_0, {\bm \tau}) & \equiv \nu^2 - |\omega({\bm \kappa}; \tau_0, {\bm \tau})|^2, \quad \text{where} \\
\label{eq:def-om}
\omega({\bm \kappa}; \tau_0, {\bm \tau}) & \equiv (\beta_1 \tau_1 - \alpha_1 {\bm \kappa} \cdot \sigma_1 {\bm \kappa}) - i(\beta_2 \tau_3 - \alpha_2 {\bm \kappa} \cdot \sigma_3 {\bm \kappa}).
\end{align}
The function $\mathcal{D}^{(1)}(\nu; {\bm \kappa}, \tau_0, {\bm \tau})$ is analytic on the domain \eqref{eq:param-dom} and satisfies
\begin{align}
\label{eq:calD-1-bd}
\mathcal{D}^{(1)}(\nu; {\bm \kappa}, \tau_0, {\bm \tau}) & =
O(|\nu|^2 |{\bm \kappa}|^2 + |\nu| |{\bm \kappa}|^4 + |\nu| |{\bm \kappa}|^2 |\tau_0| + |\nu| |{\bm \kappa}|^2 |{\bm \tau}| + |\nu| |\tau_0|^2 + |\nu| |{\bm \tau}|^2 \\
& \qquad \quad + |{\bm \kappa}|^6 + |{\bm \kappa}|^4 |\tau_0| + |{\bm \kappa}|^4 |{\bm \tau}| + |{\bm \kappa}|^2 |\tau_0|^2 + |{\bm \kappa}|^2 |{\bm \tau}|^2 + |\tau_0|^3 + |{\bm \tau}|^3), \nonumber \\
\label{eq:dcalD-1-bd}
\partial_\nu \mathcal{D}^{(1)}(\nu; {\bm \kappa}, \tau_0, {\bm \tau}) & = O(|\nu| |{\bm \kappa}|^2 + |{\bm \kappa}|^4 + |{\bm \kappa}|^2 |\tau_0| + |{\bm \kappa}|^2 |{\bm \tau}| + |\tau_0|^2 + |{\bm \tau}|^2)
\end{align}
as $|\nu|$, $|{\bm \kappa}|$, $|\tau_0|$, ${|{\bm \tau}| \to 0}$.
\end{enumerate}
\end{proposition}

\begin{proof}
Part \ref{itm:calD-ev} of Proposition \ref{prop:calD-pr} follows from part \ref{itm:calM-ev} of Proposition \ref{prop:calM-pr}. Part \ref{itm:calD-ex} is proven in Appendix \ref{apx:pf-calD-ex}.
\end{proof}

\subsection{Locating the zeros of $\mathcal{D}(\nu; {\bm \kappa}, \tau_0, {\bm \tau})$ in a disk}
\label{sec:pf-M-srf-loc}

In this section, we use the expansion from conclusion \ref{itm:calD-ex} of Proposition \ref{prop:calD-pr} to locate the zeros of $\mathcal{D}(\nu; {\bm \kappa}, \tau_0, {\bm \tau})$ relative to the explicit zeros of $\mathcal{D}^{(0)}(\nu; {\bm \kappa}, \tau_0, {\bm \tau})$ by using the bounds on $\mathcal{D}^{(1)}(\nu; {\bm \kappa}, \tau_0, {\bm \tau})$ and applying Rouch\'{e}'s theorem.

\begin{proposition}
\label{prop:calD-zs}
There exist $\kappa^{\star}$, $\tau_0^{\star}$, ${\tau^{\star} > 0}$ and a non-negative function $\nu^\star({\bm \kappa}; \tau_0, {\bm \tau})$ such that, for ${|{\bm \kappa}| < \kappa^\star}$, ${|\tau_0| < \tau_0^\star}$, ${|{\bm \tau}| < \tau^\star}$, the mapping ${\nu \mapsto \mathcal{D}(\nu; {\bm \kappa}, \tau_0, {\bm \tau})}$ has exactly two zeros ${\nu = \nu_\pm({\bm \kappa}; \tau_0, {\bm \tau})}$, counting multiplicity, in the open disk ${|\nu| < \nu^\star({\bm \kappa}; \tau_0, {\bm \tau})}$.
\end{proposition}

\noindent {\it Proof of Proposition \ref{prop:calD-zs}.} First, we note that there exist $C_1$, ${C_2 > 0}$ such that $|\omega({\bm \kappa}; \tau_0, {\bm \tau})|^2$, where $\omega({\bm \kappa}; \tau_0, {\bm \tau})$ is displayed in \eqref{eq:def-om}, has the upper bound
\begin{equation}
\label{eq:om-sq-bd}
|\omega({\bm \kappa}, \tau_0, {\bm \tau})|^2 = (\beta_1 \tau_1 - \alpha_1 {\bm \kappa} \cdot \sigma_1 {\bm \kappa})^2 + (\beta_2 \tau_3 - \alpha_2 {\bm \kappa} \cdot \sigma_3 {\bm \kappa})^2 \leq C_1 |{\bm \kappa}|^4 + C_2 |{\bm \tau}|^2.
\end{equation}
Therefore, we have the lower bound
\begin{equation}
\label{eq:calD-0-bd}
|\mathcal{D}^{(0)}(\nu; {\bm \kappa}, \tau_0, {\bm \tau})| = \bigl|\nu^2 - |\omega({\bm \kappa}, \tau_0, {\bm \tau})|^2\bigr| \geq |\nu|^2 - (C_1 |{\bm \kappa}|^4 + C_2 |{\bm \tau}|^2).
\end{equation}
Let \begin{equation}
\label{eq:def-nu-s}
\nu^\star({\bm \kappa}; \tau_0, {\bm \tau}) \equiv \sqrt{|\tau_0|^2 + 2(C_1 |{\bm \kappa}|^4 + C_2 |{\bm \tau}|^2)}.
\end{equation}
Then,  on the circle ${|\nu| = \nu^\star({\bm \kappa}; \tau_0, {\bm \tau})}$
we have the lower bound
\begin{equation}
\label{eq:calD-0-c-bd}
|\mathcal{D}^{(0)}(\nu; {\bm \kappa}, \tau_0, {\bm \tau})| \geq  C_1 |{\bm \kappa}|^4 + |\tau_0|^2 + C_2 |{\bm \tau}|^2.
\end{equation}

Immediate consequences of \eqref{eq:calD-1-bd} and 
\eqref{eq:dcalD-1-bd}
are the following upper bounds on $\mathcal{D}^{(1)}(\nu; {\bm \kappa}, \tau_0, {\bm \tau})$ and $\partial_\nu \mathcal{D}^{(1)}(\nu; {\bm \kappa}, \tau_0, {\bm \tau})$ on $nu| = \nu^\star$:
\begin{lemma}
\label{lem:calD-1-c-bd}
There exist $D_1$, $D_2$, ${D_3 > 0}$, and $D_1'$, $D_2'$, ${D_3' > 0}$ such that, for ${|\nu| = \nu^\star({\bm \kappa}; \tau_0, {\bm \tau})}$ and $|{\bm \kappa}|$, $|\tau_0|$, $|{\bm \tau}|$ sufficiently small,
\begin{align}
\label{eq:calD-1-c-bd}
|\mathcal{D}^{(1)}(\nu; {\bm \kappa}, \tau_0, {\bm \tau})| & \leq D_1 |{\bm \kappa}|^6 + D_2 |\tau_0|^3 + D_3 |{\bm \tau}|^3, \\
\label{eq:dcalD-1-c-bd}
|\partial_\nu \mathcal{D}^{(1)}(\nu; {\bm \kappa}, \tau_0, {\bm \tau})| & \leq D_1' |{\bm \kappa}|^4 + D_2' |\tau_0|^2 + D_3' |{\bm \tau}|^2.
\end{align}
\end{lemma}

We conclude the proof of Proposition \ref{prop:calD-zs} by applying Rouch\'{e}'s theorem. First, note that $\mathcal{D}^{(0)}(\nu; {\bm \kappa}, \tau_0, {\bm \tau})$, displayed in \eqref{eq:def-calD-0}, has two zeros (counting multiplicity):
\begin{equation}
\label{eq:calD-0-zs}
\nu = \pm |\omega({\bm \kappa}; \tau_0, {\bm \tau})|.
\end{equation}
By the upper bound \eqref{eq:om-sq-bd}, these zeros lie in $|\nu| < \nu^\star({\bm \kappa}; \tau_0, {\bm \tau})$.
Let \begin{equation}
\label{eq:def-kap-ta-s}
0 < \kappa^\star \leq \sqrt{\frac{C_1}{D_1}}, \quad 0 < \tau_0^\star \leq \frac{1}{D_2}, \quad
0 < \tau^\star \leq \frac{C_2}{D_3}.
\end{equation}
Assume ${|\nu| = \nu^\star({\bm \kappa}; \tau_0, {\bm \tau})}$ and ${|{\bm \kappa}| < \kappa^\star}$, ${|\tau_0| < \tau_0^\star}$, ${|{\bm \tau}| < \tau^\star}$,
Using the upper bound \eqref{eq:calD-1-c-bd} for  $|\mathcal{D}^{(1)}|$ and the lower bound  \eqref{eq:calD-0-c-bd} for
$|\mathcal{D}^{(0)}|$, we obtain 
 for  we have the strict inequality on the circle $|\nu| = \nu^\star({\bm \kappa}; \tau_0, {\bm \tau})$:
\begin{align}
\label{eq:rou-ineq}
|\mathcal{D}^{(1)}(\nu; {\bm \kappa}, \tau_0, {\bm \tau})| &  < |\mathcal{D}^{(0)}(\nu; {\bm \kappa}, \tau_0, {\bm \tau})| 
\end{align}
By Rouch\'{e}'s theorem, for $|{\bm \kappa}| < \kappa^\star$, $|\tau_0| < \tau_0^\star$, $|{\bm \tau}| < \tau^\star$,
the functions $\mathcal{D}^{(0)}(\nu; {\bm \kappa}, \tau_0, {\bm \tau})$ and\\ ${\mathcal{D}(\nu; {\bm \kappa}, \tau_0, {\bm \tau}) = \mathcal{D}^{(0)}(\nu; {\bm \kappa}, \tau_0, {\bm \tau}) + \mathcal{D}^{(1)}(\nu; {\bm \kappa}, \tau_0, {\bm \tau})}$ have the same number of zeros (two, counting multiplicity) in ${|\nu| < \nu^\star({\bm \kappa}; \tau_0, {\bm \tau})}$. 
The proof of Proposition \ref{prop:calD-zs} is now complete.
\medskip 

We denote the two zeros of $\mathcal{D}(\nu; {\bm \kappa}, \tau_0, {\bm \tau})$ in ${|\nu| < \nu^\star({\bm \kappa}; \tau_0, {\bm \tau})}$ by ${\nu = \nu_\pm({\bm \kappa}; \tau_0, {\bm \tau})}$.

\subsection[TEXT]{Expressions for the zeros of $\mathcal{D}(\nu; {\bm \kappa}, \tau_0, {\bm \tau})$}
\label{sec:pf-M-srf-exp}

Assume ${|{\bm \kappa}| < \kappa^\star}$, ${|\tau_0| < \tau_0^\star}$, ${|{\bm \tau}| < \tau^\star}$, such that the conclusion of Proposition \ref{prop:calD-zs} holds. Our next goal is to obtain explicit expressions for the zeros  $\nu = \nu_\pm({\bm \kappa}; \tau_0, {\bm \tau})$ of $\mathcal{D}(\nu; {\bm \kappa}, \tau_0, {\bm \tau})$.

Since $\nu = \nu_\pm({\bm \kappa}; \tau_0, {\bm \tau})$ are the only two zeros in the disc,  there exists an analytic function $\tilde{\mathcal{D}}(\nu; {\bm \kappa}, \tau_0, {\bm \tau})$ such that
\begin{equation}
\label{eq:calD-fac}
\mathcal{D}(\nu; {\bm \kappa}, \tau_0, {\bm \tau}) = ( \nu - \nu_+({\bm \kappa}; \tau_0, {\bm \tau}) ) ( \nu - \nu_-({\bm \kappa}; \tau_0, {\bm \tau}) )\ \tilde{\mathcal{D}}(\nu; {\bm \kappa}, \tau_0, {\bm \tau}),
\end{equation}
where ${\tilde{\mathcal{D}}(\nu; {\bm \kappa}, \tau_0, {\bm \tau}) \neq 0}$ for ${|\nu| < \nu^\star({\bm \kappa}; \tau_0, {\bm \tau})}$. Therefore, for ${|\nu| < \nu^\star({\bm \kappa}; \tau_0, {\bm \tau})}$, we have $\mathcal{D}(\nu; {\bm \kappa}, \tau_0, {\bm \tau}) = 0$ if and only if
\begin{equation}
\label{eq:calD-qu}
(\nu - \nu_+({\bm \kappa}; \tau_0, {\bm \tau}))(\nu - \nu_-({\bm \kappa}; \tau_0, {\bm \tau})) = 0.
\end{equation}
The solutions to \eqref{eq:calD-qu} are
\begin{gather}
\label{eq:calD-qu-sol}
\nu = \nu_\pm({\bm \kappa}; \tau_0, {\bm \tau}) = \frac{1}{2} \mathfrak{t}({\bm \kappa}; \tau_0, {\bm \tau}) \pm \sqrt{-\mathfrak{d}({\bm \kappa}; \tau_0, {\bm \tau}) + \frac{1}{4} \mathfrak{t}({\bm \kappa}; \tau_0, {\bm \tau})^2} \, ,
\end{gather}
where we define
\begin{align}
\label{eq:def-frak-t}
\mathfrak{t}({\bm \kappa}; \tau_0, {\bm \tau}) & \equiv \nu_+({\bm \kappa}; \tau_0, {\bm \tau}) + \nu_-({\bm \kappa}; \tau_0, {\bm \tau}), \\
\label{eq:def-frak-d}
\mathfrak{d}({\bm \kappa}; \tau_0, {\bm \tau}) & \equiv \nu_+({\bm \kappa}; \tau_0, {\bm \tau}) \, \nu_-({\bm \kappa}; \tau_0, {\bm \tau}).
\end{align}
Moreover, we observe that 
\begin{equation}
\label{eq:frak-d-sum}
\mathfrak{d}({\bm \kappa}; \tau_0, {\bm \tau}) = \frac{1}{2} (\nu_+({\bm \kappa}; \tau_0, {\bm \tau}) + \nu_-({\bm \kappa}; \tau_0, {\bm \tau}))^2 - \frac{1}{2} (\nu_+({\bm \kappa}; \tau_0, {\bm \tau})^2 + \nu_-({\bm \kappa}; \tau_0, {\bm \tau})^2).
\end{equation}
Therefore, to study the expressions defined in \eqref{eq:def-frak-t} and \eqref{eq:def-frak-d}, it suffices to study ${\bm \kappa}, \tau_0, {\bm \tau})\mapsto\nu_+^\ell + \nu_-^\ell$ for $\ell=1,2$. By analyticity of ${\nu \mapsto \mathcal{D}(\nu; {\bm \kappa}, \tau_0, {\bm \tau})}$ we have 
\begin{equation}
\label{eq:zs-arg}
\nu_+({\bm \kappa}; \tau_0, {\bm \tau})^\ell + \nu_-({\bm \kappa}; \tau_0, {\bm \tau})^\ell = \frac{1}{2 \pi i} \int_{|\nu| \, = \, \nu^\star({\bm \kappa}; \tau_0, {\bm \tau})} \frac{\nu^\ell \, \partial_\nu \mathcal{D}(\nu; {\bm \kappa}, \tau_0, {\bm \tau})}{\mathcal{D}(\nu; {\bm \kappa}, \tau_0, {\bm \tau})} \, {\rm d}\nu.
\end{equation}
Therefore, by  \eqref{eq:def-frak-t}, \eqref{eq:frak-d-sum}, and \eqref{eq:zs-arg} we have:
\begin{proposition}
\label{prop:frak-t&d-pr}
The functions functions $\mathfrak{t}({\bm \kappa}; \tau_0, {\bm \tau})$ and $\mathfrak{d}({\bm \kappa}; \tau_0, {\bm \tau})$
 are analytic for ${|{\bm \kappa}| < \kappa^\star}$, ${|\tau_0| < \tau_0^\star}$, ${|{\bm \tau}| < \tau^\star}$ and satisfy the symmetry property:
\begin{equation}
\label{eq:frak-t&d-pr}
\mathfrak{t}({\bm \kappa}; \tau_0, {\bm \tau}) = \mathfrak{t}(-{\bm \kappa}; \tau_0, {\bm \tau}), \quad \mathfrak{d}({\bm \kappa}; \tau_0, {\bm \tau}) = \mathfrak{d}(-{\bm \kappa}; \tau_0, {\bm \tau}).
\end{equation}
\end{proposition}
\begin{proof}
Symmetry is an immediate consequence of the representation \eqref{eq:zs-arg}, and the corresponding symmetry of 
$\mathcal{D}(\nu; {\bm \kappa}, \tau_0, {\bm \tau})$; see Proposition \ref{prop:calD-pr}.
\end{proof}

We use the following result to obtain expansions of $\mathfrak{t}({\bm \kappa}; \tau_0, {\bm \tau})$ and $\mathfrak{d}({\bm \kappa}; \tau_0, {\bm \tau})$:

\begin{proposition}
\label{prop:calD-zs-int}
We have
\begin{align}
\label{eq:calD-zs-sum}
\nu_+({\bm \kappa}; \tau_0, {\bm \tau}) + \nu_-({\bm \kappa}; \tau_0, {\bm \tau}) & = O(|{\bm \kappa}|^4 + |\tau_0|^2 + |{\bm \tau}|^2), \\
\label{eq:cald-zs2-sum}
\nu_+({\bm \kappa}; \tau_0, {\bm \tau})^2 + \nu_-({\bm \kappa}; \tau_0, {\bm \tau})^2 & = 2 |\omega({\bm \kappa}; \tau_0, {\bm \tau})|^2 + O(|{\bm \kappa}|^6 + |\tau_0|^3 + |{\bm \tau}|^3)
\end{align}
as $|{\bm \kappa}|$, $|\tau_0|$, ${|{\bm \tau}| \to 0}$.
\end{proposition}

\begin{proof}
We expand the right-hand side of \eqref{eq:zs-arg}, starting with its integrand. By Proposition \ref{prop:calD-pr},
\begin{align}
\label{eq:dcalD-ex}
\partial_\nu \mathcal{D}(\nu; {\bm \kappa}, \tau_0, {\bm \tau}) & = \partial_\nu \mathcal{D}^{(0)}(\nu; {\bm \kappa}, \tau_0, {\bm \tau}) + \partial_\nu \mathcal{D}^{(1)}(\nu; {\bm \kappa}, \tau_0, {\bm \tau}) \\
& = 2 \nu + \partial_\nu \mathcal{D}^{(1)}(\nu; {\bm \kappa}, \tau_0, {\bm \tau}). \nonumber
\end{align}
Furthermore,
\begin{equation}
\label{eq:calD-inv-ex}
\frac{1}{\mathcal{D}(\nu; {\bm \kappa}, \tau_0, {\bm \tau})} = \frac{1}{\mathcal{D}^{(0)}(\nu; {\bm \kappa}, \tau_0, {\bm \tau})} - \frac{\mathcal{D}^{(1)}(\nu; {\bm \kappa}, \tau_0, {\bm \tau})}{\mathcal{D}^{(0)}(\nu; {\bm \kappa}, \tau_0, {\bm \tau}) \mathcal{D}(\nu; {\bm \kappa}, \tau_0, {\bm \tau})}.
\end{equation}
Therefore, the right-hand side of \eqref{eq:zs-arg} is equal to
\begin{align}
\label{eq:zs-arg_2}
& \frac{1}{2 \pi i} \int_{|\nu| \, = \, \nu^\star({\bm \kappa}, \tau_0, {\bm \tau})} \frac{2 \nu^{\ell + 1}}{\mathcal{D}^{(0)}(\nu; {\bm \kappa}, \tau_0, {\bm \tau})} \, {\rm d}\nu \\
& \qquad - \frac{1}{2 \pi i} \int_{|\nu| \, = \, \nu^\star({\bm \kappa}, \tau_0, {\bm \tau})} \frac{2 \nu^{\ell + 1} \, \mathcal{D}^{(1)}(\nu; {\bm \kappa}, \tau_0, {\bm \tau})}{\mathcal{D}(\nu; {\bm \kappa}, \tau_0, {\bm \tau})\mathcal{D}^{(0)}(\nu; {\bm \kappa}, \tau_0, {\bm \tau})} \, {\rm d}\nu \nonumber \\
& \qquad + \frac{1}{2 \pi i} \int_{|\nu| \, = \, \nu^\star({\bm \kappa}, \tau_0, {\bm \tau})} \frac{\nu^\ell \partial_\nu \mathcal{D}^{(1)}(\nu; {\bm \kappa}, \tau_0, {\bm \tau})}{\mathcal{D}(\nu; {\bm \kappa}, \tau_0, {\bm \tau})} \, {\rm d}\nu. \nonumber
\end{align}
We evaluate the first of these contour integrals:
\begin{align}
\label{eq:zs-int-1}
& \frac{1}{2 \pi i} \int_{|\nu| \, = \, \nu^\star({\bm \kappa}; \tau_0, {\bm \tau})} \frac{2 \nu^{\ell + 1}}{\mathcal{D}^{(0)}(\nu; {\bm \kappa}, \tau_0, {\bm \tau})} \, {\rm d}\nu \\
& \qquad = \frac{1}{2 \pi i} \int_{|\nu| \, = \, \nu^\star({\bm \kappa}; \tau_0, {\bm \tau})} \frac{2 \nu^{\ell + 1}}{\nu^2 - |\omega({\bm \kappa}; \tau_0, {\bm \tau})|^2} \, {\rm d}\nu =
\begin{cases}
0, & \ell = 1 \\
2 |\omega({\bm \kappa}; \tau_0, {\bm \tau})|^2, & \ell = 2
\end{cases} \, . \nonumber
\end{align}
The remaining contour integrals can be bounded. First, using \eqref{eq:calD-0-c-bd}, \eqref{eq:calD-1-c-bd}, and the following lower bound, valid for ${|\nu| = \nu^\star({\bm \kappa}, \tau_0, {\bm \tau})}$:
\begin{align}
\label{eq:calD-c-bd}
|\mathcal{D}(\nu; {\bm \kappa}, \tau_0, {\bm \tau})| & \geq |\mathcal{D}^{(0)}(\nu; {\bm \kappa}, \tau_0, {\bm \tau})| - |\mathcal{D}^{(1)}(\nu; {\bm \kappa}, \tau_0, {\bm \tau})| \\ 
& \geq (C_1 - D_1 \kappa^\star{}^2)|{\bm \kappa}|^4 + (1 - D_2 \tau_0^\star)|\tau_0|^2 + (C_2 - D_3 \tau^\star)|{\bm \tau}|^2, \nonumber
\end{align}
we have, for ${\ell \smallin \N}$,
\begin{align}
\label{eq:zs-int-2}
\biggl|\frac{1}{2\pi i} \int_{|\nu| \, = \, \nu^\star({\bm \kappa}; \tau_0, {\bm \tau})} \frac{2 \nu^{\ell + 1} \, \mathcal{D}^{(1)}(\nu; {\bm \kappa}, \tau_0, {\bm \tau})}{\mathcal{D}(\nu; {\bm \kappa}, \tau_0, {\bm \tau}) \mathcal{D}^{(0)}(\nu; {\bm \kappa}, \tau_0, {\bm \tau})} \, {\rm d}\nu \biggr| = O(|{\bm \kappa}|^{2\ell + 2} + |\tau_0|^{\ell + 1} + |{\bm \tau}|^{\ell + 1})
\end{align}
as $|{\bm \kappa}|$, $|\tau_0|$, ${|{\bm \tau}| \to 0}$. Finally, using \eqref{eq:dcalD-1-c-bd} and \eqref{eq:calD-c-bd} above,
\begin{equation}
\label{eq:zs-int-3}
\biggl|\frac{1}{2\pi i} \int_{|\nu| \, = \, \nu^\star({\bm \kappa}; \tau_0, {\bm \tau})} \frac{\nu^\ell \partial_\nu \mathcal{D}^{(1)}(\nu; {\bm \kappa}, \tau_0, {\bm \tau})}{\mathcal{D}(\nu; {\bm \kappa}, \tau_0, {\bm \tau})} \, {\rm d}\nu \biggr| = O(|{\bm \kappa}|^{2\ell + 2} + |\tau_0|^{\ell + 1} + |{\bm \tau}|^{\ell + 1})
\end{equation}
as $|{\bm \kappa}|$, $|\tau_0|$, ${|{\bm \tau}| \to 0}$. Collecting terms appropriately yields the assertions of Proposition \ref{prop:calD-zs-int}.
\end{proof}

Hence, by \eqref{eq:def-frak-t}, \eqref{eq:frak-d-sum}, and Proposition \ref{prop:calD-zs-int}, the analytic functions $\mathfrak{t}({\bm \kappa}; \tau_0, {\bm \tau})$ and $\mathfrak{d}({\bm \kappa}; \tau_0, {\bm \tau})$ have the expansions:
\begin{align}
\label{eq:frak-t-ex}
\mathfrak{t}({\bm \kappa}; \tau_0, {\bm \tau}) & = O(|{\bm \kappa}|^4 + |\tau_0|^2 + |{\bm \tau}|^2), \\
\label{eq:frak-d-ex}
\mathfrak{d}({\bm \kappa}; \tau_0, {\bm \tau}) & =  -|\omega({\bm \kappa}, \tau_0, {\bm \tau})|^2 + O(|{\bm \kappa}|^6 + |\tau_0|^3 + |{\bm \tau}|^3)
\end{align}
as $|{\bm \kappa}|$, $|\tau_0|$, ${|{\bm \tau}| \to 0}$. Therefore, by \eqref{eq:calD-qu-sol}, the zeros $\nu_\pm({\bm \kappa}; \tau_0, {\bm \tau})$ of $\mathcal{D}(\nu; {\bm \kappa}, \tau_0, {\bm \tau})$ are given by
\begin{equation}
\label{eq:nu-zs-ex}
\nu_\pm({\bm \kappa}; \tau_0, {\bm \tau}) = Q_4({\bm \kappa}, \tau_0, {\bm \tau}) \pm \sqrt{|\omega({\bm \kappa}, \tau_0, {\bm \tau})|^2 + Q_6({\bm \kappa}, \tau_0, {\bm \tau})}.
\end{equation}
Here, we define, for ${|{\bm \kappa}| < \kappa^\star}$, ${|\tau_0| < \tau_0^\star}$, ${|{\bm \tau}| < \tau^\star}$, the analytic functions
\begin{align}
\label{eq:def-Q4}
Q_4({\bm \kappa}; \tau_0, {\bm \tau}) & \equiv \frac{1}{2} \mathfrak{t}({\bm \kappa}, \tau_0, {\bm \tau}) \\
\label{eq:def-Q6}
Q_6({\bm \kappa}; \tau_0, {\bm \tau}) & \equiv -\mathfrak{d}({\bm \kappa}, \tau_0, {\bm \tau}) + \frac{1}{4} \mathfrak{t}({\bm \kappa}, \tau_0, {\bm \tau})^2 - |\omega({\bm \kappa}; \tau_0, {\bm \tau})|^2.
\end{align}
By Proposition \ref{prop:frak-t&d-pr}, \eqref{eq:frak-t-ex}, and \eqref{eq:frak-d-ex}, these satisfy, for ${n \smallin \{ 4, \, 6 \}}$,
\begin{align}
\label{eq:Qn-ev}
Q_n({\bm \kappa}; \tau_0, {\bm \tau}) & = Q_n(-{\bm \kappa}; \tau_0, {\bm \tau}), \\
\label{eq:Qn-bd}
Q_n({\bm \kappa}; \tau_0, {\bm \tau}) & = O(|{\bm \kappa}|^n + |\tau_0|^{n/2} + |{\bm \tau}|^{n/2}) \ \ \text{as} \ \ |{\bm \kappa}|, \, |\tau_0|, \, |{\bm \tau}| \to 0.
\end{align}

Finally, from the change of variables \eqref{eq:def-nu}, we define
\begin{equation}
\label{eq:eps-zs-ex}
\varepsilon_\pm({\bm \kappa}; \tau_0, {\bm \tau}) \equiv \beta_0 \tau_0 + (1 - \alpha_0) |{\bm \kappa}|^2 + \nu_\pm({\bm \kappa}; \tau_0, {\bm \tau}),
\end{equation}
such that the solutions to \eqref{eq:det-calM-0_3} are ${\varepsilon = \varepsilon_\pm({\bm \kappa}; \tau_0, {\bm \tau})}$.
Further, we define
\begin{equation}
\label{eq:def-Epm}
E_\pm({\bm \kappa}; \tau_0, {\bm \tau}) \equiv E_S + \varepsilon_\pm({\bm \kappa}; \tau_0, {\bm \tau})
\end{equation}
such that 
\begin{align}
E_\pm({\bm \kappa}; \tau_0, {\bm \tau}) - E_S & = \beta_0 \tau_0 + (1 - \alpha_0)|{\bm \kappa}|^2 + Q_4({\bm \kappa}, \tau_0, {\bm \tau}) \\
& \qquad \pm \sqrt{(\beta_1 \tau_1 - \alpha_1 {\bm \kappa} \cdot \sigma_1 {\bm \kappa})^2 + (\beta_2 \tau_3 - \alpha_2 {\bm \kappa} \cdot \sigma_3 {\bm \kappa})^2 + Q_6({\bm \kappa}, \tau_0, {\bm \tau})}. \nonumber
\end{align}
By Proposition \ref{prop:red-M}, $E_\pm({\bm \kappa}; \tau_0, {\bm \tau})$ are $L^2_{{\bm M} + {\bm \kappa}}$ eigenvalues of $H^{\tau_0, {\bm \tau}}$.

\begin{remark}
Observe that ${E_\pm({\bm \kappa}; \tau_0, {\bm \tau}) = E_\pm(-{\bm \kappa}; \tau_0, {\bm \tau})}$, which reflects the global band structure symmetry discussed in Proposition \ref{prop:gl-P-sym}. For this local result, the global symmetry arises first at the level of $\mathcal{M}(\varepsilon, {\bm \kappa}; \tau_0, {\bm \tau})$ (see conclusion \ref{itm:calM-ev} of Proposition \ref{prop:calM-pr}), then in $\mathcal{D}(\nu; {\bm \kappa}, \tau_0, {\bm \tau})$ (see conclusion \ref{itm:calD-ev} of Proposition \ref{prop:calD-pr}), and finally in the roots of $\mathcal{D}(\nu; {\bm \kappa}, \tau_0, {\bm \tau})$ (as discussed above).
\end{remark}

The proof of Theorem \ref{thm:M-srf} is now complete.

\bigskip

\section{Quadratic degeneracies split into Dirac points; Proof of Theorem \ref{thm:M-dgn}}
\label{sec:pf-M-dgn}

\setcounter{equation}{0}
\setcounter{figure}{0}

By conclusion \ref{itm:calM-0} of Proposition \ref{prop:red-M}, for ${|\varepsilon| < \varepsilon^\flat}$, ${|{\bm \kappa}| < \kappa^\flat}$, ${|\tau_0| < \tau_0^\flat}$, ${|{\bm \tau}| < \tau^\flat}$, ${E_S + \varepsilon}$ is an $L^2_{{\bm M} + {\bm \kappa}}$ eigenvalue of $H^{\tau_0, {\bm \tau}}$, of multiplicity two, if and only if
\begin{equation}
\label{eq:calM-0_3}
\mathcal{M}(\varepsilon; {\bm \kappa}, \tau_0, {\bm \tau}) =
 \left(\mathcal{M}_{ij}(\varepsilon; {\bm \kappa}, \tau_0, {\bm \tau})\right)_{1\le i,j\le2} = 0.
\end{equation}
Here, $\mathcal{M}(\varepsilon; {\bm \kappa}, \tau_0, {\bm \tau})$ is the ${2 \times 2}$ matrix introduced in \eqref{eq:calM-a-0}. To prove Theorem \ref{thm:M-dgn}, we construct solutions ${(\varepsilon, \, {\bm \kappa}) = (\varepsilon(\tau_0, {\bm \tau}), \, {\bm \kappa}(\tau_0, {\bm \tau}))}$ to the system \eqref{eq:calM-0_3}.

First, we note a convenient decomposition of the left-hand side of $\mathcal{M}(\varepsilon; {\bm \kappa}, \tau_0, {\bm \tau})$ (Proposition \ref{prop:calM-pr}), we have by Lemma \ref{lem:Ahermsym}, that $\mathcal{M}(\varepsilon; {\bm \kappa}, \tau_0, {\bm \tau})$ has the expansion in terms of the orthonormal basis of Pauli matrices:
\begin{equation}
\label{eq:calM-ex-2}
\mathcal{M}(\varepsilon; {\bm \kappa}, \tau_0, {\bm \tau}) = \sum_{\ell \, = \, 0}^3 h_\ell(\varepsilon, {\bm \kappa}; \tau_0, {\bm \tau}) \, \sigma_\ell.
\end{equation}
 The coefficients $h_\ell(\varepsilon, {\bm \kappa}; \tau_0, {\bm \tau})$, expressed in terms of the entries of $\mathcal{M}(\varepsilon; {\bm \kappa}, \tau_0, {\bm \tau})$ are :
\begin{equation}
\label{eq:def-hs}
\begin{aligned}
h_0(\varepsilon, {\bm \kappa}; \tau_0, {\bm \tau}) & \equiv \frac{1}{2}(\mathcal{M}_{1, 1}(\varepsilon; {\bm \kappa}, \tau_0, {\bm \tau}) + \mathcal{M}_{2, 2}(\varepsilon; {\bm \kappa}, \tau_0, {\bm \tau})), \\
h_1(\varepsilon, {\bm \kappa}; \tau_0, {\bm \tau}) & \equiv {\rm Re}(\mathcal{M}_{1, 2}(\varepsilon; {\bm \kappa}, \tau_0, {\bm \tau})), \\
h_2(\varepsilon, {\bm \kappa}; \tau_0, {\bm \tau}) & \equiv - {\rm Im}\bigl(\mathcal{M}_{1, 2}(\varepsilon; {\bm \kappa}, \tau_0, {\bm \tau})), \\
h_3(\varepsilon, {\bm \kappa}; \tau_0, {\bm \tau}) & \equiv \frac{1}{2}(\mathcal{M}_{1, 1}(\varepsilon; {\bm \kappa}, \tau_0, {\bm \tau}) - \mathcal{M}_{2, 2}(\varepsilon; {\bm \kappa}, \tau_0, {\bm \tau})).
\end{aligned}
\end{equation}
The following proposition summarizes the properties of $h_\ell(\varepsilon, {\bm \kappa};\tau_0,{\bm \tau})$, ${0 \leq \ell \leq 3}$.
\begin{proposition}
\label{prop:hs-pr}
The functions $h_\ell(\varepsilon, {\bm \kappa}; \tau_0, {\bm \tau})$, ${0 \leq \ell \leq 3}$ are analytic functions; they have power series expansions in $(\varepsilon, {\bm \kappa}; \tau_0, {\bm \tau})$ which converge uniformly for ${|\varepsilon| < \varepsilon^\flat}$, ${|{\bm \kappa}| < \kappa^\flat}$, ${|\tau_0| < \tau}_0^\flat$, ${|{\bm \tau}| < \tau^\flat}$. Additionally, these functions have the following properties:
\begin{enumerate}
\item \label{itm:hs-ev} By $\mathcal{P}$ (or $\mathcal{C}$) symmetry,
\begin{equation}
\label{eq:hs-ev}
h_\ell(\varepsilon, {\bm \kappa}; \tau_0, {\bm \tau}) = h_\ell(\varepsilon, -{\bm \kappa}; \tau_0, {\bm \tau}).
\end{equation}
\item \label{itm:h-0&3-PC} By $\mathcal{PC}$ symmetry,
\begin{align}
\label{eq:h-0&3-PC}
h_0(\varepsilon, {\bm \kappa}; \tau_0, {\bm \tau}) &= \mathcal{M}_{1, 1}(\varepsilon; {\bm \kappa}, \tau_0, {\bm \tau})\quad {\rm and}\\
h_3(\varepsilon, {\bm \kappa}; \tau_0, {\bm \tau}) &= 0. \label{eq:h30}
\end{align}
\item \label{itm:hs-ex} 
For ${0 \leq \ell \leq 2}$, $h_\ell(\varepsilon, {\bm \kappa}; \tau_0, {\bm \tau})$ have the expansions:
\begin{equation}
\label{eq:hs-ex}
\begin{aligned}
h_0(\varepsilon, {\bm \kappa}; \tau_0, {\bm \tau}) & = \varepsilon - \beta_0 \tau_0 - (1 - \alpha_0) |{\bm \kappa}|^2 + \widetilde{\mathcal{M}}_{1, 1}(\varepsilon; {\bm \kappa}, \tau_0, {\bm \tau}), \\
h_1(\varepsilon, {\bm \kappa}; \tau_0, {\bm \tau}) & = \beta_1 \tau_1 - \alpha_1 {\bm \kappa} \cdot \sigma_1 {\bm \kappa} + {\rm Re}(\widetilde{\mathcal{M}}_{1, 2}(\varepsilon; {\bm \kappa}, \tau_0, {\bm \tau})), \\
h_2(\varepsilon, {\bm \kappa}; \tau_0, {\bm \tau}) & = \beta_2 \tau_3 - \alpha_2 {\bm \kappa} \cdot \sigma_3 {\bm \kappa} - {\rm Im}(\widetilde{\mathcal{M}}_{1, 2}(\varepsilon; {\bm \kappa}, \tau_0, {\bm \tau})).
\end{aligned}
\end{equation}
The entries of $\widetilde{\mathcal{M}}(\varepsilon; {\bm \kappa}, \tau_0, {\bm \tau})$ satisfy the bounds:
\begin{equation}
\label{eq:t-calM-bd_2}
\widetilde{\mathcal{M}}_{j, k}(\varepsilon; {\bm \kappa}, \tau_0, {\bm \tau}) = O ( |\varepsilon| |{\bm \kappa}|^2 + |{\bm \kappa}|^4 + |{\bm \kappa}|^2 |\tau_0| + |{\bm \kappa}|^2 |{\bm \tau}| + |\tau_0|^2 + |{\bm \tau}|^2 )
\end{equation}
as $|\varepsilon|$, $|{\bm \kappa}|$, $|\tau_0|$, ${|{\bm \tau}| \to 0}$; see conclusion \ref{itm:calM-ex-1} of Proposition \ref{prop:calM-pr}
\end{enumerate}
\end{proposition}

\begin{proof}
These follow immediately from the definitions \eqref{eq:def-hs} and
Proposition \ref{prop:calM-pr}.
\end{proof}

\noindent Applying Proposition \ref{prop:hs-pr}, we find that the condition \eqref{eq:calM-0_3} for the existence of twofold band structure degeneracies of $H^{\tau_0, {\bm \tau}}$ in a neighborhood of $(E_S, \, {\bm M})$ is equivalent to a system of three equations:
\begin{equation}
\label{eq:hs-0}
\begin{aligned}
h_0(\varepsilon, {\bm \kappa}; \tau_0, {\bm \tau}) & = \mathcal{M}_{1, 1}(\varepsilon; {\bm \kappa}, \tau_0, {\bm \tau}) = 0, \\
h_1(\varepsilon, {\bm \kappa}; \tau_0, {\bm \tau}) & = {\rm Re}(\mathcal{M}_{1, 2}(\varepsilon; {\bm \kappa}, \tau_0, {\bm \tau})) = 0, \\
h_2(\varepsilon, {\bm \kappa}; \tau_0, {\bm \tau}) & = -{\rm Im}(\mathcal{M}_{1, 2}(\varepsilon; {\bm \kappa}, \tau_0, {\bm \tau})) = 0.
\end{aligned}
\end{equation}
In what follows, we use the expansions from conclusion \ref{itm:hs-ex} of Proposition \ref{prop:hs-pr} to construct solutions to the system \eqref{eq:hs-0} via the implicit function theorem.

\subsection{Set-up for the implicit function theorem}
\label{sec:hs-0-ift}

Introduce polar coordinates as in 
\eqref{eq:def-tau-polar}:
$ {\bm \tau} = |{\bm \tau}| \hat{\bm\tau}(\varphi) \equiv |{\bm \tau}| [\cos(\varphi), \, \sin(\varphi)]^\mathsf{T}$. Then, by part  \ref{itm:hs-ex} of Proposition \ref{prop:hs-pr}, the system \eqref{eq:hs-0} can be expressed as:
\begin{equation}
\label{eq:hs-0_2}
\begin{aligned}
\varepsilon - \beta_0 \tau_0 - (1 - \alpha_0) |{\bm \kappa}|^2 + \tilde{\mathcal{M}}_{1, 1}(\varepsilon; {\bm \kappa}, \tau_0, |{\bm \tau}| \hat{\bm \tau}(\varphi)) & = 0, \\
\beta_1 |{\bm \tau}| \cos(\varphi) - \alpha_1 {\bm \kappa} \cdot \sigma_1 {\bm \kappa} + {\rm Re}(\tilde{\mathcal{M}}_{1, 2}(\varepsilon; {\bm \kappa}, \tau_0, |{\bm \tau}| \hat{\bm \tau}(\varphi))) & = 0, \\
\beta_2 |{\bm \tau}| \sin(\varphi) - \alpha_2 {\bm \kappa} \cdot \sigma_3 {\bm \kappa} - {\rm Im}(\tilde{\mathcal{M}}_{1, 2}(\varepsilon; {\bm \kappa}, \tau_0, |{\bm \tau}| \hat{\bm \tau}(\varphi))) & = 0.
\end{aligned}
\end{equation}

For any fixed ${\varphi \smallin [-\pi, \, \pi]}$, we shall construct solutions to \eqref{eq:hs-0_2} by first rescaling ${\bm \kappa}$:
\begin{align}
\label{eq:def-t-kap}
{\bm \kappa} = \sqrt{|{\bm \tau}|} \tilde{\bm \kappa}.
\end{align}
The scaling \eqref{eq:def-t-kap} is natural since: the leading order terms of $h_\ell$, ${1 \leq \ell \leq 2}$, do not depend on $\tau_0$, are linear in $|{\bm \tau}|$, and are quadratic in ${\bm \kappa}$. Moreover, recalling conclusion \ref{itm:hs-ev} of Proposition \ref{prop:hs-pr}, the functions $h_\ell(\varepsilon, {\bm \kappa}; \tau_0, {\bm \tau})$, ${0 \leq \ell \leq 2}$, are even in ${\bm \kappa}$, implying that their power series expansions contain only even powers of ${\bm \kappa}$. \\

\noindent {\bf N.B.} For  the remainder of the proof we shall, for convenience, omit dependence on the fixed parameter ${\varphi \smallin [-\pi, \, \pi]}$ in certain expressions. \\

\noindent Substitution of the polar representation of $\bm\tau$ and \eqref{eq:def-t-kap} into \eqref{eq:hs-0_2} yields:
\begin{equation}
\label{eq:hs-0_3}
\begin{aligned}
\varepsilon - \beta_0 \tau_0 - (1 - \alpha_0) |\tilde{\bm \kappa}|^2 |{\bm \tau}| + g_0(\varepsilon, \tilde{\bm \kappa}; \tau_0, |{\bm \tau}|) & = 0, \\
( \beta_1 \cos(\varphi) - \alpha_1 \tilde{\bm \kappa} \cdot \sigma_1 \tilde{\bm \kappa} ) |{\bm \tau}| + g_1(\varepsilon, \tilde{\bm \kappa}; \tau_0, |{\bm \tau}|) & = 0, \\
( \beta_2 \sin(\varphi) - \alpha_2 \tilde{\bm \kappa} \cdot \sigma_3 \tilde{\bm \kappa} ) |{\bm \tau}| + g_2(\varepsilon, \tilde{\bm \kappa}; \tau_0, |{\bm \tau}|) & = 0,
\end{aligned}
\end{equation}
where
\begin{equation}
\label{eq:def-gs}
\begin{aligned}
g_0(\varepsilon, \tilde{\bm \kappa}; \tau_0, |{\bm \tau}|) & \equiv \tilde{\mathcal{M}}_{1, 1}(\varepsilon, \sqrt{|{\bm \tau}|} \tilde{\bm \kappa}; \tau_0, |{\bm \tau}| \hat{\bm \tau}), \\
g_1(\varepsilon, \tilde{\bm \kappa}; \tau_0, |{\bm \tau}|) & \equiv {\rm Re}(\tilde{\mathcal{M}}_{1, 2}(\varepsilon, \sqrt{|{\bm \tau}|} \tilde{\bm \kappa}; \tau_0, |{\bm \tau}| \hat{\bm \tau})), \\
g_2(\varepsilon, \tilde{\bm \kappa}; \tau_0, |{\bm \tau}|) & \equiv -{\rm Im}(\tilde{\mathcal{M}}_{1, 2}(\varepsilon, \sqrt{|{\bm \tau}|} \tilde{\bm \kappa}; \tau_0, |{\bm \tau}| \hat{\bm \tau})).
\end{aligned}
\end{equation}

\noindent The expressions $g_\ell(\varepsilon, \tilde{\bm \kappa}; \tau_0, s)$, ${0 \leq \ell \leq 2}$, where $s$ is allowed to take on negative values, define analytic functions in a neighborhood of ${(\tau_0, s) = (0, 0)}$. We shall first solve: 
\begin{equation}
\label{eq:hs-0_4}
\begin{aligned}
\varepsilon - \beta_0 \tau_0 - (1 - \alpha_0) |\tilde{\bm \kappa}|^2 s + g_0(\varepsilon, \tilde{\bm \kappa}; \tau_0, s) & = 0, \\
( \beta_1 \cos(\varphi) - \alpha_1 \tilde{\bm \kappa} \cdot \sigma_1 \tilde{\bm \kappa} ) s + g_1(\varepsilon, \tilde{\bm \kappa}; \tau_0, s) & = 0, \\
( \beta_2 \sin(\varphi) - \alpha_2 \tilde{\bm \kappa} \cdot \sigma_3 \tilde{\bm \kappa} ) s + g_2(\varepsilon, \tilde{\bm \kappa}; \tau_0, s) & = 0,
\end{aligned}
\end{equation}
for ${(\varepsilon, \, \tilde{\bm \kappa}) = (\varepsilon(\tau_0, s), \, \tilde{\bm \kappa}(\tau_0, s))}$ in a neighborhood of ${(\tau_0, \, s) = (0, \, 0)}$. The desired solution of \eqref{eq:hs-0_3} is then ${(\varepsilon, \, \tilde{\bm \kappa}) = (\varepsilon(\tau_0, |\bm\tau|), \, \tilde{\bm \kappa}(\tau_0, |\bm\tau|))}$.

\begin{proposition}
\label{prop:gs-pr}
The  functions $g_\ell(\varepsilon, \tilde{\bm \kappa}; \tau_0, s)$, ${0 \leq \ell \leq 2}$ are analytic and  have the following properties:
\begin{enumerate}
\item \label{itm:g-0} For ${\ell = 0}$, $g_0(\varepsilon, \tilde{\bm \kappa}; \tau_0, s)$ satisfies:
\begin{equation}
\label{eq:g-0-bd}
g_0(\varepsilon, \tilde{\bm \kappa}; \tau_0, s) = O ( |\tau_0|^2 + |s| ) \ \ \text{as} \ \ |\tau_0|, \, |s| \to 0,
\end{equation}
uniformly for $(\varepsilon, \, \tilde{\bm \kappa})$ in any sufficiently small neighborhood of ${(\varepsilon, \, \tilde{\bm \kappa}) = (0, \, {\bm 0})}$. Furthermore,
\begin{equation}
\label{eq:f-bd}
g_0(\varepsilon, \tilde{\bm \kappa}; \tau_0, 0) \equiv f(\varepsilon; \tau_0) = O ( |\tau_0|^2 + |\varepsilon||\tau_0|^2 + |\tau_0|^3 ) \ \ \text{as} \ \ |\varepsilon|, \, |\tau_0| \to 0.
\end{equation}
Here, $f(\varepsilon; \tau_0)$ is analytic for $(\varepsilon; \tau_0)$ in a neighborhood of ${(\varepsilon; \tau_0) = (0; 0)}$ and does not depend on $\tilde{\bm \kappa}$ (or $\varphi$). 

\item \label{itm:g-1&2} For ${\ell \smallin \{ 1, \, 2 \}}$, there exist analytic functions $\tilde{g}_\ell(\varepsilon, \tilde{\bm \kappa}; \tau_0, s)$ such that
\begin{equation}
\label{eq:g-1&2-ex}
g_\ell(\varepsilon, {\bm \kappa}_1; \tau_0, s) = s \, \tilde{g}_\ell(\varepsilon, \tilde{\bm \kappa}; \tau_0, s).
\end{equation}
The functions $\tilde{g}_\ell(\varepsilon, \tilde{\bm \kappa}; \tau_0, s)$ have expansions
\begin{align}
\label{eq:g-1&2-1-ex}
\tilde{g}_\ell(\varepsilon, \tilde{\bm \kappa}; \tau_0, s) & = \tilde{g}_\ell^{(0)}(\varepsilon, \tilde{\bm \kappa}) + \tilde{g}_\ell^{(1)}(\varepsilon, \tilde{\bm \kappa}; \tau_0, s), \quad \text{where} \\
\label{eq:g-1&2-0&1-bd}
\tilde{g}_\ell^{(0)}(\varepsilon, \tilde{\bm \kappa}) & = O(|\varepsilon|) \ \ \text{as} \ \ |\varepsilon| \to 0 \\ \tilde{g}_\ell^{(1)}(\varepsilon, \tilde{\bm \kappa}; \tau_0, s) & = O(|\tau_0| + |s|) \ \ \text{as} \ \ |\tau_0|, \, |s| \to 0,
\end{align}
uniformly for $(\varepsilon, \, \tilde{\bm \kappa})$ in any sufficiently small neighborhood of ${(\varepsilon, \, \tilde{\bm \kappa}) = (0, \, {\bm 0})}$.
\end{enumerate}
\end{proposition}

\begin{proof}
The proof of Proposition \ref{prop:gs-pr} can be found in Appendix \ref{apx:gs-pr}.
\end{proof}

We summarize the content of this subsection with the following:

\begin{proposition}
\label{prop:hs-0_fin}
For ${|\tau_0| < \tau_0^\flat}$, ${|{\bm \tau}| < \tau^\flat}$, and any fixed ${\varphi \smallin [-\pi, \, \pi]}$, the energy-quasimomentum pair 
\[ {(E, \, {\bm k}) = (E_S + \varepsilon, \, {\bm M} + {\bm \kappa})}\]
is a twofold degenerate eigenpair of of $H^{\tau_0, {\bm \tau}}$ if and only if
\begin{equation}
\label{eq:calM-0-sol}
(\varepsilon, \, {\bm \kappa}) = (\varepsilon(\tau_0, s), \, \sqrt{s} \tilde{\bm \kappa}(\tau_0, s)) \bigr\rvert_{s \, = \, |{\bm \tau}|}
\end{equation}
where ${(\varepsilon, \, \tilde{\bm \kappa}) = (\varepsilon(\tau_0, s), \, \tilde{\bm \kappa}(\tau_0, s))}$ is a solution to the system of equations:
\begin{equation}
\label{eq:hs-0_fin}
\begin{aligned}
\varepsilon - \beta_0 \tau_0 - (1 - \alpha_0) |\tilde{\bm \kappa}|^2 s + g_0(\varepsilon, \tilde{\bm \kappa}; \tau_0, s)  & = 0, \\
s \cdot \bigl( \beta_1 \cos(\varphi) - \alpha_1 \tilde{\bm \kappa} \cdot \sigma_1 \tilde{\bm \kappa}  +  \tilde{g}_1(\varepsilon, \tilde{\bm \kappa}; \tau_0, s) \bigr) & = 0, \\
s \cdot \bigl( \beta_2 \sin(\varphi) - \alpha_2 \tilde{\bm \kappa} \cdot \sigma_3 \tilde{\bm \kappa} + \tilde{g}_2(\varepsilon, \tilde{\bm \kappa}; \tau_0, s) \bigr) & = 0.
\end{aligned}
\end{equation}
\end{proposition}

\begin{proof}
This follows immediately after substituting \eqref{eq:g-1&2-ex} into \eqref{eq:hs-0_4}.
\end{proof}

\subsection{The set of solutions to (\ref{eq:hs-0_fin}) for $\tau_0$, $s$ sufficiently small}
\label{sec:hs-0-sol}

We divide our study of the system \eqref{eq:hs-0_fin} into two cases: ${s = 0}$ and ${s \neq 0}$.

If ${s = 0}$, the second and third equations of \eqref{eq:hs-0_fin} are satisfied identically, and \eqref{eq:hs-0_fin} reduces to the single nonlinear equation:
\begin{equation}
\label{eq:F-0}
F(\varepsilon; \tau_0) \equiv \varepsilon - \beta_0 \tau_0 + f(\varepsilon; \tau_0) = 0.
\end{equation}
Here, $f(\varepsilon; \tau_0)$ is defined in conclusion \ref{itm:g-0} of Proposition \ref{prop:gs-pr} and does not depend on $\tilde{\bm \kappa}$ (or $\varphi$). The function $F(\varepsilon; \tau_0)$ is analytic in a neighborhood of ${(\varepsilon; \tau_0) = (0; 0)}$.

For ${s \neq 0}$, \eqref{eq:hs-0_fin} is equivalent to the system:
\begin{equation}
\label{eq:G-0}
\begin{aligned}
G_1(\varepsilon, \tilde{\bm \kappa}; \tau_0, s) & \equiv \varepsilon - \beta_0 \tau_0 - (1 - \alpha_0) |\tilde{\bm \kappa}|^2 s + g_0(\varepsilon, \tilde{\bm \kappa}; \tau_0, s) = 0, \\
G_2(\varepsilon, \tilde{\bm \kappa}; \tau_0, s) & \equiv \beta_1 \cos(\varphi) - \alpha_1 \tilde{\bm \kappa} \cdot \sigma_1 \tilde{\bm \kappa} + \tilde{g}_1(\varepsilon, \tilde{\bm \kappa}; \tau_0, s) = 0, \\
G_3(\varepsilon, \tilde{\bm \kappa}; \tau_0, s) & \equiv \beta_2 \sin(\varphi) - \alpha_2 \tilde{\bm \kappa} \cdot \sigma_3 \tilde{\bm \kappa} + \tilde{g}_2(\varepsilon, \tilde{\bm \kappa}; \tau_0, s) = 0.
\end{aligned}
\end{equation}
The vector-valued function $G(\varepsilon, \tilde{\bm \kappa}; \tau_0, s)$, with components defined in \eqref{eq:G-0}, is analytic in a neighborhood of ${(\varepsilon, \tilde{\bm \kappa}; \tau_0, s) = (0, {\bm 0}; 0, 0)}$.

\subsubsection{Solutions to (\ref{eq:hs-0_fin}) for $\tau_0$ sufficiently small, ${s = 0}$}
\label{sec:F-0}

We first solve \eqref{eq:F-0}, for $\tau_0$ in a neighborhood of ${\tau_0 = 0}$, using the implicit function theorem. 
For ${\tau_0 = 0}$, \eqref{eq:F-0} reads:
$F(\varepsilon; 0) = \varepsilon = 0$.
Clearly, ${\varepsilon = 0}$ is the only solution. Furthermore, 
$\partial_\varepsilon F(0; 0) = 1 \neq 0$.
Hence, by the implicit function theorem, there exists ${\tau_0^{\star\star} > 0}$, and a unique, analytic function $\varepsilon_M(\tau_0)$, defined for $|\tau_0| < \tau_0^{\star \star}$, such that
\begin{equation}
\label{eq:F-0-ift}
\varepsilon_M(0) = 0, \quad F(\varepsilon_M(\tau_0); \tau_0) = 0,
\quad|\tau_0|<\tau_0^{\star\star}.
\end{equation}
Expanding $\varepsilon_M(\tau_0)$ about $\tau_0 = 0$ yields:$
\varepsilon_M(\tau_0) = \beta_0 \tau_0 + \varepsilon_M^{(1)}(\tau_0)$, where
$
\varepsilon_M^{(1)}(\tau_0) = O(|\tau_0|^2) \ \ \text{as} \ \ |\tau_0| \to 0$.
Here, the coefficient of 
$\tau_0$ is obtained via implicit differentiation of \eqref{eq:F-0-ift}. Finally, define
\begin{align}
\label{eq:def-EM_2}
E_M(\tau_0) & \equiv E_S + \varepsilon_M(\tau_0) = E_S + \beta_0 \tau_0 + \varepsilon_M^{(1)}(\tau_0).
\end{align}
By Proposition \ref{prop:hs-0_fin}, the energy-quasimomentum pair ${(E_M(\tau_0), \, {\bm M})}$ of $H^{\tau_0, {\bm 0}}$ is twofold-degenerate. 
 
\begin{proposition}
\label{prop:M-dgn-quad}
${(E_M(\tau_0), {\bm M})}$ is a quadratic band degeneracy point of $H^{\tau_0, {\bm 0}}$.
\end{proposition}

\begin{proof}
The proof of Proposition \ref{prop:M-dgn-quad} can be found in Appendix \ref{apx:M-dgn-quad}.
\end{proof}

The proof of conclusion \ref{itm:M-dgn-1} of Theorem \ref{thm:M-dgn} is now complete.

\subsubsection{Solutions to (\ref{eq:hs-0_fin}) for $\tau_0$, ${s \neq 0}$ sufficiently small}
\label{sec:G-0}

We now solve \eqref{eq:G-0} for ${(\tau_0, \, s)}$ in a neighborhood of ${(\tau_0, \, s) = (0, \, 0)}$. 
 For ${(\tau_0, \, s) = (0, \, 0)}$, \eqref{eq:G-0} reduces to:
\begin{equation}
\label{eq:G-0-0}
\begin{aligned}
\varepsilon & = 0, \\
\smash{\beta_1 \cos(\varphi) - \alpha_1 \tilde{\bm \kappa} \cdot \sigma_1 \tilde{\bm \kappa} + \tilde{g}_1^{(0)}(\varepsilon, \tilde{\bm \kappa})} & = 0, \\
\smash{\beta_2 \sin(\varphi) - \alpha_2 \tilde{\bm \kappa} \cdot \sigma_3 \tilde{\bm \kappa} + \tilde{g}_2^{(0)}(\varepsilon, \tilde{\bm \kappa})} & = 0.
\end{aligned}
\end{equation}
The functions $\smash{\tilde{g}_\ell^{(0)}(\varepsilon, \tilde{\bm \kappa})}$, ${\ell \smallin \{ 1, \, 2 \}}$, defined in conclusion \ref{itm:g-1&2} of Proposition \ref{prop:gs-pr}, satisfy ${\smash{\tilde{g}_\ell^{(0)}(\varepsilon, \tilde{\bm \kappa}) = O(|\varepsilon|)}}$ as ${|\varepsilon| \to 0}$. Therefore, the system \eqref{eq:G-0-0} has two solutions:
\begin{equation}
\label{eq:G-0-0-sol}
(\varepsilon, \, \tilde{\bm \kappa}) = (0, \, \pm {\bm \kappa}_D^{(0)}(\varphi)),
\end{equation}
 where ${\bm \kappa}^{(0)}_D(\varphi)$ is defined in \eqref{eq:def-kap-D-0-1}.

Next, we proceed to construct a continuation of the solution $(\varepsilon, \, \tilde{\bm \kappa}) = (0, \, {\bm \kappa}_D^{(0)}(\varphi))$ for ${(\tau_0, \, s)}$ sufficiently small; the solution ${(\varepsilon, \, \tilde{\bm \kappa}) = (0, \, -{\bm \kappa}_D^{(0)}(\varphi))}$ can be treated analogously. The Jacobian of $G(\varepsilon, \tilde{\bm \kappa}; 0, 0)$, evaluated at ${(\varepsilon, \, \tilde{\bm \kappa}) = (0, \, {\bm \kappa}_D^{(0)}(\varphi))}$, is
\begin{equation}
\label{eq:D-G-0-0}
D_{(\varepsilon, \tilde{\bm \kappa})} G(0, {\bm \kappa}_D^{(0)}(\varphi); 0, 0) =
\begin{bmatrix}
1 & 0 & 0 \\
\partial_\varepsilon \tilde{g}_1^{(0)}(0, {\bm \kappa}_D^{(0)}(\varphi)) & \alpha_1 \kappa_D^{(0)}(\varphi)_2 & \alpha_1 \kappa_D^{(0)}(\varphi)_1 \\
\partial_\varepsilon \tilde{g}_2^{(0)}(0, {\bm \kappa}_D^{(0)}(\varphi)) & \alpha_2 \kappa_D^{(0)}(\varphi)_1 & -\alpha_2 \kappa_D^{(0)}(\varphi)_2
\end{bmatrix} \!\! .
\end{equation}
This yields:
\begin{equation}
\label{eq:det-D-G-0-0}
\det(D_{(\varepsilon, \tilde{\bm \kappa})} G(0, {\bm \kappa}_D^{(0)}(\varphi); 0, 0)) = - \alpha_1 \alpha_2 |{\bm \kappa}_D^{(0)}(\varphi)|^2 \neq 0.
\end{equation}
Indeed, $\alpha_1$, ${\alpha_2 \neq 0}$ by property \ref{itm:quad-dgn-5}, and $|{\bm \kappa}_D^{(0)}(\varphi)|$ is defined and non-zero  by hypotheses  \ref{itm:quad-dgn-5} and \ref{itm:quad-dgn-6}. Therefore $D_{(\varepsilon, \tilde{\bm \kappa})} G(0, {\bm \kappa}_D^{(0)}(\varphi); 0, 0)$ is invertible. By the implicit function theorem, there exist $\tau_0^{\star\star}$, ${\tau^{\star\star} > 0}$, and unique, analytic functions $\varepsilon_D(\tau_0, s)$, $\tilde{\bm \kappa}_D(\tau_0, s)$, defined for ${|\tau_0| < \tau_0^{\star \star}}$, ${|s| < \tau^{\star \star}}$, such that
\begin{equation}
\label{eq:G-0-ift}
\varepsilon_D(0, 0) = 0, \quad \tilde{\bm \kappa}_D(0, 0) = {\bm \kappa}_D^{(0)}(\varphi), \quad G(\varepsilon_D(\tau_0, s), \tilde{\bm \kappa}_D(\tau_0, s); \tau_0, s) = 0.
\end{equation}
Note that, for all ${\varphi \smallin [-\pi, \, \pi]}$, we have the lower bound:
\begin{equation}
\label{eq:G-0-ift-bd}
|\det(D_{(\varepsilon, \, \tilde{\bm \kappa})} G(0, {\bm \kappa}_D^{(0)}(\varphi); 0, 0))| \geq C > 0,
\end{equation}
where $C$ does not depend on $\varphi$. Therefore, $\tau_0^{\star\star}$, $\tau^{\star\star}$ may be chosen independently of $\varphi$. Expanding $\varepsilon_D(\tau_0, s)$, $\tilde{\bm \kappa}_D(\tau_0, s)$ about ${(\tau_0, \, s) = (0, \, 0)}$ yields:
\begin{align}
\label{eq:G-0-tay-1}
\varepsilon_D(\tau_0, s) & = \beta_0 \tau_0 + (1 - \alpha_0) |{\bm \kappa}_D^{(0)}(\varphi)|^2 s + \varepsilon_D^{(1)}(\tau_0, s), \\
\tilde{\bm \kappa}_D(\tau_0, s) & = {\bm \kappa}_D^{(0)}(\varphi) + {\bm \kappa}_D^{(1)}(\tau_0, s),
\end{align}
where  $\varepsilon_D^{(1)}(\tau_0, s)$, ${\bm \kappa}_D^{(1)}(\tau_0, s)$, are defined and analytic for ${|\tau_0| < \tau_0^{\star \star}}$, ${|s| < \tau^{\star \star}}$, and satisfy:
\begin{equation}
\label{eq:G-0-tay-2}
\varepsilon_D^{(1)}(\tau_0, s) = O(|\tau_0|^2 + |s|^2), \quad |{\bm \kappa}_D^{(1)}(\tau_0, s)| = O(|\tau_0| + |s|) \ \ \text{as} \ \ |\tau_0|, \, |s| \to 0.
\end{equation}
The coefficients of the linear terms of \eqref{eq:G-0-tay-1} are obtained via implicit differentiation of \eqref{eq:G-0-ift}.  Finally, we define:
\begin{align}
\label{eq:def-ED_2}
E_D(\tau_0, |{\bm \tau}|; \varphi) & \equiv E_S + \varepsilon_D(\tau_0, |{\bm \tau}|; \varphi) \\
& = E_S + \beta_0 \tau_0 + (1 - \alpha_0) |{\bm \kappa}_D^{(0)}|^2 |{\bm \tau}| + \varepsilon_D^{(1)}(\tau_0, |{\bm \tau}|; \varphi), \nonumber \\
\label{eq:def-Dp}
{\bm D}^+(\tau_0, |{\bm \tau}|; \varphi) & \equiv {\bm M} + \sqrt{|{\bm \tau}|} \, \tilde{\bm \kappa}_D(\tau_0, |{\bm \tau}|; \varphi) \\
& = {\bm M} + \sqrt{|{\bm \tau}|} ( {\bm \kappa}_D^{(0)}(\varphi) + {\bm \kappa}_D^{(1)}(\tau_0, |{\bm \tau}|; \varphi) ). \nonumber
\end{align}
By Proposition \ref{prop:hs-0_fin}, the energy-quasimomentum pair $(E_D(\tau_0, |{\bm \tau}|; \varphi), \, {\bm D}^+(\tau_0, |{\bm \tau}|; \varphi))$ of $H^{\tau_0, {\bm \tau}}$ is twofold degenerate.

A second twofold degeneracy can be obtained analogously, beginning from the solution ${(\varepsilon, \, \tilde{\bm \kappa}) = (0, \, -{\bm \kappa}_D^{(0)}(\varphi))}$ of \eqref{eq:G-0-0-sol} and applying the implicit function theorem. On the other hand, from conclusion \ref{itm:calM-ev} of Proposition \ref{prop:calM-pr}, observe that another solution to \eqref{eq:calM-0_3} is immediately given by
\begin{equation}
\label{eq:calM-0-sol-3}
(\varepsilon, \, {\bm \kappa}) = (\varepsilon_D(\tau_0, |{\bm \tau}|), \, -{\bm \kappa}_D(\tau_0, |{\bm \tau}|)).
\end{equation}
By uniqueness, the two solutions obtained via each of these methods coincide. We therefore define:
\begin{align}
\label{eq:def-Dm}
{\bm D}^-(\tau_0, |{\bm \tau}|; \varphi) & \equiv {\bm M} - \sqrt{|{\bm \tau}|} \, \tilde{\bm \kappa}_D(\tau_0, |{\bm \tau}|; \varphi) \\
& = {\bm M} - \sqrt{|{\bm \tau}|} \, ( {\bm \kappa}_D^{(0)}(\varphi) + {\bm \kappa}_D^{(1)}(\tau_0, |{\bm \tau}|; \varphi) ). \nonumber
\end{align}

\begin{proposition}
\label{prop:M-dgn-dir}
${(E_D(\tau_0, |{\bm \tau}|; \varphi), {\bm D}^\pm(\tau_0, |{\bm \tau}|; \varphi))}$ are Dirac points of $H^{\tau_0, {\bm \tau}}$.
\end{proposition}

\begin{proof}
The proof of Proposition \ref{prop:M-dgn-dir} can be found in Appendix \ref{apx:M-dgn-dir}.
\end{proof}

\noindent This concludes the proof of conclusion \ref{itm:M-dgn-2} of Theorem \ref{thm:M-dgn}.

\subsection{Special deformations; Proof of Theorems \ref{thm:M-dgn-ref-1} and \ref{thm:M-dgn-ref-2}}
\label{sec:pf-M-dgn-ref}

The proofs of Theorems \ref{thm:M-dgn-ref-1} and \ref{thm:M-dgn-ref-2} proceed analogously to that of conclusion \ref{itm:M-dgn-2} of Theorem \ref{thm:M-dgn}.

In these cases, the solutions to \eqref{eq:hs-0} are more easily constructed due to the presence of additional symmetries; see the table in Figure \ref{fig:deform-1}. The following proposition, proved in Appendix \ref{apx:M-dgn-ref}, indicates simplifications arising from  additional symmetries.

\begin{proposition}
\label{prop:M-dgn-ref}
Consider the functions $h_\ell(\varepsilon, {\bm \kappa}; \tau_0, {\bm \tau})$, ${\ell \smallin \{ 1, \, 2 \}}$, defined in \eqref{eq:def-hs}. The following hold:
\begin{enumerate}
\item \label{itm:M-dgn-ref-1} Assume ${{\bm \tau} = (\tau_1, 0)}$ with ${\tau_1 \neq 0}$. Then, if ${\bm \kappa}$ satisfies $\sigma_1 {\bm \kappa} = \pm {\bm \kappa}$, it follows that
\begin{equation}
\label{eq:h2-0}
h_2(\varepsilon, {\bm \kappa}; \tau_0, {\bm \tau}) = 0.
\end{equation}
\item \label{itm:M-dgn-ref-2} Assume ${{\bm \tau} = (0, \tau_3)}$ with ${\tau_3 \neq 0}$. Then, if ${\bm \kappa}$ satisfies $\sigma_3 {\bm \kappa} = \pm {\bm \kappa}$, it follows that
\begin{equation}
h_1(\varepsilon, {\bm \kappa}; \tau_0, {\bm \tau}) = 0.
\end{equation}
\end{enumerate}
\end{proposition}

\begin{proof}
The proof can be found in Appendix \ref{apx:M-dgn-ref}.
\end{proof}

\noindent Proposition \ref{prop:M-dgn-ref} implies that, following the proof of Theorem \ref{thm:M-dgn}, the ansatz \eqref{eq:def-t-kap} can be improved by setting ${\bm \kappa}$ to be proportional to one of the two eigenvectors of either $\sigma_1$ or $\sigma_3$, immediately eliminating one equation from the system \eqref{eq:hs-0}.

Note that, in each case, Proposition \ref{prop:M-dgn-ref} provides two possibilities for the direction of ${\bm \kappa}$. However, only one of these two possible directions will lead to a solution of \eqref{eq:hs-0}, which can be determined at leading order: For concreteness, consider the case ${{\bm \tau} = (\tau_1, 0)}$ with ${\tau_1 \neq 0}$. By part \ref{itm:M-dgn-ref-1} of Proposition \ref{prop:M-dgn-ref}, if ${{\bm \kappa} = \pm \sigma_1 {\bm \kappa}}$, then ${h_2(\varepsilon; {\bm \kappa}; \tau_0, {\bm \tau}) = 0}$. Let us write ${\sigma_1 {\bm \kappa} = \chi {\bm \kappa}}$, with ${\chi = \pm 1}$ to be determined. Next, by the expansion \eqref{eq:hs-ex},
\begin{equation}
h_1(\varepsilon, {\bm \kappa}; \tau_0, {\bm \tau}) = \beta_1 \tau_1 - \alpha_1 {\bm \kappa} \cdot \sigma_1 {\bm \kappa} + {\rm Re}\bigl( \tilde{\mathcal{M}}_{1, 2}(\varepsilon; {\bm \kappa}, \tau_0, {\bm \tau}) \bigr) = \beta_1 \tau_1 - \chi \alpha_1 |{\bm \kappa}|^2 + {\rm Re}\bigl( \tilde{\mathcal{M}}_{1, 2}(\varepsilon; {\bm \kappa}, \tau_0, {\bm \tau}) \bigr).
\end{equation}
Neglecting higher order terms, we see that ${h_1(\varepsilon, {\bm \kappa}; \tau_0, {\bm \tau}) = 0}$ to leading order if and only if
\begin{equation}
|{\bm \kappa}|^2 = \frac{\beta_1 \tau_1}{\chi \alpha_1}.
\end{equation}
Therefore, we must chose ${\chi = {\rm sgn}(\alpha_1 \beta_1 \tau_1)}$ to obtain a real solution. A completely analogous consideration holds in the case ${{\bm \tau} = (0, \tau_3)}$ with ${\tau_3 \neq 0}$, considered by part \ref{itm:M-dgn-ref-2}.

In each case, we obtain a pair of Dirac points by following the remainder of the proof of Theorem \ref{thm:M-dgn}. However, by conclusion \ref{itm:M-dgn-2} of Theorem \ref{thm:M-dgn}, there are only two Dirac points in this neighborhood of ${(E_S, {\bm M})}$, so these must coincide. 

\bigskip

\section{Breaking parity or time-reversal symmetries}
\label{sec:break-PC}

\setcounter{equation}{0}
\setcounter{figure}{0}

So far, we have studied band structure degeneracies of deformed Schr\"{o}dinger operators which commute with the parity ($\mathcal{P}$) and time-reversal ($\mathcal{C}$) operators; see Definition \ref{def:syms}. In this section, we discuss how these band structure degeneracies are impacted by perturbations which break either $\mathcal{P}$ or $\mathcal{C}$ symmetry. In all cases, we show that the formerly degenerate bands, which touch either at a quadratic band degeneracy point or at pairs of Dirac points, become locally nondegenerate, or locally {\it gapped.} These gapped bands possess topological information which relates to edge structures via the {\it bulk-edge correspondence;} see Section \ref{sec:edge-states} for discussion.

We work with the general periodic, elliptic operator $L$ introduced in \eqref{eq:def-L}:
\begin{equation}
\label{eq:def-L_2}
L = -\nabla \cdot A \nabla + V,
\end{equation}
where $A$ is symmetric and positive definite. Note that the ``pushforward'' deformed Schr\"{o}dinger operators ${T_* H}$, in \eqref{eq:def-TH}, are of this form. Further, observe that $L$ is $\mathcal{P}$- and $\mathcal{C}$-symmetric, and therefore $\mathcal{PC}$-symmetric. We consider two scenarios: 
\begin{enumerate}
\item $L$ has a quadratic band degeneracy point (arising, e.g., from Theorem \ref{thm:M-dgn}, part \ref{itm:M-dgn-1} when ${L = T_* H}$). Nearby, the band structure of $L$ is approximated by an effective operator $L^{\bm M}_{\rm eff}$ with Fourier symbol
\begin{equation}
\label{eq:Leff-quad}
L^{\bm M}_{\rm eff}({\bm \kappa}) =  (1 - \alpha_0) |{\bm \kappa}|^2 \, I - \alpha_1 ({\bm \kappa} \cdot \sigma_1 {\bm \kappa}) \, \sigma_1 - \alpha_2 ({\bm \kappa} \cdot \sigma_3 {\bm \kappa}) \, \sigma_2.
\end{equation}
Here, ${{\bm \kappa} = {\bm k} - {\bm M}}$ is the quasimomentum displacement from ${\bm M}$.

\item $L$ has a symmetric pair of Dirac points (arising, e.g., from Theorem \ref{thm:M-dgn}, part \ref{itm:M-dgn-2} when ${L = T_* H}$). Near each, the band structure of $L$ is approximated by effective operators $L^{{\bm D}^\pm}_{\rm eff}$ with Fourier symbols
\begin{align}
\label{eq:Leff-dir-p}
L^{{\bm D}^+}_{\rm eff}({\bm q}) & = ({\bm \gamma}^+_0 \cdot {\bm q}) \, I + ({\bm \gamma}^+_1 \cdot {\bm q}) \, \sigma_1 + ({\bm \gamma}^+_2 \cdot {\bm q}) \, \sigma_2, \\
\label{eq:Leff-dir-m}
L^{{\bm D}^-}_{\rm eff}({\bm q}) & = ({\bm \gamma}^-_0 \cdot {\bm q}) \, I + ({\bm \gamma}^-_1 \cdot {\bm q}) \, \sigma_1 + ({\bm \gamma}^-_2 \cdot {\bm q}) \, \sigma_2 = - L_{\rm eff}^{{\bm D}^+}({\bm q}).
\end{align}
Here, ${{\bm q} = {\bm k} - {\bm D}^\pm}$ is the quasimomentum displacement from either ${\bm D}^\pm$.
\end{enumerate}

\begin{remark}
\label{rmk:Leff-break-PC}
For brevity, we will demonstrate the impact of $\mathcal{P}$- or $\mathcal{C}$-breaking perturbations only on these effective operators; it is possible to rigorously justify the approximation by effective operators using a strategy similar to that presented in Sections \ref{sec:set-up} and \ref{sec:pf-M-srf}. Detailed proofs can be found in \cite{chaban2025thesis}.
\end{remark}

In all cases below, we derive perturbed effective operators $L_{\rm eff}$ with Fourier symbols of the general form
\begin{equation}
L_{\rm eff}({\bm \kappa}) = h_0({\bm \kappa}) I + {\bm h}({\bm \kappa}) \cdot {\bm \sigma}.
\end{equation}
Here, ${\bm h}({\bm \kappa}) = [h_1({\bm \kappa}), \, h_2({\bm \kappa}), \, h_3({\bm \kappa})]^\mathsf{T}$ is a vector-valued function and ${\bm \sigma}$ is a ``vector'' of Pauli matrices; see \eqref{eq:def-pauli-vec}. The two dispersion relations of $L_{\rm eff}$ are given by
\begin{equation}
\varepsilon_\pm({\bm \kappa}) = h_0({\bm \kappa}) \pm |{\bm h}({\bm \kappa})|.
\end{equation}

\begin{proposition}
\label{prop:Leff-nondgn}
{\rm (Condition for nondegeneracy of bands of $L_{\rm eff}$.)}
The two dispersion surfaces of $L_{\rm eff}$ are nondegenerate if and only if ${{\bm h}({\bm \kappa}) \neq {\bm 0}}$ for all ${\bm \kappa}$.
\end{proposition}

\begin{remark}
{\rm (Spectral gaps.)}
Note that nondegenerate dispersion surfaces does not imply that the effective operator $L_{\rm eff}$ has a spectral gap; see Figure \ref{fig:break-PC-1a} for a counterexample.
\end{remark}

The remainder of the section is organized as follows: In Section \ref{sec:break-P}, we break $\mathcal{P}$ symmetry by perturbing $L$ with an odd potential, then discuss the effect on quadratic band degeneracies (Section \ref{sec:break-P-quad}) and symmetric pairs of Dirac points (Section \ref{sec:break-P-dir}). In Section \ref{sec:break-C}, we break $\mathcal{C}$ symmetry by perturbing $L$ with a magneto-optic term, then similarly discuss the effect on quadratic band degeneracies (Section \ref{sec:break-C-quad}) and symmetric pairs of Dirac points (Section \ref{sec:break-C-dir}).

\begin{figure}[t]
\centering
\includegraphics[height = 4.3cm]{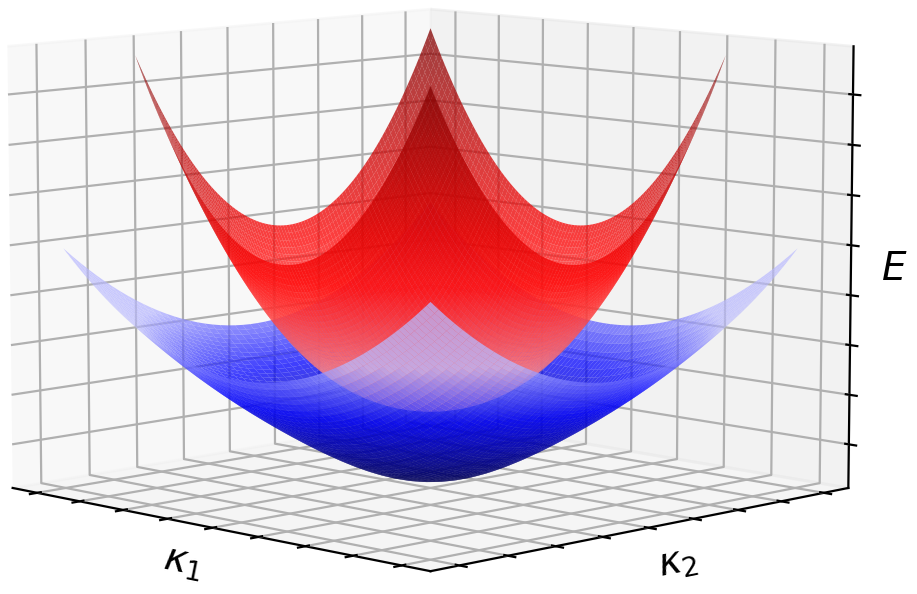}
\vspace{0.3cm}
\caption{Two nondegenerate dispersion surfaces need not give rise to a gap in the energy spectrum.}
\label{fig:break-PC-1a}
\end{figure}

\subsection{Breaking $\mathcal{P}$ symmetry}
\label{sec:break-P}

Let $W$ denote a nonzero smooth function which is $\Z^2$-periodic, real-valued, and odd:
\begin{equation}
\label{eq:def-W}
\mathcal{P}[W]({\bm x}) = W(-{\bm x}) = -W({\bm x}) \quad \text{and} \quad \mathcal{C}[W]({\bm x}) = W({\bm x})^\star = W({\bm x}).
\end{equation}
Hence, as a multiplication operator, $W$ anticommutes with $\mathcal{P}$ and commutes with $\mathcal{C}$.

\begin{remark}
Note that such $W$ cannot satisfy ${\mathcal{R}[W]({\bm x}) = W(R^\mathsf{T} {\bm x}) = W({\bm x})}$. If it did, then
\begin{equation}
\mathcal{P}[W]({\bm x}) = \mathcal{R}^2[W]({\bm x}) = W({\bm x})
\end{equation}
which contradicts ${\mathcal{P}[W]({\bm x}) = -W({\bm x})}$.
\end{remark}

\noindent For ${\delta > 0}$ small, we consider
\begin{equation}
\label{eq:def-L-P}
L^\delta \equiv -\nabla \cdot A \nabla + V + \delta W = L + \delta W.
\end{equation}

\subsubsection{Effect of $\mathcal{P}$-breaking on quadratic band degeneracies}
\label{sec:break-P-quad}

\begin{proposition}
\label{prop:Leff-quad-P}
Suppose ${L^0 = L}$ has a quadratic band degeneracy point. Then, for ${\delta > 0}$ sufficiently small, the Fourier symbol of $L^{\bm M}_{\rm eff}$, defined in \eqref{eq:Leff-quad}, perturbs as:
\begin{equation}
\label{eq:Leff-quad-P}
L^{\bm M}_{\rm eff}({\bm \kappa}; \delta) = (\delta^2 \vartheta^{\bm M}_0 + (1 - \alpha_0) |{\bm \kappa}|^2) I + (\delta^2 \vartheta^{\bm M}_1 - \alpha_1 {\bm \kappa} \cdot \sigma_1 {\bm \kappa}) \sigma_1 + (\delta^2 \vartheta^{\bm M}_2 - \alpha_2 {\bm \kappa} \cdot \sigma_3 {\bm \kappa}) \sigma_2 + \delta ({\bm \vartheta}^{\bm M}_3 \cdot {\bm \kappa}) \sigma_3.
\end{equation}
Here, $\vartheta^{\bm M}_0$, $\vartheta^{\bm M}_1$, ${\vartheta^{\bm M}_2 \in \R}$ and ${{\bm \vartheta}^{\bm M}_3 \in \R^2}$ are defined by:
\begin{equation}
\begin{aligned}
\label{eq:def-thetas}
\vartheta^{\bm M}_0 & \equiv \langle W \Phi^{\bm M}_1, \, \mathscr{R}(E_S) W \Phi^{\bm M}_1 \rangle, \\
\vartheta^{\bm M}_1 & \equiv {\rm Re} \, \langle W \Phi^{\bm M}_1, \, \mathscr{R}(E_S) W \Phi^{\bm M}_2 \rangle, \\
\vartheta^{\bm M}_2 & \equiv - {\rm Im} \, \langle W \Phi^{\bm M}_1, \, \mathscr{R}(E_S) W \Phi^{\bm M}_2 \rangle, \\
{\bm \vartheta}^{\bm M}_3 & \equiv 2 \, {\rm Re} \, \langle W \Phi^{\bm M}_1, \, \mathscr{R}(E_S) (2 i \nabla) \Phi^{\bm M}_1 \rangle.
\end{aligned}
\end{equation}
\end{proposition}

\begin{proposition}
Assume either ${\vartheta^{\bm M}_1 \neq 0}$ or ${\vartheta^{\bm M}_2 \neq 0}$, and that ${{\bm \vartheta}^{\bm M}_3 \neq {\bm 0}}$. Then, Proposition \ref{prop:Leff-nondgn} implies that the dispersion surfaces of $L^{\bm M}_{\rm eff}(\delta)$ are nondegenerate.
\end{proposition}

\subsubsection{Effect of $\mathcal{P}$-breaking on Dirac points}
\label{sec:break-P-dir}

\begin{proposition}
\label{prop:Leff-dir-P}
Suppose ${L^0 = L}$ has a symmetric pair of Dirac points. Then, for ${\delta > 0}$ sufficiently small, the Fourier symbols of $L^{{\bm D}^\pm}_{\rm eff}$, defined in \eqref{eq:Leff-dir-p} and \eqref{eq:Leff-dir-m}, perturb as:
\begin{align}
\label{eq:Leff-dir-p-P}
L^{{\bm D}^+}_{\rm eff}({\bm q}; \delta) & = ({\bm q} \cdot {\bm \gamma}^+_0) \, I + ({\bm q} \cdot {\bm \gamma}^+_1) \, \sigma_1 + ({\bm q} \cdot {\bm \gamma}^+_2) \, \sigma_2 + \delta \vartheta^{{\bm D}^+} \sigma_3, \\
\label{eq:Leff-dir-m-P}
L^{{\bm D}^-}_{\rm eff}({\bm q}; \delta) & = ({\bm q} \cdot {\bm \gamma}_0^-) \, I + ({\bm q} \cdot {\bm \gamma}^-_1) \, \sigma_1 + ({\bm q} \cdot {\bm \gamma}^-_2) \, \sigma_2 + \delta \vartheta^{{\bm D}^-} \sigma_3.
\end{align}
Here, ${\vartheta^{{\bm D}^\pm} \! \smallin \R}$ are defined by:
\begin{equation}
\label{eq:def-vth}
\vartheta^{{\bm D}^\pm} \equiv \langle \Phi^{{\bm D}^\pm}_1, \, W \Phi^{{\bm D}^\pm}_1 \rangle.
\end{equation}
\end{proposition}

\noindent The parameters $\vartheta^{{\bm D}^\pm} \!$, as well as the perturbed effective operators, are related via symmetry:

\begin{proposition}
\label{prop:vth-sym}
Consider the pair of effective operators \eqref{eq:Leff-dir-p-P} and \eqref{eq:Leff-dir-m-P}. We have:
\begin{align}
\label{eq:vth-sym}
\vartheta^{{\bm D}^-} \! = - \vartheta^{{\bm D}^+} \quad \text{which implies} \quad L^{{\bm D}^-}_{\rm eff}({\bm q}; \delta) = - L^{{\bm D}^+}_{\rm eff}({\bm q}; \delta).
\end{align}
\end{proposition}

\begin{proof}
Since ${\Phi^{{\bm D}^-}_1 = \mathcal{P}[\Phi^{{\bm D}^+}_1]}$ and $\mathcal{P}$ is unitary,
\begin{align}
\label{eq:vth-sym_2}
\vartheta^{{\bm D}^-} & = \langle \Phi^{{\bm D}^-}_1, \, W \Phi^{{\bm D}^-}_1 \rangle \\
& = \langle \mathcal{P}[\Phi^{{\bm D}^+}_1], \, W \mathcal{P}[\Phi^{{\bm D}^+}_1] \rangle \nonumber \\
& = \langle \mathcal{P}[\Phi^{{\bm D}^+}_1], \, \mathcal{P}[-W \Phi^{{\bm D}^+}_1] \rangle \nonumber \\
& = - \langle \Phi^{{\bm D}^+}_1, \, W \Phi^{{\bm D}^+}_1 \rangle = - \vartheta^{{\bm D}^+}. \nonumber
\end{align}
That $L^{{\bm D}^-}_{\rm eff}({\bm q}; \delta) = - L^{{\bm D}^+}_{\rm eff}({\bm q}; \delta)$ then immediately follows. An argument using $\mathcal{C}$ is also possible.
\end{proof}

\begin{proposition}
Assume ${\vartheta^{{\bm D}^+} = -\vartheta^{{\bm D}^-} \neq 0}$. Then, by Proposition \ref{prop:Leff-nondgn}, the dispersion surfaces of the perturbed effective operators $L^{{\bm D}^\pm}_{\rm eff}(\delta)$ are nondegenerate.
\end{proposition}

\begin{figure}[t]
\centering
\begin{subfigure}{0.42\textwidth}
\subcaption{$\qquad$}
\vspace{0.1cm}
\includegraphics[height = 4.3cm]{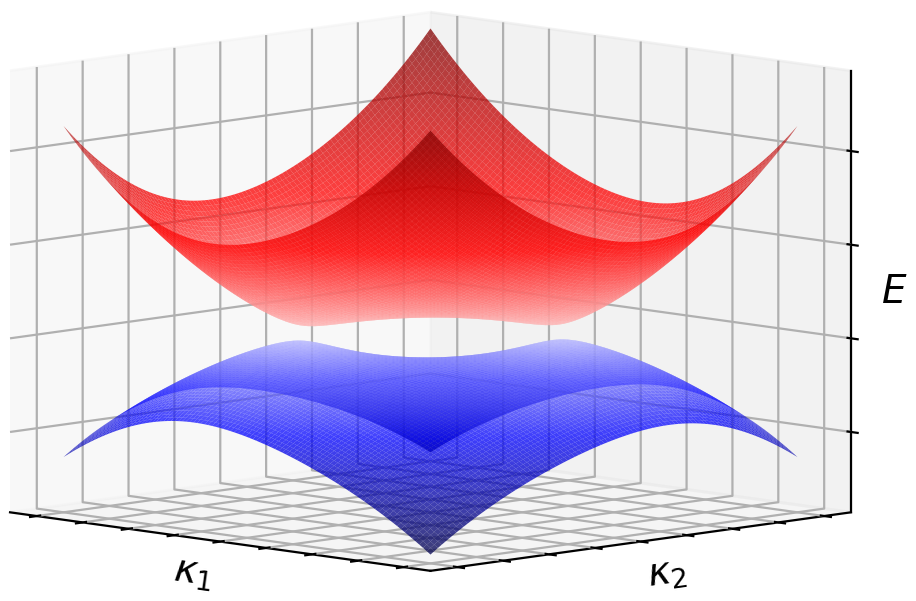}
\end{subfigure}
\hspace{0.4cm}
\begin{subfigure}{0.42\textwidth}
\subcaption{$\qquad$}
\vspace{0.1cm}
\includegraphics[height = 4.3cm]{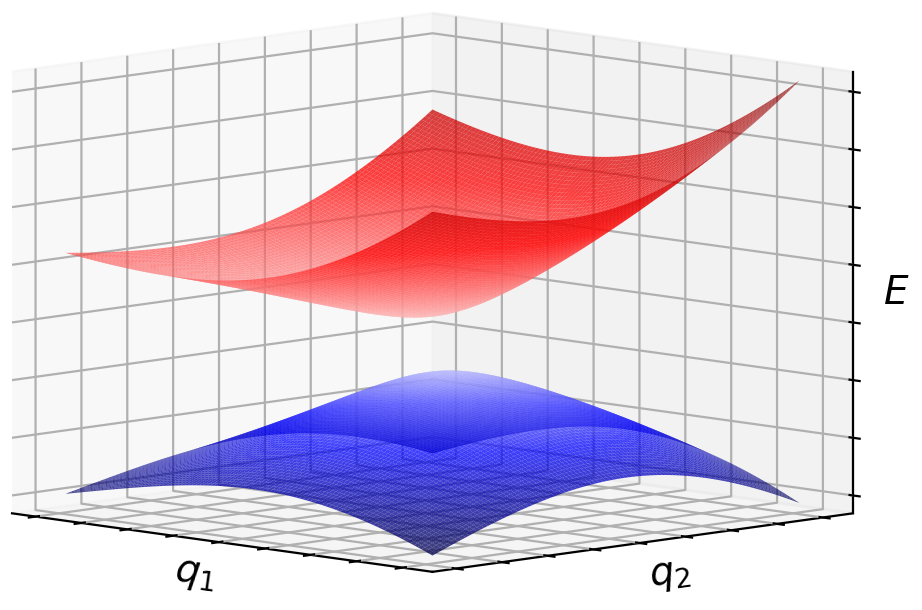}
\end{subfigure}
\vspace{0.3cm}
\caption{Nondegenerate (i.e. gapped) dispersion surfaces of the perturbed $\mathcal{P}$-breaking effective operators {\bf (a)} about a quadratic band degeneracy point; see $L^{\bm M}_{\rm eff}(\delta)$ in \eqref{eq:Leff-quad-P}; and {\bf (b)} about one of two Dirac points from a pair; see $L^{{\bm D}^+}_{\rm eff}(\delta)$ in \eqref{eq:Leff-dir-p-P}. Compare {\bf (a)} with panel (\subref{fig:intro-1b}) of Figure \ref{fig:intro-1} and {\bf (b)} with panel (\subref{fig:intro-2b}) of Figure \ref{fig:intro-2}.}
\end{figure}

\subsection{Breaking $\mathcal{C}$ symmetry}
\label{sec:break-C}

Let $a$ denote a nonzero smooth function which is $\Z^2$-periodic, real-valued, and even. In particular:
\begin{equation}
\label{eq:def-a}
P[a]({\bm x}) = a(-{\bm x}) = a({\bm x}) \quad \text{and} \quad \mathcal{C}[a]({\bm x}) = a({\bm x})^* = a({\bm x}).
\end{equation}
We introduce a magneto-optic type term of the form ${\nabla \cdot a \sigma_2 \nabla}$, which satisfies:
\begin{equation}
\label{eq:def-a-op}
\mathcal{P} \circ \nabla \cdot a \sigma_2 \nabla = \nabla \cdot a \sigma_2 \nabla \circ \mathcal{P} \quad \text{and} \quad \mathcal{C} \circ \nabla \cdot a \sigma_2 \nabla = - \nabla \cdot a \sigma_2 \nabla \circ \mathcal{C}.
\end{equation}
Here, $\sigma_2$ denotes the second Pauli matrix; see \eqref{eq:def-pauli}. For ${\delta > 0}$ sufficiently small, we consider:
\begin{equation}
\label{eq:def-L-C}
\tilde{L}^\delta \equiv -\nabla \cdot A \nabla + V + \delta \nabla \cdot a \sigma_2 \nabla = L + \delta \nabla \cdot a \sigma_2 \nabla.
\end{equation}

\subsubsection{Effect of $\mathcal{C}$-breaking on quadratic band degeneracies}
\label{sec:break-C-quad}

\begin{proposition}
\label{prop:Leff-quad-C}
Suppose ${\tilde{L}^0 = L}$ has a quadratic band degeneracy point. Then, for ${\delta > 0}$ sufficiently small, the Fourier symbol of $L^{\bm M}_{\rm eff}$, defined in \eqref{eq:Leff-quad}, perturbs as
\begin{equation}
\label{eq:Leff-quad-C}
\tilde{L}^{\bm M}_{\rm eff}({\bm \kappa}; \delta) = (1 - \alpha_0) |{\bm \kappa}|^2 \, I - \alpha_1 ({\bm \kappa} \cdot \sigma_1 {\bm \kappa}) \, \sigma_1 - \alpha_2 ({\bm \kappa} \cdot \sigma_3 {\bm \kappa}) \, \sigma_2 + \delta \tilde{\vartheta}^{\bm M} \, \sigma_3.
\end{equation}
Here, ${\tilde{\vartheta}^{\bm M} \smallin \R}$ is defined by:
\begin{equation}
\label{eq:def-tl-th}
\tilde{\vartheta}^{\bm M} \equiv \langle \Phi^{\bm M}_1, \, \nabla \cdot a \sigma_2 \nabla \Phi^{\bm M}_1 \rangle.
\end{equation}
\end{proposition}

\begin{proposition}
Assume ${\tilde{\vartheta}^{\bm M} \neq 0}$. Then, by Proposition \ref{prop:Leff-nondgn}, the dispersion surfaces of the perturbed effective operator $\tilde{L}^{\bm M}_{\rm eff}(\delta)$ are nondegenerate.
\end{proposition}

\begin{figure}[t]
\centering
\begin{subfigure}{0.42\textwidth}
\subcaption{$\qquad$}
\vspace{0.1cm}
\includegraphics[height = 4.3cm]{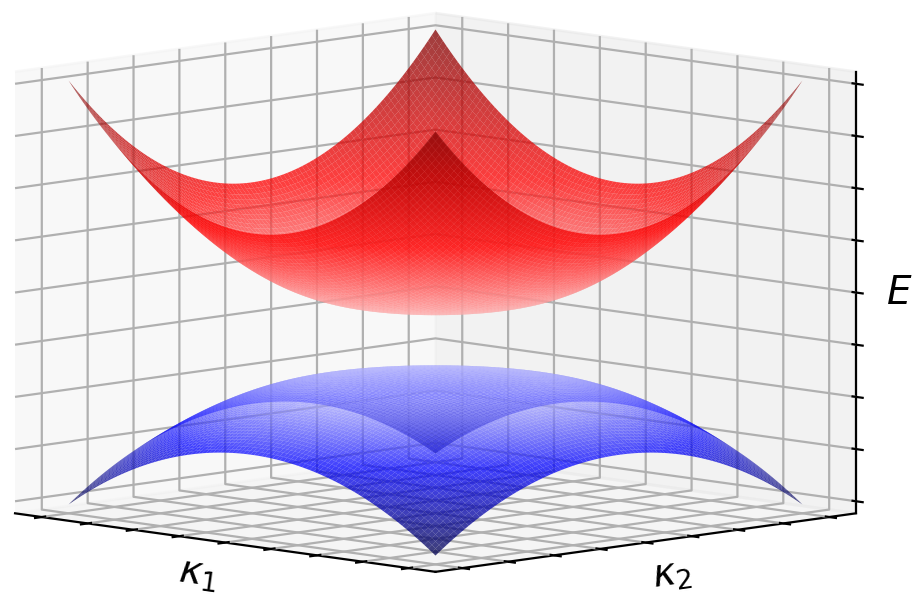}
\end{subfigure}
\hspace{0.4cm}
\begin{subfigure}{0.42\textwidth}
\subcaption{$\qquad$}
\vspace{0.1cm}
\includegraphics[height = 4.3cm]{fig_break-PC-2b.png}
\end{subfigure}
\vspace{0.3cm}
\caption{Nondegenerate (i.e. gapped) dispersion surfaces of the perturbed $\mathcal{C}$-breaking effective operators {\bf (a)} about a quadratic band degeneracy point; see $\tilde{L}^{\bm M}_{\rm eff}(\delta)$ in \eqref{eq:Leff-quad-C}; and {\bf (b)} about one of two Dirac points from a pair; see $\tilde{L}^{{\bm D}^+}_{\rm eff}(\delta)$ in \eqref{eq:Leff-dir-p-C}. Compare {\bf (a)} with panel (\subref{fig:intro-1b}) of Figure \ref{fig:intro-1} and {\bf (b)} with panel (\subref{fig:intro-2b}) of Figure \ref{fig:intro-2}.}
\end{figure}

\subsubsection{Effect of $\mathcal{C}$-breaking on Dirac points}
\label{sec:break-C-dir}

\begin{proposition}
\label{prop:Leff-dir-C}
Suppose ${\tilde{L}^0 = L}$ has a symmetric pair of Dirac points. Then, for ${\delta > 0}$ sufficiently small, the Fourier symbols of $L^{{\bm D}^\pm}_{\rm eff}$, defined in \eqref{eq:Leff-dir-p} and \eqref{eq:Leff-dir-m}, perturb as
\begin{align}
\label{eq:Leff-dir-p-C}
\tilde{L}^{{\bm D}^+}_{\rm eff}({\bm q}; \delta) & = ({\bm q} \cdot {\bm \gamma}^+_0) \, I + ({\bm q} \cdot {\bm \gamma}^+_1) \, \sigma_1 + ({\bm q} \cdot {\bm \gamma}^+_2) \, \sigma_2 + \delta \tilde{\vartheta}^{{\bm D}^+} \sigma_3, \\
\label{eq:Leff-dir-m-C}
\tilde{L}^{{\bm D}^-}_{\rm eff}({\bm q}; \delta) & = ({\bm q} \cdot {\bm \gamma}^-_0) \, I + ({\bm q} \cdot {\bm \gamma}^-_1) \, \sigma_1 + ({\bm q} \cdot {\bm \gamma}^-_2) \, \sigma_2 + \delta \tilde{\vartheta}^{{\bm D}^-} \sigma_3.
\end{align}
Here, ${\tilde{\vartheta}^{{\bm D}^\pm} \! \smallin \R}$ are defined by:
\begin{equation}
\label{eq:def-tl-vth}
\tilde{\vartheta}^{{\bm D}^\pm} \equiv \langle \Phi^{{\bm D}^\pm}_1, \, \nabla \cdot a \sigma_2 \nabla \Phi^{{\bm D}^\pm}_1 \rangle.
\end{equation}
\end{proposition}

\noindent The parameters $\tilde{\vartheta}^{{\bm D}^\pm}$ are again related by symmetry:

\begin{proposition}
\label{prop:tl-vth-sym}
Consider the pair of effective operators \eqref{eq:Leff-dir-p-C} and \eqref{eq:Leff-dir-m-C}. We have:
\begin{equation}
\tilde{\vartheta}^{{\bm D}^-} = \tilde{\vartheta}^{{\bm D}^+}.
\end{equation}
\end{proposition}

\begin{proof}
Since ${\Phi^{{\bm D}^-}_1 = \mathcal{P}[\Phi^{{\bm D}^+}_1]}$ and $\mathcal{P}$ is unitary,
\begin{align}
\label{eq:tl-vth-sym_2}
\tilde{\vartheta}^{{\bm D}^-} & = \langle \Phi^{{\bm D}^-}_1, \, \nabla \cdot a \sigma_2 \nabla \Phi^{{\bm D}^-}_1 \rangle \\
& = \langle \mathcal{P}[\Phi^{{\bm D}^+}_1], \, \nabla \cdot a \sigma_2 \nabla \mathcal{P}[\Phi^{{\bm D}^+}_1] \rangle \nonumber \\
& = \langle \mathcal{P}[\Phi^{{\bm D}^+}_1], \, \mathcal{P}[\nabla \cdot a \sigma_2 \nabla \Phi^{{\bm D}^+}_1] \rangle \nonumber \\
& = \langle \Phi^{{\bm D}^+}_1, \, \nabla \cdot a \sigma_2 \nabla \Phi^{{\bm D}^+}_1 \rangle = \tilde{\vartheta}^{{\bm D}^+}. \nonumber
\end{align}
An argument using $\mathcal{C}$ is also possible.
\end{proof}

\noindent The symmetry relating effective operators $L^{{\bm D}^\pm}_{\rm eff}({\bm q}; \delta)$, as discussed in Proposition \ref{prop:vth-sym}, has no analog.

\begin{proposition}
Assume ${\tilde{\vartheta}^+ = \tilde{\vartheta}^- \neq 0}$. Then, by Proposition \ref{prop:Leff-nondgn}, the dispersion surfaces of the perturbed effective operators $\tilde{L}^{{\bm D}^\pm}_{\rm eff}(\delta)$ are nondegenerate.
\end{proposition}

\begin{remark}
\label{rmk:Leff-dir-PC}
Note that the perturbed effective operators \eqref{eq:Leff-dir-p-P} and \eqref{eq:Leff-dir-m-P} in the $\mathcal{P}$-breaking scenario appear similar to \eqref{eq:Leff-dir-p-C} and \eqref{eq:Leff-dir-m-C} in the $\mathcal{C}$-breaking scenario. The important difference lies in the symmetry relationship between $\vartheta^{{\bm D}^\pm}$ and between $\tilde{\vartheta}^{{\bm D}^\pm}$ with consequences on the topological index $c_1$.
\end{remark}

\subsection{Band structure topology}
\label{sec:band-top}

As remarked in Section \ref{sec:previous}, breaking certain symmetries can open a local spectral gap about a band structure degeneracy, separating the dispersion surfaces in a neighborhood of the degenerate energy-quasimomentum into nondegenerate, or {\it isolated} bands. To such isolated bands, one can assign an integer topological index, which is nonzero when when time-reversal symmetry is broken and zero when it is unbroken. This topological index plays an important role in our interpretation of results on edge states via the bulk-edge correspondence; see Section \ref{sec:edge-states}.

The integer topological index, i.e. the {\it first Chern number} $c_1$, is defined as follows: To any isolated band, associated with a family of Floquet-Bloch eigenstates $\phi(\cdot; {\bm k})$, where ${\bm k} \smallin \mathcal{B}$, we define the {\it Berry connection}
\begin{equation}
\label{eq:def-berry-con}
{\bm A}({\bm k}) \equiv i \langle \phi(\cdot; {\bm k}), \, \nabla_{\bm k} \phi(\cdot; {\bm k}) \rangle
\end{equation}
as well as the {\it Berry curvature}
\begin{equation}
\label{eq:def-berry-cur}
{\bm F}({\bm k}) = \nabla_{\bm k} \times {\bm A}({\bm k}).
\end{equation}
For models with two spatial dimensions, ${\bm F}({\bm k})$ has just one component which we denote as $F({\bm k})$. This allows us to define the first Chern number as
\begin{equation}
\label{eq:def-chern}
c_1 = \frac{1}{2 \pi} \int_\mathcal{B} F({\bm k}) \, {\rm d}{\bm k}.
\end{equation}
This is a well-defined integer.

Above, we considered small (size ${\delta > 0}$) parity or time-reversal symmetry-breaking perturbations. In this regime, the integral expression \eqref{eq:def-chern} for $c_1$ is dominated by contributions from the Berry curvature $F({\bm k})$ in quasimomentum neighborhoods of gap openings (i.e., about former degeneracies). In the limit ${\delta \to 0}$, the gap closes and the the Berry curvature blows up. This limit is explored in \cite{drouot2019a} in the context of honeycomb media, and it is shown that the Chern number can actually be read off from the perturbed effective Dirac Hamiltonian. That the Chern number is explicit in the expression for effective Hamiltonians is also demonstrated in \cite{chong2008effective} in the context of square lattice media and their quadratic band degeneracies, and tilted Dirac points emerging from deformations of such media.

In each of these works, it is shown that ${c_1 = 0}$ for structures which preserve $\mathcal{C}$ symmetry, while ${c_1 = \pm 1}$ if $\mathcal{C}$ symmetry is broken. We demonstrate this, in the context of our model of deformations, as follows:

\begin{remark}
\label{rmk:chern-break-P}
{\rm (Chern number; $\mathcal{P}$-breaking, $\mathcal{C}$-preserving perturbation.)} Consider the scenario \eqref{eq:def-L-P} of a $\mathcal{P}$-breaking, $\mathcal{C}$-preserving perturbation to $L$. The effective operator \eqref{eq:Leff-quad-P} exhibits a gap opening about a former quadratic band degeneracy of $L$. The Chern number is
\begin{equation}
\label{eq:chern-break-P-quad}
c_1 = 0.
\end{equation}
On the other hand, the pair of effective operators \eqref{eq:Leff-dir-p-P} and \eqref{eq:Leff-dir-m-P} exhibit gap openings about former Dirac points. The Chern number, obtained by summing contributions from both Dirac points, is
\begin{equation}
\label{eq:chern-break-P-dir}
c_1 = - \frac{{\rm sgn}(\vartheta^{{\bm D}^+})}{2} - \frac{{\rm sgn}(\vartheta^{{\bm D}^-})}{2} = - \frac{{\rm sgn}(\vartheta^{{\bm D}^+})}{2} + \frac{{\rm sgn}(\vartheta^{{\bm D}^+})}{2} = 0.
\end{equation}
\end{remark}

\begin{remark}
\label{rmk:chern-break-C}
{\rm (Chern number; $\mathcal{P}$-preserving, $\mathcal{C}$-breaking perturbation.)}
Consider the scenario \eqref{eq:def-L-C} of a $\mathcal{P}$-preserving, $\mathcal{C}$-breaking perturbation to $L$. The effective operator \eqref{eq:Leff-quad-C} exhibits a gap opening about a former quadratic band degeneracy of $L$. The Chern number is
\begin{equation}
\label{eq:chern-break-C-quad}
c_1 = - {\rm sgn}(\tilde{\vartheta}^{\bm M}).
\end{equation}
On the other hand, the pair of effective operators \eqref{eq:Leff-dir-p-C} and \eqref{eq:Leff-dir-m-C} exhibit gap openings about former Dirac points. The Chern number, obtained by summing contributions from both Dirac points, is
\begin{equation}
\label{eq:chern-break-C-dir}
c_1 = - \frac{{\rm sgn}(\tilde{\vartheta}^{{\bm D}^+})}{2} - \frac{{\rm sgn}(\tilde{\vartheta}^{{\bm D}^-})}{2} = - \frac{{\rm sgn}(\tilde{\vartheta}^{{\bm D}^+})}{2} - \frac{{\rm sgn}(\tilde{\vartheta}^{{\bm D}^+})}{2} = - {\rm sgn}(\tilde{\vartheta}^{{\bm D}_+}).
\end{equation}
\end{remark}

\noindent Note that, regardless of whether the Berry curvature receives contributions from gap closings at quadratic band degeneracy points or at pairs of Dirac points, the total Chern number is unchaged; these two scenarios correspond to continuous deformations of dispersion surfaces which maintain the band gap.

\bigskip

\appendix

\section{Fourier series of square lattice potentials}
\label{apx:fourier}

\setcounter{equation}{0}
\setcounter{figure}{0}

In this appendix, we discuss the structure of the Fourier series of $\Z^2$-periodic potentials with additional symmetries. In the course of this discussion, we correct the statement and proof of \cite[Proposition 2.9]{keller2018spectral}. 

First, recall that any ${f \smallin L^2(\R^2/\Z^2)}$ can be expressed as a uniformly convergent Fourier series:
\begin{align}
\label{eq:def-fourier}
f({\bm x}) = \sum_{{\bm m} \smallin \Z^2} f_{\bm m} \, e^{i {\bm m} {\bm k} \cdot {\bm x}}, \quad \text{where} \quad f_{\bm m} = \frac{1}{|\Omega|} \int_\Omega e^{-i {\bm m} {\bm k} \cdot {\bm x}} \, f({\bm x}) \, {\rm d}{\bm x}.
\end{align}
Here, $\Omega$ is a choice of fundamental cell for ${\R^2/\Z^2}$. Throughout this appendix, we use the compressed notation ${{\bm m} {\bm k} = m_1 {\bm k}_1 + m_2 {\bm k}_2}$, where ${{\bm m} = (m_1, m_2) \smallin \Z^2}$ and ${\bm k}_1$, ${\bm k}_2$ are the dual lattice basis vectors; see \eqref{eq:def-sql-dual}.

Additional symmetries of functions in $L^2(\R^2/\Z^2)$ imply constraints on their Fourier coefficients, as detailed in the following.

\begin{proposition}
\label{prop:fourier-sym}
Suppose ${f \smallin L^2(\R^2/\Z^2)}$ with Fourier series  \eqref{eq:def-fourier}. Then the following hold:
\begin{enumerate}
\item \label{itm:fourier-sym-1} ${f({\bm x}) = \mathcal{P}[f]({\bm x}) = f(-{\bm x})}$ if and only if ${f_{\bm m} = f_{-{\bm m}}}$.
\item \label{itm:fourier-sym-2} ${f({\bm x}) = \mathcal{C}[f]({\bm x}) = f({\bm x})^*}$ if and only if ${f_{\bm m} = f_{-{\bm m}}^*}$.
\item \label{itm:fourier-sym-3} ${f({\bm x}) = \mathcal{R}[f]({\bm x}) = f(R^\mathsf{T}{\bm x})}$ if and only if ${f_{\bm m} = f_{\tilde{\mathcal{R}} {\bm m}}}$, where ${\tilde{\mathcal{R}}(m_1, m_2) = (m_2, -m_1)}$.
\item \label{itm:fourier-sym-4} ${f({\bm x}) = \Sigma_1[f]({\bm x}) = f(\sigma_1 {\bm x})}$ if and only if ${f_{\bm m} = f_{\tilde{\Sigma}_1 {\bm m}}}$, where ${\tilde{\Sigma}_1 (m_1, m_2) = (m_2, m_1)}$.
\end{enumerate}
\end{proposition}

\begin{proof}
Parts \ref{itm:fourier-sym-1}, \ref{itm:fourier-sym-2}, and \ref{itm:fourier-sym-3} are proven in  \cite[Proposition 2.8]{keller2018spectral}. Part \ref{itm:fourier-sym-4} is proven analogously, using the relations ${\sigma_1 {\bm k}_1 = {\bm k}_2}$ and ${\sigma_1 {\bm k}_2 = {\bm k}_1}$.
\end{proof}

\noindent Parts \ref{itm:fourier-sym-1} and \ref{itm:fourier-sym-2} of Proposition \ref{prop:fourier-sym} together imply that ${f = \mathcal{P}[f]}$ and ${f = \mathcal{C}[f]}$ if and only if ${f_{\bm m} \smallin \R}$.

In Part \ref{itm:fourier-sym-3}, the map ${\tilde{\mathcal{R}}: \Z^2 \to \Z^2}$ has ${\bm 0}$ as a unique fixed point. We partition ${\Z^2 \setminus \{ {\bm 0} \}}$ into orbits under $\tilde{\mathcal{R}}$ of length exactly four, and write ${{\bm m} \sim {\bm n}}$ if ${\bm m}$ and ${\bm n}$ belong to the same orbit. The relation $\sim$ defines an equivalence relation, and we denote by $\tilde{S}$ a set of representatives from each of the equivalence classes of $(\Z^2 \setminus \{ {\bm 0} \})/\sim$. For clarity, we here fix a choice of $\tilde{S}$:
\begin{equation}
\label{eq:def-til-S}
\tilde{S} \equiv \{ (m_1, m_2) \smallin \Z^2 : m_1 \geq 1, \, m_2 \geq 0 \};
\end{equation}
see the red-shaded region in Figure \ref{fig:fourier-1}, panel (\subref{fig:fourier-1a}).

We obtain the following Fourier series representation of $\Z^2$-periodic, $\pi/2$-rotationally invariant potentials:

\begin{proposition}
\label{prop:fourier-rot}
Suppose ${V \smallin L^2(\R^2/\Z^2)}$ such that ${\mathcal{P}[V] = V}$, ${\mathcal{C}[V] = V}$, and ${\mathcal{R}[V] = V}$. Then
\begin{equation}
\label{eq:fourier-rot}
V({\bm x}) = V_{\bm 0} + \sum_{{\bm m} \smallin \tilde{S}} 2 V_{\bm m} \bigl( \cos({\bm m}{\bm k} \cdot {\bm x}) + \cos((\tilde{\mathcal{R}}{\bm m}) {\bm k} \cdot {\bm x}) \bigr)
\end{equation}
where $\tilde{S}$ is defined in \eqref{eq:def-til-S}.
\end{proposition}

\begin{proof}
This is proven as part (a) of \cite[Proposition 2.9]{keller2018spectral}.
\end{proof}

\begin{figure}[t]
\centering
\begin{subfigure}{0.4\textwidth}
\centering
\subcaption{$\quad$}
\label{fig:fourier-1a}
\vspace{0.1cm}
    \begin{tikzpicture}[x = 0.85cm, y = 0.85cm]
    \draw[white] (0, 0) rectangle (7.5, 7.25);
    \draw[black, thick, ->, opacity = 0.5] (0.5, 3) -- (7, 3) node[anchor = north] {$m_1$};
    \draw[black, thick, ->, opacity = 0.5] (3, 0.5) -- (3, 6.75) node[anchor = east] {$m_2 \,$};
    \foreach \x in {1,...,6}
        \foreach \y in {1,...,6}
            \fill[black] (\x, \y) circle (0.08cm);
    \fill[gray, opacity = 0.3] (2.5, 2.5) rectangle (3.5, 3.5);
    \fill[red, opacity = 0.3] (3.5, 2.5) rectangle (6.5, 6.5);
    \end{tikzpicture}
\end{subfigure}
\hspace{0.3cm}
\begin{subfigure}{0.4\textwidth}
\subcaption{$\quad$}
\label{fig:fourier-1b}
\vspace{0.1cm}
    \begin{tikzpicture}[x = 0.85cm, y = 0.85cm]
    \draw[white] (0, 0) rectangle (7.5, 7.25);
    \draw[black, thick, ->, opacity = 0.5] (0.5, 3) -- (7, 3) node[anchor = north] {$m_1$};
    \draw[black, thick, ->, opacity = 0.5] (3, 0.5) -- (3, 6.75) node[anchor = east] {$m_2 \,$};
    \foreach \x in {1,...,6}
        \foreach \y in {1,...,6}
            \fill[black] (\x, \y) circle (0.08cm);
    \fill[gray, opacity = 0.3] (2.5, 2.5) rectangle (3.5, 3.5);
    \fill[red, opacity = 0.3] (3.5, 2.5) rectangle (6.5, 3.5);
    \fill[green, opacity = 0.3] (4.5, 3.5) rectangle (6.5, 4.5);
    \fill[green, opacity = 0.3] (5.5, 4.5) rectangle (6.5, 5.5);
    \fill[green!20!blue, opacity = 0.3] (3.5, 3.5) rectangle (4.5, 4.5);
    \fill[green!20!blue, opacity = 0.3] (4.5, 4.5) rectangle (5.5, 5.5);
    \fill[green!20!blue, opacity = 0.3] (5.5, 5.5) rectangle (6.5, 6.5);
    \end{tikzpicture}
\end{subfigure}
\vspace{0.3cm}
\caption{Shaded regions correspond to representatives of equivalence classes of indices $(m_1,m_2)\in\Z^2$ of Fourier coefficients for {\bf (a)} $\Z^2$-periodic, $\pi/2$-rotationally invariant potentials, and {\bf (b)} $\Z^2$-periodic, $\pi/2$-rotationally invariant, and reflection invariant potentials. In {\bf (a)}, the fixed point ${\bm 0}$ of $\tilde{\mathcal{R}}$ is shaded gray, and the set $\tilde{S}$ is shaded red. In {\bf (b)}, the common fixed point ${\bm 0}$ of $\tilde{\mathcal{R}}$ and $\tilde{\Sigma}_1$ is shaded gray, the set $\tilde{S}_1$ is shaded red, the set $\tilde{S}_2$ is shaded blue, and the set $\tilde{S}_3$ is shaded green.}
\label{fig:fourier-1}
\end{figure}

The potentials of interest in this article are $\Z^2$-periodic, $\pi/2$-rotationally invariant, and reflection invariant. Next, we refine the Fourier series representation \eqref{eq:fourier-rot} using the constraint on Fourier coefficients from part \ref{itm:fourier-sym-4} of Proposition \ref{prop:fourier-sym}. We identify three subsets of $\tilde{S}$: 
\begin{enumerate}
\item First, we define
\begin{equation}
\label{eq:def-tilS-1}
\tilde{S}_1 \equiv \{ (m_1, m_2) \smallin \Z^2 : m_1 \geq 1, \, m_2 = 0 \};
\end{equation}
see the red-shaded region in Figure \ref{fig:fourier-1}, panel (\subref{fig:fourier-1b}). Note that $\tilde{S}_1$ satisfies ${\tilde{\Sigma}_1 \tilde{S}_1 = \tilde{\mathcal{R}} \tilde{S}_1}$. Therefore, for ${{\bm m} \smallin \tilde{S}_1}$, the constraint from part \ref{itm:fourier-sym-4} of Proposition \ref{prop:fourier-sym} on $V_{\bm m}$ is already explicit in the expression \eqref{eq:fourier-rot}.
\item Next, we define
\begin{equation}
\label{eq:def-tilS-2}
\tilde{S}_2 \equiv \{ (m_1, m_2) \smallin \Z^2 : m_1 \geq 1, \, m_2 = m_1 \};
\end{equation}
see the blue-shaded region in Figure \ref{fig:fourier-1}, panel (\subref{fig:fourier-1b}). $\tilde{S}_2$ consists of fixed points of $\tilde{\Sigma}_1$ belonging to $\tilde{S}$. 
\item Finally, we partition ${\tilde{S} \setminus (\tilde{S}_1 \cup \tilde{S}_2)}$ into orbits under $\tilde{\Sigma}_1$ of length exactly two, and write ${{\bm m} \approx {\bm n}}$ if ${\bm m}$ and ${\bm n}$ belong to the same orbit. The relation $\approx$ is an equivalence relation, and we denote by $\tilde{S}_3$ a set of representatives from each of the equivalence classes of ${(\tilde{S} \setminus (\tilde{S}_1 \cup \tilde{S}_2))/\approx}$. For clarity, we fix
\begin{equation}
\label{eq:def-tilS-3}
\tilde{S}_3 \equiv \{ (m_1, m_2) \smallin \Z^2 : m_1 \geq 2, \, 1 \leq m_2 < m_1 \};
\end{equation}
see the green-shaded region in Figure \ref{fig:fourier-1}, panel (\subref{fig:fourier-1b}).
\end{enumerate}

\noindent We obtain the following Fourier series representation of $\Z^2$-periodic, $\pi/2$-rotationally invariant, and reflection invariant potentials (i.e., square lattice potentials, in the sense of Definition \ref{def:sql-pot}):

\begin{proposition}
\label{prop:fourier-sql}
Let ${V \smallin L^2(\R^2/\Z^2)}$ be $\mathcal{P}$, $\mathcal{C}$ and $\mathcal{R}$ invariant as in Proposition \ref{prop:fourier-rot}, and additionally assume that ${\Sigma_1[V] = V}$. Then,
\begin{align}
V({\bm x}) & = V_{\bm 0} + \sum_{{\bm m} \smallin \tilde{S}_1 \cup \, \tilde{S}_2} 2 V_{\bm m} \bigl( \cos({\bm m}{\bm k} \cdot {\bm x}) + \cos((\tilde{\mathcal{R}}{\bm m}) {\bm k} \cdot {\bm x}) \bigr) \\
& \qquad \, + \sum_{{\bm m} \smallin \tilde{S}_3} 2 V_{\bm m} \bigl( \cos({\bm m}{\bm k} \cdot {\bm x}) + \cos((\tilde{\mathcal{R}}{\bm m}) {\bm k} \cdot {\bm x}) + \cos((\tilde{\Sigma}_1 {\bm m}) {\bm k} \cdot {\bm x}) + \cos((\tilde{\mathcal{R}} \tilde{\Sigma}_1 {\bm m}) {\bm k} \cdot {\bm x}) \bigr) \nonumber
\end{align}
where $\tilde{S}_1$, $\tilde{S}_2$, and $\tilde{S}_3$ are defined in \eqref{eq:def-tilS-1}, \eqref{eq:def-tilS-2}, and \eqref{eq:def-tilS-3}, respectively. 
\end{proposition}

\begin{proof}
This is proven analogously to Proposition \ref{prop:fourier-rot}, using the above discussion to implement the additional constraint ${V_{\bm m} = V_{\tilde{\Sigma}_1 {\bm m}}}$ for each ${{\bm m} \smallin \tilde{S}}$.
\end{proof}

The following remark concerns the term ``generic'' used in the statement of Theorem \ref{thm:kmow}, as it relates to a nondegeneracy condition on distinguished Fourier coefficients.

\begin{remark}
\label{rmk:kmow-generic}
In part \ref{itm:kmow-exist} of Theorem \ref{thm:kmow}, the term ``generic'' has the following meaning: Let $V_{m_1, m_2}$ denote the $(m_1, m_2)$ Fourier coefficient of the square lattice potential $V$; see Definition \ref{eq:def-fourier}. Suppose that the distinguished Fourier coefficients $V_{0, 1}$ and $V_{1, 1}$ satisfy $V_{0, 1}$, ${V_{1, 1} \neq 0}$ and ${V_{0, 1} \neq \pm V_{1, 1}}$. Then, for all ${\delta \smallin \R}$, except for $\delta$ in a discrete set containing ${\delta = 0}$, the band structures of the operators ${H^\delta = - \Delta + \delta V}$ contain quadratic band degeneracy points.
\end{remark}

\noindent See Theorem 6.1 of \cite{keller2018spectral} for further discussion.

\bigskip

\section{Sketch of proof of Theorem \ref{thm:dir-srf}}
\label{apx:pf-dir-srf}

\setcounter{equation}{0}
\setcounter{figure}{0}

In this appendix, we provide a sketch for the proof of Theorem \ref{thm:dir-srf}, which concerns solutions to the Floquet-Bloch eigenvalue problem
\begin{equation}
\label{eq:fb-evp-L}
L \Phi = E \Phi, \quad \Phi \smallin L^2_{\bm k}, \quad {\bm k} \smallin \mathcal{B}
\end{equation}
for $(E, {\bm k})$ in a suitably small neighborhood of an energy-quasimomentum pair $(E_D, {\bm D})$ satisfying properties \ref{itm:dir-pt-1} -- \ref{itm:dir-pt-4}, i.e., a Dirac point of $L$. Throughout this appendix, $(\varepsilon, {\bm q})$ denotes a displacement from $(E_D, {\bm D})$.

First, in analogy with Proposition \ref{prop:red-M} part \ref{itm:det-calM-0}, we have the following:

\begin{proposition}
\label{prop:red-D}
Suppose ${(E_D, {\bm D})}$ satisfies properties \ref{itm:dir-pt-1} -- \ref{itm:dir-pt-4}, and is therefore a Dirac point of $L$. Then, there exist $\varepsilon^\flat$, ${q^\flat > 0}$, and a ${2 \times 2}$ matrix-valued, analytic function $\mathcal{N}(\varepsilon; {\bm q})$, defined for ${|\varepsilon| < \varepsilon^\flat}$ and ${|{\bm q}| < q^\flat}$, such that ${E_D + \varepsilon}$ is an $L^2_{{\bm D} + {\bm q}}$ eigenvalue of $L$, with ${|\varepsilon| < \varepsilon^\flat}$ and ${|{\bm q}| < q^\flat}$, if and only if
\begin{equation}
\label{eq:det-calN-0}
\det(\mathcal{N}(\varepsilon; {\bm q})) = 0.
\end{equation}
\end{proposition}

\begin{proof}
The proof follows a Schur complement/Lyapunov-Schmidt reduction strategy analogous to that outlined in Section \ref{sec:set-up}. We provide the following sketch:

Let ${{\bm q} = {\bm k} - {\bm D}}$ denote the quasimomentum displacement from ${\bm D}$. We consider the family of eigenvalue problems, equivalent to \eqref{eq:fb-evp-L},
\begin{equation}
\label{eq:fb-evp-L_2}
L({\bm D} + {\bm q}) \, (a_1 \phi_1 + a_2 \phi_2 + \phi^{(1)}) = (E_D + \varepsilon) \, (a_1 \phi_1 + a_2 \phi_2 + \phi^{(1)}),
\end{equation} 
where $a_1$, ${a_2 \smallin \C}$, ${\phi^{(1)} \smallin {\rm ker}(L({\bm D}) - E_D)^\perp}$, and ${\varepsilon \smallin \R}$. Using the expression
\begin{equation}
\label{eq:L-phi-k}
L({\bm D} + {\bm q}) = L({\bm D}) - i {\bm q} \cdot (\nabla_{\bm D} \cdot A + A \nabla_{\bm D}) + {\bm q} \cdot A {\bm q},
\end{equation}
we obtain, in analogy to \eqref{eq:phi-1} and \eqref{eq:red-M},
\begin{align}
\label{eq:phi-1_4}
& (L({\bm D}) - E_D) \, \phi^{(1)} = \Pi^\perp_{\bm D} \, F(a_1, a_2, \phi^{(1)}, \varepsilon; {\bm q}), \\
\label{eq:red-D}
& \Pi^\parallel_{\bm D} \, F(a_1, a_2, \phi^{(1)}, \varepsilon; {\bm q}) = 0,
\end{align}
where 
\begin{equation}
\label{eq:def-F_2}
F(a_1, a_2, \phi^{(1)}, \varepsilon; {\bm q}) \equiv ( \varepsilon + i {\bm q} \cdot (\nabla_{\bm D} \cdot A + A \nabla_{\bm D}) - {\bm q} \cdot A {\bm q}) \, (a_1 \phi_1 + a_2 \phi_2 + \phi^{(1)}).
\end{equation}
Here, $\Pi^\parallel_{\bm D}$ and $\Pi^\perp_{\bm D}$ denote, respectively, orthogonal projections onto ${{\rm ker}(L({\bm D}) - E_D)}$ and ${{\rm ker}(L({\bm D}) - E_D)^\perp}$. 

Following the procedure of Section \ref{sec:phi-1}, there exist constants $\varepsilon^\flat$, $q^\flat > 0$ such that \eqref{eq:phi-1_4} has a solution, for ${|\varepsilon| < \varepsilon^\flat}$ and ${|{\bm q}| < q^\flat}$, given by a mapping ${(a_1, a_2, \varepsilon) \mapsto \phi^{(1)}[a_1, a_2, \varepsilon]}$ which satisfies the bounds
\begin{equation}
\label{eq:phi-1-bd_2}
\phi^{(1)}[a_1, a_2, \varepsilon] = O(|{\bm q}| + |\varepsilon| |{\bm q}|) \ \ \text{as} \ \ |\varepsilon|, \, |{\bm q}| \to 0. 
\end{equation}

Finally, following the procedure of Section \ref{sec:reduction}, we have that \eqref{eq:red-D} has a solution, for ${|\varepsilon| < \varepsilon^\flat}$ and ${|{\bm q}| < q^\flat}$, if and only if
\begin{equation}
\label{eq:reduced_2}
\langle \phi_\ell, \, F(a_1, a_2, \phi^{(1)}[a_1, a_2, \varepsilon], \varepsilon; {\bm q}) \rangle = 0, \quad \ell \smallin \{ 1, \, 2 \}.
\end{equation}
Note that \eqref{eq:reduced_2} has the expansion
\begin{equation}
\label{eq:reduced_3}
\sum_{k = 1}^2 \bigl( \varepsilon I_{\ell, k} + {\bm q} \cdot \langle \phi_\ell, \, i \nabla_{\bm D} \cdot (A \phi_k) + A \nabla_{\bm D} \phi_k \rangle + O(|{\bm q}|^2 + |\varepsilon| |{\bm q}|^2) \bigr) = 0 \ \ \text{as} \ \ |\varepsilon|, \, |{\bm q}| \to 0, \quad \ell \smallin \{ 1, \, 2 \};
\end{equation}
see Proposition \ref{prop:calN-pr}, part \ref{itm:calN-ex}. The system \eqref{eq:reduced_2} is of the form
\begin{equation}
\mathcal{N}(\varepsilon; {\bm q})
\begin{bmatrix}
a_1 \\
a_2
\end{bmatrix}
= 
\begin{bmatrix}
0 \\
0
\end{bmatrix} \! ,
\end{equation}
and is therefore solvable if and only if
\begin{equation}
{\rm det}(\mathcal{N}(\varepsilon; {\bm q})) = 0.
\end{equation}
This concludes the sketch of the proof of Proposition \ref{prop:red-D}; for further details, see \cite{chaban2025thesis}.
\end{proof}

In analogy with Proposition \ref{prop:calM-pr}, we record properties of $\mathcal{N}(\varepsilon; {\bm q})$.

\begin{proposition}
\label{prop:calN-pr}
The ${2 \times 2}$ matrix-valued, analytic function $\mathcal{N}(\varepsilon; {\bm q})$ satisfies the following:
\begin{enumerate}
\item \label{itm:calN-sa} $\mathcal{N}(\varepsilon; {\bm q})$ is self-adjoint.
\item \label{itm:calN-PC} By $\mathcal{PC}$ symmetry,
\begin{equation}
\label{eq:calN-PC}
\mathcal{N}_{1, 1}(\varepsilon; {\bm q}) = \mathcal{N}_{2, 2}(\varepsilon; {\bm q}).
\end{equation}
\item \label{itm:calN-ex} $\mathcal{N}(\varepsilon; {\bm q})$ has the expansion
\begin{equation}
\label{eq:calN-ex}
\mathcal{N}(\varepsilon; {\bm q}) = \varepsilon I - L^{\bm D}_{\rm eff}({\bm q}) + \tilde{\mathcal{N}}(\varepsilon; {\bm q}).
\end{equation}
Here, $L^{\bm D}_{\rm eff}({\bm q})$ is the Fourier symbol of the effective Hamiltonian $L^{\bm D}_{\rm eff}$, discussed in Remark \ref{rmk:dir-srf-eff}:
\begin{equation}
\label{eq:def-Leff_2}
L^{\bm D}_{\rm eff}({\bm q}) = ({\bm \gamma}_0 \cdot {\bm q}) I + ({\bm \gamma}_1 \cdot {\bm q}) \sigma_1 + ({\bm \gamma}_2 \cdot {\bm q}) \sigma_2.
\end{equation}
The entries of the ${2 \times 2}$ matrix $\tilde{\mathcal{N}}(\varepsilon; {\bm q})$ satisfy the bound, for $j$, $k \smallin \{ 1, \, 2 \}$,
\begin{equation}
\label{eq:til-calN-bd}
\tilde{\mathcal{N}}_{j, k}(\varepsilon; {\bm q}) = O(|{\bm q}|^2 + |\varepsilon| |{\bm q}|^2) \ \ \text{as} \ \ |\varepsilon|, \, |{\bm q}| \to 0.
\end{equation}
\end{enumerate}
\end{proposition}

\begin{proof}
Part \ref{itm:calN-sa} is proven analogously to Proposition \ref{prop:calM-pr} part \ref{itm:calM-sa}; see Appendix \ref{apx:calM-sa}. Part \ref{itm:calN-PC} is proven analogously to Proposition \ref{prop:calM-pr} part \ref{itm:calM-PC}; see Appendix \ref{apx:calM-PC}. Part \ref{itm:calN-ex} is proven analogously to Proposition \ref{prop:calM-pr} part \ref{itm:calM-ex-1}; see Appendix \ref{apx:calM-ex-1}.
\end{proof}

\noindent Theorem \ref{thm:dir-srf} is then obtained through a procedure analogous to that presented in Section \ref{sec:pf-M-srf}. In particular, the expression \eqref{eq:dir-srf} describes the locus of solutions to \eqref{eq:det-calN-0} for ${|{\bm q}| < q^\star}$, where ${0 < q^\star \leq q^\flat}$.

We conclude this appendix by stating a result concerning the relationship between reductions about a pair of Dirac points, which are related to one another by $\mathcal{P}$ symmetry.

\begin{proposition}
\label{prop:dir-pair}
Suppose ${(E_D, {\bm D})}$ and ${(E_D, {\bm D}^\prime)}$ are Dirac points of $L$, and let $\mathcal{N}^{\bm D}(\varepsilon; {\bm q})$ and $\mathcal{N}^{{\bm D}^\prime}(\varepsilon; {\bm q}')$, respectively, denote their reduction matrices; see Proposition \ref{prop:red-D}. Furthermore, assume that the corresponding normalized eigenstates satisfy $\Phi^{{\bm D}^\prime}_\ell = \mathcal{P}[\Phi^{\bm D}_\ell]$, $\ell \smallin \{ 1, \, 2 \}$. Then
\begin{equation}
\label{eq:dir-pair-red}
\mathcal{N}^{{\bm D}^\prime}(\varepsilon; {\bm q}') = \mathcal{N}^{\bm D}(\varepsilon; -{\bm q}').
\end{equation}
In particular, we have
\begin{equation}
\label{eq:dir-pair-eff}
L^{{\bm D}^\prime}_{\rm eff}({\bm q}') = L^{\bm D}_{\rm eff}(-{\bm q}') = - L^{\bm D}_{\rm eff}({\bm q}').
\end{equation}
\end{proposition}

\begin{proof}
A detailed proof can be found in \cite{chaban2025thesis}.
\end{proof}

\bigskip

\section{Change of variables; Proof of Proposition \ref{prop:fb-evp-HT-TH}}
\label{sec:ch-var}

\setcounter{equation}{0}
\setcounter{figure}{0}

{\it Proof of Proposition \ref{prop:fb-evp-HT-TH}.}
Suppose ${\tilde{\Phi}({\bm x}) \smallin L^2_{(T^{-1})^\mathsf{T}{\bm k}}(\R^2/T\Z^2)}$ satisfies the eigenvalue problem:
\begin{equation}
\label{eq:evp-HT}
H_{V \circ \, T^{-1}} \, \tilde{\Phi}({\bm x}) = ( -\Delta_{\bm x} + V(T^{-1}{\bm x}) ) \, \tilde{\Phi}({\bm x}) = E \, \tilde{\Phi}({\bm x}), \quad {\bm x} \smallin \R^2/T\Z^2.
\end{equation}
Consider the change of variables ${{\bm y} = T^{-1} {\bm x}}$. Then, for any function $f({\bm x})$,
\begin{equation}
f({\bm x}) = f(T{\bm y}) \bigr|_{{\bm y} \, = \, T^{-1}{\bm x}} = (f \circ T)({\bm y}) \bigr|_{{\bm y} \, = \, T^{-1}{\bm x}}.
\end{equation}
By the chain rule,
\begin{align}
\nabla_{\bm x} f({\bm x}) & = (T^{-1})^\mathsf{T} \nabla_{\bm y} (f \circ T)({\bm y}) \bigr|_{{\bm y} \, = \, T^{-1}{\bm x}}, \\
\Delta_{\bm x} f({\bm x}) & = \nabla_{\bm y} \cdot T^{-1}(T^{-1})^\mathsf{T} \nabla_{\bm y} (f \circ T)({\bm y}) \bigr|_{{\bm y} \, = \, T^{-1}{\bm x}}.
\end{align}
Substituting into \eqref{eq:evp-HT} and using that ${(T^{-1})^\mathsf{T} T^{-1} = (T^\mathsf{T} T)^{-1}}$ yields
\begin{equation}
\label{eq:evp-TH}
\bigl( - \nabla_{\bm y} \cdot (T^\mathsf{T} T)^{-1} \nabla_{\bm y} + V({\bm y}) \bigr) \, (\tilde{\Phi} \circ T)({\bm y})\bigr|_{{\bm y} \, = \, T^{-1}{\bm x}} = E \, (\tilde{\Phi} \circ T)({\bm y})\bigr|_{{\bm y} \, = \, T^{-1}{\bm x}}.
\end{equation}
Define ${\Phi({\bm y}) \equiv (\tilde{\Phi} \circ T)({\bm y})}$ (in particular, ${\tilde{\Phi} = \Phi \circ T^{-1}}$). We claim that ${\Phi({\bm y}) \smallin L^2_{\bm k}(\R^2/\Z^2)}$; indeed, for any ${{\bm v} \smallin \Z^2}$,
\begin{equation}
\Phi({\bm y} + {\bm v}) = (\tilde{\Phi} \circ T)({\bm y} + {\bm v}) = \tilde{\Phi}(T {\bm y} + T {\bm v}) = e^{i (T^{-1})^\mathsf{T} {\bm k} \, \cdot \, T {\bm v}} \tilde{\Phi}(T {\bm y}) = e^{i {\bm k} \cdot T^{-1} T {\bm v}} \Phi({\bm y}) = e^{i {\bm k} \cdot {\bm v}} \Phi({\bm y}),
\end{equation}
where we used that ${\tilde{\Phi} \smallin L^2_{(T^{-1})^\mathsf{T} {\bm k}}(\R^2/T\Z^2)}$. Now, since \eqref{eq:evp-TH} holds for all ${{\bm x} \smallin \R^2/T\Z^2}$, it follows that ${\Phi \smallin L^2_{\bm k}(\R^2/\Z^2)}$ satisfies the eigenvalue problem
\begin{equation}
T_* H_V \, \Phi({\bm y}) = \bigl( - \nabla_{\bm y} \cdot (T^\mathsf{T} T)^{-1} \nabla_{\bm y} + V({\bm y}) \bigr) \, \Phi({\bm y}) = E \, \Phi({\bm y}), \quad {\bm y} \smallin \R^2/\Z^2.
\end{equation}
The proof is now complete. \qed

\bigskip

\section{Computation of $\beta_0$, $\beta_1$, and $\beta_2$ for small amplitude potentials}
\label{apx:b-par-small}

\setcounter{equation}{0}
\setcounter{figure}{0}

Our main results, stated in Section \ref{sec:results}, require the additional nondegeneracy hypothesis \ref{itm:quad-dgn-6} on a quadratic band degeneracy of the undeformed Schr\"{o}dinger operator ${H = -\Delta + V}$ to ensure the effect of the deformation enters at leading order in perturbation theory. In this appendix, we compute the parameters $\beta_0$, $\beta_1$, and $\beta_2$ in \eqref{eq:def-b-par}, in the scenario of small amplitude potentials; this computation is analogous to that presented in \cite[Appendix C]{keller2018spectral} regarding the parameters $\alpha_0$, $\alpha_1$, and $\alpha_2$ in \eqref{eq:def-a-par}. We then verify that \ref{itm:quad-dgn-6}, concerning $\beta_1$ and $\beta_2$, holds in this setting.

The following result, which demonstrates the emergence of a quadratic band degeneracy point from the band structure of the Laplacian when perturbed by a small amplitude square lattice potential, is proven in \cite[Theorem 5.3]{keller2018spectral}.

\begin{theorem}
\label{thm:kmow-small}
{\rm (\cite[Theorem 5.3]{keller2018spectral}.)}
Consider the Schr\"{o}dinger operator ${H^\delta = -\Delta + \delta V}$, where $\delta$ is small and $V$ is a square lattice potential; see Definition \ref{def:sql-pot}. Assume
\begin{equation}
\label{eq:small-V-ndgn}
\quad V_{1, 0} \neq \pm V_{1, 1}.
\end{equation}
Then, the fourfold-degenerate $L^2_{\bm M}$ eigenvalue ${E_S^{(0)} = |{\bm M}|^2}$ of ${H^0 = -\Delta}$ splits, at order $\delta$, into a twofold-degenerate $L^2_{\bm M}$ eigenvalue $E_S^\delta$ of $H^\delta$, simple in each of the subspaces $L^2_{{\bm M}, +i}$ and $L^2_{{\bm M}, -i}$, and two simple $L^2_{\bm M}$ eigenvalues with corresponding Floquet-Bloch eigenstates in $L^2_{{\bm M}, +1}$ and $L^2_{{\bm M}, -1}$. 

The twofold-degenerate $L^2_{\bm M}$ eigenvalue $E_S^\delta$ has the expansion
\begin{equation}
\label{eq:small-eig}
E_S^\delta = E_S^{(0)} + \delta (V_{0, 0} - V_{1, 1}) + O(\delta^2) \ \ \text{as} \ \ \delta \to 0.
\end{equation}
The Floquet-Bloch eigenstates ${\Phi^\delta_{(+i)} \smallin L^2_{{\bm M}, +i}}$ and ${\Phi^\delta_{(-i)} \smallin L^2_{{\bm M}, -i}}$ associated with $E_S^\delta$ are given by
\begin{equation}
\label{eq:small-phi-1&2}
\begin{aligned}
\Phi^\delta_{(+i)}({\bm x}) & = \sum_{{\bm m} \smallin \mathcal{S}} c_{(+i)}^\delta({\bm m}) \bigl( e^{i ({\bm M} + {\bm m} {\bm k}) \cdot {\bm x}} - i e^{i ({\bm M} + \mathcal{R} {\bm m} {\bm k}) \cdot {\bm x}} - e^{i ({\bm M} + \mathcal{R}^2 {\bm m} {\bm k}) \cdot {\bm x}} + i e^{i ({\bm M} + \mathcal{R}^3 {\bm m} {\bm k}) \cdot {\bm x}} \bigr), \\
\Phi^\delta_{(-i)}({\bm x}) & = \sum_{{\bm m} \smallin \mathcal{S}} c_{(-i)}^\delta({\bm m}) \bigl( e^{i ({\bm M} + {\bm m} {\bm k}) \cdot {\bm x}} + i e^{i ({\bm M} + \mathcal{R} {\bm m} {\bm k}) \cdot {\bm x}} - e^{i ({\bm M} + \mathcal{R}^2 {\bm m} {\bm k}) \cdot {\bm x}} - i e^{i ({\bm M} + \mathcal{R}^3 {\bm m} {\bm k}) \cdot {\bm x}} \bigr) ,
\end{aligned}
\end{equation}
where we use the notation ${{\bm m}{\bm k} = m_1 {\bm k}_1 + m_2 {\bm k}_2}$ of Appendix \ref{apx:fourier}. Here, ${\mathcal{R}(m_1, m_2) = (m_2, -m_1 - 1)}$ and $\mathcal{S}$ is a set of representatives of equivalence classes consisting of orbits of $\mathcal{R}$. The coefficients satisfy
\begin{align}
\label{eq:small-phi-PC}
& c^\delta_{(-i)}({\bm m}) = c^\delta_{(+i)}({\bm m})^*, \\
\label{eq:small-phi-ref}
& c^\delta_{(\ell)}(\Sigma_1 {\bm m}) = c^\delta_{(\ell)}({\bm m})^* \quad \text{where} \quad \Sigma_1 (m_1, m_2) = (m_2, m_1), \\
\label{eq:small-phi-ex}
& c_{(\ell)}^\delta({\bm 0}) = 1 \quad \text{and} \quad c_{(\ell)}^\delta({\bm m}) = O(\delta) \ \ \text{as} \ \ \delta \to 0, \quad {\bm m} \smallin \mathcal{S}^\perp \equiv \mathcal{S} \setminus \{ {\bm 0} \}
\end{align}
for ${\ell \smallin \{ +i, -i \}}$.
\end{theorem}

Further, in \cite[Corollary 5.4]{keller2018spectral}, the parameters $\alpha_0$, $\alpha_1$, and $\alpha_2$ in \eqref{eq:def-a-par} are calculated and the nondegeneracy condition \ref{itm:quad-dgn-5} is verified.

\begin{corollary}
\label{cor:a-par-small}
{\rm (\cite[Corollary 5.4]{keller2018spectral}.)}
Consider the setup of Theorem \ref{thm:kmow-small}. In addition, assume
\begin{equation}
\label{eq:small-V-ndgn_2}
V_{1, 0} \neq 0 \quad \text{and} \quad V_{1, 1} \neq 0.
\end{equation}
Then
\begin{equation}
\label{eq:a-par-small}
\begin{aligned}
\alpha_0^\delta & = 4 \langle \partial_{x_1} \Phi_{(+i)}^\delta, \, \mathscr{R}^\delta(E_S^\delta) \, \partial_{x_1} \Phi_{(+i)}^\delta \rangle = \frac{32 \pi^2}{\delta} \frac{V_{1, 1}}{V_{1, 1}^2 - V_{1, 0}^2} + O(1), \\
\alpha_1^\delta & = 4 \langle \partial_{x_1} \Phi_{(+i)}^\delta, \, \mathscr{R}^\delta(E_S^\delta) \, \partial_{x_2} \Phi_{(-i)}^\delta \rangle = \frac{32 \pi^2}{\delta} \frac{V_{1, 1}}{V_{1, 1}^2 - V_{1, 0}^2} + O(1), \\
\alpha_2^\delta & = 4i \langle \partial_{x_1} \Phi_{(+i)}^\delta, \, \mathscr{R}^\delta(E_S^\delta) \, \partial_{x_1} \Phi_{(-i)}^\delta \rangle = \frac{32 \pi^2}{\delta} \frac{V_{1, 0}}{V_{1, 1}^2 - V_{1, 0}^2} + O(1)
\end{aligned}
\end{equation}
as ${\delta \to 0}$. In particular, ${\alpha_1^\delta \neq 0}$ and ${\alpha_2^\delta \neq 0}$ so that the nondegeneracy hypothesis \ref{itm:quad-dgn-5} is satisfied.
\end{corollary}

Now, let us extend Corollary \ref{cor:a-par-small} to the parameters $\beta_0$, $\beta_1$, and $\beta_2$, defined in \eqref{eq:def-b-par}, and verify the additional nondegeneracy condition \ref{itm:quad-dgn-6}.

\begin{proposition}
\label{prop:b-par-small}
Consider the setup of Theorem \ref{thm:kmow-small}. Then
\begin{equation}
\label{eq:b-par-small}
\begin{aligned}
\beta_0^\delta & = \langle \partial_{x_1} \Phi_{(+i)}^\delta, \, \partial_{x_1} \Phi_{(+i)}^\delta \rangle = 4 \pi^2 + O(\delta^2), \\
\beta_1^\delta & = \langle \partial_{x_1} \Phi_{(+i)}^\delta, \, \partial_{x_2} \Phi_{(-i)}^\delta \rangle = 4 \pi^2 + O(\delta^2), \\
\beta_2^\delta & = i \langle \partial_{x_1} \Phi_{(+i)}^\delta, \, \partial_{x_1} \Phi_{(-i)}^\delta \rangle = O(\delta^2)
\end{aligned}
\end{equation}
as ${\delta \to 0}$. In particular, ${\beta_1^\delta \neq 0}$ so that the additional nondegeneracy hypothesis \ref{itm:quad-dgn-6} is satisfied.
\end{proposition}

\begin{proof}
\allowdisplaybreaks
We first derive the expansion of $\beta_0^\delta$. Using \eqref{eq:small-phi-1&2}, we have the exact expression
\begin{align}
\label{eq:b0-ex}
\beta_0^\delta & = \langle \partial_{x_1} \Phi_{(+i)}^\delta, \, \partial_{x_1} \Phi_{(+i)}^\delta \rangle = \int_\Omega \partial_{x_1} \Phi^\delta_{(+i)}({\bm x})^* \, \partial_{x_1} \Phi_{(+i)}^\delta({\bm x}) \, {\rm d}{\bm x} \\
& = \int_\Omega \Bigl( \partial_{x_1} \sum_{{\bm m} \smallin \mathcal{S}} c_{(+i)}^\delta({\bm m}) \bigl( e^{i ({\bm M} + {\bm m} {\bm k}) \cdot {\bm x}} - i e^{i ({\bm M} + \mathcal{R} {\bm m} {\bm k}) \cdot {\bm x}} - e^{i ({\bm M} + \mathcal{R}^2 {\bm m} {\bm k}) \cdot {\bm x}} + i e^{i ({\bm M} + \mathcal{R}^3 {\bm m} {\bm k}) \cdot {\bm x}} \bigr) \Bigr)^{\! *} \nonumber \\
& \qquad \qquad \cdot \Bigl( \partial_{x_1} \sum_{{\bm n} \smallin \mathcal{S}} c_{(+i)}^\delta({\bm n}) \bigl( e^{i ({\bm M} + {\bm n} {\bm k}) \cdot {\bm x}} - i e^{i ({\bm M} + \mathcal{R} {\bm n} {\bm k}) \cdot {\bm x}} - e^{i ({\bm M} + \mathcal{R}^2 {\bm n} {\bm k}) \cdot {\bm x}} + i e^{i ({\bm M} + \mathcal{R}^3 {\bm n} {\bm k}) \cdot {\bm x}} \bigr) \Bigr) \, {\rm d}{\bm x} \nonumber \\
& = \sum_{{\bm m} \smallin \mathcal{S}} c^\delta_{(+i)}({\bm m})^* \, c^\delta_{(+i)}({\bm m}) \bigl( ({\bm M} + {\bm m} {\bm k})_1^2 + ({\bm M} + \mathcal{R} {\bm m} {\bm k})_1^2 + ({\bm M} + \mathcal{R}^2 {\bm m} {\bm k})_1^2 + ({\bm M} + \mathcal{R}^3 {\bm m} {\bm k})_1^2 \bigr) \nonumber \\
& = \sum_{{\bm m} \smallin \mathcal{S}} 2 \pi^2 \bigl| c_{(+i)}^\delta({\bm m}) \bigr|^2 \bigl( (1 + 2 m_1)^2 + (1 + 2 m_2)^2 \bigr). \nonumber
\end{align}
Applying the expansions in \eqref{eq:small-phi-ex} then yields
\begin{equation}
\label{eq:b0-ex_2}
\beta_0^\delta = 4 \pi^2 + O(\delta^2) \ \ \text{as} \ \ \delta \to 0.
\end{equation}

For the expansion of $\beta_1^\delta$, similarly
\begin{align}
\label{eq:b1-ex}
\beta_1^\delta & = \langle \partial_{x_1} \Phi_{(+i)}^\delta, \, \partial_{x_2} \Phi_{(-i)}^\delta \rangle = \int_\Omega \partial_{x_1} \Phi^\delta_{(+i)}({\bm x})^* \, \partial_{x_2} \Phi_{(-i)}^\delta({\bm x}) \, {\rm d}{\bm x} \\
& = \int_\Omega \Bigl( \partial_{x_1} \sum_{{\bm m} \smallin \mathcal{S}} c_{(+i)}^\delta({\bm m}) \bigl( e^{i ({\bm M} + {\bm m} {\bm k}) \cdot {\bm x}} - i e^{i ({\bm M} + \mathcal{R} {\bm m} {\bm k}) \cdot {\bm x}} - e^{i ({\bm M} + \mathcal{R}^2 {\bm m} {\bm k}) \cdot {\bm x}} + i e^{i ({\bm M} + \mathcal{R}^3 {\bm m} {\bm k}) \cdot {\bm x}} \bigr) \Bigr)^{\! *} \nonumber \\
& \qquad \qquad \cdot \Bigl( \partial_{x_2} \sum_{{\bm n} \smallin \mathcal{S}} c_{(-i)}^\delta({\bm n}) \bigl( e^{i ({\bm M} + {\bm n} {\bm k}) \cdot {\bm x}} + i e^{i ({\bm M} + \mathcal{R} {\bm n} {\bm k}) \cdot {\bm x}} - e^{i ({\bm M} + \mathcal{R}^2 {\bm n} {\bm k}) \cdot {\bm x}} - i e^{i ({\bm M} + \mathcal{R}^3 {\bm n} {\bm k}) \cdot {\bm x}} \bigr) \Bigr) \, {\rm d}{\bm x} \nonumber \\
& = \sum_{{\bm m} \smallin \mathcal{S}} c^\delta_{(+i)}({\bm m})^* \, c^\delta_{(-i)}({\bm m}) \, \bigl( ({\bm M} + {\bm m} {\bm k})_1 ({\bm M} + {\bm m} {\bm k})_2 - ({\bm M} + \mathcal{R} {\bm m} {\bm k})_1 ({\bm M} + \mathcal{R} {\bm m} {\bm k})_2 \nonumber \\
& \qquad \qquad + ({\bm M} + \mathcal{R}^2 {\bm m} {\bm k})_1 ({\bm M} + \mathcal{R}^2 {\bm m} {\bm k})_2 - ({\bm M} + \mathcal{R}^3 {\bm m} {\bm k})_1 ({\bm M} + \mathcal{R}^3 {\bm m} {\bm k})_2 \bigr). \nonumber \\[1ex]
& = \sum_{{\bm m} \smallin \mathcal{S}} 4 \pi^2 c^\delta_{(-i)}({\bm m})^2 \, (1 + 2 m_1) (1 + 2 m_2). \nonumber
\end{align}
Let us partition $\mathcal{S}$ as ${\mathcal{S} = \mathcal{S}_1 \cup \mathcal{S}_2 \cup \Sigma_1 \mathcal{S}_2}$, where
\begin{align}
\label{eq:def-S1}
\mathcal{S}_1 & \equiv \{ (m_1, m_2) \smallin \mathcal{S} : m_1 = m_2 \}, \\
\label{eq:def-S2}
\mathcal{S}_2 & \equiv \{ (m_1, m_2) \smallin \mathcal{S} : m_1 \geq 1, \ 0 \leq m_2 < m_1 \}.
\end{align}
By \eqref{eq:small-phi-ref}, the Fourier coefficients with indices ${{\bm m} \smallin \mathcal{S}_1}$ satisfy
\begin{equation}
\label{eq:small-phi-S1}
c^{\delta}_{(-i)}({\bm m})^* = c^{\delta}_{(-i)}(\Sigma_1 {\bm m}) = c^{\delta}_{(-i)}({\bm m}), \quad \text{implying} \quad c^{\delta}_{(-i)}({\bm m}) \smallin \R.
\end{equation}
Hence, isolating the terms of \eqref{eq:b1-ex} with indices ${{\bm m} \smallin \mathcal{S}_1}$, we have
\begin{equation}
\label{eq:b1-S1}
\sum_{{\bm m} \smallin \mathcal{S}_1} 4 \pi^2 c^\delta_{(-i)}({\bm m})^2 \, (1 + 2 m_1) (1 + 2 m_2) = {\rm Re} \sum_{{\bm m} \smallin \mathcal{S}_1} 4 \pi^2 c^\delta_{(-i)}({\bm m})^2 \, (1 + 2 m_1)^2.
\end{equation}
We pair the remaining terms with indices ${\bm m}$ and $\Sigma_1 {\bm m}$, ${{\bm m} \smallin \mathcal{S}_2}$, then use \eqref{eq:small-phi-ref} to obtain
\begin{align}
\label{eq:b1-S2}
& \sum_{{\bm m} \smallin \mathcal{S}_2 \cup \Sigma_1 \mathcal{S}_2} 4 \pi^2 c^\delta_{(-i)}({\bm m})^2 (1 + 2 m_1) (1 + 2 m_2) \\
& \qquad = \sum_{{\bm m} \smallin \mathcal{S}_2} 4 \pi^2 \bigl( c^\delta_{(-i)}({\bm m})^2 \, (1 + 2 m_1) (1 + 2 m_2) + c^\delta_{(-i)}(\Sigma_1 {\bm m})^2 (1 + 2 m_2) (1 + 2 m_1) \bigr) \nonumber \\
& \qquad = \sum_{{\bm m} \smallin \mathcal{S}_2} 4 \pi^2 \bigl( c^\delta_{(-i)}({\bm m})^2 + c^\delta_{(-i)}({\bm m})^*{}^2 \bigr) (1 + 2 m_1) (1 + 2 m_2) \nonumber \\
& \qquad = 2 \, {\rm Re} \sum_{{\bm m} \smallin \mathcal{S}_2} 4 \pi^2 c^\delta_{(-i)}({\bm m})^2 (1 + 2 m_1) (1 + 2 m_2). \nonumber
\end{align}
Recombining the terms, \eqref{eq:b1-ex} reads
\begin{align}
\label{eq:b1-ex_2}
\beta_1^\delta & = \sum_{{\bm m} \smallin \mathcal{S}} 4 \pi^2 c^\delta_{(-i)}({\bm m})^2 \, (1 + 2 m_1) (1 + 2 m_2) \\
& = {\rm Re} \sum_{{\bm m} \smallin \mathcal{S}_1} 4 \pi^2 c^\delta_{(-i)}({\bm m})^2 \, (1 + 2 m_1)^2 + 2 \, {\rm Re} \sum_{{\bm m} \smallin \mathcal{S}_2} 4 \pi^2 c^\delta_{(-i)}({\bm m})^2 (1 + 2 m_1) (1 + 2 m_2). \nonumber
\end{align}
Using the expansions in \eqref{eq:small-phi-ex} then yields
\begin{equation}
\label{eq:b1-ex_3}
\beta_1^\delta = 4 \pi^2 + O(\delta^2) \ \ \text{as} \ \ \delta \to 0.
\end{equation}

Finally, for the expansion of $\beta_2^\delta$, we have
\begin{align}
\label{eq:b2-ex}
\beta_2^\delta & = i \langle \partial_{x_1} \Phi^\delta_{(+i)}, \, \partial_{x_1} \Phi^\delta_{(-i)} \rangle = i \int_\Omega \partial_{x_1} \Phi^\delta_{(+i)}({\bm x})^* \, \partial_{x_1} \Phi^\delta_{(-i)}({\bm x}) \, {\rm d}{\bm x} \\
& = i \int_\Omega \Bigl( \partial_{x_1} \sum_{{\bm m} \smallin \mathcal{S}} c_{(+i)}^\delta({\bm m}) \bigl( e^{i ({\bm M} + {\bm m} {\bm k}) \cdot {\bm x}} - i e^{i ({\bm M} + \mathcal{R} {\bm m} {\bm k}) \cdot {\bm x}} - e^{i ({\bm M} + \mathcal{R}^2 {\bm m} {\bm k}) \cdot {\bm x}} + i e^{i ({\bm M} + \mathcal{R}^3 {\bm m} {\bm k}) \cdot {\bm x}} \bigr) \Bigr)^{\! *} \nonumber \\
& \qquad \qquad \cdot \Bigl( \partial_{x_1} \sum_{{\bm n} \smallin \mathcal{S}} c_{(-i)}^\delta({\bm n}) \bigl( e^{i ({\bm M} + {\bm n} {\bm k}) \cdot {\bm x}} + i e^{i ({\bm M} + \mathcal{R} {\bm n} {\bm k}) \cdot {\bm x}} - e^{i ({\bm M} + \mathcal{R}^2 {\bm n} {\bm k}) \cdot {\bm x}} - i e^{i ({\bm M} + \mathcal{R}^3 {\bm n} {\bm k}) \cdot {\bm x}} \bigr) \Bigr) \, {\rm d}{\bm x} \nonumber \\
& = i \sum_{{\bm m} \smallin \mathcal{S}} 2 \pi^2 c_{(-i)}^\delta({\bm m})^2 \bigl( (1 + 2 m_1)^2 - (1 + 2 m_2)^2 \bigr). \nonumber
\end{align}
Note that the terms of \eqref{eq:b2-ex} with indices ${{\bm m} \smallin \mathcal{S}_1}$ vanish. We pair the remaining terms with indices ${\bm m}$ and $\Sigma_1 {\bm m}$, ${{\bm m} \smallin \mathcal{S}_2}$, then use \eqref{eq:small-phi-ref} to obtain
\begin{align}
\label{eq:b2-S2}
& i \sum_{{\bm m} \smallin \mathcal{S}_2 \cup \Sigma_1 \mathcal{S}_2} 2 \pi^2 c_{(-i)}^\delta({\bm m})^2 \bigl( (1 + 2 m_1)^2 - (1 + 2 m_2)^2 \bigr) \\
& \qquad = i \sum_{{\bm m} \smallin \mathcal{S}_2} 2 \pi^2 \bigl( c_{(-i)}^\delta({\bm m})^2 \bigl( (1 + 2 m_1)^2 - (1 + 2 m_2)^2 \bigr) + c_{(-i)}^\delta(\Sigma_1{\bm m})^2 \bigl( (1 + 2 m_2)^2 - (1 + 2 m_1)^2 \bigr) \bigr) \nonumber \\
& \qquad = i \sum_{{\bm m} \smallin \mathcal{S}_2} 2 \pi^2 \bigl( c_{(-i)}^\delta({\bm m})^2 - c_{(-i)}^\delta({\bm m})^*{}^2 \bigr) \bigl( (1 + 2 m_1)^2 - (1 + 2 m_2)^2 \bigr) \nonumber \\
& \qquad = -2 \, {\rm Im} \sum_{{\bm m} \smallin \mathcal{S}_2} 2 \pi^2 c_{(-i)}^\delta({\bm m})^2 \bigl( (1 + 2 m_1)^2 - (1 + 2 m_2)^2 \bigr). \nonumber
\end{align}
Hence
\begin{equation}
\label{eq:b2-ex_2}
\beta_2^\delta = -2 \, {\rm Im} \sum_{{\bm m} \smallin \mathcal{S}_2} 2 \pi^2 c_{(-i)}^\delta({\bm m})^2 \bigl( (1 + 2 m_1)^2 - (1 + 2 m_2)^2 \bigr).
\end{equation}
Applying the expansions in \eqref{eq:small-phi-ex} then yields
\begin{equation}
\label{eq:b2-ex_3}
\beta_2^\delta = O(\delta^2) \ \ \text{as} \ \ \delta \to 0.
\end{equation}
The proof is now complete.
\end{proof}

We conclude by providing an example of a square lattice potential for which ${\alpha_\ell \neq 0}$ and ${\beta_\ell \neq 0}$, ${0 \leq \ell \leq 2}$, in the small amplitude scaling.

\begin{example}
\label{ex:a-b-par-small}
Consider the square lattice potential
\begin{equation}
\label{eq:sql-pot-b-par}
V({\bm x}) = 2 \, V_{0, 1} \bigl( \cos({\bm k}_1 \cdot {\bm x}) + \cos({\bm k}_2 \cdot {\bm x}) \bigr) + 2 \, V_{1, 1} \bigl( \cos(({\bm k}_1 + {\bm k}_2) \cdot {\bm x}) + \cos(({\bm k}_1 - {\bm k}_2) \cdot {\bm x}) \bigr)
\end{equation}
where ${V_{0, 1} \neq 0}$, ${V_{1, 1} \neq 0}$, and ${V_{0, 1} \neq \pm V_{1, 1}}$. By Corollary \ref{cor:a-par-small}, ${\alpha_\ell^\delta \neq 0}$ for ${0 \leq \ell \leq 2}$. Furthermore, by Proposition \ref{prop:b-par-small}, ${\beta_0^\delta \neq 0}$ and ${\beta_1^\delta \neq 0}$. Finally, we claim that, for $\delta$ small,
\begin{align}
\beta_2^\delta & = - \frac{4}{\pi^2} V_{0, 1} V_{1, 1} \, \delta^2 + O(\delta^3)\ne0.
\end{align}
This can be seen as follows:
\begin{align}
c^\delta_{(-i)}({\bm m}) & = \delta F_{(-i)}({\bm m}, \mu_S^\delta) + O(\delta^2) \nonumber \\
& = - \delta \frac{\mathcal{K}_{(-i)}({\bm m}; {\bm 0})}{|{\bm M} + {\bm m}{\bm k}|^2 - \mu_S^\delta} + O(\delta^2) \nonumber \\
& = - \delta \frac{1}{|{\bm M} + {\bm m}{\bm k}|^2 - |{\bm M}|^2 + O(\delta)} (V_{\bm m} + i V_{{\bm m} - \mathcal{R}{\bm 0}} - V_{{\bm m} - \mathcal{R}^2{\bm 0}} - i V_{{\bm m} - \mathcal{R}^3{\bm 0}}) + O(\delta^2) \nonumber \\
& = - \delta \frac{1}{|{\bm M} + {\bm m}{\bm k}|^2 - |{\bm M}|^2} (V_{\bm m} + i V_{{\bm m} - \mathcal{R}{\bm 0}} - V_{{\bm m} - \mathcal{R}^2{\bm 0}} - i V_{{\bm m} - \mathcal{R}^3{\bm 0}}) + O(\delta^2) \nonumber \\
& = - \frac{\delta}{4 \pi^2} \frac{1}{(m_1 + m_2) + (m_1^2 + m_2^2)} (V_{m_1, m_2} + i V_{m_1, m_2 + 1} - V_{m_1 + 1, m_2 + 1} - i V_{m_1 + 1, m_2}) + O(\delta^2). \nonumber
\end{align}
For $V$ given by \eqref{eq:sql-pot-b-par}, only the following Fourier coefficients are nonzero at $O(\delta)$:
\begin{align}
c_{(-i)}^\delta(0, 0) & = 1 \\
c_{(-i)}^\delta(1, 0) & = - \frac{\delta}{8 \pi^2} (V_{1, 0} + i V_{1, 1}) + O(\delta^2) = - \frac{\delta}{8 \pi^2} (V_{0, 1} + i V_{1, 1}) + O(\delta^2) \\
c_{(-i)}^\delta(0, 1) & = - \frac{\delta}{8 \pi^2} (V_{0, 1} - i V_{1, 1}) + O(\delta^2) = c_{(-i)}^\delta(1, 0)^* \\
c_{(-i)}^\delta(1, 1) & = - \frac{\delta}{16 \pi^2} V_{1, 1} + O(\delta^2).
\end{align}
Of these, only ${(1, 0) \smallin \mathcal{S}_2}$. Thus,
\begin{align}
\beta_2^\delta & = -2 {\rm Im} \bigl( 16 \pi^2 c_{(-i)}^\delta(1, 0)^2 + O(\delta^3) \bigr) \\
& = -2 \, {\rm Im} \Bigl( 16 \pi^2 \frac{\delta^2}{(8 \pi^2)^2} (V_{0, 1}^2 - V_{1, 1}^2 + 2 i V_{0, 1} V_{1, 1}) + O(\delta^3) \Bigr) + O(\delta^3) \nonumber \\
& = - \frac{4}{\pi^2} V_{0, 1} V_{1, 1} \, \delta^2 + O(\delta^3)
\end{align}
as ${\delta \to 0}$. By assumption, $V_{0, 1} \neq 0$ and ${V_{1, 1} \neq 0}$, and therefore ${\beta_2^\delta \neq 0}$.
\end{example}

\bigskip

\section{The matrix $\mathcal{M}(\varepsilon; {\bm \kappa}, \tau_0, {\bm \tau})$}
\label{apx:calM-pr}

\setcounter{equation}{0}
\setcounter{figure}{0}

This appendix contains a technical study of the ${2 \times 2}$ matrix-valued, analytic function $\mathcal{M}(\varepsilon; {\bm \kappa}, \tau_0, {\bm \tau})$ arising in the Lyapunov-Schmidt reduction of Section \ref{sec:set-up}. First, in Appendix \ref{apx:calM-ent}, we display explicit expressions for its entries. Next, in Appendix \ref{apx:calM-sa}, we prove self-adjointness. In Section \ref{apx:calM-0-pr}, we then somewhat simplify the leading order terms of the entries. Finally, in Section \ref{apx:pf-calM-pr}, we proceed to prove Proposition \ref{prop:calM-pr}, which extends properties observable at leading order to the exact expression for the matrix.

\subsection{Entries of $\mathcal{M}(\varepsilon; {\bm \kappa}, \tau_0, {\bm \tau})$}
\label{apx:calM-ent}

The following decomposition is convenient:
\begin{equation}
\label{eq:calM-ex_1}
\mathcal{M}(\varepsilon; {\bm \kappa}, \tau_0, {\bm \tau}) = \mathcal{M}^{(0)}(\varepsilon; {\bm \kappa}, \tau_0, {\bm \tau}) + \mathcal{M}^{(1)}(\varepsilon; {\bm \kappa}, \tau_0, {\bm \tau}).
\end{equation}
Here, the entries of $\mathcal{M}^{(0)}(\varepsilon; {\bm \kappa}, \tau_0, {\bm \tau})$ are, for $j$, ${k \smallin \{ 1, 2 \}}$,
\begin{align}
\label{eq:def-calM-0}
\mathcal{M}^{(0)}_{j, k}(\varepsilon; {\bm \kappa}, \tau_0, {\bm \tau}) & = \langle \Phi_j, \, F_0(-i\nabla) \Phi_k \rangle, \quad \text{where} \\
F_0(-i\nabla) & \equiv \varepsilon - |{\bm \kappa}|^2 - \tau_0 (\Delta - |{\bm \kappa}|^2) - |{\bm \tau}| (\nabla \cdot \sigma_\varphi \nabla - {\bm \kappa} \cdot \sigma_\varphi {\bm \kappa})=F_0(-i\nabla)^*,
\end{align}
and the entries of $\mathcal{M}^{(1)}(\varepsilon; {\bm \kappa}, \tau_0, {\bm \tau})$ are
\begin{align}
\label{eq:def-calM-1}
\mathcal{M}^{(1)}_{j, k}(\varepsilon; {\bm \kappa}, \tau_0, {\bm \tau}) & = \langle \Phi_j, \, F_1(-i\nabla)^* (1 - \mathcal{A}(\varepsilon; {\bm \kappa}, \tau_0, {\bm \tau}))^{-1} \mathscr{R}(E_S) F_1(-i\nabla) \Phi_k \rangle, \quad \text{where} \\
F_1(-i\nabla) & \equiv 2i{\bm \kappa}\cdot\nabla - \tau_0 (\Delta + 2i{\bm \kappa} \cdot \nabla) - |{\bm \tau}| (\nabla \cdot \sigma_\varphi \nabla + 2i {\bm \kappa} \cdot \sigma_\varphi \nabla) = F_1(-i\nabla)^*
\label{eq:def-F} \\
\mathcal{A}(\varepsilon; {\bm \kappa}, \tau_0, {\bm \tau}) & \equiv \mathscr{R}(E_S) F_2(-i\nabla),
\label{eq:def-calA} \\
\mathscr{R}(E_S) & \equiv \Pi^\perp (H - E_S)^{-1} \Pi^\perp,
\label{eq:def-res} \\
F_2(-i\nabla) & \equiv \varepsilon + 2i {\bm \kappa} \cdot \nabla - |{\bm \kappa}|^2 - \tau_0 (\nabla + i{\bm \kappa})^2 - |{\bm \tau}| (\nabla + i{\bm \kappa}) \cdot \sigma_\varphi (\nabla + i{\bm \kappa}) = F_2(-i\nabla)^*,
\end{align}
and $\Pi^\parallel$ and $\Pi^\perp$ are orthogonal projections onto ${\rm ker}(H - E_S)$ and ${\rm ker}(H - E_S)^\perp$, respectively, given by:
\begin{equation}
\Pi^\parallel \equiv \Phi_1 \langle \Phi_1, \, \cdot \rangle + \Phi_2 \langle \Phi_2, \, \cdot \rangle \quad \text{and} \quad \Pi^\perp  \equiv 1 - \Pi^\parallel.
\end{equation}

\subsection{$\mathcal{M}(\varepsilon; {\bm \kappa}, \tau_0, {\bm \tau})$ is self-adjoint}
\label{apx:calM-sa}

Clearly, $\mathcal{M}^{(0)}(\varepsilon; {\bm \kappa}, \tau_0, {\bm \tau})$ is self-adjoint since $F_0(-i\nabla)$ is self-adjoint. So, by \eqref{eq:calM-ex_1}, it suffices to show that $\mathcal{M}^{(1)}(\varepsilon; {\bm \kappa}, \tau_0, {\bm \tau})$ is self-adjoint. Since $F_1(-i\nabla)$ is self-adjoint, this reduces to showing that $(1 - \mathcal{A}(\varepsilon; {\bm \kappa}, \tau_0, {\bm \tau}))^{-1} \mathscr{R}(E_S)$ is self-adjoint, or equivalently,
\begin{equation}
(1 - \mathscr{R}(E_S) F_2(-i\nabla))^{-1} \mathscr{R}(E_S) = ( (1 - \mathscr{R}(E_S) F_2(-i\nabla))^{-1} \mathscr{R}(E_S) )^*.
\end{equation}
For any ${g \smallin L^2_{\bm M}}$, define:
\begin{align}
\label{eq:def-u}
u & \equiv (1 - \mathscr{R}(E_S) F_2(-i\nabla))^{-1} \mathscr{R}(E_S) g, \\
\label{eq:def-v}
v & \equiv ( (1 - \mathscr{R}(E_S) F_2(-i\nabla))^{-1} \mathscr{R}(E_S) )^* g = \mathscr{R}(E_S) (1 -  F_2(-i\nabla)\mathscr{R}(E_S))^{-1} g.
\end{align}
We claim that ${u = v}$, which implies self-adjointness of $\mathcal{M}^{(1)}(\varepsilon; {\bm \kappa}, \tau_0, {\bm \tau})$. First, ${u \in L^2_{\bm M}}$ uniquely solves
\begin{equation}
\label{eq:u-sol}
(1 - \mathscr{R}(E_S) F_2(-i\nabla)) u = \mathscr{R}(E_S) g.
\end{equation}
Applying ${\Pi^\perp (H - E_S) \Pi^\perp}$ gives 
\begin{equation}
\label{eq:u-sol-2}
\bigl( \Pi^\perp (H - E_S) \Pi^\perp - F_2(-i\nabla) \bigr) u = g.
\end{equation}
Next, ${v \in L^2_{\bm M}}$ uniquely solves 
\begin{equation}
\label{eq:v-sol}
\Pi^\perp (H - E_S) \Pi^\perp v = (1 -  F_2(-i\nabla) \mathscr{R}(E_S))^{-1} g.
\end{equation}
Applying ${(1 - F_2(-i\nabla) \mathscr{R}(E_S))}$, we obtain
\begin{equation}
\label{eq:v-sol-2}
(1 -  F_2(-i\nabla) \mathscr{R}(E_S)) \Pi^\perp (H - E_S) \Pi^\perp v = g.
\end{equation}
Finally, observe that \eqref{eq:v-sol-2} is equivalent to \eqref{eq:u-sol-2}. It follows by uniqueness that ${u = v}$.

\subsection{Simplifying $\mathcal{M}^{(0)}(\varepsilon; {\bm \kappa}, \tau_0, {\bm \tau})$}
\label{apx:calM-0-pr}

Here, we use symmetry arguments to simplify $\mathcal{M}^{(0)}(\varepsilon; {\bm \kappa}, \tau_0, {\bm \tau})$, defined in \eqref{eq:def-calM-0}.

\begin{proposition}
\label{prop:calM-0-ex}
The ${2 \times 2}$ self-adjoint matrix $\mathcal{M}^{(0)}(\varepsilon; {\bm \kappa}, \tau_0, {\bm \tau})$ is given by
\begin{equation}
\label{eq:calM-0-ex}
\mathcal{M}^{(0)}(\varepsilon; {\bm \kappa}, \tau_0, {\bm \tau}) = (\varepsilon - \beta_0 \tau_0 - |{\bm \kappa}|^2 + \tau_0 |{\bm \kappa}|^2 + |{\bm \tau}| {\bm \kappa} \cdot \sigma_\varphi {\bm \kappa}) I - |{\bm \tau}| (\beta_1 \cos(\varphi) \sigma_1 + \beta_2 \sin(\varphi) \sigma_2).
\end{equation}
The parameters $\beta_0$, $\beta_1$, ${\beta_2 \smallin \R}$ are defined in \eqref{eq:def-b-par}. Under \ref{itm:quad-dgn-6}, we assume $\beta_1$, ${\beta_2 \neq 0}$.
\end{proposition}

\begin{proof}
We begin with the expression \eqref{eq:def-calM-0}. Note, from properties \ref{itm:quad-dgn-2}, \ref{itm:quad-dgn-3}, and from \eqref{eq:L2M-dsum}, that $\{ \Phi_1, \Phi_2 \}$ is an orthonormal basis of ${\rm ker}(H - E_S)$. Hence, for $j$, ${k \smallin \{ 1, 2 \}}$,
\begin{equation}
\label{eq:calM-0-ex_2}
\mathcal{M}^{(0)}(\varepsilon; {\bm \kappa}, \tau_0, {\bm \tau})_{j, k} = (\varepsilon - |{\bm \kappa}|^2 + \tau_0 |{\bm \kappa}|^2 + |{\bm \tau}| {\bm \kappa} \cdot \sigma_\varphi {\bm \kappa}) I_{j, k} - \tau_0 \langle \Phi_j, \, \Delta \Phi_k \rangle - |{\bm \tau}| \langle \Phi_j, \, \nabla \cdot \sigma_\varphi \nabla \Phi_k \rangle.
\end{equation}
It therefore suffices to show
\begin{equation}
\label{eq:calM-0-ex_3}
\langle \Phi_j, \, \Delta \Phi_k \rangle = \beta_0 I_{j, k} \quad \text{and} \quad \langle \Phi_j, \, \nabla \cdot \sigma_\varphi \nabla \Phi_k \rangle = \beta_1 \cos(\varphi) (\sigma_1)_{j, k} + \beta_2 \sin(\varphi) (\sigma_2)_{j, k}.
\end{equation}
For convenience, we define ${2 \times 2}$ self-adjoint matrices $A$, $B$, and $C$ with entries
\begin{equation}
\label{eq:def-A-B-C}
A_{j, k} \equiv \langle \Phi_j, \, \Delta \Phi_k \rangle, \quad B_{j, k} \equiv \langle \Phi_j, \, \nabla \cdot \sigma_1 \nabla \Phi_k \rangle, \quad \text{and} \quad C_{j, k} \equiv \langle \Phi_j, \, \nabla \cdot \sigma_3 \nabla \Phi_k \rangle.
\end{equation}
To derive the equalities in \eqref{eq:calM-0-ex_3}, it therefore suffices to show that
\begin{equation}
\label{eq:A-B-C-ex}
A = \beta_0 I, \quad B = \beta_1 \sigma_1, \quad \text{and} \quad C = \beta_2 \sigma_2.
\end{equation}

\begin{enumerate}
\item First, we prove the equality in \eqref{eq:A-B-C-ex} regarding $A$:
\begin{enumerate}
\item Since $\mathcal{R}$ is unitary and $\mathcal{R}[\Phi_\ell] = i^{2\ell - 1} \Phi_\ell$ for $\ell \smallin \{ 1, \, 2\}$,
\begin{align}
A_{j, k} & = \langle \Phi_j, \, \Delta \Phi_k \rangle = \langle \mathcal{R}[\Phi_j], \, \mathcal{R}[\Delta \Phi_k] \rangle = \langle \mathcal{R}[\Phi_j], \, \Delta \mathcal{R}[\Phi_k] \rangle \\
& = \langle (i^{2j - 1} \Phi_j), \, \Delta (i^{2k - 1} \Phi_k) \rangle = i^{2(k - j)} \langle \Phi_j, \, \Delta \Phi_k \rangle = (-1)^{k - j} A_{j, k}. \nonumber
\end{align}
Thus ${A_{1, 2} = -A_{1, 2}}$, which implies ${A_{1, 2} = 0}$.
\item Since $\Sigma_1$ is unitary, ${\Sigma_1[\Phi_1] = \Phi_2}$, and ${\sigma_1^2 = I}$,
\begin{align}
A_{1, 1} & = \langle \Phi_1, \, \Delta \Phi_1 \rangle = \langle \Sigma_1[\Phi_1], \, \Sigma_1[\Delta \Phi_1] \rangle \\
& = \langle \Sigma_1[\Phi_1], \, \nabla \cdot \sigma_1^2 \nabla \Sigma_1[\Phi_1] \rangle = \langle \Phi_2, \, \Delta \Phi_2 \rangle = A_{2, 2}. \nonumber
\end{align}
(Alternately, $\mathcal{PC}$ can be used.)
\item It follows that ${A = A_{1, 1} I}$, where
\begin{align}
A_{1, 1} & = \langle \Phi_1, \, \Delta \Phi_1 \rangle = \langle \Phi_1, \, (\partial_{x_1}^2 + \partial_{x_2}^2) \Phi_1 \rangle = \langle \Phi_1, \, \partial_{x_1}^2 \Phi_1 \rangle + \langle \Phi_1, \, \partial_{x_2}^2 \Phi_1 \rangle \\
& = - \langle \partial_{x_1} \Phi_1, \, \partial_{x_1} \Phi_1 \rangle - \langle \partial_{x_2} \Phi_1, \, \partial_{x_2} \Phi_1 \rangle. \nonumber
\end{align}
Since $\mathcal{R}$ is unitary, ${\mathcal{R} \circ \partial_{x_2} = \partial_{x_1} \circ \mathcal{R}}$, and ${\mathcal{R}[\Phi_1] = i \Phi_1}$,
\begin{align}
\langle \partial_{x_2} \Phi_1, \, \partial_{x_2} \Phi_1 \rangle & = \langle \mathcal{R}[\partial_{x_2} \Phi_1], \, \mathcal{R}[\partial_{x_2} \Phi_1] \rangle = \langle \partial_{x_1} \mathcal{R}[\Phi_1], \, \partial_{x_1} \mathcal{R}[\Phi_1] \rangle \\
& = \langle \partial_{x_1} (i \Phi_1), \, \partial_{x_1} (i \Phi_1) \rangle = \langle \partial_{x_1} \Phi_1, \, \partial_{x_1} \Phi_1 \rangle, \nonumber
\end{align}
which implies ${A_{1, 1} = -2 \langle \partial_{x_1} \Phi_1, \, \partial_{x_1} \Phi_1 \rangle = -2 \lVert \partial_{x_1} \Phi_1 \rVert = \beta_0}$. Thus ${A = \beta_0 I}$.
\end{enumerate}

\item Next, we prove the equality in \eqref{eq:A-B-C-ex} regarding $B$:
\begin{enumerate}
\item Since $\mathcal{R}$ is unitary, ${\mathcal{R}[\Phi_\ell] = i^{2\ell - 1} \Phi_\ell}$ for ${\ell \smallin \{ 1, \, 2\}}$, and ${R \sigma_1 R^\mathsf{T} = - \sigma_1}$,
\begin{align}
B_{\ell, \ell} & = \langle \Phi_\ell, \, \nabla \cdot \sigma_1 \nabla \Phi_\ell \rangle = \langle \mathcal{R}[\Phi_\ell], \, \mathcal{R}[\nabla \cdot \sigma_1 \nabla \Phi_\ell] \rangle = \langle \mathcal{R}[\Phi_\ell], \, \nabla \cdot R \sigma_1 R^\mathsf{T} \nabla \mathcal{R}[\Phi_\ell] \rangle \\
& = \langle (i^{2\ell - 1} \Phi_\ell), \, - \nabla \cdot \sigma_1 \nabla (i^{2\ell - 1} \Phi_\ell) \rangle = - \langle \Phi_\ell, \, \nabla \cdot \sigma_1 \nabla \Phi_\ell \rangle = - B_{\ell, \ell}. \nonumber
\end{align}
Therefore $B_{1, 1}$, ${B_{2, 2} = 0}$.
\item Since $\Sigma_1$ is unitary, ${\Sigma_1[\Phi_1] = \Phi_2}$, ${\Sigma_1[\Phi_2] = \Phi_1}$, and ${\sigma_1^3 = \sigma_1}$,
\begin{align}
B_{1, 2}^* & = \langle \Phi_1, \, \nabla \cdot \sigma_1 \nabla \Phi_2 \rangle^* = \langle \Sigma_1[\Phi_1], \, \Sigma_1[\nabla \cdot \sigma_1 \nabla \Phi_2] \rangle^* \\
& = \langle \Sigma_1[\Phi_1], \, \nabla \cdot \sigma_1^3 \nabla \Sigma_1[\Phi_2] \rangle^* = \langle \Phi_2, \, \nabla \cdot \sigma_1 \nabla \Phi_1 \rangle^* = B_{2, 1}^* = B_{1, 2}. \nonumber
\end{align}
Therefore ${B_{1, 2} \smallin \R}$. 
\item It follows that ${B = B_{1, 2} \sigma_1}$, where
\begin{equation}
B_{1, 2} = \langle \Phi_1, \, \nabla \cdot \sigma_1 \nabla \Phi_2 \rangle_{L^2_{\bm M}} = \langle \Phi_1, \, 2 \partial_{x_1} \partial_{x_2} \Phi_2 \rangle_{L^2_{\bm M}} = - 2 \langle \partial_{x_1} \Phi_1, \, \partial_{x_2} \Phi_2 \rangle_{L^2_{\bm M}} = \beta_1.
\end{equation}
Thus ${B = \beta_1 \sigma_1}$.
\end{enumerate}

\item Finally, we prove the equality in \eqref{eq:A-B-C-ex} regarding $C$:
\begin{enumerate}
\item Since $\mathcal{R}$ is unitary, $\mathcal{R}[\Phi_\ell] = i^{2\ell - 1} \Phi_\ell$ for $\ell \smallin \{ 1, \, 2\}$, and $R \sigma_3 R^\mathsf{T} = - \sigma_3$
\begin{align}
C_{\ell, \ell} & = \langle \Phi_\ell, \, \nabla \cdot \sigma_3 \nabla \Phi_\ell \rangle = \langle \mathcal{R}[\Phi_\ell], \, \mathcal{R}[\nabla \cdot \sigma_3 \nabla \Phi_\ell] \rangle = \langle \mathcal{R}[\Phi_\ell], \, \nabla \cdot R \sigma_3 R^\mathsf{T} \nabla \mathcal{R}[\Phi_\ell] \rangle \\
& = \langle (i^{2\ell - 1} \Phi_\ell), \, - \nabla \cdot \sigma_3 \nabla (i^{2\ell - 1} \Phi_\ell) \rangle_{L^2_{\bm M}} = - \langle \Phi_\ell, \, \nabla \cdot \sigma_3 \nabla \Phi_\ell \rangle = - C_{\ell, \ell}. \nonumber
\end{align}
Therefore $C_{1, 1}$, ${C_{2, 2} = 0}$.
\item Since $\Sigma_1$ is unitary, $\Sigma_1[\Phi_1] = \Phi_2$, $\Sigma_1[\Phi_2] = \Phi_1$, and $\sigma_1 \sigma_3 \sigma_1 = -\sigma_3$
\begin{align}
C_{1, 2}^* & = \langle \Phi_1, \, \nabla \cdot \sigma_3 \nabla \Phi_2 \rangle^* = \langle \Sigma_1[\Phi_1], \, \Sigma_1[\nabla \cdot \sigma_3 \nabla \Phi_2] \rangle^* \\
& = \langle \Sigma_1[\Phi_1], \, \nabla \cdot \sigma_1 \sigma_3 \sigma_1 \nabla \Sigma_1[\Phi_2] \rangle^* = \langle \Phi_2, \, -\nabla \cdot \sigma_3 \nabla \Phi_1 \rangle^* = - C_{2, 1}^* = - C_{1, 2}. \nonumber
\end{align}
Therefore $C_{1, 2} = i \tilde{C}_{1, 2}$ where $\tilde{C}_{1, 2} \smallin \R$.
\item Hence $C = -\tilde{C}_{1, 2} \sigma_2$, where
\begin{align}
i \tilde{C}_{1, 2} & = \langle \Phi_1, \, \nabla \cdot \sigma_3 \nabla \Phi_2 \rangle = \langle \Phi_1, \, ( \partial_{x_1}^2 - \partial_{x_2}^2) \Phi_2 \rangle = \langle \Phi_1, \, \partial_{x_1}^2 \Phi_2 \rangle - \langle \Phi_1, \, \partial_{x_2}^2 \Phi_2 \rangle \\
& = - \langle \partial_{x_1} \Phi_1, \, \partial_{x_1} \Phi_2 \rangle + \langle \partial_{x_2} \Phi_1, \, \partial_{x_2} \Phi_2 \rangle. \nonumber
\end{align}
Since $\mathcal{R}$ is unitary, $\mathcal{R} \circ \partial_{x_2} = \partial_{x_1} \circ \mathcal{R}$, $\mathcal{R}[\Phi_1] = i \Phi_1$, and $\mathcal{R}[\Phi_2] = -i \Phi_2$,
\begin{align}
\langle \partial_{x_2} \Phi_1, \, \partial_{x_2} \Phi_2 \rangle_{L^2_{\bm M}} & = \langle \mathcal{R}[\partial_{x_2} \Phi_1], \, \mathcal{R}[\partial_{x_2} \Phi_2] \rangle = \langle \partial_{x_1} \mathcal{R}[\Phi_1], \, \partial_{x_1} \mathcal{R}[\Phi_2] \rangle \\
& = \langle \partial_{x_1} (i \Phi_1), \, \partial_{x_1} (- i \Phi_2) \rangle = - \langle \partial_{x_1} \Phi_1, \, \partial_{x_1} \Phi_2 \rangle, \nonumber
\end{align}
which implies ${i \tilde{C}_{1, 2} = -2 \langle \partial_{x_1} \Phi_1, \, \partial_{x_1} \Phi_2 \rangle_{L^2_{\bm M}} = -i \beta_2}$.
\end{enumerate}
Therefore $C = \beta_2 \sigma_2$.
\end{enumerate}
This concludes the proof of Proposition \ref{prop:calM-0-ex}.
\end{proof}

\begin{proposition}
\label{prop:calM-0-pr}
The ${2 \times 2}$ self-adjoint matrix $\mathcal{M}^{(0)}(\varepsilon; {\bm \kappa}, \tau_0, {\bm \tau})$ satisfies the following:
\begin{enumerate}
\item The entries of $\mathcal{M}^{(0)}(\varepsilon; {\bm \kappa}, \tau_0, {\bm \tau})$ are analytic (i.e., polynomials) in $(\varepsilon; {\bm \kappa}, \tau_0, {\bm \tau})$.
\item \label{itm:calM-0-ev} ${\mathcal{M}^{(0)}(\varepsilon; {\bm \kappa}, \tau_0, {\bm \tau}) = \mathcal{M}^{(0)}(\varepsilon; -{\bm \kappa}, \tau_0, {\bm \tau})}$.
\item \label{itm:calM-0-PC} ${\mathcal{M}^{(0)}(\varepsilon; {\bm \kappa}, \tau_0, {\bm \tau})_{1, 1} = \mathcal{M}^{(0)}(\varepsilon; {\bm \kappa}, \tau_0, {\bm \tau})_{2, 2}}$.
\end{enumerate}
\end{proposition}

\begin{proof}
These follow immediately from the expression \eqref{eq:calM-0-ex} of Proposition \ref{prop:calM-0-ex}.
\end{proof}

\subsection{Proof of Proposition \ref{prop:calM-pr}}
\label{apx:pf-calM-pr}

We now proceed to prove the conclusions of Proposition \ref{prop:calM-pr}, which extend the properties of $\mathcal{M}^{(0)}(\varepsilon; {\bm \kappa}, \tau_0, {\bm \tau})$, listed in Proposition \ref{prop:calM-0-pr}, to $\mathcal{M}(\varepsilon; {\bm \kappa}, \tau_0, {\bm \tau})$. 

\subsubsection{Proof of Proposition \ref{prop:calM-pr}, part \ref{itm:calM-an}}
\label{apx:calM-an}

{\it Proof of Proposition \ref{prop:calM-pr}, part \ref{itm:calM-an}.}
The entries of $\mathcal{M}^{(0)}(\varepsilon; {\bm \kappa}, \tau_0, {\bm \tau})$ are analytic (i.e., polynomials) in $(\varepsilon; {\bm \kappa}, \tau_0, {\bm \tau})$. Therefore, it suffices to show that the entries of $\mathcal{M}^{(1)}(\varepsilon; {\bm \kappa}, \tau_0, {\bm \tau})$ are analytic in $(\varepsilon; {\bm \kappa}, \tau_0, {\bm \tau})$.

This is due to the Neumann series expansion of inverted operator in \eqref{eq:phi-1_2}.

\subsubsection{Proof of Proposition \ref{prop:calM-pr}, part \ref{itm:calM-ev}}
\label{apx:calM-ev}

{\it Proof of Proposition \ref{prop:calM-pr}, part \ref{itm:calM-ev}.}
By the expansion \eqref{eq:calM-ex_1} and part \ref{itm:calM-0-ev} of Proposition \ref{prop:calM-0-pr}, it suffices to show that
\begin{equation}
\mathcal{M}^{(1)}(\varepsilon, -{\bm \kappa}; \tau_0, {\bm \tau}) = \mathcal{M}^{(1)}(\varepsilon, {\bm \kappa}; \tau_0, {\bm \tau}).
\end{equation}
We use the $\mathcal{P}$ operator. First,
\begin{align}
& \mathcal{P} \circ \mathcal{A}(\varepsilon, {\bm \kappa}; \tau_0, {\bm \tau}) \\
& \qquad = \mathcal{P} \circ \mathscr{R}(E_S) (\varepsilon + 2i{\bm \kappa} \cdot \nabla - |{\bm \kappa}|^2 + \tau_0(\Delta + 2i {\bm \kappa} \cdot \nabla - |{\bm \kappa}|^2) + |{\bm \tau}|(\nabla\cdot\sigma_\varphi\nabla + 2i{\bm \kappa}\cdot\sigma_\varphi\nabla - {\bm \kappa}\cdot\sigma_\varphi{\bm \kappa})) \nonumber \\
& \qquad = \mathscr{R}(E_S) (\varepsilon - 2i{\bm \kappa} \cdot \nabla - |{\bm \kappa}|^2 + \tau_0(\Delta - 2i {\bm \kappa} \cdot \nabla - |{\bm \kappa}|^2) + |{\bm \tau}|(\nabla\cdot\sigma_\varphi\nabla - 2i{\bm \kappa}\cdot\sigma_\varphi\nabla - {\bm \kappa}\cdot\sigma_\varphi{\bm \kappa})) \circ \mathcal{P} \nonumber \\
& \qquad = \mathcal{A}(\varepsilon, -{\bm \kappa}; \tau_0, {\bm \tau}) \circ \mathcal{P}, \nonumber
\end{align}
which implies
\begin{equation}
\mathcal{P} \circ (1 - \mathcal{A}(\varepsilon, {\bm \kappa}; \tau_0, {\bm \tau}))^{-1} = (1 - \mathcal{A}(\varepsilon, -{\bm \kappa}; \tau_0, {\bm \tau}))^{-1} \circ \mathcal{P}.
\end{equation}
Then, since $\mathcal{P}$ is unitary, we have, for $j$, $k \smallin \{ 1, 2 \}$,
\begin{align}
& \mathcal{M}^{(1)}_{j, k}(\varepsilon, {\bm \kappa}; \tau_0, {\bm \tau}) \\
& \qquad = \langle (2i{\bm \kappa}\cdot\nabla + \tau_0 (\Delta + 2i{\bm \kappa} \cdot \nabla) + |{\bm \tau}| ( \nabla\cdot\sigma_\varphi\nabla + 2i{\bm \kappa}\cdot\sigma_\varphi\nabla)) \Phi_j, \nonumber \\ 
& \qquad \qquad \quad (1 - \mathcal{A}(\varepsilon, {\bm \kappa}; \tau_0, {\bm \tau}))^{-1} \mathscr{R}(E_S) (2i{\bm \kappa} \cdot \nabla + \tau_0(\Delta + 2i {\bm \kappa} \cdot \nabla) + |{\bm \tau}|(\nabla\cdot\sigma_\varphi\nabla + 2i{\bm \kappa}\cdot\sigma_\varphi\nabla)) \Phi_k \rangle \nonumber \\
& \qquad = \langle \mathcal{P} [(2i{\bm \kappa}\cdot\nabla + \tau_0 (\Delta + 2i{\bm \kappa} \cdot \nabla) + |{\bm \tau}| ( \nabla\cdot\sigma_\varphi\nabla + 2i{\bm \kappa}\cdot\sigma_\varphi\nabla)) \Phi_j], \nonumber \\
& \qquad \qquad \quad \mathcal{P}[(1 - \mathcal{A}(\varepsilon, {\bm \kappa}; \tau_0, {\bm \tau}))^{-1} \mathscr{R}(E_S) (2i{\bm \kappa} \cdot \nabla + \tau_0(\Delta + 2i {\bm \kappa} \cdot \nabla) + |{\bm \tau}|(\nabla\cdot\sigma_\varphi\nabla + 2i{\bm \kappa}\cdot\sigma_\varphi\nabla)) \Phi_k] \rangle \nonumber \\
& \qquad = \langle (-2i{\bm \kappa}\cdot\nabla + \tau_0 (\Delta - 2i{\bm \kappa} \cdot \nabla) + |{\bm \tau}| ( \nabla\cdot\sigma_\varphi\nabla - 2i{\bm \kappa}\cdot\sigma_\varphi\nabla)) \mathcal{P}[\Phi_j], \nonumber \\
& \qquad \qquad \quad (1 - \mathcal{A}(\varepsilon, -{\bm \kappa}; \tau_0, {\bm \tau}))^{-1} \mathscr{R}(E_S) (-2i{\bm \kappa} \cdot \nabla + \tau_0(\Delta - 2i {\bm \kappa} \cdot \nabla) + |{\bm \tau}|(\nabla\cdot\sigma_\varphi\nabla - 2i{\bm \kappa}\cdot\sigma_\varphi\nabla)) \mathcal{P}[\Phi_k] \rangle \nonumber \\
& \qquad = \langle (-2i{\bm \kappa}\cdot\nabla + \tau_0 (\Delta - 2i{\bm \kappa} \cdot \nabla) + |{\bm \tau}| ( \nabla\cdot\sigma_\varphi\nabla - 2i{\bm \kappa}\cdot\sigma_\varphi\nabla)) (-\Phi_j), \nonumber \\
& \qquad \qquad \quad (1 - \mathcal{A}(\varepsilon, -{\bm \kappa}; \tau_0, {\bm \tau}))^{-1} \mathscr{R}(E_S) (-2i{\bm \kappa} \cdot \nabla + \tau_0(\Delta - 2i {\bm \kappa} \cdot \nabla) + |{\bm \tau}|(\nabla\cdot\sigma_\varphi\nabla - 2i{\bm \kappa}\cdot\sigma_\varphi\nabla)) (-\Phi_k) \rangle \nonumber \\
& \qquad = \mathcal{M}^{(1)}_{j, k}(\varepsilon, -{\bm \kappa}; \tau_0, {\bm \tau}), \nonumber
\end{align}
where we used that ${\mathcal{P}[\Phi_\ell] = - \Phi_\ell}$, ${\ell \smallin \{ 1, 2 \}}$.

An argument using the $\mathcal{C}$ operator is also possible. \qed

\subsubsection{Proof of Proposition \ref{prop:calM-pr}, part \ref{itm:calM-PC}}
\label{apx:calM-PC}

{\it Proof of Proposition \ref{prop:calM-pr}, part \ref{itm:calM-PC}.}
By the expansion \eqref{eq:calM-ex_1} and part \ref{itm:calM-0-PC} of Proposition \ref{prop:calM-0-pr}, it suffices to show that 
\begin{equation}
\mathcal{M}^{(1)}_{1, 1}(\varepsilon, {\bm \kappa}; \tau_0, {\bm \tau}) = \mathcal{M}^{(1)}_{2, 2}(\varepsilon, {\bm \kappa}; \tau_0, {\bm \tau}).
\end{equation}
We use the $\mathcal{PC}$ operator. First,
\begin{align}
& \mathcal{PC} \circ \mathcal{A}(\varepsilon, {\bm \kappa}; \tau_0, {\bm \tau}) \\
& \qquad = \mathcal{PC} \circ \mathscr{R}(E_S) (\varepsilon + 2i{\bm \kappa} \cdot \nabla - |{\bm \kappa}|^2 + \tau_0 (\Delta + 2i {\bm \kappa} \cdot \nabla - |{\bm \kappa}|^2) + |{\bm \tau}|(\nabla \cdot \sigma_\varphi \nabla + 2i{\bm \kappa} \cdot \sigma_\varphi \nabla - {\bm \kappa} \cdot \sigma_\varphi {\bm \kappa})) \nonumber \\
& \qquad = \mathscr{R}(E_S) (\varepsilon + 2i{\bm \kappa} \cdot \nabla - |{\bm \kappa}|^2 + \tau_0 (\Delta + 2i {\bm \kappa} \cdot \nabla - |{\bm \kappa}|^2) + |{\bm \tau}|(\nabla \cdot \sigma_\varphi \nabla + 2i{\bm \kappa} \cdot \sigma_\varphi \nabla - {\bm \kappa} \cdot \sigma_\varphi {\bm \kappa})) \circ \mathcal{PC} \nonumber \\
& \qquad = \mathcal{A}(\varepsilon, {\bm \kappa}; \tau_0, {\bm \tau}) \circ \mathcal{PC}, \nonumber
\end{align}
which implies
\begin{equation}
\label{eq:calA-PC}
\mathcal{PC} \circ (1 - \mathcal{A}(\varepsilon, {\bm \kappa}; \tau_0, {\bm \tau}))^{-1} = (1 - \mathcal{A}(\varepsilon, {\bm \kappa}; \tau_0, {\bm \tau}))^{-1} \circ \mathcal{PC}.
\end{equation} 
Then, since $\mathcal{P}$ is unitary (while $\mathcal{C}$ is not), we have, for $j$, ${k \smallin \{ 1, 2 \}}$,
\begin{align}
& \mathcal{M}^{(1)}_{1, 1}(\varepsilon, {\bm \kappa}; \tau_0, {\bm \tau}) \\
& \qquad = \langle (2i{\bm \kappa}\cdot\nabla + \tau_0 (\Delta + 2i{\bm \kappa} \cdot \nabla) + |{\bm \tau}| ( \nabla\cdot\sigma_\varphi\nabla + 2i{\bm \kappa}\cdot\sigma_\varphi\nabla)) \Phi_1, \nonumber \\
& \qquad \qquad \quad (1 - \mathcal{A}(\varepsilon, {\bm \kappa}; \tau_0, {\bm \tau}))^{-1} \mathscr{R}(E_S) (2i{\bm \kappa} \cdot \nabla + \tau_0(\Delta + 2i {\bm \kappa} \cdot \nabla) + |{\bm \tau}|(\nabla\cdot\sigma_\varphi\nabla + 2i{\bm \kappa}\cdot\sigma_\varphi\nabla)) \Phi_1 \rangle \nonumber \\
& \qquad = \langle (\mathcal{PC}[(2i{\bm \kappa}\cdot\nabla + \tau_0 (\Delta + 2i{\bm \kappa} \cdot \nabla) + |{\bm \tau}| ( \nabla\cdot\sigma_\varphi\nabla + 2i{\bm \kappa}\cdot\sigma_\varphi\nabla)) \Phi_1], \nonumber \\
& \qquad \qquad \quad \mathcal{PC}[(1 - \mathcal{A}(\varepsilon, {\bm \kappa}; \tau_0, {\bm \tau}))^{-1} \mathscr{R}(E_S) (2i{\bm \kappa} \cdot \nabla + \tau_0(\Delta + 2i {\bm \kappa} \cdot \nabla) + |{\bm \tau}|(\nabla\cdot\sigma_\varphi\nabla + 2i{\bm \kappa}\cdot\sigma_\varphi\nabla)) \Phi_1] \rangle^* \nonumber \\
& \qquad = \langle (2i{\bm \kappa}\cdot\nabla + \tau_0 (\Delta + 2i{\bm \kappa} \cdot \nabla) + |{\bm \tau}| ( \nabla\cdot\sigma_\varphi\nabla + 2i{\bm \kappa}\cdot\sigma_\varphi\nabla)) \mathcal{PC}[\Phi_1], \nonumber \\
& \qquad \qquad \quad (1 - \mathcal{A}(\varepsilon, {\bm \kappa}; \tau_0, {\bm \tau}))^{-1} \mathscr{R}(E_S) (2i{\bm \kappa} \cdot \nabla + \tau_0(\Delta + 2i {\bm \kappa} \cdot \nabla) + |{\bm \tau}|(\nabla\cdot\sigma_\varphi\nabla + 2i{\bm \kappa}\cdot\sigma_\varphi\nabla)) \mathcal{PC}[\Phi_1] \rangle^* \nonumber \\
& \qquad = \langle (2i{\bm \kappa}\cdot\nabla + \tau_0 (\Delta + 2i{\bm \kappa} \cdot \nabla) + |{\bm \tau}| ( \nabla\cdot\sigma_\varphi\nabla + 2i{\bm \kappa}\cdot\sigma_\varphi\nabla)) (\Phi_2), \nonumber \\
& \qquad \qquad \quad (1 - \mathcal{A}(\varepsilon, {\bm \kappa}; \tau_0, {\bm \tau}))^{-1} \mathscr{R}(E_S) (2i{\bm \kappa} \cdot \nabla + \tau_0(\Delta + 2i {\bm \kappa} \cdot \nabla) + |{\bm \tau}|(\nabla\cdot\sigma_\varphi\nabla + 2i{\bm \kappa}\cdot\sigma_\varphi\nabla)) (\Phi_2) \rangle^* \nonumber \\
& \qquad = \mathcal{M}^{(1)}_{2, 2}(\varepsilon, {\bm \kappa}; \tau_0, {\bm \tau})^* = \mathcal{M}^{(1)}_{2, 2}(\varepsilon, {\bm \kappa}; \tau_0, {\bm \tau}), \nonumber
\end{align}
where, in the last equality, we used that $\mathcal{M}^{(1)}(\varepsilon, {\bm \kappa}; \tau_0, {\bm \tau})$ is self-adjoint; see Appendix \ref{apx:calM-sa}. \qed

\subsubsection{Proof of Proposition \ref{prop:calM-pr}, part \ref{itm:calM-ex-1}}
\label{apx:calM-ex-1}

{\it Proof of Proposition \ref{prop:calM-pr}, part \ref{itm:calM-ex-1}.} 
From \eqref{eq:def-calM-1}, we have
\begin{align}
\label{eq:calM-1-ex}
& \mathcal{M}^{(1)}_{j, k}(\varepsilon; {\bm \kappa}, \tau_0, {\bm \tau}) \\
& \qquad = \langle (2i{\bm \kappa}\cdot\nabla + \tau_0 (\Delta + 2i{\bm \kappa} \cdot \nabla) + |{\bm \tau}| ( \nabla\cdot\sigma_\varphi\nabla + 2i{\bm \kappa}\cdot\sigma_\varphi\nabla)) \Phi_j, \nonumber \\ 
& \qquad \qquad \quad (1 - \mathcal{A}(\varepsilon, {\bm \kappa}; \tau_0, {\bm \tau}))^{-1} \mathscr{R}(E_S) (2i{\bm \kappa} \cdot \nabla + \tau_0(\Delta + 2i {\bm \kappa} \cdot \nabla) + |{\bm \tau}|(\nabla\cdot\sigma_\varphi\nabla + 2i{\bm \kappa}\cdot\sigma_\varphi\nabla)) \Phi_k \rangle \nonumber \\
& \qquad = \langle 2i {\bm \kappa} \cdot \nabla \Phi_j, \,
\mathscr{R}(E_S) (2i {\bm \kappa} \cdot \nabla) \Phi_k \rangle + \langle (2i{\bm \kappa}\cdot\nabla \Phi_j, \,((1 - \mathcal{A}(\varepsilon, {\bm \kappa}; \tau_0, {\bm \tau}))^{-1} - 1) \mathscr{R}(E_S) (2i{\bm \kappa} \cdot \nabla) \Phi_k \rangle \nonumber \\
& \qquad \qquad \quad + \langle (2i{\bm \kappa}\cdot\nabla \Phi_j, \, (1 - \mathcal{A}(\varepsilon, {\bm \kappa}; \tau_0, {\bm \tau}))^{-1} \mathscr{R}(E_S) (\tau_0(\Delta + 2i {\bm \kappa} \cdot \nabla) + |{\bm \tau}|(\nabla\cdot\sigma_\varphi\nabla + 2i{\bm \kappa}\cdot\sigma_\varphi\nabla)) \Phi_k \rangle \nonumber \\
& \qquad \qquad \quad + \langle (\tau_0 (\Delta + 2i{\bm \kappa} \cdot \nabla) + |{\bm \tau}| ( \nabla\cdot\sigma_\varphi\nabla + 2i{\bm \kappa}\cdot\sigma_\varphi\nabla)) \Phi_j, \nonumber \\ 
& \qquad \qquad \qquad \qquad (1 - \mathcal{A}(\varepsilon, {\bm \kappa}; \tau_0, {\bm \tau}))^{-1} \mathscr{R}(E_S) (2i{\bm \kappa} \cdot \nabla + \tau_0(\Delta + 2i {\bm \kappa} \cdot \nabla) + |{\bm \tau}|(\nabla\cdot\sigma_\varphi\nabla + 2i{\bm \kappa}\cdot\sigma_\varphi\nabla)) \Phi_k \rangle. \nonumber
\end{align}
The results of Section 4.1.3 of \cite{keller2018spectral} imply that the first term of \eqref{eq:calM-1-ex} is
\begin{equation}
\langle 2i {\bm \kappa} \cdot \nabla \Phi_j, \, \mathscr{R}(E_S) (2i {\bm \kappa} \cdot \nabla) \Phi_k \rangle = \alpha_0 |{\bm \kappa}|^2 I_{j, k} + (\alpha_1 {\bm \kappa} \cdot \sigma_1 {\bm \kappa}) (\sigma_1)_{j, k} + (\alpha_2 {\bm \kappa} \cdot \sigma_3 {\bm \kappa}) (\sigma_2)_{j, k},
\end{equation}
where the parameters $\alpha_0$, $\alpha_1$, and $\alpha_2$ are defined in \eqref{eq:def-a-par}. Using the expansion \eqref{eq:calM-ex_1}, then substituting \eqref{eq:calM-0-ex} from Proposition \ref{prop:calM-0-ex} and \eqref{eq:calM-1-ex} above, we obtain
\begin{align}
& \mathcal{M}(\varepsilon; {\bm \kappa}, \tau_0, {\bm \tau}) \\
& \qquad = (\varepsilon - \beta_0 \tau_0 - (1 - \alpha_0) |{\bm \kappa}|^2) I - (\beta_1 \tau_1 - \alpha_1 {\bm \kappa} \cdot \sigma_1 {\bm \kappa}) \sigma_1 - (\beta_2 \tau_3 - \alpha_2 {\bm \kappa} \cdot \sigma_3 {\bm \kappa}) \sigma_2 + \tilde{\mathcal{M}}(\varepsilon; {\bm \kappa}, \tau_0, {\bm \tau}) \nonumber \\
& \qquad = \varepsilon I - H^{\bm M}_{\rm eff}({\bm \kappa}; \tau_0, {\bm \tau}) + \tilde{\mathcal{M}}(\varepsilon; {\bm \kappa}, \tau_0, {\bm \tau}), \nonumber
\end{align}
where $H^{\bm M}_{\rm eff}({\bm \kappa}; \tau_0, {\bm \tau})$ is defined in \eqref{eq:M-eff} of Remark \ref{rmk:M-srf-eff}. The ${2 \times 2}$ matrix $\tilde{\mathcal{M}}(\varepsilon; {\bm \kappa}, \tau_0, {\bm \tau})$ has entries
\begin{align}
\label{eq:def-t-calM}
& \tilde{\mathcal{M}}_{j, k}(\varepsilon; {\bm \kappa}, \tau_0, {\bm \tau}) \\
& \qquad \equiv (-\tau_0 |{\bm \kappa}|^2 - |{\bm \tau}| {\bm \kappa} \cdot \sigma_\varphi {\bm \kappa}) I_{j, k} + \langle (2i{\bm \kappa}\cdot\nabla \Phi_j, \,((1 - \mathcal{A}(\varepsilon, {\bm \kappa}; \tau_0, {\bm \tau}))^{-1} - 1) \mathscr{R}(E_S) (2i{\bm \kappa} \cdot \nabla) \Phi_k \rangle \nonumber \\
& \qquad \qquad \quad + \langle (2i{\bm \kappa}\cdot\nabla \Phi_j, \, (1 - \mathcal{A}(\varepsilon, {\bm \kappa}; \tau_0, {\bm \tau}))^{-1} \mathscr{R}(E_S) (\tau_0(\Delta + 2i {\bm \kappa} \cdot \nabla) + |{\bm \tau}|(\nabla\cdot\sigma_\varphi\nabla + 2i{\bm \kappa}\cdot\sigma_\varphi\nabla)) \Phi_k \rangle \nonumber \\
& \qquad \qquad \quad + \, \langle (\tau_0 (\Delta + 2i{\bm \kappa} \cdot \nabla) + |{\bm \tau}| ( \nabla\cdot\sigma_\varphi\nabla + 2i{\bm \kappa}\cdot\sigma_\varphi\nabla)) \Phi_j, \nonumber \\ 
& \qquad \qquad \qquad \qquad (1 - \mathcal{A}(\varepsilon, {\bm \kappa}; \tau_0, {\bm \tau}))^{-1} \mathscr{R}(E_S) (2i{\bm \kappa} \cdot \nabla + \tau_0(\Delta + 2i {\bm \kappa} \cdot \nabla) + |{\bm \tau}|(\nabla\cdot\sigma_\varphi\nabla + 2i{\bm \kappa}\cdot\sigma_\varphi\nabla)) \Phi_k \rangle. \nonumber
\end{align}
Standard approximations then yield
\begin{equation}
\label{eq:t-calM-bd_3}
\tilde{\mathcal{M}}_{j, k}(\varepsilon; {\bm \kappa}, \tau_0, {\bm \tau}) = O( |\varepsilon| |{\bm \kappa}|^2 + |{\bm \kappa}|^4 + |{\bm \kappa}|^2 |\tau_0| + |{\bm \kappa}|^2 |{\bm \tau}| + |\tau_0|^2 + |{\bm \tau}|^2) \ \ \text{as} \ \ |\varepsilon|, |{\bm \kappa}|, |\tau_0|, |{\bm \tau}| \to 0.
\end{equation}
The proof is now complete. \qed

\bigskip

\section{Proof of Proposition \ref{prop:calD-pr}, part \ref{itm:calD-ex}}
\label{apx:pf-calD-ex}

\setcounter{equation}{0}
\setcounter{figure}{0}

{\it Proof of Proposition \ref{prop:calD-pr}, part \ref{itm:calD-ex}.} First, using ${\nu = \varepsilon - \beta_0 \tau_0 - (1 - \alpha_0)|{\bm \kappa}|^2}$ and the expansion from part \ref{itm:calM-ex-1} of Proposition \ref{prop:calM-pr}, we have
\begin{align}
\mathcal{D}(\nu; {\bm \kappa}, \tau_0, {\bm \tau}) & = (\nu + \tilde{\mathcal{M}}_{1, 1}(\nu + \beta_0 \tau_0 + (1 - \alpha_0) |{\bm \kappa}|^2; {\bm \kappa}, \tau_0, {\bm \tau}))^2 \\
& \qquad - |\omega({\bm \kappa}; \tau_0, {\bm \tau}) + \tilde{\mathcal{M}}_{1, 2}(\nu + \beta_0 \tau_0 + (1 - \alpha_0) |{\bm \kappa}|^2; {\bm \kappa}, \tau_0, {\bm \tau})|^2 \nonumber \\
& = \nu^2 + 2 \nu \tilde{\mathcal{M}}_{1, 1}(\nu + \beta_0 \tau_0 + (1 - \alpha_0) |{\bm \kappa}|^2; {\bm \kappa}, \tau_0, {\bm \tau}) + \tilde{\mathcal{M}}_{1, 1}(\nu + \beta_0 \tau_0 + (1 - \alpha_0) |{\bm \kappa}|^2; {\bm \kappa}, \tau_0, {\bm \tau})^2 \nonumber \\
& \qquad - |\omega({\bm \kappa}; \tau_0, {\bm \tau})|^2 - 2 \, {\rm Re}(\omega({\bm \kappa}; \tau_0, {\bm \tau}) \tilde{\mathcal{M}}_{1, 2}(\nu + \beta_0 \tau_0 + (1 - \alpha_0) |{\bm \kappa}|^2; {\bm \kappa}, \tau_0, {\bm \tau})^*) \nonumber \\
& \qquad - |\tilde{\mathcal{M}}_{1, 2}(\nu + \beta_0 \tau_0 + (1 - \alpha_0) |{\bm \kappa}|^2; {\bm \kappa}, \tau_0, {\bm \tau})|^2 \nonumber \\
& = \mathcal{D}^{(0)}(\nu; {\bm \kappa}, \tau_0, {\bm \tau}) + \mathcal{D}^{(1)}(\nu; {\bm \kappa}, \tau_0, {\bm \tau}),
\end{align}
where $\mathcal{D}^{(0)}(\nu; {\bm \kappa}, \tau_0, {\bm \tau})$ is defined in \eqref{eq:def-calD-0}. The analytic function $\mathcal{D}^{(1)}(\nu; {\bm \kappa}, \tau_0, {\bm \tau})$ is given by
\begin{align}
& \mathcal{D}^{(1)}(\nu; {\bm \kappa}, \tau_0, {\bm \tau}) \\
& \qquad = 2 \nu \tilde{\mathcal{M}}_{1, 1}(\nu + \beta_0 \tau_0 + (1 - \alpha_0) |{\bm \kappa}|^2; {\bm \kappa}, \tau_0, {\bm \tau}) + \tilde{\mathcal{M}}_{1, 1}(\nu + \beta_0 \tau_0 + (1 - \alpha_0) |{\bm \kappa}|^2; {\bm \kappa}, \tau_0, {\bm \tau})^2 \nonumber \\
& \qquad \qquad - 2 \, {\rm Re}(\omega({\bm \kappa}; \tau_0, {\bm \tau}) \tilde{\mathcal{M}}_{1, 2}(\nu + \beta_0 \tau_0 + (1 - \alpha_0) |{\bm \kappa}|^2; {\bm \kappa}, \tau_0, {\bm \tau})^*) \nonumber \\
& \qquad \qquad - |\tilde{\mathcal{M}}_{1, 2}(\nu + \beta_0 \tau_0 + (1 - \alpha_0) |{\bm \kappa}|^2; {\bm \kappa}, \tau_0, {\bm \tau})|^2. \nonumber
\end{align}
The bounds \eqref{eq:til-calM-bd} then yield
\begin{align}
\mathcal{D}^{(1)}(\nu; {\bm \kappa}, \tau_0, {\bm \tau}) & = O(|\nu|^2 |{\bm \kappa}|^2 + |\nu| |{\bm \kappa}|^4 + |\nu| |{\bm \kappa}|^2 |\tau_0| + |\nu| |{\bm \kappa}|^2 |{\bm \tau}| + |\nu| |\tau_0|^2 + |\nu| |{\bm \tau}|^2 \\
& \qquad \quad + |{\bm \kappa}|^6 + |{\bm \kappa}|^4 |\tau_0| + |{\bm \kappa}|^4 |{\bm \tau}| + |{\bm \kappa}|^2 |\tau_0|^2 + |{\bm \kappa}|^2 |{\bm \tau}|^2 + |\tau_0|^3 + |{\bm \tau}|^3) \nonumber
\end{align}
as $|\nu|$, $|{\bm \kappa}|$, $|\tau_0|$, ${|{\bm \tau}| \to 0}$. 

Next, differentiating the power series expansion of $\mathcal{D}^{(1)}(\nu; {\bm \kappa}, \tau_0, {\bm \tau})$ term-by-term with respect to $\nu$ yields
\begin{equation}
\partial_\nu \mathcal{D}^{(1)}(\nu; {\bm \kappa}, \tau_0, {\bm \tau}) = O(|\nu| |{\bm \kappa}|^2 + |{\bm \kappa}|^4 + |{\bm \kappa}|^2 |\tau_0| + |{\bm \kappa}|^2 |{\bm \tau}| + |\tau_0|^2 + |{\bm \tau}|^2)
\end{equation}
as $|\nu|$, $|{\bm \kappa}|$, $|\tau_0|$, ${|{\bm \tau}| \to 0}$. \qed

\bigskip

\section{Proof of Proposition \ref{prop:gs-pr}}
\label{apx:gs-pr}

\setcounter{equation}{0}
\setcounter{figure}{0}

\subsection{Proof of Proposition \ref{prop:gs-pr}, part \ref{itm:g-0}}

{\it Proof of Proposition \ref{prop:gs-pr}, part \ref{itm:g-0}.} 
First, \eqref{eq:g-0-bd} follows immediately from \eqref{eq:t-calM-bd_2}, \eqref{eq:def-t-kap}, and \eqref{eq:def-gs}.

Next, to show \eqref{eq:f-bd}, we recall from \eqref{eq:def-gs} that
\begin{equation}
g_0(\varepsilon, {\bm \kappa}_1; \tau_0, 0) = \tilde{\mathcal{M}}_{1, 1}(\varepsilon; {\bm 0}, \tau_0, {\bm 0}).
\end{equation}
Using \eqref{eq:def-t-calM}, which displays the entries of $\tilde{\mathcal{M}}(\varepsilon; {\bm \kappa}, \tau_0, {\bm \tau})$, we obtain
\begin{equation}
\tilde{\mathcal{M}}_{1, 1}(\varepsilon; {\bm 0}, \tau_0, {\bm 0}) = \tau_0^2 \langle \Delta \Phi_1, \, (1 - \mathcal{A}(\varepsilon, {\bm 0}; \tau_0, {\bm 0}))^{-1} \mathscr{R}(E_S) \Delta \Phi_1 \rangle.
\end{equation}
Note that ${\bm \kappa}_1$ does not enter due to the ansatz \eqref{eq:def-t-kap}. From the definition \eqref{eq:def-calA} of $\mathcal{A}(\varepsilon; {\bm \kappa}, \tau_0, {\bm \tau})$, we have
\begin{equation}
\mathcal{A}(\varepsilon; {\bm 0}, 0, {\bm 0}) = \mathscr{R}(E_S)(\varepsilon + \tau_0 \Delta), \quad \text{where} \quad \lVert \mathscr{R}(E_S)(\varepsilon + \tau_0 \Delta) \rVert = O(|\varepsilon| + |\tau_0|) \ \ \text{as} \ \ |\varepsilon|, \, |\tau_0| \to 0.
\end{equation}
Using the Neumann series representation of $(1 - \mathcal{A}(\varepsilon; {\bm 0}, 0, {\bm 0}))^{-1}$ (and preserving leading order dependence on $\varepsilon$), we have
\begin{equation}
\lVert (1 - \mathcal{A}(\varepsilon; {\bm 0}, 0, {\bm 0}))^{-1} \rVert = O(1 + |\varepsilon| + |\tau_0|).
\end{equation}
It follows that
\begin{equation}
f(\varepsilon; \tau_0) \equiv g_0(\varepsilon, {\bm \kappa}_1; \tau_0, 0) = O(|\tau_0|^2 + |\varepsilon| |\tau_0|^2 + |\tau_0|^3) \ \ \text{as} \ \ |\varepsilon|, \, |\tau_0| \to 0.
\end{equation}
This establishes \eqref{eq:f-bd}. \qed

\subsection{Proof of Proposition \ref{prop:gs-pr}, part \ref{itm:g-1&2}}

{\it Proof of Proposition \ref{prop:gs-pr}, part \ref{itm:g-1&2}.}
We claim, for ${\ell \smallin \{ 1, \, 2 \}}$, that ${g_\ell(\varepsilon, {\bm \kappa}_1; \tau_0, 0) = 0}$. By analyticity, it then follows that ${s \mapsto g_\ell(\varepsilon, {\bm \kappa}_1; \tau_0, s)}$ has a zero of some finite order, greater than or equal to one, at ${s = 0}$. In particular, there exists an analytic mapping ${s \mapsto g_{\ell, 1}(\varepsilon, {\bm \kappa}_1; \tau_0, s)}$ such that ${g_\ell(\varepsilon, {\bm \kappa}_1; \tau_0, s) = s \, g_{\ell, 1}(\varepsilon, {\bm \kappa}_1; \tau_0, s)}$, which is \eqref{eq:g-1&2-ex}. 

That ${g_\ell(\varepsilon, {\bm \kappa}_1; \tau_0, 0) = 0}$ follows from a symmetry argument: From \eqref{eq:def-gs}, observe that
\begin{equation}
g_1(\varepsilon, {\bm \kappa}_1; \tau_0, 0) = {\rm Re}(\tilde{\mathcal{M}}_{1, 2}(\varepsilon, {\bm 0}; \tau_0, {\bm 0})) \quad \text{and} \quad g_2(\varepsilon, {\bm \kappa}_1; \tau_0, 0) = -{\rm Im}(\tilde{\mathcal{M}}_{1, 2}(\varepsilon, {\bm 0}; \tau_0, {\bm 0})).
\end{equation}
It therefore suffices to evaluate $\tilde{\mathcal{M}}_{1, 2}(\varepsilon, {\bm 0}; \tau_0, {\bm 0})$. We use the commutation
\begin{equation}\label{eq:commR}
\mathcal{R} \circ \mathcal{A}(\varepsilon, {\bm 0}; \tau_0, {\bm 0}) = \mathcal{R} \circ \mathscr{R}(E_S) (\varepsilon + \tau_0 \Delta) = \mathscr{R}(E_S) (\varepsilon + \tau_0 \Delta) \circ \mathcal{R} = \mathcal{A}(\varepsilon, {\bm 0}; \tau_0, {\bm 0}) \circ \mathcal{R},
\end{equation}
where $\mathcal{R}$ is introduced in Definition \ref{def:syms} and $\mathcal{A}(\varepsilon; {\bm \kappa}, \tau_0, {\bm \tau})$ is defined in \eqref{eq:def-calA}. Hence, by application of \eqref{eq:commR} to the terms of the Neumann series,
\begin{equation}
\mathcal{R} \circ (1 - \mathcal{A}(\varepsilon, {\bm 0}; \tau_0, {\bm 0}))^{-1} = (1 - \mathcal{A}(\varepsilon, {\bm 0}; \tau_0, {\bm 0}))^{-1} \circ \mathcal{R}.
\end{equation}
We now use \eqref{eq:def-t-calM}, which displays an explicit expression for the entries of $\tilde{\mathcal{M}}(\varepsilon; {\bm \kappa}, \tau_0, {\bm \tau})$. Since $\mathcal{R}$ is unitary and, for ${k \smallin \{ 1, \, 2\}}$, $\mathcal{R}[\Phi_k] = i^{2 k - 1} \Phi_k$, it follows that
\begin{align}
\tilde{\mathcal{M}}_{1, 2}(\varepsilon, {\bm 0}; \tau_0, {\bm 0}) & = \langle \tau_0 \Delta \Phi_1, \, \bigl(1 - \mathcal{A}(\varepsilon, {\bm 0}; \tau_0, {\bm 0})\bigr)^{-1} \mathscr{R}(E_S) (\tau_0 \Delta) \Phi_2 \rangle \\
& = \langle \mathcal{R}[\tau_0 \Delta \Phi_1], \, \mathcal{R}\bigl[\bigl(1 - \mathcal{A}(\varepsilon, {\bm 0}; \tau_0, {\bm 0})\bigr)^{-1} \mathscr{R}(E_S) (\tau_0 \Delta) \Phi_2\bigr] \rangle \nonumber \\
& = \langle \tau_0 \Delta (i \Phi_1), \, \bigl(1 - \mathcal{A}(\varepsilon, {\bm 0}; \tau_0, {\bm 0})\bigr)^{-1} \mathscr{R}(E_S) (\tau_0 \Delta) (-i \Phi_2) \rangle \nonumber \\
& = (-i)^2 \langle \tau_0 \Delta \Phi_1, \, \bigl(1 - \mathcal{A}(\varepsilon_0, {\bm 0}; \tau_0, {\bm 0})\bigr)^{-1} \mathscr{R}(E_S) (\tau_0 \Delta) \Phi_2 \rangle = - \tilde{\mathcal{M}}(\varepsilon, {\bm 0}; \tau_0, {\bm 0})_{1, 2}, \nonumber
\end{align}
and therefore $\tilde{\mathcal{M}}(\varepsilon; {\bm 0}, \tau_0, {\bm 0})_{1, 2} = 0$. This establishes \eqref{eq:g-1&2-ex}.

To prove \eqref{eq:g-1&2-1-ex}, we expand $g_{\ell, 1}(\varepsilon, {\bm \kappa}_1; \tau_0, s)$ in power series about ${(\tau_0, s) = (0, 0)}$ to obtain:
\begin{equation}
g_{\ell, 1}(\varepsilon, {\bm \kappa}_1; \tau_0, s) = g_{\ell, 1}^{(0)}(\varepsilon, {\bm \kappa}_1) + g_{\ell}^{(1)}(\varepsilon, {\bm \kappa}_1; \tau_0, s),
\end{equation}
where
\begin{equation}
g_{\ell, 1}^{(0)}(\varepsilon, {\bm \kappa}_1) \equiv g_{\ell, 1}(\varepsilon, {\bm \kappa}_1; 0, 0) \quad \text{and} \quad g_{\ell, 1}^{(1)}(\varepsilon, {\bm \kappa}_1; \tau_0, s)  = O(|\tau_0| + |s|) \ \ \ \text{as} \ \ |\tau_0|, \, |s| \to 0.
\end{equation}
It remains to bound $\smash{g_{\ell, 1}^{(0)}(\varepsilon, {\bm \kappa}_1)}$. We differentiate \eqref{eq:g-1&2-ex} with respect to $s$ and evaluate at ${s = 0}$ to obtain 
\begin{equation}
\label{eq:g-l-1-0-tay}
g_{\ell, 1}^{(0)}(\varepsilon, {\bm \kappa}_1) = \partial_s g_\ell(\varepsilon, {\bm \kappa}_1; 0, 0),\ \ell=1,2.
\end{equation}
To compute the right-hand side of \eqref{eq:g-l-1-0-tay}, observe from \eqref{eq:def-gs} that 
\begin{equation}
\label{eq:g-l-tau0-0}
g_1(\varepsilon, {\bm \kappa}_1; 0, s) = {\rm Re}(\tilde{\mathcal{M}}(\varepsilon, \sqrt{s} {\bm \kappa}_1; 0, s \hat{\bm \tau}(\varphi))_{1, 2})  \quad \text{and} \quad g_2(\varepsilon, {\bm \kappa}_1; 0, s) = -{\rm Im}(\tilde{\mathcal{M}}(\varepsilon, \sqrt{s} {\bm \kappa}_1; 0, s \hat{\bm \tau}(\varphi))_{1, 2}).
\end{equation}
It therefore suffices to consider $\tilde{\mathcal{M}}(\varepsilon, \sqrt{s} {\bm \kappa}_1; 0, s \hat{\bm \tau}(\varphi))_{1, 2}$ and then take real and imaginary parts. From the expression \eqref{eq:def-t-calM} for the entries of $\tilde{\mathcal{M}}(\varepsilon; {\bm \kappa}, \tau_0, {\bm \tau})$, we expand ${s \mapsto \tilde{\mathcal{M}}(\varepsilon; \sqrt{s} {\bm \kappa}_1, 0, s \hat{\bm \tau}(\varphi))}$ about ${s = 0}$:
\begin{equation}
\label{eq:tM12exp}
\tilde{\mathcal{M}}(\varepsilon; \sqrt{s} {\bm \kappa}_1, 0, s \hat{\bm \tau}(\varphi)) = s\ \langle 2 i {\bm \kappa}_1 \cdot \nabla \Phi_1, \, ((1 - \mathcal{A}(\varepsilon; {\bm 0}, 0, {\bm 0}))^{-1} - 1) \mathscr{R}(E_S)(2i {\bm \kappa}_1 \cdot \nabla) \Phi_2 \rangle + O(|s|^2).
\end{equation}
Therefore,
\begin{equation}
\label{eq:DstM12exp}
\partial_s\tilde{\mathcal{M}}(\varepsilon; \sqrt{s} {\bm \kappa}_1, 0, s \hat{\bm \tau}(\varphi)) =  \langle 2 i {\bm \kappa}_1 \cdot \nabla \Phi_1, \, ((1 - \mathcal{A}(\varepsilon; {\bm 0}, 0, {\bm 0}))^{-1} - 1) \mathscr{R}(E_S)(2i {\bm \kappa}_1 \cdot \nabla) \Phi_2 \rangle + O(|s|).
\end{equation}
From the definition \eqref{eq:def-calA} of $\mathcal{A}(\varepsilon; {\bm \kappa}, \tau_0, {\bm \tau})$, we have
$\mathcal{A}(\varepsilon; {\bm 0}, 0, 0) = \mathscr{R}(E_S) \varepsilon$, where $\lVert \mathscr{R}(E_S) \varepsilon \rVert = O(|\varepsilon|)$  as  $|\varepsilon| \to 0$.
Expanding $(1 - \mathcal{A}(\varepsilon; {\bm 0}, 0, {\bm 0}))^{-1}$ in a Neumann series  we have
$\lVert (1 - \mathcal{A}(\varepsilon; {\bm 0}, 0, {\bm 0}))^{-1} - 1\rVert = O(|\varepsilon|)$.
It follows that
\begin{align}
\label{eq:dstM0}
\partial_s\tilde{\mathcal{M}}(\varepsilon; \sqrt{s} {\bm \kappa}_1, 0, s \hat{\bm \tau}(\varphi))\Big|_{s=0} = \langle 2 i {\bm \kappa}_1 \cdot \nabla \Phi_1, \, ((1 - \mathcal{A}(\varepsilon; {\bm 0}, 0, {\bm 0}))^{-1} - 1) \mathscr{R}(E_S)(2i {\bm \kappa}_1 \cdot \nabla) \Phi_2 \rangle = O(|\varepsilon|) \ \ \text{as} \ \ |\varepsilon| \to 0.
\end{align}
Finally equations \eqref{eq:g-l-1-0-tay}, \eqref{eq:g-l-tau0-0} and \eqref{eq:dstM0} imply
\begin{equation}
g_{\ell, 1}^{(0)}(\varepsilon, {\bm \kappa}_1) = O(|\varepsilon|) \ \ \text{as} \ \ |\varepsilon| \to 0.
\end{equation}
The proof is now complete. \qed

\bigskip

\section{Proofs of Propositions \ref{prop:M-dgn-quad} and \ref{prop:M-dgn-dir}}

\setcounter{equation}{0}
\setcounter{figure}{0}

In this appendix, we prove Propositions \ref{prop:M-dgn-quad} and \ref{prop:M-dgn-dir}, which together classify the band structure degeneracies of the deformed Schr\"{o}dinger operator $H^{\tau_0, {\bm \tau}}$ for ${(\tau_0, {\bm \tau})}$ in a sufficiently small neighborhood of ${(0, {\bm 0})}$. First, in Appendix \ref{apx:M-dgn-quad}, we prove Proposition \ref{prop:M-dgn-quad}, which classifies ${(E_S(\tau_0), {\bm M})}$ (arising in the case ${{\bm \tau} = {\bm 0}}$) as a quadratic band degeneracy point. Then, in Appendix \ref{apx:M-dgn-dir}, we prove Proposition \ref{prop:M-dgn-dir}, which classifies ${(E_D(\tau_0, |{\bm \tau}|; \varphi), {\bm D}^\pm(\tau_0, |{\bm \tau}|; \varphi))}$ (arising in the case ${{\bm \tau} \neq {\bm 0}}$) as Dirac points. \\

\noindent {\bf N.B.} Here, we shall again omit the dependence of certain expressions on the fixed parameter $\varphi$.

\subsection{Proof of Proposition \ref{prop:M-dgn-quad}}
\label{apx:M-dgn-quad}

We here prove Proposition \ref{prop:M-dgn-quad}, which asserts that ${(E_S(\tau_0), {\bm M})}$ is a quadratic band degeneracy point of $H^{\tau_0, {\bm 0}}$. \\

\noindent {\it Proof of Proposition \ref{prop:M-dgn-quad}.}
It suffices to verify that ${(E_S(\tau_0), {\bm M})}$ satisfies properties \ref{itm:quad-dgn-1} -- \ref{itm:quad-dgn-5}:

First, by construction, ${(E_S(\tau_0), {\bm M})}$ is twofold degenerate and therefore satisfies \ref{itm:quad-dgn-1}; see Proposition \ref{prop:hs-0_fin} and the subsequent discussion in Section \ref{sec:F-0}.

Next, recall that the degenerate eigenspace corresponding to ${(E_S(\tau_0), {\bm M})}$ is, from \eqref{eq:evp-asz_1} and \eqref{eq:phi-1_3},
\begin{equation}
a_1 \Phi_1 + a_2 \Phi_2 + (1 - \mathcal{A}(\varepsilon_S(\tau_0); {\bm 0}, \tau_0, {\bm 0}))^{-1} \mathscr{R}(E_S) (- \tau_0 \Delta) (a_1 \Phi_1 + a_2 \Phi_2), \quad a_1, \, a_2 \smallin \C.
\end{equation}
Here, the operator $\mathcal{A}(\varepsilon_S(\tau_0); {\bm 0}, \tau_0, {\bm 0})$, initially defined in \eqref{eq:def-calA}, simplifies to
\begin{equation}
\mathcal{A}(\varepsilon_S(\tau_0); {\bm 0}, \tau_0, {\bm 0}) = \mathscr{R}(E_S) (\varepsilon_S(\tau_0) - \tau_0 \Delta).
\end{equation}
We claim that ${(a_1, a_2) = (1, 0)}$ yields an $L^2_{{\bm M}, +i}$ eigenstate and ${(a_1, a_2) = (0, 1)}$ yields an $L^2_{{\bm M}, -i}$ eigenstate. Let us define
\begin{equation}
\begin{aligned}
\Phi^{\bm M}_1(\tau_0) & \equiv C^{\bm M}_1(\tau_0) \bigl( \Phi_1 + (1 - \mathcal{A}(\varepsilon_S(\tau_0); {\bm 0}, \tau_0, {\bm 0}))^{-1} \mathscr{R}(E_S) (- \tau_0 \Delta) \Phi_1 \bigr), \\
\Phi^{\bm M}_2(\tau_0) & \equiv C^{\bm M}_2(\tau_0) \bigl( \Phi_2 + (1 - \mathcal{A}(\varepsilon_S(\tau_0); {\bm 0}, \tau_0, {\bm 0}))^{-1} \mathscr{R}(E_S) (- \tau_0 \Delta) \Phi_2 \bigr),
\end{aligned}
\end{equation}
where $C^{\bm M}_1(\tau_0)$, ${C^{\bm M}_2(\tau_0) > 0}$ are normalization constants. Note that
\begin{align}
\mathcal{R} \circ \mathcal{A}(\varepsilon_S(\tau_0); {\bm 0}, \tau_0, {\bm 0}) & = \mathcal{R} \circ \mathscr{R}(E_S) (\varepsilon_S(\tau_0) - \tau_0 \Delta) \\
& = \mathscr{R}(E_S) (\varepsilon_S(\tau_0) - \tau_0 \Delta) \circ \mathcal{R} = \mathcal{A}(\varepsilon_S(\tau_0); {\bm 0}, \tau_0, {\bm 0}) \circ \mathcal{R}, \nonumber
\end{align}
which implies
\begin{equation}
\mathcal{R} \circ (1 - \mathcal{A}(\varepsilon_S(\tau_0); {\bm 0}, \tau_0, {\bm 0}))^{-1} = (1 - \mathcal{A}(\varepsilon_S(\tau_0); {\bm 0}, \tau_0, {\bm 0}))^{-1} \circ \mathcal{R}.
\end{equation}
It follows that ${\Phi^{\bm M}_1(\tau_0) \smallin L^2_{{\bm M}, +i}}$ and ${\Phi^{\bm M}_2(\tau_0) \smallin L^2_{{\bm M}, -i}}$, since
\begin{align}
\mathcal{R}[\Phi^{\bm M}_1(\tau_0)] & = \mathcal{R}\bigl[ C^{\bm M}_1(\tau_0) \bigl( \Phi_1 + (1 - \mathcal{A}(\varepsilon_S(\tau_0); {\bm 0}, \tau_0, {\bm 0}))^{-1} \mathscr{R}(E_S) (- \tau_0 \Delta) \Phi_1 \bigr) \bigr] \\
& = C^{\bm M}_1(\tau_0) \bigl( i \Phi_1 + (1 - \mathcal{A}(\varepsilon_S(\tau_0); {\bm 0}, \tau_0, {\bm 0}))^{-1} \mathscr{R}(E_S) (- \tau_0 \Delta) (i \Phi_1) \bigr) = i \Phi^{\bm M}_1(\tau_0), \nonumber
\end{align}
and
\begin{align}
\mathcal{R}[\Phi^{\bm M}_2(\tau_0)] & = \mathcal{R}\bigl[ C^{\bm M}_2(\tau_0) \bigl( \Phi_2 + (1 - \mathcal{A}(\varepsilon_S(\tau_0); {\bm 0}, \tau_0, {\bm 0}))^{-1} \mathscr{R}(E_S) (- \tau_0 \Delta) \Phi_2 \bigr) \bigr] \\
& = C^{\bm M}_2(\tau_0) \bigl( -i \Phi_2 + (1 - \mathcal{A}(\varepsilon_S(\tau_0); {\bm 0}, \tau_0, {\bm 0}))^{-1} \mathscr{R}(E_S) (- \tau_0 \Delta) (-i \Phi_2) \bigr) = -i \Phi^{\bm M}_1(\tau_0). \nonumber
\end{align}
Moreover, ${\Phi^{\bm M}_2(\tau_0) = \mathcal{PC}[\Phi^{\bm M}_1(\tau_0)]}$. With these eigenstates, ${(E_S(\tau_0), {\bm M})}$ satisfies \ref{itm:quad-dgn-2} -- \ref{itm:quad-dgn-4}. 

We note that the normalization constants satisfy
\begin{equation}
C^{\bm M}_\ell(\tau_0) = 1 + O(|\tau_0|^2) \ \ \text{as} \ \ |\tau_0| \to 0, \quad \ell \smallin \{ 1, \, 2 \},
\end{equation}
and we therefore have the expansion
\begin{equation}
\label{eq:phi-quad-ex}
\Phi^{\bm M}_\ell(\tau_0) = \Phi_\ell + O(|\tau_0|) \ \ \text{as} \ \ |\tau_0| \to 0, \quad \ell \smallin \{ 1, \, 2 \}.
\end{equation}

We are now ready to address the nondegeneracy conditon \ref{itm:quad-dgn-5}. From \eqref{eq:def-a-par}, let us accordingly define
\begin{equation}
\label{eq:def-a-quad-par}
\begin{aligned}
\alpha_0(\tau_0) & \equiv \langle \partial_1 \Phi^{\bm M}_1, \, \mathscr{R}^{\tau_0, {\bm 0}}(E_S(\tau_0)) \partial_1 \Phi^{\bm M}_1 \rangle, \\
\alpha_1(\tau_0) & \equiv \langle \partial_1 \Phi^{\bm M}_1(\tau_0), \, \mathscr{R}^{\tau_0, {\bm 0}}(E_S(\tau_0)) \partial_2 \Phi^{\bm M}_2(\tau_0) \rangle, \\
\alpha_2(\tau_0) & \equiv i \langle \partial_1 \Phi^{\bm M}_1(\tau_0), \, \mathscr{R}^{\tau_0, {\bm 0}}(E_S(\tau_0)) \partial_1 \Phi^{\bm M}_2(\tau_0) \rangle,
\end{aligned}
\end{equation}
where $\mathscr{R}^{\tau_0, {\bm 0}}(E_S(\tau_0))$ denotes the resolvent ${(H^{\tau_0, {\bm 0}} - E_S(\tau_0))^{-1}}$ acting in ${{\rm ker}(H^{\tau_0, {\bm 0}} - E_S(\tau_0))^\perp}$. We prove:

\begin{proposition}
\label{prop:a-quad-ex}
The parameters $\alpha_0(\tau_0)$, $\alpha_1(\tau_0)$, and ${\alpha_2(\tau_0) \smallin \R}$, defined in \eqref{eq:def-a-quad-par}, have expansions
\begin{equation}
\label{eq:a-quad-ex}
\begin{aligned}
\alpha_0(\tau_0) & = \alpha_0 + O(|\tau_0|), \\
\alpha_1(\tau_0) & = \alpha_1 + O(|\tau_0|), \\
\alpha_2(\tau_0) & = \alpha_2 + O(|\tau_0|),
\end{aligned}
\end{equation}
as ${|\tau_0| \to 0}$, where the parameters $\alpha_0$, $\alpha_1$, and ${\alpha_2 \smallin \R}$ are defined in \eqref{eq:def-a-par}.
\end{proposition}

\begin{proof}
We first derive the expansion
\begin{equation}
\label{eq:res-quad-exp}
\mathscr{R}^{\tau_0, {\bm 0}}(E_S(\tau_0)) \partial_j \Phi^{\bm M}_k = \mathscr{R}(E_S) \partial_j \Phi_k + O(|\tau_0|) \ \ \text{as} \ \ |\tau_0| \to 0, \quad j, \, k \smallin \{ 1, \, 2 \}.
\end{equation}
Fix $j$, ${k \smallin \{ 1, \, 2 \}}$ and define
\begin{align}
f & \equiv \partial_j \Phi^{\bm M}_k, \\
u & \equiv \mathscr{R}^{\tau_0, {\bm 0}}(E_S(\tau_0)) f.
\end{align}
Hence, $u$ is the unique solution to
\begin{equation}
(H^{\tau_0, {\bm 0}} - E_S(\tau_0)) \, u = f.
\end{equation}
or, equivalently,
\begin{equation}
(H - E_S) u - (\varepsilon_S(\tau_0) - \tau_0 \Delta) u = f.
\end{equation}
Note that $f$ is orthogonal to both $L^2_{{\bm M}, +i}$ and $L^2_{{\bm M}, -i}$, which implies ${f \smallin L^2_{{\bm M}, +1} \oplus L^2_{{\bm M}, -1}}$. Further, note that $H^{\tau_0, {\bm 0}}$, and thus $\mathscr{R}^{\tau_0, {\bm 0}}(E_S(\tau_0))$, maps ${L^2_{{\bm M}, +1} \oplus L^2_{{\bm M}, -1}}$ to itself. Thus, ${u \smallin L^2_{{\bm M}, +1} \oplus L^2_{{\bm M}, -1}}$ and
\begin{equation}
u = (1 - \mathscr{R}(E_S)(\varepsilon_S(\tau_0) - \tau_0 \Delta))^{-1} \mathscr{R}(E_S) \, f .
\end{equation}
Using the Neumann series representation of the inverted operator above and applying the expansion \eqref{eq:phi-quad-ex} yields
\begin{equation}
u = \mathscr{R}(E_S) \partial_j \Phi_k + O(|\tau_0|) \ \ \text{as} \ \ |\tau_0| \to 0,
\end{equation}
as claimed by \eqref{eq:res-quad-exp}.

Finally, applying the expansions \eqref{eq:phi-quad-ex} and \eqref{eq:res-quad-exp}, we obtain
\begin{align}
\alpha_0(\tau_0) & = \langle \partial_1 \Phi^{\bm M}_1(\tau_0), \, \mathscr{R}^{\tau_0, {\bm 0}}(E_S(\tau_0)) \partial_1 \Phi^{\bm M}_1(\tau_0) \rangle \\
& = \langle \partial_1 \Phi_1, \, \mathscr{R}(E_S) \partial_1 \Phi_1 \rangle + O(|\tau_0|) = \alpha_0 + O(|\tau_0|), \nonumber \\
\alpha_1(\tau_0) & = \langle \partial_1 \Phi^{\bm M}_1(\tau_0), \, \mathscr{R}^{\tau_0, {\bm 0}}(E_S(\tau_0)) \partial_2 \Phi^{\bm M}_2(\tau_0) \rangle \\
& = \langle \partial_1 \Phi_1, \, \mathscr{R}(E_S) \partial_2 \Phi_2 \rangle + O(|\tau_0|) = \alpha_1 + O(|\tau_0|), \quad \text{and} \nonumber \\ 
-i \alpha_2(\tau_0) & = \langle \partial_1 \Phi^{\bm M}_1(\tau_0), \, \mathscr{R}^{\tau_0, {\bm 0}}(E_S(\tau_0)) \partial_1 \Phi^{\bm M}_2(\tau_0) \rangle \\
& = \langle \partial_1 \Phi_1, \, \mathscr{R}(E_S) \partial_1 \Phi_2 \rangle + O(|\tau_0|) = -i \alpha_2 + O(|\tau_0|) \nonumber
\end{align}
as ${|\tau_0| \to 0}$.
\end{proof}

\noindent In particular, since ${\alpha_1 \neq 0}$ and ${\alpha_2 \neq 0}$ by the assumption \ref{itm:quad-dgn-5}, it follows from Proposition \ref{prop:a-quad-ex} that $\alpha_1(\tau_0)$, ${\alpha_2(\tau_0) \neq 0}$. Therefore ${(E_S(\tau_0), {\bm M})}$ satisfies \ref{itm:quad-dgn-5}. 

The proof of Proposition \ref{prop:M-dgn-quad} is now complete. \qed

\subsection{Proof of Proposition \ref{prop:M-dgn-dir}}
\label{apx:M-dgn-dir}

We now prove Proposition \ref{prop:M-dgn-dir}, which asserts that ${(E_D(\tau_0, |{\bm \tau}|), {\bm D}^\pm(\tau_0, |{\bm \tau}|))}$ are Dirac points of $H^{\tau_0, {\bm \tau}}$. For brevity, we only prove that ${(E_D(\tau_0, |{\bm \tau}|), {\bm D}^+(\tau_0, |{\bm \tau}|))}$ is a Dirac point; the proof for ${(E_D(\tau_0, |{\bm \tau}|), {\bm D}^-(\tau_0, |{\bm \tau}|))}$ follows analogously, as discussed in Remark \ref{rmk:proof-dir-m} ahead. \\

\noindent {\it Proof of Proposition \ref{prop:M-dgn-dir}.} 
It suffices to verify that ${(E_D(\tau_0, |{\bm \tau}|), {\bm D}^+(\tau_0, |{\bm \tau}|))}$ satisfies properties \ref{itm:dir-pt-1} -- \ref{itm:dir-pt-4}:

First, by construction, ${(E_D(\tau_0, |{\bm \tau}|), {\bm D}(\tau_0, |{\bm \tau}|))}$ is twofold degenerate and therefore satisfies \ref{itm:dir-pt-1}; see Proposition \ref{prop:hs-0_fin} and the subsequent discussion in Section \ref{sec:G-0}.

Next, the degenerate eigenspace corresponding to ${(E_D(\tau_0, |{\bm \tau}|), {\bm D}^+(\tau_0, |{\bm \tau}|))}$ is, from \eqref{eq:evp-asz_1} and \eqref{eq:phi-1_3},
\begin{align}
& a_1 \Phi_1 + a_2 \Phi_2 + (1 - \mathcal{A}(\varepsilon_D(\tau_0, |{\bm \tau}|); {\bm \kappa}_D(\tau_0, |{\bm \tau}|), \tau_0, {\bm \tau}))^{-1} \mathscr{R}(E_S) (2i {\bm \kappa}_D(\tau_0, |{\bm \tau}|) \cdot \nabla \\
& \qquad - \tau_0 (\Delta + 2 i {\bm \kappa}_D(\tau_0, |{\bm \tau}|) \cdot \nabla) - |{\bm \tau}| (\nabla \cdot \sigma_\varphi \nabla + 2 i {\bm \kappa}_D(\tau_0, |{\bm \tau}|) \cdot \sigma_\varphi \nabla)) (a_1 \Phi_1 + a_2 \Phi_2), \quad a_1, a_2 \smallin \C. \nonumber
\end{align}
We claim that ${(\alpha_1, \alpha_2) = (1, 1)}$ yields a $Y_{{\bm D}(\tau_0, |{\bm \tau}|), +1}$ eigenstate, and ${(\alpha_1, \alpha_2) = (1, -1)}$ yields a $Y_{{\bm D}(\tau_0, |{\bm \tau}|), -1}$ eigenstate. Let us define
\begin{equation}
\begin{aligned}
\tilde{\Phi}^{\bm D}_1(\tau_0, |{\bm \tau}|) & \equiv \tilde{C}^{\bm D}_1(\tau_0, |{\bm \tau}|) \bigl( \Phi_1 + \Phi_2 + (1 - \mathcal{A}(\varepsilon_D(\tau_0, |{\bm \tau}|); {\bm \kappa}_D(\tau_0, |{\bm \tau}|), \tau_0, {\bm \tau}))^{-1} \mathscr{R}(E_S) (2i {\bm \kappa}_D(\tau_0, |{\bm \tau}|) \cdot \nabla \\
& \qquad - \tau_0 (\Delta + 2 i {\bm \kappa}_D(\tau_0, |{\bm \tau}|) \cdot \nabla) - |{\bm \tau}| (\nabla \cdot \sigma_\varphi \nabla + 2 i {\bm \kappa}_D(\tau_0, |{\bm \tau}|) \cdot \sigma_\varphi \nabla)) (\Phi_1 + \Phi_2) \bigr), \\
\tilde{\Phi}^{\bm D}_2(\tau_0, |{\bm \tau}|) & \equiv \tilde{C}^{\bm D}_2(\tau_0, |{\bm \tau}|) \bigl( \Phi_1 - \Phi_2 + (1 - \mathcal{A}(\varepsilon_D(\tau_0, |{\bm \tau}|); {\bm \kappa}_D(\tau_0, |{\bm \tau}|), \tau_0, {\bm \tau}))^{-1} \mathscr{R}(E_S) (2i {\bm \kappa}_D(\tau_0, |{\bm \tau}|) \cdot \nabla \\
& \qquad - \tau_0 (\Delta + 2 i {\bm \kappa}_D(\tau_0, |{\bm \tau}|) \cdot \nabla) - |{\bm \tau}| (\nabla \cdot \sigma_\varphi \nabla + 2 i {\bm \kappa}_D(\tau_0, |{\bm \tau}|) \cdot \sigma_\varphi \nabla)) (\Phi_1 - \Phi_2) \bigr),
\end{aligned}
\end{equation}
where $\tilde{C}^{\bm D}_1(\tau_0, |{\bm \tau}|)$, ${\tilde{C}^{\bm D}_2(\tau_0, |{\bm \tau}|) > 0}$ are normalization constants. Note that \eqref{eq:calA-PC} implies
\begin{equation}
\mathcal{PC} \circ (1 - \mathcal{A}(\varepsilon_D(\tau_0, |{\bm \tau}|); {\bm \kappa}_D(\tau_0, |{\bm \tau}|), \tau_0, {\bm \tau}))^{-1} = (1 - \mathcal{A}(\varepsilon_D(\tau_0, |{\bm \tau}|); {\bm \kappa}_D(\tau_0, |{\bm \tau}|), \tau_0, {\bm \tau}))^{-1} \circ \mathcal{PC}.
\end{equation}
It follows that ${\tilde{\Phi}^{\bm D}_1(\tau_0, |{\bm \tau}|) \smallin Y_{{\bm D}(\tau_0, |{\bm \tau}|), +1}}$ and ${\tilde{\Phi}^{\bm D}_2(\tau_0, |{\bm \tau}|) \smallin Y_{{\bm D}(\tau_0, |{\bm \tau}|), -1}}$, since
\begin{align}
& \mathcal{PC}[\tilde{\Phi}^{\bm D}_1(\tau_0, |{\bm \tau}|)] \\
& \qquad = \mathcal{PC}[\tilde{C}^{\bm D}_1 \bigl( \Phi_1 + \Phi_2 + (1 - \mathcal{A}(\varepsilon_D(\tau_0, |{\bm \tau}|); {\bm \kappa}_D(\tau_0, |{\bm \tau}|), \tau_0, {\bm \tau}))^{-1} \mathscr{R}(E_S) \nonumber \\
& \qquad \qquad \cdot (2i {\bm \kappa}_D(\tau_0, |{\bm \tau}|) \cdot \nabla - \tau_0 (\Delta + 2 i {\bm \kappa}_D(\tau_0, |{\bm \tau}|) \cdot \nabla) - |{\bm \tau}| (\nabla \cdot \sigma_\varphi \nabla + 2 i {\bm \kappa}_D(\tau_0, |{\bm \tau}|) \cdot \sigma_\varphi \nabla)) (\Phi_1 + \Phi_2) \bigr)] \nonumber \\
& \qquad = \tilde{C}^{\bm D}_1 \bigl( \Phi_2 + \Phi_1 + (1 - \mathcal{A}(\varepsilon_D(\tau_0, |{\bm \tau}|); {\bm \kappa}_D(\tau_0, |{\bm \tau}|), \tau_0, {\bm \tau}))^{-1} \mathscr{R}(E_S) \nonumber \\
& \qquad \qquad \cdot (2i {\bm \kappa}_D(\tau_0, |{\bm \tau}|) \cdot \nabla - \tau_0 (\Delta + 2 i {\bm \kappa}_D(\tau_0, |{\bm \tau}|) \cdot \nabla) - |{\bm \tau}| (\nabla \cdot \sigma_\varphi \nabla + 2 i {\bm \kappa}_D(\tau_0, |{\bm \tau}|) \cdot \sigma_\varphi \nabla)) (\Phi_2 + \Phi_1) \bigr) \nonumber \\
& \qquad = \tilde{\Phi}^{\bm D}_1(\tau_0, |{\bm \tau}|), \nonumber
\end{align}
and
\begin{align}
& \mathcal{PC}[\tilde{\Phi}^{\bm D}_2(\tau_0, |{\bm \tau}|)] \\
& \qquad = \mathcal{PC}[\tilde{C}^{\bm D}_2 \bigl( \Phi_1 + \Phi_2 + (1 - \mathcal{A}(\varepsilon_D(\tau_0, |{\bm \tau}|); {\bm \kappa}_D(\tau_0, |{\bm \tau}|), \tau_0, {\bm \tau}))^{-1} \mathscr{R}(E_S) \nonumber \\
& \qquad \qquad \cdot (2i {\bm \kappa}_D(\tau_0, |{\bm \tau}|) \cdot \nabla - \tau_0 (\Delta + 2 i {\bm \kappa}_D(\tau_0, |{\bm \tau}|) \cdot \nabla) - |{\bm \tau}| (\nabla \cdot \sigma_\varphi \nabla + 2 i {\bm \kappa}_D(\tau_0, |{\bm \tau}|) \cdot \sigma_\varphi \nabla)) (\Phi_1 - \Phi_2) \bigr)] \nonumber \\
& \qquad = \tilde{C}^{\bm D}_2 \bigl( \Phi_2 + \Phi_1 + (1 - \mathcal{A}(\varepsilon_D(\tau_0, |{\bm \tau}|); {\bm \kappa}_D(\tau_0, |{\bm \tau}|), \tau_0, {\bm \tau}))^{-1} \mathscr{R}(E_S) \nonumber \\
& \qquad \qquad \cdot (2i {\bm \kappa}_D(\tau_0, |{\bm \tau}|) \cdot \nabla - \tau_0 (\Delta + 2 i {\bm \kappa}_D(\tau_0, |{\bm \tau}|) \cdot \nabla) - |{\bm \tau}| (\nabla \cdot \sigma_\varphi \nabla + 2 i {\bm \kappa}_D(\tau_0, |{\bm \tau}|) \cdot \sigma_\varphi \nabla)) (\Phi_2 - \Phi_1) \bigr) \nonumber \\
& \qquad = - \tilde{\Phi}^{\bm D}_2(\tau_0, |{\bm \tau}|). \nonumber
\end{align}
With these eigenstates, ${(E_D(\tau_0, |{\bm \tau}|), {\bm D}(\tau_0, |{\bm \tau}|))}$ satisfies \ref{itm:dir-pt-2} and \ref{itm:dir-pt-3}. 

Rotating, we define
\begin{equation}
\begin{aligned}
\Phi^{\bm D}_1(\tau_0, |{\bm \tau}|) & \equiv C^{\bm D}_1(\tau_0, |{\bm \tau}|) \bigl( \Phi_1 + (1 - \mathcal{A}(\varepsilon_D(\tau_0, |{\bm \tau}|); {\bm \kappa}_D(\tau_0, |{\bm \tau}|), \tau_0, {\bm \tau}))^{-1} \mathscr{R}(E_S) (2i {\bm \kappa}_D(\tau_0, |{\bm \tau}|) \cdot \nabla \\
& \qquad - \tau_0 (\Delta + 2 i {\bm \kappa}_D(\tau_0, |{\bm \tau}|) \cdot \nabla) - |{\bm \tau}| (\nabla \cdot \sigma_\varphi \nabla + 2 i {\bm \kappa}_D(\tau_0, |{\bm \tau}|) \cdot \sigma_\varphi \nabla)) \Phi_1 \bigr), \\
\Phi^{\bm D}_2(\tau_0, |{\bm \tau}|) & \equiv C^{\bm D}_2(\tau_0, |{\bm \tau}|) \bigl( \Phi_2 + (1 - \mathcal{A}(\varepsilon_D(\tau_0, |{\bm \tau}|); {\bm \kappa}_D(\tau_0, |{\bm \tau}|), \tau_0, {\bm \tau}))^{-1} \mathscr{R}(E_S) (2i {\bm \kappa}_D(\tau_0, |{\bm \tau}|) \cdot \nabla \\
& \qquad - \tau_0 (\Delta + 2 i {\bm \kappa}_D(\tau_0, |{\bm \tau}|) \cdot \nabla) - |{\bm \tau}| (\nabla \cdot \sigma_\varphi \nabla + 2 i {\bm \kappa}_D(\tau_0, |{\bm \tau}|) \cdot \sigma_\varphi \nabla)) \Phi_2 \bigr),
\end{aligned}
\end{equation}
such that $\Phi^{\bm D}_2(\tau_0, |{\bm \tau}|) = \mathcal{PC}[\Phi^{\bm D}_1(\tau_0, |{\bm \tau}|)]$. Here, $C^{\bm D}_1(\tau_0, |{\bm \tau}|)$, ${C^{\bm D}_2(\tau_0, |{\bm \tau}|) > 0}$ are new normalization constants. We note that these satisfy
\begin{equation}
C^{\bm D}_\ell = 1 + O(|\tau_0| + |{\bm \tau}|) \ \ \text{as} \ \ |\tau_0|, |{\bm \tau}| \to 0, \quad \ell \smallin \{ 1, \, 2 \},
\end{equation}
and we therefore have the expansion
\begin{equation}
\label{eq:phi-dir-p-ex}
\Phi^{{\bm D}^+}_\ell(\tau_0, |{\bm \tau}|) = \Phi_\ell + \sqrt{|{\bm \tau}|} \, \mathscr{R}(E_S) (2 i {\bm \kappa}_D^{(0)}(\varphi) \cdot \nabla) \Phi_\ell + O(|\tau_0| + |{\bm \tau}|) \ \ \text{as} \ \ |\tau_0|, \, |{\bm \tau}| \to 0, \quad \ell \smallin \{ 1, \, 2 \}.
\end{equation}

We are now ready to address the nondegeneracy condition \ref{itm:dir-pt-4}. From \eqref{eq:def-gam-par}, let us accordingly define
\begin{equation}
\label{eq:def-gam-dir-p}
\begin{aligned}
{\bm \gamma}^{{\bm D}^+}_0(\tau_0, |{\bm \tau}|) & \equiv \langle \Phi^{\bm D}_1, \, 2i ((1 - \tau_0) I - |{\bm \tau}| \sigma_\varphi) \nabla \Phi^{\bm D}_1 \rangle, \\
{\bm \gamma}^{{\bm D}^+}_1(\tau_0, |{\bm \tau}|) & \equiv {\rm Re} \, \langle \Phi^{\bm D}_1, \, 2i ((1 - \tau_0) I - |{\bm \tau}| \sigma_\varphi) \nabla \Phi^{\bm D}_2 \rangle, \\
{\bm \gamma}^{{\bm D}^+}_2(\tau_0, |{\bm \tau}|) & \equiv - {\rm Im} \, \langle \Phi^{\bm D}_1, \, 2i ((1 - \tau_0) I - |{\bm \tau}| \sigma_\varphi) \nabla \Phi^{\bm D}_2 \rangle.
\end{aligned}
\end{equation}
We prove:

\begin{proposition}
\label{prop:gam-dir-p-ex}
The parameters ${\bm \gamma}^{{\bm D}^+}_0(\tau_0, |{\bm \tau}|)$, ${\bm \gamma}^{{\bm D}^+}_1(\tau_0, |{\bm \tau}|)$, and ${{\bm \gamma}^{{\bm D}^+}_2(\tau_0, |{\bm \tau}|) \smallin \R^2}$, defined in \eqref{eq:def-gam-dir-p}, have expansions
\begin{equation}
\label{eq:gam-dir-p-ex}
\begin{aligned}
{\bm \gamma}^{{\bm D}^+}_0(\tau_0, |{\bm \tau}|) & = 2 \sqrt{|{\bm \tau}|} (\alpha_0 I + 4 a^{1, 1}_{1, 2} R) {\bm \kappa}_D^{(0)}(\varphi) + O(|\tau_0| + |{\bm \tau}|), \\
{\bm \gamma}^{{\bm D}^+}_1(\tau_0, |{\bm \tau}|) & = 2 \sqrt{|{\bm \tau}|} \, \alpha_1 \sigma_1 {\bm \kappa}^{(0)}_D(\varphi) + O(|\tau_0| + |{\bm \tau}|), \\
{\bm \gamma}^{{\bm D}^+}_2(\tau_0, |{\bm \tau}|) & = 2 \sqrt{|{\bm \tau}|} \, \alpha_2 \sigma_3 {\bm \kappa}^{(0)}_D(\varphi) + O(|\tau_0| + |{\bm \tau}|)
\end{aligned}
\end{equation}
as $|\tau_0|$, ${|{\bm \tau}| \to 0}$, where $\alpha_0$, $\alpha_1$, and ${\alpha_2 \smallin \R}$ are defined in \eqref{eq:def-a-par} and ${{\bm \kappa}^{(0)}_D(\varphi) \smallin \R^2}$ is defined in \eqref{eq:def-kap-D-0-1}. The parameter ${a^{1, 1}_{1, 2} \smallin \R}$ is defined in \cite[Section 4.1.3]{keller2018spectral}, but does not enter in final expressions.
\end{proposition}

\begin{proof}
Applying the expansion \eqref{eq:phi-dir-p-ex}, we obtain
\begin{align}
& {\bm \gamma}_0^{{\bm D}^+}(\tau_0, |{\bm \tau}|) \\
& \qquad = \langle \Phi^{{\bm D}^+}_1, \, 2i ((1 - \tau_0) I - |{\bm \tau}| \sigma_\varphi) \nabla \Phi^{{\bm D}^+}_1 \rangle \nonumber \\
& \qquad = \langle \Phi^{{\bm D}^+}_1, \, 2i \nabla \Phi^{{\bm D}^+}_1 \rangle + O(|\tau_0| + |{\bm \tau}|) \nonumber \\
& \qquad = \langle \Phi_1 + \sqrt{|{\bm \tau}|} \, \mathscr{R}(E_S) (2 i {\bm \kappa}_D^{(0)}(\varphi) \cdot \nabla) \Phi_1, \, 2i \nabla (\Phi_1 + \sqrt{|{\bm \tau}|} \, \mathscr{R}(E_S) (2 i {\bm \kappa}_D^{(0)}(\varphi) \cdot \nabla) \Phi_1) \rangle + O(|\tau_0| + |{\bm \tau}|) \nonumber \\
& \qquad = \langle \Phi_1, \, 2i \nabla \Phi_1 \rangle + \sqrt{|{\bm \tau}|} \bigl( \langle \Phi_1, \, (2i \nabla) \mathscr{R}(E_S) (2 i {\bm \kappa}_D^{(0)}(\varphi) \cdot \nabla) \Phi_1 \rangle + \langle \mathscr{R}(E_S) (2 i {\bm \kappa}_D^{(0)}(\varphi) \cdot \nabla) \Phi_1, \, 2i \nabla \Phi_1 \rangle \bigr) \nonumber \\
& \qquad \qquad \quad + O(|\tau_0| + |{\bm \tau}|) \nonumber \\
& \qquad = \sqrt{|{\bm \tau}|} \bigl( \langle 2i \nabla \Phi_1, \, \mathscr{R}(E_S) (2 i {\bm \kappa}_D^{(0)}(\varphi) \cdot \nabla) \Phi_1 \rangle + \langle 2i \nabla \Phi_1, \, \mathscr{R}(E_S) (2 i {\bm \kappa}_D^{(0)}(\varphi) \cdot \nabla) \Phi_1 \rangle^* \bigr) + O(|\tau_0| + |{\bm \tau}|) \nonumber \\
& \qquad = 2 \sqrt{|{\bm \tau}|} \langle 2i \nabla \Phi_1, \, \mathscr{R}(E_S) (2 i {\bm \kappa}_D^{(0)}(\varphi) \cdot \nabla) \Phi_1 \rangle + O(|\tau_0| + |{\bm \tau}|) \nonumber \\
& \qquad = 2 \sqrt{|{\bm \tau}|} (\alpha_0 I + 4 a^{1, 1}_{1, 2} R) {\bm \kappa}_D^{(0)}(\varphi) + O(|\tau_0| + |{\bm \tau}|) \nonumber
\end{align}
as $|\tau_0|$, ${|{\bm \tau}| \to 0}$. Next, we similarly obtain
\begin{align}
& \langle \Phi^{\bm D}_1, \, 2i ((1 - \tau_0) I - |{\bm \tau}| \sigma_\varphi) \nabla \Phi^{\bm D}_2 \rangle \\ 
& \qquad = \langle \Phi^{{\bm D}^+}_1, \, 2i \nabla \Phi^{{\bm D}^+}_2 \rangle + O(|\tau_0| + |{\bm \tau}|) \nonumber \\
& \qquad = \langle \Phi_1 + \sqrt{|{\bm \tau}|} \, \mathscr{R}(E_S) (2 i {\bm \kappa}_D^{(0)}(\varphi) \cdot \nabla) \Phi_1, \, 2 i \nabla (\Phi_2 + \sqrt{|{\bm \tau}|} \, \mathscr{R}(E_S) (2 i {\bm \kappa}_D^{(0)}(\varphi) \cdot \nabla) \Phi_2) \rangle + O(|\tau_0| + |{\bm \tau}|) \nonumber \\
& \qquad = \langle \Phi_1, \, 2i \nabla \Phi_2 \rangle + \sqrt{|{\bm \tau}|} \bigl( \langle \Phi_1, \, (2i \nabla) \mathscr{R}(E_S) (2 i {\bm \kappa}_D^{(0)}(\varphi) \cdot \nabla) \Phi_2 \rangle + \langle \mathscr{R}(E_S) (2i {\bm \kappa}^{(0)}_D(\varphi) \cdot \nabla) \Phi_1, \, 2i \nabla \Phi_2 \rangle \bigr) \nonumber \\
& \qquad \qquad + O(|\tau_0| + |{\bm \tau}|) \nonumber \\
& \qquad = \sqrt{|{\bm \tau}|} \bigl( \langle 2i \nabla \Phi_1, \, \mathscr{R}(E_S) (2 i {\bm \kappa}_D^{(0)}(\varphi) \cdot \nabla) \Phi_2 \rangle + \langle 2i \nabla \Phi_2, \, \mathscr{R}(E_S) (2i {\bm \kappa}^{(0)}_D(\varphi) \cdot \nabla) \Phi_ 1 \rangle^* \bigr) + O(|\tau_0| + |{\bm \tau}|) \nonumber \\
& \qquad = 2 \sqrt{|{\bm \tau}|} \langle 2i \nabla \Phi_1, \, \mathscr{R}(E_S) (2 i {\bm \kappa}_D^{(0)}(\varphi) \cdot \nabla) \Phi_2 \rangle + O(|\tau_0| + |{\bm \tau}|) \nonumber \\
& \qquad = 2 \sqrt{|{\bm \tau}|} (\alpha_1 \sigma_1 - i\alpha_2 \sigma_3) {\bm \kappa}_D^{(0)}(\varphi) + O(|\tau_0| + |{\bm \tau}|) \nonumber
\end{align}
as $|\tau_0|$, ${|{\bm \tau}| \to 0}$. Taking real and imaginary parts yields
\begin{align}
{\bm \gamma}^{\bm D}_1(\tau_0, |{\bm \tau}|) & = 2 \sqrt{|{\bm \tau}|} \, \alpha_1 \sigma_1 {\bm \kappa}^{(0)}_D(\varphi) + O(|\tau_0| + |{\bm \tau}|), \\
{\bm \gamma}^{\bm D}_2(\tau_0, |{\bm \tau}|) & = 2 \sqrt{|{\bm \tau}|} \, \alpha_2 \sigma_3 {\bm \kappa}^{(0)}_D(\varphi) + O(|\tau_0| + |{\bm \tau}|)
\end{align}
as $|\tau_0|$, ${|{\bm \tau}| \to 0}$.
\end{proof}

\noindent In particular, since ${\alpha_1 \neq 0}$ and ${\alpha_2 \neq 0}$ by \ref{itm:quad-dgn-5}, and ${{\bm \kappa}_D^{(0)}(\varphi) \neq {\bm 0}}$ by \ref{itm:quad-dgn-5} and \ref{itm:quad-dgn-6}, we have that ${{\bm \gamma}^{{\bm D}^+}_1(\tau_0, |{\bm \tau}|) \neq {\bm 0}}$ and ${{\bm \gamma}^{{\bm D}^+}_2(\tau_0, |{\bm \tau}|) \neq {\bm 0}}$. Therefore ${(E_D(\tau_0, |{\bm \tau}|), {\bm D}^+(\tau_0, |{\bm \tau}|))}$ satisfies \ref{itm:dir-pt-4}.

The proof of Proposition \ref{prop:M-dgn-dir} is now complete. \qed

\begin{remark}
\label{rmk:proof-dir-m}
The proof for ${(E_D(\tau_0, |{\bm \tau}|), {\bm D}^-(\tau_0, |{\bm \tau}|))}$ proceeds analogously: We take as eigenstates
\begin{equation}
\Phi^{{\bm D}^-}_1 = \mathcal{P}[\Phi^{{\bm D}^+}_1] \quad \text{and} \quad \Phi^{{\bm D}^-}_2 = \mathcal{P}[\Phi^{{\bm D}^+}_2],
\end{equation}
and note that ${\Phi^{{\bm D}^-}_2 = \mathcal{PC}[\Phi^{{\bm D}^-}_1]}$. Since $\mathcal{P}$ is unitary,
\begin{align}
\langle \Phi^{{\bm D}^-}_j, \, 2i ((1 - \tau_0) I - |{\bm \tau}| \sigma_\varphi) \nabla \Phi^{{\bm D}^-}_k \rangle & = \langle \mathcal{P}[\Phi^{{\bm D}^+}_1], \, 2i ((1 - \tau_0) I - |{\bm \tau}| \sigma_\varphi) \nabla \mathcal{P}[\Phi^{{\bm D}^+}_2] \rangle \\
& = \langle \mathcal{P}[\Phi^{{\bm D}^+}_j], \, \mathcal{P}[2i ((1 - \tau_0) I - |{\bm \tau}| \sigma_\varphi) (-\nabla \Phi^{{\bm D}^+}_k)] \rangle \nonumber \\
& = - \langle \Phi^{{\bm D}^+}_j, \, 2i ((1 - \tau_0) I - |{\bm \tau}| \sigma_\varphi) \nabla \Phi^{{\bm D}^+}_k \rangle. \nonumber
\end{align}
for $j$, ${k \smallin \{ 1, \, 2 \}}$. Thus
\begin{equation}
\label{eq:gam-dir-m-ex}
\begin{aligned}
{\bm \gamma}^{{\bm D}^-}_0(\tau_0, |{\bm \tau}|) & = - {\bm \gamma}^{{\bm D}^+}_0(\tau_0, |{\bm \tau}|) = -2 \sqrt{|{\bm \tau}|} (\alpha_0 I + 4 a^{1, 1}_{1, 2} R) {\bm \kappa}_D^{(0)}(\varphi) + O(|\tau_0| + |{\bm \tau}|), \\
{\bm \gamma}^{{\bm D}^-}_1(\tau_0, |{\bm \tau}|) & = - {\bm \gamma}^{{\bm D}^+}_1(\tau_0, |{\bm \tau}|) = -2 \sqrt{|{\bm \tau}|} \, \alpha_1 \sigma_1 {\bm \kappa}^{(0)}_D(\varphi) + O(|\tau_0| + |{\bm \tau}|), \\
{\bm \gamma}^{{\bm D}^-}_2(\tau_0, |{\bm \tau}|) & = - {\bm \gamma}^{{\bm D}^+}_2(\tau_0, |{\bm \tau}|) = -2 \sqrt{|{\bm \tau}|} \, \alpha_2 \sigma_3 {\bm \kappa}^{(0)}_D(\varphi) + O(|\tau_0| + |{\bm \tau}|),
\end{aligned}
\end{equation}
as $|\tau_0|$, ${|{\bm \tau}| \to 0}$, and so ${(E_D(\tau_0, |{\bm \tau}|), {\bm D}^-(\tau_0, |{\bm \tau}|))}$ also satisfies \ref{itm:dir-pt-1} -- \ref{itm:dir-pt-4}.
\end{remark}

\begin{remark}
The expressions \eqref{eq:gam-dir-p-ex} and \eqref{eq:gam-dir-m-ex} can also be obtained by expanding the effective Hamiltonian \eqref{eq:M-eff} about its degeneracies and extracting the leading order terms.
\end{remark}

\bigskip

\section{Proof of Proposition \ref{prop:M-dgn-ref}}
\label{apx:M-dgn-ref}

\setcounter{equation}{0}
\setcounter{figure}{0}

In this appendix, we prove Proposition \ref{prop:M-dgn-ref}. For brevity, we prove only part \ref{itm:M-dgn-ref-1} here; the proof of part \ref{itm:M-dgn-ref-2} follows analogously (as noted in Remark \ref{rmk:M-dgn-ref-2} ahead). \\

\noindent {\it Proof of Proposition \ref{prop:M-dgn-ref}, part \ref{itm:M-dgn-ref-1}.} Assume ${{\bm \tau} = (\tau_1, 0)}$ with ${\tau_1 \neq 0}$. First, suppose that $\sigma_1 {\bm \kappa} = {\bm \kappa}$. Using the definition \eqref{eq:def-hs} and the expansion \eqref{eq:calM-ex_1}, we have
\begin{equation}
h_2(\varepsilon, {\bm \kappa}; \tau_0, {\bm \tau}) = -{\rm Im}(\mathcal{M}_{1, 2}(\varepsilon; {\bm \kappa}, \tau_0, {\bm \tau}))
= -{\rm Im}(\mathcal{M}^{(0)}_{1, 2}(\varepsilon; {\bm \kappa}, \tau_0, {\bm \tau})) - {\rm Im}(\mathcal{M}^{(1)}_{1, 2}(\varepsilon; {\bm \kappa}, \tau_0, {\bm \tau})).
\end{equation}
Using the expression \eqref{eq:calM-0-ex} and substituting ${\tau_3 = 0}$, the first term on the right-hand side simplifies to
\begin{equation}
\mathcal{M}^{(0)}_{1, 2}(\varepsilon; {\bm \kappa}, \tau_0, {\bm \tau}) = \beta_1 \tau_1.
\end{equation}
In particular, ${{\rm Im}(\mathcal{M}^{(0)}_{1, 2}(\varepsilon; {\bm \kappa}, \tau_0, {\bm \tau})) = 0}$ for all ${\bm \kappa}$. To prove ${{\rm Im}(\mathcal{M}^{(1)}_{1, 2}(\varepsilon; {\bm \kappa}, \tau_0, {\bm \tau})) = 0}$, it suffices to show that
\begin{equation}
\mathcal{M}^{(1)}_{1, 2}(\varepsilon; {\bm \kappa}, \tau_0, {\bm \tau}) = \mathcal{M}^{(1)}_{1, 2}(\varepsilon; {\bm \kappa}, \tau_0, {\bm \tau})^*.
\end{equation} 
We use the $\Sigma_1$ operator. First, note that the operator $\mathcal{A}(\varepsilon; {\bm \kappa}, \tau_0, {\bm \tau})$, defined in \eqref{eq:def-calA}, simplifies to
\begin{equation}
\mathcal{A}(\varepsilon; {\bm \kappa}, \tau_0, {\bm \tau}) = \mathscr{R}(E_S) ( \varepsilon + 2i {\bm \kappa} \cdot \nabla - |{\bm \kappa}|^2 - \tau_0 (\nabla + i{\bm \kappa})^2 - \tau_1 (\nabla + i{\bm \kappa}) \cdot \sigma_1 (\nabla + i{\bm \kappa}) ).
\end{equation}
Thus,
\begin{align}
& \Sigma_1 \circ \mathcal{A}(\varepsilon; {\bm \kappa}, \tau_0, {\bm \tau}) \\
& \qquad = \Sigma_1 \circ \mathscr{R}(E_S) ( \varepsilon + 2i {\bm \kappa} \cdot \nabla - |{\bm \kappa}|^2 - \tau_0 (\nabla + i{\bm \kappa})^2 - \tau_1 (\nabla + i{\bm \kappa}) \cdot \sigma_1 (\nabla + i{\bm \kappa}) ) \nonumber \\
& \qquad = \mathscr{R}(E_S) ( \varepsilon + 2i \sigma_1 {\bm \kappa} \cdot \nabla - |\sigma_1 {\bm \kappa}|^2 - \tau_0 (\nabla + i \sigma_1 {\bm \kappa})^2 - \tau_1 (\nabla + i \sigma_1 {\bm \kappa}) \cdot \sigma_1 (\nabla + i \sigma_1 {\bm \kappa}) ) \circ \Sigma_1 \nonumber \\
& \qquad = \mathcal{A}(\varepsilon; \sigma_1 {\bm \kappa}, \tau_0, {\bm \tau}) \circ \Sigma_1, \nonumber
\end{align}
which implies
\begin{equation}
\Sigma_1 \circ (1 - \mathcal{A}(\varepsilon; {\bm \kappa}, \tau_0, {\bm \tau}))^{-1} = (1 - \mathcal{A}(\varepsilon; \sigma_1 {\bm \kappa}, \tau_0, {\bm \tau}))^{-1} \circ \Sigma_1.
\end{equation}
Then, since $\Sigma_1$ is unitary, we have
\begin{align}
\label{eq:calM-1-ref-1}
& \mathcal{M}^{(1)}_{1, 2}(\varepsilon; {\bm \kappa}, \tau_0, {\bm \tau}) \\
& \qquad = \langle (2i{\bm \kappa}\cdot\nabla - \tau_0 (\Delta + 2i{\bm \kappa} \cdot \nabla) - \tau_1 (\nabla \cdot \sigma_1 \nabla + 2i{\bm \kappa} \cdot \sigma_1 \nabla)) \Phi_1, \nonumber \\
& \qquad \qquad \quad (1 - \mathcal{A}(\varepsilon; {\bm \kappa}, \tau_0, {\bm \tau}))^{-1} \mathscr{R}(E_S) (2i{\bm \kappa} \cdot \nabla - \tau_0(\Delta + 2i {\bm \kappa} \cdot \nabla) - \tau_1 (\nabla \cdot \sigma_1 \nabla + 2i{\bm \kappa} \cdot \sigma_1 \nabla)) \Phi_2 \rangle \nonumber \\
& \qquad = \langle \Sigma_1[ (2i{\bm \kappa}\cdot\nabla - \tau_0 (\Delta + 2i{\bm \kappa} \cdot \nabla) - \tau_1 ( \nabla \cdot \sigma_1 \nabla + 2i{\bm \kappa} \cdot \sigma_1 \nabla)) \Phi_1 ], \nonumber \\
& \qquad \qquad \quad \Sigma_1[ (1 - \mathcal{A}(\varepsilon; {\bm \kappa}, \tau_0, {\bm \tau}))^{-1} \mathscr{R}(E_S) (2i{\bm \kappa} \cdot \nabla - \tau_0(\Delta + 2i {\bm \kappa} \cdot \nabla) - \tau_1 (\nabla \cdot \sigma_1 \nabla + 2i{\bm \kappa} \cdot \sigma_1 \nabla)) \Phi_2 ] \rangle \nonumber \\
& \qquad = \langle (2i \sigma_1 {\bm \kappa} \cdot \nabla - \tau_0 (\Delta + 2i \sigma_1 {\bm \kappa} \cdot \nabla) - \tau_1 (\nabla \cdot \sigma_1 \nabla + 2i \sigma_1 {\bm \kappa} \cdot \nabla)) (\Phi_2), \nonumber \\
& \qquad \qquad \quad (1 - \mathcal{A}(\varepsilon; \sigma_1 {\bm \kappa}, \tau_0, {\bm \tau}))^{-1} \mathscr{R}(E_S) (2i \sigma_1 {\bm \kappa} \cdot \nabla - \tau_0(\Delta + 2i \sigma_1 {\bm \kappa} \cdot \nabla) - \tau_1 (\nabla \cdot \sigma_1 \nabla + 2i \sigma_1 {\bm \kappa} \cdot \sigma_1 \nabla)) (\Phi_1) \rangle \nonumber \\
& \qquad = \mathcal{M}^{(1)}_{2, 1}(\varepsilon; \sigma_1 {\bm \kappa}, \tau_0, {\bm \tau}) = \mathcal{M}^{(1)}_{1, 2}(\varepsilon; \sigma_1 {\bm \kappa}, \tau_0, {\bm \tau})^*. \nonumber
\end{align}
Thus, if ${\sigma_1 {\bm \kappa} = {\bm \kappa}}$, we have ${{\rm Im}(\mathcal{M}^{(1)}_{1, 2}(\varepsilon; {\bm \kappa}, \tau_0, {\bm \tau})) = 0}$.

Next, we consider the case ${\sigma_1 {\bm \kappa} = - {\bm \kappa}}$. Using \eqref{eq:calM-1-ref-1} and applying part \ref{itm:calM-ev} of Proposition \ref{prop:calM-pr}, it follows that
\begin{equation}
\mathcal{M}^{(1)}_{1, 2}(\varepsilon; \sigma_1 {\bm \kappa}, \tau_0, {\bm \tau}) = \mathcal{M}^{(1)}_{1, 2}(\varepsilon; {\bm \kappa}, \tau_0, {\bm \tau})^* = \mathcal{M}^{(1)}_{1, 2}(\varepsilon; -{\bm \kappa}, \tau_0, {\bm \tau})^*.
\end{equation}
Thus, if ${\sigma_1 {\bm \kappa} = -{\bm \kappa}}$, we have ${{\rm Im}(\mathcal{M}^{(1)}_{1, 2}(\varepsilon; {\bm \kappa}, \tau_0, {\bm \tau})) = 0}$. \qed

\begin{remark}
\label{rmk:M-dgn-ref-2}
The proof of part \ref{itm:M-dgn-ref-2} follows analogously, instead using the $\Sigma_3$ operator when ${\sigma_3 {\bm \kappa} = {\bm \kappa}}$.
\end{remark}

\bigskip

\nocite{*}
\bibliographystyle{plain}
\bibliography{bibliography.bib}

\end{document}